\documentclass[a4paper,11pt]{article}
\input{Preamble.tex}

\begin{document}

\pagenumbering{Roman}

\hypersetup{pageanchor=false}
\title{Linear-Time Algorithms for $k$-Edge-Connected Components, $k$-Lean Tree Decompositions, and More
\thanks{The research leading to these results was supported by Grant Number 16582, Basic Algorithms Research Copenhagen (BARC), from the VILLUM Foundation.}
}

\author{Tuukka Korhonen\thanks{Department of Computer Science, University of Copenhagen, Denmark. \texttt{tuko@di.ku.dk}}}


\maketitle

\thispagestyle{empty}

\begin{abstract}
We present $k^{\OO(k^2)} m$ time algorithms for various problems about decomposing a given undirected graph by edge cuts or vertex separators of size $<k$ into parts that are ``well-connected'' with respect to cuts or separators of size $<k$; here, $m$ is the total number of vertices and edges of the graph.
As an application of our results, we obtain for every fixed $k$ a linear-time algorithm for computing the $k$-edge-connected components of a given graph, solving a long-standing open problem.
More generally, we obtain a $k^{\OO(k^2)} m$ time algorithm for computing a $k$-Gomory-Hu tree of a given graph, which is a structure representing pairwise minimum cuts of size~$<k$.

Our main technical result, from which the other results follow, is a $k^{\OO(k^2)} m$ time algorithm for computing a $k$-lean tree decomposition of a given graph.
This is a tree decomposition with adhesion size $<k$ that captures the existence of separators of size $<k$ between subsets of its bags.
A $k$-lean tree decomposition is also an unbreakable tree decomposition with optimal unbreakability parameters for the adhesion size bound $k$.

As further applications, we obtain $k^{\OO(k^2)} m$ time algorithms for $k$-vertex connectivity and for element connectivity $k$-Gomory-Hu tree.
All of our algorithms are deterministic.

Our techniques are inspired by the tenth paper of the Graph Minors series of Robertson and Seymour and by Bodlaender's parameterized linear-time algorithm for treewidth.
\end{abstract}

\thispagestyle{empty}

\newpage

\pagenumbering{roman}

\setcounter{page}{1}
\setcounter{tocdepth}{2}

\tableofcontents

\newpage

\pagenumbering{arabic}

\hypersetup{pageanchor=true}

\clearpage
\setcounter{page}{1}

\section{Introduction}
Fifty years ago, Aho, Hopcroft, and Ullman~\cite[Problem~5.30]{aho1974design} asked the following question:
\begin{displayquote}
``Hopcroft and Tarjan~\cite{HopcroftT73} have given a linear time algorithm to find $3$-connected components.
It is natural to conjecture that there exists linear (in numbers of vertices and edges) time algorithm to find $k$-connected components for each $k$. Can you find one?''
\end{displayquote}
It is not clear whether they meant that the algorithm is supposed to run in linear-time for every fixed constant $k$, or for unbounded $k$.
It is also unclear what they meant by finding $k$-connected components, since no agreed-upon definition of them in the context of vertex connectivity exists even today.
Nevertheless, in this paper we answer various forms of this question affirmatively in the setting when $k$ is an arbitrary fixed constant.

In particular, we give $k^{\OO(k^2)} n + \OO(m)$ time algorithms for various problems about decomposing an undirected graph along edge cuts or vertex separators of size less than $k$.
Our main theorem is technical to state, so let us start with a simple-to-state application of it, a linear-time algorithm for $k$-edge-connected components for every fixed $k$.

\paragraph{Computing the $k$-edge-connected components.}
Two vertices of a graph are in the same $k$-edge-connected component if they cannot be separated from each other by cutting less than $k$ edges, or equivalently, if there is a set of $k$ edge-disjoint paths between them.
It can be observed that this yields an equivalence relation among the vertices, so the $k$-edge-connected components of a graph $G$ form a uniquely defined partition of the vertex set $V(G)$.

Computing the $k$-edge-connected components of a graph is a fundamental and well-studied problem.
For $k=1$, this corresponds to computing the connected components, for which linear-time algorithms are considered folklore.
For $k=2$, Tarjan~\cite{DBLP:journals/siamcomp/Tarjan72} (see also~\cite{DBLP:journals/cacm/Paton71,DBLP:journals/cacm/HopcroftT73}) gave a linear-time algorithm for computing the biconnected components, from which $2$-edge-connected components can be easily retrieved.
For $k=3$, Hopcroft and Tarjan~\cite{HopcroftT73} gave a linear-time algorithm for computing the triconnected components, which Galil and Italiano~\cite{DBLP:journals/sigact/GalilI91} observed to yield also a linear-time algorithm for $3$-edge-connected components via a simple reduction.
Simplified linear-time algorithms for $3$-edge-connected components were later given by multiple authors~\cite{nagamochi1992linear,DBLP:journals/mst/Tsin07,DBLP:books/cu/NI2008,DBLP:journals/jda/Tsin09,DBLP:journals/ipl/NorouziT14}.
For $4$-edge-connected components, the first linear-time algorithms were given by Nadara, Radecki, Smulewicz, and Soko{\l}owski~\cite{DBLP:conf/esa/NadaraRSS21}, and concurrently by Georgiadis, Italiano, and Kosinas~\cite{DBLP:conf/esa/GeorgiadisIK21}.
Very recently, Kosinas~\cite{DBLP:conf/soda/Kosinas24} obtained a linear-time algorithm for $5$-edge-connected components.

For $k > 5$, the best known algorithm when $k$ is a constant is an algorithm of Hariharan, Kavitha, and Panigrahi~\cite{DBLP:conf/soda/HariharanKP07}, which runs in time $\OO(k^3 n \log^2 n + m)$.
When $k$ is unbounded, the best running time is is obtained by the almost-linear $m^{1+o(1)}$ time algorithm of Abboud, Li, Panigrahi, and Saranurak~\cite{DBLP:conf/focs/Abboud0PS23}.
In fact, the latter algorithm computes a Gomory-Hu tree of the input graph, which is a structure that represents pairwise minimum cuts~\cite{gomory1961multi}, and the former computes a $k$-Gomory-Hu tree\footnote{In~\cite{DBLP:conf/soda/HariharanKP07} this structure is called a partial Gomory-Hu tree, but we use $k$-Gomory-Hu tree, following~\cite{DBLP:conf/stoc/PettieSY22}.}, which represents pairwise minimum cuts of size less than~$k$.
Our algorithm also computes a $k$-Gomory-Hu tree, so let us define it formally.

A $k$-Gomory-Hu tree of a graph $G$ is a triple $(T, \gamma, \alpha)$, where $T$ is a tree, $\gamma \colon V(G) \to V(T)$ maps vertices of $G$ to nodes of $T$, $\alpha \colon E(T) \to 2^{E(G)}$ maps edges of $T$ to sets of edges of $G$, and
\begin{enumerate}
\item for all $u,v \in V(G)$ and $e \in E(T)$ so that $e$ is on the $(\gamma(u), \gamma(v))$-path in $T$, the set $\alpha(e)$ is an $(u,v)$-cutset of size $|\alpha(e)| < k$, and
\item for all $u,v \in V(G)$, if there exists an $(u,v)$-cutset of size $<k$, then there exists $e \in E(T)$ on the $(\gamma(u),\gamma(v))$-path in $T$ so that $\alpha(e)$ is a minimum-size $(u,v)$-cutset.
\end{enumerate}

It can be observed that the $k$-edge-connected components of $G$ correspond to the non-empty sets $\{v \in V(G) : \gamma(v) = t\}$ over all $t \in V(T)$, and therefore any algorithm for computing a $k$-Gomory-Hu gives also an algorithm for $k$-edge-connected components.
Let us then state our first theorem.

\begin{restatable}{theorem}{thekgomoryhu}
\label{the:mainkgomoryhu}
There is an algorithm that, given an $n$-vertex $m$-edge graph $G$ and an integer $k \ge 1$, in time $k^{\OO(k^2)} n + \OO(m)$ returns a $k$-Gomory-Hu-tree of $G$.
\end{restatable}

The algorithm of \Cref{the:mainkgomoryhu}, like all of our algorithms, is deterministic.

For every constant $k \ge 6$, the algorithm of \Cref{the:mainkgomoryhu} is the first linear-time algorithm for computing the $k$-edge-connected components.
In fact, no linear-time algorithms were known for $k \ge 6$ even for the simpler problem of finding a cut of size less than $k$ if one exists, i.e., the edge connectivity problem.
For edge connectivity, an $\OO(k^2 n \log n + m)$ time algorithm was given by Gabow~\cite{DBLP:journals/jcss/Gabow95}, and an $\OO(m \log^3 n)$ time algorithm by Karger~\cite{DBLP:journals/jacm/Karger00} (see also~\cite{DBLP:journals/jacm/KawarabayashiT19,DBLP:journals/siamcomp/HenzingerRW20,DBLP:conf/soda/Ghaffari0T20}).

As further related work, let us mention that computing the $k$-edge-connected components is also a well-studied problem in the context of directed graphs~\cite{DBLP:conf/icalp/Georgiadis10,DBLP:journals/siamcomp/GeorgiadisIP20,DBLP:conf/focs/GIK24}, and maintaining the edge connectivity and the pairwise $k$-edge-connectivity relation has received a lot of attention in the dynamic setting, see for example~\cite{DBLP:conf/focs/KanevskyTBC91,DBLP:journals/algorithmica/WestbrookT92,DBLP:journals/siamcomp/GalilI93,DBLP:journals/dm/PoutreLO93,DBLP:journals/jacm/EppsteinGIN97,DBLP:journals/algorithmica/DinitzW98,DBLP:journals/siamcomp/PoutreW98,DBLP:journals/siamcomp/Poutre00,DBLP:conf/swat/DinitzN00,DBLP:journals/jacm/HolmLT01,DBLP:conf/soda/HolmRT18,DBLP:conf/wads/PengSS19,DBLP:conf/focs/0001S21,DBLP:conf/soda/GoranciHNSTW23}.

\paragraph{Lean and unbreakable tree decompositions.}
\Cref{the:mainkgomoryhu} is a corollary of our main theorem, which is about computing a certain type of a \emph{tree decomposition} of a graph~\cite{DBLP:journals/jct/RobertsonS84}.
A tree decomposition of a graph $G$ is a pair $(T,\bag)$, where $T$ is a tree and $\bag \colon V(T) \to 2^{V(G)}$ is a mapping that assigns a ``bag'' of vertices for each node of $T$, and must satisfy that
\begin{enumerate}
\item for every edge $uv \in E(G)$, there is $t \in V(T)$ with $\{u,v\} \subseteq \bag(t)$, and
\item for every vertex $v \in V(G)$, the set $\{t \in V(T) : v \in \bag(t)\}$ induces a non-empty connected subtree of $T$.
\end{enumerate}

Conventionally, a tree decomposition is associated with its \emph{width}, i.e., the maximum size of a bag (minus one)~\cite{DBLP:journals/jct/RobertsonS84}.
Tree decompositions of small width exist only for graphs of small treewidth, but in this paper we are interested in tree decompositions that exist for all graphs.
Instead of computing a tree decomposition of small width, we will compute a tree decomposition that has small \emph{adhesions} and whose bags satisfy certain connectivity conditions.
The adhesion at an edge $st \in E(T)$ of a tree decomposition $(T,\bag)$ is the intersection $\adh(st) = \bag(s) \cap \bag(t)$.
The definition of tree decompositions implies that $\adh(st)$ is a separator of $G$ that separates $G$ in the same manner as $st$ separates $T$.

Our main theorem is about computing a \emph{$k$-lean tree decomposition}~\cite{DBLP:journals/siamdm/CarmesinDHH14}, which is a tree decomposition $(T,\bag)$ that satisfies
\begin{enumerate}
\item every adhesion of $(T,\bag)$ has size $<k$, and
\item for all $t_1, t_2 \in V(T)$, $X_1 \subseteq \bag(t_1)$, and $X_2 \subseteq \bag(t_2)$, if there exists an $(X_1,X_2)$-separator\footnote{An $(X_1,X_2)$-separator is a vertex set that intersects all paths starting in $X_1$ and ending in $X_2$, including the one-vertex paths in $X_1 \cap X_2$.} of size $<\min(k, |X_1|, |X_2|)$, then there exists $e \in E(T)$ on the $(t_1,t_2)$-path in $T$ so that $\adh(e)$ is a minimum-size $(X_1,X_2)$-separator.\label{enum:klendefintro:leannes}
\end{enumerate}

We remark that in \Cref{enum:klendefintro:leannes}, it can be that $t_1 = t_2$, and this is indeed an important case.

Let us note the analogy to the $k$-Gomory-Hu tree: The edges of the tree $T$ correspond to cuts/separators of $G$ of size $<k$, and there is a condition stating that these cuts/separators should capture minimum cuts/separators of size $<k$ between certain vertices/sets of vertices.

Our main theorem is the following algorithm for computing a $k$-lean tree decomposition.

\begin{restatable}{theorem}{themainalgklean}
\label{the:mainalgklean}
There is an algorithm that, given an $n$-vertex $m$-edge graph $G$ and an integer $k \ge 1$, in time $k^{\OO(k^2)} n + \OO(m)$ returns a $k$-lean tree decomposition of $G$.
\end{restatable}

\Cref{the:mainalgklean} implies \Cref{the:mainkgomoryhu} via a simple reduction that first applies the Nagamochi-Ibaraki sparsifier~\cite{DBLP:journals/algorithmica/NagamochiI92} to reduce the number of edges to $\le kn$, then subdivides each edge, and then replaces the original vertices by cliques of size $k$.

Even the fact that a $k$-lean tree decompositions exist for every graph is a non-trivial theorem.
Lean tree decompositions, which correspond to $k$-lean with $k=\infty$, were introduced and proved to exist by Thomas~\cite{DBLP:journals/jct/Thomas90}.
Bellenbaum and Diestel~\cite{bellenbaum2002two} simplified the proof of Thomas and introduced the term ``lean''.
The definition of $k$-lean tree decompositions was introduced by Carmesin, Diestel, Hamann, and Hundertmark~\cite{DBLP:journals/siamdm/CarmesinDHH14}, who observed that the Bellenbaum-Diestel proof can be directly adapted to show the existence of $k$-lean tree decompositions.

To the author's knowledge, there was no prior algorithmic work about computing $k$-lean tree decompositions, but instead the weaker notion of \emph{unbreakable tree decompositions}~\cite{DBLP:journals/siamcomp/CyganLPPS19} has received a lot of attention in the literature of parameterized graph algorithms.
For integers $s \ge k \ge 1$, a tree decomposition $(T,\bag)$ is \emph{$(s,k)$-unbreakable} if for every node $t \in V(T)$, there is no vertex cut $(A,B)$\footnote{A pair $(A,B)$ with $A \cup B = V(G)$ is a \emph{vertex cut} if there are no edges between $A \setminus B$ and $B \setminus A$.} of $G$ so that $|A \cap B| < k$, $|A \cap \bag(t)| \ge s$, and $|B \cap \bag(t)| \ge s$.

Unbreakable tree decompositions were introduced by Cygan, Lokshtanov, Pilipczuk, Pilipczuk, and Saurabh~\cite{DBLP:journals/siamcomp/CyganLPPS19}, who gave an $2^{\OO(k^2)} n^2 m$ time algorithm for computing a $(2^{\OO(k)}, k)$-unbreakable tree decomposition with adhesion size $2^{\OO(k)}$, and used it to show that minimum bisection is fixed-parameter tractable.
Their work was motivated by an earlier technique called ``recursive understanding''~\cite{DBLP:conf/stoc/GroheKMW11,DBLP:conf/focs/KawarabayashiT11,DBLP:journals/siamcomp/ChitnisCHPP16}, which can be seen as a predecessor of unbreakable tree decompositions in this context.
After~\cite{DBLP:journals/siamcomp/CyganLPPS19}, unbreakable tree decompositions became a standard tool in the area and were used in several works, for example~\cite{DBLP:conf/icalp/LokshtanovR0Z18,DBLP:conf/focs/Lokshtanov0S20,DBLP:journals/siamdm/AgrawalKPRS22,DBLP:conf/icalp/PilipczukSSTV22,DBLP:conf/esa/0002L0S23,DBLP:conf/lics/SchirrmacherSST24}.

As for computing unbreakable tree decompositions, Cygan, Komosa, Lokshtanov, Pilipczuk, Pilipczuk, Saurabh, and Wahlström~\cite{DBLP:journals/talg/CyganKLPPSW21} used the Bellenbaum-Diestel technique~\cite{bellenbaum2002two} to give a $k^{\OO(k)} n^{\OO(1)}$ time algorithm for computing a tree decomposition with adhesion size $<k$ that is $(i,i)$-unbreakable for all $i \le k$.
However, the $n^{\OO(1)}$ factor is far from linear, and the authors asked as an open problem whether a parameterized ``near-linear'' time algorithm for unbreakable tree decompositions exists.
Very recently, Anand, Lee, Li, Long, and Saranurak~\cite{DBLP:journals/corr/abs-2408-09368} gave a parameterized ``close-to-linear'' time algorithm for unbreakable tree decompositions, computing an $(\OO(k/\varepsilon), k)$-unbreakable tree decomposition with adhesion size $\OO(k/\varepsilon)$ in time $(k/\varepsilon)^{\OO(k/\varepsilon)} m^{1+\varepsilon}$, for any given $k \ge 1$ and $\varepsilon < 1$.

It can be observed that a $k$-lean tree decomposition is $(i,i)$-unbreakable for all $i \le k$, and therefore \Cref{the:mainalgklean} gives a $k^{\OO(k^2)} n + \OO(m)$ time algorithm for computing a tree decomposition that is $(i,i)$-unbreakable for all $i \le k$ and has adhesion size $<k$.
This is the first parameterized linear-time algorithm for unbreakable tree decompositions, answering the question of~\cite{DBLP:journals/talg/CyganKLPPSW21} in a strong sense.
The unbreakability bounds match the bounds in~\cite{DBLP:journals/talg/CyganKLPPSW21}, which are indeed the best possible for the adhesion size bound $k$, while the running time dependence on the parameter $k$ is worse than in the previous algorithms.

As further related work, let us mention the work of Grohe and Schweitzer~\cite{DBLP:journals/siamdm/GroheS16,DBLP:conf/icalp/Grohe16} on algorithms with running time $n^{\OO(k)}$ for computing a \emph{$k$-tangle tree decomposition}~\cite{RobertsonS91}, which is a structure that is in some ways similar to lean and unbreakable tree decompositions, and is closely related to the techniques we use in the proof of \Cref{the:mainalgklean}.

\paragraph{Vertex connectivity.}
The \emph{vertex connectivity} of a graph $G$ is the minimum number of vertices that need to be deleted from $G$ to disconnect it into at least two components.
Equivalently, it is the minimum size of a \emph{proper vertex separator} of $G$, which is a set $S \subseteq V(G)$ so that $G \setminus S$ has at least two connected components.
We obtain the following algorithm as a straightforward application of~\Cref{the:mainalgklean}.

\begin{restatable}{corollary}{theglobalvertexcut}
\label{the:mainglobalvertexcut}
There is an algorithm that, given an $n$-vertex $m$-edge graph $G$ and an integer $k \ge 1$, in time $k^{\OO(k^2)} n + \OO(m)$ returns a proper vertex separator of $G$ of size $<k$ if any exist.
\end{restatable}

The algorithm of \Cref{the:mainglobalvertexcut} is the first algorithm for testing if vertex connectivity is less than $k$ that works in linear-time for every fixed $k$.
Previously, linear-time algorithms were known only for $k \le 3$~\cite{DBLP:journals/siamcomp/Tarjan72,HopcroftT73}.
Even though there was a long line of work for the problem for $k \ge 4$~\cite{kleitman1969methods,DBLP:journals/siamcomp/Even75,DBLP:journals/siamcomp/EvenT75,DBLP:journals/combinatorica/LinialLW88,DBLP:journals/jcss/KanevskyR91,DBLP:journals/algorithmica/NagamochiI92,DBLP:journals/jal/HenzingerRG00,DBLP:journals/jacm/Gabow06}, before 2019 the best running time for $k=4$ (and for any constant $k$) was $\OO(k^2 n^2)$ achieved by~\cite{DBLP:journals/jcss/KanevskyR91,DBLP:journals/algorithmica/NagamochiI92}.
The quadratic barrier was broken by Nanongkai, Saranurak, and Yingchareonthawornchai~\cite{DBLP:conf/stoc/NanongkaiSY19}, and shortly afterwards Forster, Nanongkai, Yang, Saranurak, and Yingchareonthawornchai achieved the running time $\OO(k^3 n \log^{\OO(1)} n + m)$~\cite{DBLP:conf/soda/ForsterNYSY20} (see also~\cite{DBLP:conf/focs/SaranurakY22,DBLP:journals/corr/abs-2308-04695}).
For unbounded $k$, Li, Nanongkai, Panigrahi, Saranurak, and Yingchareonthawornchai~\cite{DBLP:conf/stoc/LiNPSY21} gave an almost-linear $m^{1+o(1)}$ time algorithm via almost-linear time max-flow~\cite{DBLP:conf/focs/ChenKLPGS22}.

\paragraph{Element connectivity.}
As a $k$-lean tree decomposition is a rather strong representation of vertex separators, it would be tempting to use \Cref{the:mainalgklean} to also obtain a $k$-Gomory-Hu tree for vertex connectivity.
However, no ($k$-)Gomory-Hu trees exist for vertex connectivity~\cite{DBLP:journals/siamcomp/Benczur95}, so we settle for the simpler notion of \emph{element connectivity}~\cite{DBLP:journals/jal/JainMVW02}.

For a set $U \subseteq V(G)$ and two vertices $u,v \in U$, we say that a set $S \subseteq (V(G) \setminus U) \cup E(G)$ is an $U$-element $(u,v)$-cutset if $u$ and $v$ are in different connected components of $G \setminus S$.
In particular, the vertices in $U$ are treated as undeletable, while we are allowed to cut any other vertices and edges.
It turns out that $k$-Gomory-Hu trees exist for element connectivity~\cite{schrijver2003combinatorial,DBLP:conf/esa/ChekuriRX15,DBLP:conf/stoc/PettieSY22}.
An \emph{element connectivity $k$-Gomory-Hu tree} of a pair $(G,U)$ is a triple $(T,\gamma,\alpha)$, where $T$ is a tree, $\gamma \colon U \to V(T)$ maps vertices in $U$ to nodes of $T$, $\alpha \colon E(T) \to 2^{(V(G) \setminus U) \cup E(G)}$ maps edges of $T$ to subsets of $(V(G) \setminus U) \cup E(G)$, and
\begin{enumerate}
\item for all $u,v \in U$ and $e \in E(T)$ so that $e$ is on the $(\gamma(u), \gamma(v))$-path in $T$, the set $\alpha(e)$ is an $U$-element $(u,v)$-cutset of size $|\alpha(e)| < k$, and
\item for all $u, v \in U$, if there exists an $U$-element $(u,v)$-cutset of size $<k$, then there exists $e \in E(T)$ on the $(\gamma(u),\gamma(v))$-path so that $\alpha(e)$ is a minimum-size $U$-element $(u,v)$-cutset.
\end{enumerate}

Via a similar reduction to \Cref{the:mainalgklean} as for \Cref{the:mainkgomoryhu}, we obtain the following algorithm for computing an element connectivity $k$-Gomory-Hu tree.

\begin{restatable}{corollary}{theelementkgomoryhu}
\label{the:mainelementgomoryhu}
There is an algorithm that, given an $n$-vertex $m$-edge graph $G$, a set $U \subseteq V(G)$, and an integer $k \ge 1$, in time $k^{\OO(k^2)} n + \OO(m)$ returns an element connectivity $k$-Gomory-Hu tree of $(G,U)$.
\end{restatable}

In fact, \Cref{the:mainkgomoryhu} is a special case of \Cref{the:mainelementgomoryhu} obtained by setting $U = V(G)$.
The algorithm of \Cref{the:mainelementgomoryhu} also implies a $k^{\OO(k^2)} (n+m) + \OO(M)$ time algorithm for computing a $k$-Gomory-Hu tree of a hypergraph (in the context of hyperedge cuts), where $m$ is the number of hyperedges and $M$ the total size of them, since hyperedge connectivity can be reduced to element connectivity (see e.g.~\cite{DBLP:journals/orl/CheriyanZ12} for the reduction).

Previously, Pettie, Saranurak, and Yin~\cite{DBLP:conf/stoc/PettieSY22} gave a $k \cdot m^{1+o(1)}$ time algorithm for computing an element connectivity $k$-Gomory-Hu tree, based on the Gomory-Hu tree algorithm of Li and Panigrahi~\cite{DBLP:conf/stoc/LiP21}.
The authors of~\cite{DBLP:conf/stoc/PettieSY22} used the element connectivity $k$-Gomory-Hu tree for constructing a data structure for answering pairwise vertex connectivity queries.
Our \Cref{the:mainelementgomoryhu} could also be applied with color coding~\cite{DBLP:journals/jacm/AlonYZ95} to construct a data structure that answers pairwise vertex connectivity queries up to the connectivity value $k$, but this would involve either randomization (yielding a data structure that is correct for a given pair with probability $0.99$ with one-sided error), or an $\OO(\log n)$-factor in the running time.

\paragraph{Our techniques.}
Previous linear-time algorithms for $k$-(edge)-connected components~\cite{DBLP:journals/siamcomp/Tarjan72,HopcroftT73,DBLP:conf/esa/NadaraRSS21,DBLP:conf/esa/GeorgiadisIK21,DBLP:conf/soda/Kosinas24} followed an approach based on analyzing the properties of depth-first search trees.
In this paper, we take a different approach, and our main ``linear-time engine'' instead comes from a matching contraction technique that Bodlaender~\cite{Bodlaender96} introduced for his parameterized linear-time algorithm for treewidth.
By generalizing Bodlaender's technique, we show that computing a $k$-lean tree decomposition reduces to a problem where we can assume that the input already contains a $(2k,k)$-unbreakable tree decomposition with adhesion size~$\le 2k$, and we only need to improve it to $k$-lean.
The sparsifier of Nagamochi and Ibaraki~\cite{DBLP:journals/algorithmica/NagamochiI92} plays a large role in our proof of this generalized version of Bodlaender's technique.

The majority of this paper is about giving the algorithm for improving the given unbreakable tree decomposition into $k$-lean.
For this algorithm, our main new insights are about exploiting so-called ``doubly well-linked'' separations.
These are separations can be used to ``greedily'' decompose the graph, analogously to decomposing along minimum cuts in the setting of computing Gomory-Hu trees~\cite{gomory1961multi}.
Furthermore, doubly well-linked separations turn out to be relatively easy to find via local-search type computations in unbreakable tree decompositions.
The definition of doubly well-linked separations and the techniques for exploiting them are inspired by ``doubly robust'' separations used by Robertson and Seymour~\cite{RobertsonS91}.

Our algorithm proceeds by improving the input tree decomposition step by step, finally arriving at a $(k^{\OO(k)}, k)$-unbreakable tree decomposition with adhesion size $<k$, whose adhesions correspond to doubly well-linked separations.
At this point, we can turn the tree decomposition into $k$-lean simply by computing $k$-lean tree decompositions of all ``torsos'' of the decomposition, which are $(k^{\OO(k)}, k)$-unbreakable graphs, and stitching them together.
For computing $k$-lean tree decomposition of unbreakable graphs, we use an approach based on the Bellenbaum-Diestel proof~\cite{bellenbaum2002two}.

We believe that an important conceptual ingredient of our algorithm is that during most of the improvement algorithm we work with hypergraphs, and in particular, with Robertson-Seymour style separations of hypergraphs~\cite{RobertsonS91} (we explain what this means in \Cref{sec:overview}).
This allows for elegant formulations of the submodularity of vertex separators and the concept of ``gluing vertices together''.
In fact, we decided to use this set of definitions because the usual definitions for vertex separators failed in subtle ways for some of our key lemmas.

A question that is natural to ask about our techniques is whether our result for $k$-edge-connected components could be obtained directly without going through notions that are about vertex connectivity, like $k$-lean tree decompositions.
The answer to this appears to be currently negative.
In particular, the first, and perhaps the most significant barrier is that the matching contraction technique of Bodlaender~\cite{Bodlaender96} seems to not work in the setting of edge connectivity, because contracting a matching approximately preserves small balanced vertex separators, but could completely destroy small balanced edge cuts.

Our paper is almost completely self-contained, the only external ingredients whose proofs are not given here are the sparsifier of Nagamochi and Ibaraki~\cite{DBLP:journals/algorithmica/NagamochiI92} and the branchwidth-tangle duality theorem of Robertson and Seymour~\cite{RobertsonS91}, both of which have proofs in less than 10 pages.

\paragraph{Organization of the paper.}
We start by giving an overview of our algorithm in \Cref{sec:overview}.
Then, in \Cref{sec:prelis} we present formal definitions and preliminary results.
In \Cref{sec:roadmap} we break the proof of \Cref{the:mainalgklean} down to 9 ``main lemmas'', then prove \Cref{the:mainalgklean} assuming these main lemmas, and then prove \Cref{the:mainkgomoryhu}, \Cref{the:mainglobalvertexcut}, and \Cref{the:mainelementgomoryhu} assuming \Cref{the:mainalgklean}.
The rest of the paper, consisting of \Cref{sec:bodl,sec:gt,sec:localsearch,sec:manipulatingsuperbranch,sec:downwl,sec:mixedkwldecomp,sec:tangubrtorsos,sec:smalleradhesions,sec:fromtubrtoubr,sec:combin,sec:lean}, is then dedicated to proving these main lemmas one by one, relatively independently of each other.
Further explanation of the organization of \Cref{sec:bodl,sec:gt,sec:localsearch,sec:manipulatingsuperbranch,sec:downwl,sec:mixedkwldecomp,sec:tangubrtorsos,sec:smalleradhesions,sec:fromtubrtoubr,sec:combin,sec:lean} is given in \Cref{sec:roadmap}.

\section{Overview}
\label{sec:overview}
We then give an overview of our proofs.
We focus only on the proof of \Cref{the:mainalgklean}; the proofs of \Cref{the:mainkgomoryhu}, \Cref{the:mainglobalvertexcut}, and \Cref{the:mainelementgomoryhu} assuming \Cref{the:mainalgklean} are relatively short and easy, and are given in \Cref{sec:roadmap}.

The algorithm of \Cref{the:mainalgklean} has two conceptually distinct parts.
The first part is a self-reduction technique that shows that to give a $k^{\OO(k^2)} n + \OO(m)$ time algorithm for computing a $k$-lean tree decomposition, it suffices to give such an algorithm for computing a $k$-lean tree decomposition with an additional assumption that the input already contains a $(2k,k)$-unbreakable tree decomposition with adhesion size $\le 2k$.
The second part is to actually give such an ``improvement algorithm''.
We sketch the proof of the first part in \Cref{subsec:ovgenbodl}, and then overview of the second part in \Cref{subsec:ovimprover}.

\subsection{Generalized Bodlaender's technique}
\label{subsec:ovgenbodl}
Assume that we have an algorithm $\mathcal{A}$, that takes as input a graph $G$, an integer $k \ge 1$, and a $(2k,k)$-unbreakable tree decomposition of $G$ with adhesion size $\le 2k$, and in time $k^{\OO(k^2)} m$ outputs a $k$-lean tree decomposition of $G$.
Let us show that $\mathcal{A}$ can be used to construct a $k^{\OO(k^2)} n + \OO(m)$ time algorithm for computing a $k$-lean tree decomposition from scratch.

First, we use the algorithm of Nagamochi and Ibaraki~\cite{DBLP:journals/algorithmica/NagamochiI92} to in time $\OO(m)$ compute a subgraph $G'$ of the input graph $G$ with $|E(G')| \le k \cdot |V(G)|$, so that any vertex cut $(A,B)$ of $G'$ with $|A \cap B| < k$ is also a vertex cut of $G$.
It is not hard to show that any $k$-lean tree decomposition $(T,\bag)$ of $G'$ is also a $k$-lean tree decomposition of $G$:
To show that $(T,\bag)$ is a tree decomposition of $G$, we observe that if the algorithm of~\cite{DBLP:journals/algorithmica/NagamochiI92} removes an edge $uv \in E(G) \setminus E(G')$, then it leaves in $G'$ a set of $k$ internally vertex-disjoint $(u,v)$-paths.
Together with the fact that $(T,\bag)$ has adhesion size $<k$, this implies that $(T,\bag)$ must contain a bag containing both $u$ and $v$.
Then, $(T,\bag)$ is $k$-lean because any separator of $G$ witnessing otherwise would also be a separator of $G'$ witnessing that $(T,\bag)$ is not a $k$-lean tree decomposition of $G'$.
Therefore, we can now assume that $m \le kn$.

We then give a recursive algorithm for computing a $k$-lean tree decomposition with the help of $\mathcal{A}$.
We start by computing in $\OO(m) = \OO(kn)$ time an inclusion-wise maximal matching $M$ of the input graph $G$.
Assume first that $|M| \ge n/\OO(k^3)$.
Now, we contract all edges in $M$ to obtain the graph $G' = G \contr M$, call our algorithm recursively with $G'$ to compute a $k$-lean tree decomposition $(T',\bag')$ of $G'$, and then ``uncontract'' $(T',\bag')$ along the matching $M$ to obtain a tree decomposition $(T^*,\bag^*)$ of $G$.
Here, the uncontraction means replacing each vertex $w_{uv}$ corresponding to the contraction of an edge $uv$ with the two vertices $u$ and $v$.
It is not hard to show that the uncontraction operation at most doubles the adhesion size, so the adhesion size of $(T^*,\bag^*)$ is at most $2k-2$, and furthermore, because each bag of $(T',\bag')$ is $(i,i)$-unbreakable for all $i \le k$, it can be shown that each bag of $(T^*,\bag^*)$ is $(2k,k)$-unbreakable.
Then, we can apply the algorithm $\mathcal{A}$ with $G$ and $(T^*,\bag^*)$ to obtain a $k$-lean tree decomposition of $G$, and then return it.
Because $|V(G')| \le |V(G)| \cdot \left(1-\frac{1}{\OO(k^3)}\right)$, the total running time of this recursive algorithm sums up to $k^{\OO(k^2)} n$.
Note that the Nagamochi-Ibaraki sparsifier needs to be applied at each level of the recursion.

The above described recursion step works when $|M| \ge n/\OO(k^3)$.
Our goal is to also achieve a similar recursion step that eliminates a large fraction of vertices when $|M| < n/\OO(k^3)$.
In the case, we would like to conclude that because $V(M)$ is a relatively small vertex cover of $G$ (i.e., every edge of $G$ is incident to $V(M)$), the graph $G$ is ``star-like'' in some sense, and we can ``eliminate'' a significant number of ``leaves'' of this star.
More formally, we define that the \emph{$k$-improved graph} $I_k(G)$ of $G$ is the supergraph of $G$ that contains all edges of $G$ and additionally an edge between every pair $u,v \in V(G)$ for which there is a collection of $k$ internally vertex-disjoint paths between $u$ and $v$.
Now, we say that a vertex $v$ is \emph{$I_k$-simplicial} if the neighborhood of $v$ is a clique in $I_k(G)$.
We manage to show, by using the Nagamochi-Ibaraki sparsifier in an interesting manner, that when $|M| < n/\OO(k^3)$, we can find in $k^{\OO(1)} n$ time an independent set $I \subseteq V(G) \setminus V(M)$ of $|I| \ge \Omega(n)$ $I_k$-simplicial vertices of degree $\le 4k$.
Then we construct the graph $G' = G \elim I$ by turning the neighborhood of each $v \in I$ into a clique and removing $I$, then apply our algorithm recursively to compute a $k$-lean tree decomposition $(T',\bag')$ of $G'$, and then use the fact that the vertices in $I$ are $I_k$-simplicial and the algorithm $\mathcal{A}$ to lift $(T',\bag')$ into a $k$-lean tree decomposition of $G$.

Let us remark that the above two cases of either contracting a large matching or eliminating a large set of $I_k$-simplicial vertices follow the outline of Bodlaender's algorithm~\cite{Bodlaender96}.
However, the analysis that a small maximal matching implies many $I_k$-simplicial vertices is significantly more complicated in our case than in Bodlaender's (where he could also make the third conclusion that the graph has treewidth $>k$), so we regard this part of our algorithm as a non-trivial generalization of the Bodlaender's technique.

\subsection{The improvement algorithm}
\label{subsec:ovimprover}
We then overview the improvement algorithm, that takes as input a graph $G$, an integer $k \ge 1$, and a $(2k,k)$-unbreakable tree decomposition of $G$ with adhesion size $\le 2k$, and in time $k^{\OO(k^2)} m$ computes a $k$-lean tree decomposition of $G$.
This algorithm has so many steps that there is no space to sketch the proofs of all of them here.
Let us instead first discuss the general methodology and ideas behind this algorithm, and then give a high-level outline of it.
A further outline of our algorithm with formally correct definitions and statements is given in \Cref{sec:roadmap}.

\subsubsection{General methodology}
\paragraph{Hypergraphs and Robertson-Seymour style separations.}
We work with hypergraphs, which are like graphs but the hyperedges are arbitrary subsets of vertices instead of subsets of size $2$.
In our case, the hyperedges will always have size at most $2k$, and moreover, we will have the extra property that every vertex belongs to at least two hyperedges.
The hypergraphs we work with may have multihyperedges, i.e., there may be many distinct hyperedges each consisting of the same set of vertices.
The intuition will be that working with a hypergraph $G$ corresponds to working with its primal graph $\primal(G)$, which is the graph with $V(\primal(G)) = V(G)$ and $E(\primal(G))$ obtained by making each hyperedge of $G$ into a clique.
This intuition will hold almost always, the exceptions turning up only in some subtle considerations.
In this overview, whenever we use definitions in the context of hypergraphs that have only been defined for graphs, assume that this definition is applied to the primal graph of the hypergraph.

The \emph{border} of a hyperedge set $X \subseteq E(G)$ is $\bd(X) = V(X) \cap V(E(G) \setminus X)$, where $V(X)$ for a set $X \subseteq E(G)$ denotes the union of the vertices of the hyperedges in $X$.
A separation of a hypergraph $G$ is defined as a bipartition $(A,B)$ of its hyperedge set $E(G)$, i.e., so that $A \cup B = E(G)$ and $A \cap B = \emptyset$, and the \emph{order} of a separation $(A,B)$ is $\bdc(A) = \bdc(B) = |\bd(A)| = |\bd(B)| = |V(A) \cap V(B)|$.
We note that the function $\bdc \colon 2^{E(G)} \to \mathbb{Z}_{\ge 0}$ is symmetric and submodular, meaning that (1) $\bdc(A) = \bdc(E(G) \setminus A)$ for all $A \subseteq E(G)$, and (2) $\bdc(A \cap B) + \bdc(A \cup B) \le \bdc(A) + \bdc(B)$ for all $A,B \subseteq E(G)$.
These definitions are inspired by~\cite{RobertsonS91}.

The advantage of these definitions is the following.
Suppose we have identified a separation $(A,B)$ of $G$ and would like to decompose $G$ along $(A,B)$.
Now, we define $G \rescliqs B$ to be the hypergraph obtained by removing all hyperedges in $B$, and inserting a hyperedge $e_B$ with $V(e_B) = \bd(B)$.
More formally, $V(G \rescliqs B) = V(A)$, and $E(G \rescliqs B) = A \cup \{e_B\}$ with $V(e_B) = \bd(B)$.
Now, the separations of $G \rescliqs B$ correspond to separations of $G$ by expanding $e_B$ into $B$, and the separations of $G$ that do not ``cross'' $B$ correspond to separations of $G \rescliqs B$ by shrinking $B$ into~$e_B$.
Similar effect could be achieved in the context of graphs by just making $\bd(B)$ into a clique.
However, this turns out to be problematic because a clique could ``provide connectivity'' to both sides of a separation, while we need to insist that $e_B$ provides connectivity only to one side of any separation of $G \rescliqs B$.

\paragraph{Well-linkedness.}
A key technique that we will employ is that if a separation $(A,B)$ of a hypergraph $G$ has sufficiently strong ``uncrossing properties'', then a $k$-lean tree decomposition of $G \rescliqs A$ can be simply stitched together with a $k$-lean tree decomposition of $G \rescliqs B$ along $\bd(A) = \bd(B)$ to obtain a $k$-lean tree decomposition of $G$.
Formally, these ``uncrossing properties'' are that $(A,B)$ is \emph{doubly well-linked}.
We define that a set $A \subseteq E(G)$ of hyperedges is \emph{well-linked} if there is no bipartition $(C_1,C_2)$ of $A$ so that $\bdc(C_i) < \bdc(A)$ for both $i \in [2]$.
Now, a separation $(A,B)$ is doubly well-linked if both $A$ and $B$ are well-linked.
It turns out that doubly well-linked separations enjoy strong uncrossing properties: If $(A,B)$ is doubly well-linked separation, then we can uncross with $(A,B)$ any separations between two subsets of $V(A)$, any separations between two subsets of $V(B)$, and any separations between a subset of $V(A)$ and a subset of $V(B)$.
It follows from these uncrossing properties that when computing a $k$-lean tree decomposition, we can indeed ``greedily'' decompose along any doubly well-linked separation of order $<k$.
We note that in a sense, this is a generalization of the classical property that when computing a Gomory-Hu tree, one can ``greedily'' decompose along any pairwise minimum cut~\cite{gomory1961multi}.

Let us furthermore note the following key property of well-linkedness.
Suppose that $A \subseteq E(G)$ is a well-linked set, and consider the hypergraph $G \rescliqs A$, containing the hyperedge $e_A$ corresponding to $A$.
Now, suppose that $B \subseteq E(G \rescliqs A)$ is a well-linked set in $G \rescliqs A$.
Let us denote by $B \orescliqs A$ the set of hyperedges of $G$ corresponding to $B$, i.e., $B \orescliqs A = B$ if $e_A \notin B$, and $B \orescliqs A = (B \cup A) \setminus \{e_A\}$ if $e_A \in B$.
It can be shown that now, $B \orescliqs A$ is well-linked in $G$.
This is quite trivial if $e_A \notin B$, and it can be proved when $e_A \in B$ by an uncrossing argument with $A$.
We call this property the \emph{transitivity} of well-linkedness.
It is useful for arguing that even after multiple iterations of decomposing a hypergraph by doubly well-linked separations, all of the corresponding separations of the original hypergraph are also doubly well-linked.

\paragraph{Superbranch decompositions.}
We did not yet define what we actually mean by a decomposition of a hypergraph, so let us define it now.
A \emph{superbranch decomposition} of a hypergraph $G$ is a pair $\Tc = (T,\lmap)$, where $T$ is a tree whose all internal nodes have degree~$\ge 3$, and $\lmap \colon \leafs(T) \to E(G)$ is a bijection from the leaves of $T$ to the hyperedges of $G$.
This is a generalization of the definition of branch decompositions by~\cite{RobertsonS91}, whose internal nodes are required to have degree exactly~$3$.
For an edge $xy$ of $T$, we denote by $\lmap(\vec{xy}) \subseteq E(G)$ the set of hyperedges mapped to leaves of $T$ closer to $x$ than $y$.
Note that $(\lmap(\vec{xy}), \lmap(\vec{yx}))$ is a separation of $G$, and in fact the collection of all such separations over all $xy \in E(T)$ forms a collection of pairwise non-crossing separations of $G$.
We say that $(\lmap(\vec{xy}), \lmap(\vec{yx}))$ is an \emph{internal separation} of $\Tc$ if both $x$ and $y$ are internal nodes of $T$.
The adhesion at $xy$ is $\adh(xy) = \bd(\lmap(\vec{xy})) = \bd(\lmap(\vec{yx}))$.
Lastly, the \emph{torso} of a node $t \in V(T)$ is the hypergraph $\torso(t)$ whose hyperedges correspond to the adhesions of the edges of $T$ incident to $t$.
Equivalently, $\torso(t) = G \rescliqs \lmap(\vec{c_1 t}) \rescliqs \lmap(\vec{c_2 t}) \rescliqs \ldots \rescliqs \lmap(\vec{c_p t})$, where $c_1,\ldots,c_p$ are the neighbors of $t$ in $T$ and the operation $\rescliqs$ is applied from left to right.

We note that if $(T,\lmap)$ is a superbranch decomposition of $G$, then $(T,\bag)$, where $\bag(t) = V(\torso(t))$ for each $t \in V(T)$, is a tree decomposition of $\primal(G)$.
Moreover, the adhesions of $(T,\bag)$ and the superbranch decomposition $(T,\lmap)$ are the same.
In particular, one may think of a superbranch decomposition of $G$ as a tree decomposition of $\primal(G)$, where we have additionally fixed a bijection between the leaves of the decomposition and the hyperedges of $G$.

\paragraph{Intermediate goal.}
We can now state the main intermediate goal of our algorithm.
A hypergraph $G$ is $(s,k)$-unbreakable if there is no separation $(A,B)$ so that $\bdc(A) < k$ and $|V(A)|,|V(B)| \ge s$.
Let $\hyperg(G)$ denote the hypergraph obtained from the input graph $G$ in the natural manner.
Our goal is to compute a superbranch decomposition $(T,\lmap)$ of $\hyperg(G)$ so that
\begin{enumerate}
\item the internal separations of $(T,\lmap)$ have order $<k$ and are doubly well-linked, and
\item for each $t \in V(T)$, the hypergraph $\torso(t)$ is $(k^{\OO(k)}, k)$-unbreakable.
\end{enumerate}

After computing such a superbranch decomposition $(T,\lmap)$, it turns out that we can compute a $k$-lean tree decomposition $\Tc_t$ of each graph $\primal(\torso(t))$ for $t \in V(T)$, and obtain a $k$-lean tree decomposition of $G$ by stitching the decompositions $\Tc_t$ together along $T$.
This follows from the uncrossing properties of doubly well-linked separations.
The graphs $\primal(\torso(t))$ are $(k^{\OO(k)}, k)$-unbreakable, so what remains after this is to design an $s^{\OO(k)} m$ time algorithm for computing a $k$-lean tree decomposition of an $(s,k)$-unbreakable graph, which can be done by essentially following the ideas of~\cite{bellenbaum2002two,DBLP:journals/talg/CyganKLPPSW21}.
Therefore, in the rest of this overview let us focus only on achieving this intermediate goal.

First, let us argue why such a superbranch decomposition $(T,\lmap)$ even exists.
Let $(T,\lmap)$ be a superbranch decomposition of a hypergraph $G$ whose internal separations have order $<k$ and are doubly well-linked.
Note that such a superbranch decomposition exists by taking the trivial star-shaped superbranch decomposition that has one internal node and no internal separations.
We will argue that such a superbranch decomposition that maximizes $|V(T)|$ must have also $(2^k,k)$-unbreakable torsos.
In a more algorithmic sense, we will argue that if there exists $t \in V(T)$ so that $\torso(t)$ is not $(2^k, k)$-unbreakable, then we can further decompose $\torso(t)$ by a doubly well-linked separation of order $<k$, increasing the number of nodes of~$T$.

Suppose now that $\torso(t)$ is not $(2^k, k)$-unbreakable, let $i \le k$ be the smallest integer so that $\torso(t)$ is not $(2^i, i)$-unbreakable, and let $(A,B)$ be a separation of $\torso(t)$ with $\bdc(A) < i$, and $|V(A)|,|V(B)| \ge 2^i$.
Now, it can be shown that the $(2^{i-1}, i-1)$-unbreakability of $\torso(t)$ implies that both $A$ and $B$ are well-linked in $\torso(t)$, i.e., $(A,B)$ is doubly well-linked.
Then, the transitivity of well-linkedness implies that the separation $(A \orescliqs \Tc, B \orescliqs \Tc)$ of $G$ corresponding to $(A,B)$ is also doubly-well-linked.
Furthermore, we have that $|A|,|B| \ge 2$.
Therefore, we can split $\torso(t)$ along $(A,B)$, separating the node $t$ of $T$ into two nodes $t_A$ and $t_B$, with $\torso(t_A) = \torso(t) \rescliqs B$ and $\torso(t_B) = \torso(t) \rescliqs A$, increasing $|V(T)|$.

The above proof is quite algorithmic, but there is an issue of how to find such separations $(A,B)$ efficiently, and even then it is not clear how the whole decomposition process could be implemented in linear-time.
However, both of these issues would be more approachable if we would have a guarantee that if such an ``unbreakability-violating'' separation $(A,B)$ exists, then either $|A|$ or $|B|$ could be bounded as a function of $k$.
Then, we could find such separations by a local-search type algorithm in time proportional to $\min(|A|,|B|)$, and also implement the decomposition process efficiently by local computations that would edit only the small side.

Now, the utility of starting with a $(2k,k)$-unbreakable superbranch decomposition instead of starting from scratch is that we indeed have strong unbreakability guarantees to start with.
However, the challenge is that the adhesions of the input superbranch decomposition can have size up to $2k$, and its internal separations do not have any well-linkedness properties.
Therefore, our plan is to first ``massage'' the internal separations of the input tree decomposition into doubly well-linked, then get rid of adhesions of size $\ge k$, and if we manage to do this while maintaining the unbreakability of the torsos, we are done.

\subsubsection{Outline of the algorithm}
Let us now give a more detailed outline of our algorithm, and in particular, let us focus on achieving the intermediate goal.
We start with a superbranch decomposition $\Tc = (T,\lmap)$ of a hypergraph $G$, so that the adhesions of $\Tc$ have size $\le 2k$, and for each $t \in V(T)$, the set $V(\torso(t))$ is $(2k,k)$-unbreakable in $G$.
Note that the unbreakability of $V(\torso(t))$ in $G$ is a stronger property than just the unbreakability of the hypergraph $\torso(t)$.
Our goal is to turn $\Tc$ into a superbranch decomposition $\Tc' = (T',\lmap')$, so that the internal separations of $\Tc'$ have order $<k$ and are doubly well-linked, and $\torso(t)$ for each $t \in V(T)$ is $(k^{\OO(k)}, k)$-unbreakable.

First, we wish to turn the internal separations of $\Tc$ into doubly well-linked.
We divide this into two parts, first making them ``downwards well-linked'', and then making them ``upwards well-linked''.
Let $T$ be rooted at some node.
We say that $\Tc$ is \emph{downwards well-linked} if for all edges $tp$ of $T$ so that $p$ is the parent of $t$, the set $\lmap(\vec{tp})$ is well-linked.

\paragraph{Downwards well-linkedness.}
The idea for making $\Tc$ downwards well-linked is that if a set $A \subseteq E(G)$ is not well-linked, then it can be partitioned into at most $2^{\bdc(A)}$ parts $B_1,\ldots,B_h$, so that each part $B_i$ is well-linked.
This partitioning works simply by initializing the partition to contain only $A$, and then as long as there is a part $B_i$ that is not well-linked, finding a bipartition $(C_1,C_2)$ of $B_i$ so that $\bdc(C_j) < \bdc(B_i)$ for both $j \in [2]$, and replacing $B_i$ by $C_1$ and $C_2$.
This can be implemented in roughly $|A| \cdot 2^{\OO(\bdc(A))}$ time.

Now, we process $\Tc$ from the leaves to the root, and for each node $t$ with a parent $p$, if the set $E(\torso(t)) \setminus \{e_p\}$ (where $e_p$ is the hyperedge corresponding to the adhesion at $tp$) is not well-linked in $\torso(t)$, we find a partition of it into $2^{\OO(k)}$ well-linked sets, and ``collapse'' the adhesion at $tp$ into adhesions corresponding to these sets.
When implementing this in the bottom-up order, the transitivity of well-linkedness implies that the resulting superbranch decomposition becomes downwards well-linked.
After incorporating some additional technical elements to this process, it can be shown that because the vertex sets of the torsos of $\Tc$ are $(2k,k)$-unbreakable in $G$, the vertex sets of the torsos of the resulting superbranch decomposition are $(2^{\OO(k)}, k)$-unbreakable in $G$.

\paragraph{Upwards well-linkedness.}
In the next step we wish to make the superbranch decomposition $\Tc$, in addition to being downwards well-linked, also ``upwards well-linked'', i.e., for each $tp \in E(T)$ so that $p$ is the parent of $t$, make the set $\lmap(\vec{pt})$ well-linked.
The intuition is that we would like to run a similar adhesion-collapsing procedure as for downwards well-linkedness, but this time top-down, from the root to the leaves.
This is however too much to ask, since we cannot afford to make for each child $t$ of $p$ an $\Omega(E(\torso(p)))$-time search for decomposing its parent.
Therefore, we settle for a notion of well-linkedness into which we can decompose by a local-search type algorithm in $(2^{\OO(k)}, k)$-unbreakable hypergraphs.

We define that a set $A \subseteq E(G)$ is \emph{$k$-well-linked}, if there is no bipartition $(C_1,C_2)$ of $A$ so that $\bdc(C_i) < \bdc(A)$ for both $i \in [2]$, and $\bdc(C_i) < k$ for at least one $i \in [2]$.
Now, we can use a local-search type algorithm, that if $E(\torso(p)) \setminus \{e_t\}$ is not $k$-well-linked, finds an implicitly represented partition of it into $2^{\OO(k)}$ $k$-well-linked sets, whose complements are well-linked in $\torso(p)$.
With this, we can, in some sense, implement a top-down procedure that makes $\Tc$ upwards $k$-well-linked, in addition to being downwards well-linked.
In reality the statement is more technical, as $\Tc$ will become ``scrambled'' in this process, but it is true enough that we can in this overview assume that $\Tc$ becomes downwards well-linked and upwards $k$-well-linked.
This procedure furthermore maintains that the torsos of $\Tc$ are $(2^{\OO(k)}, k)$-unbreakable.

We note that $k$-well-linkedness alone is a quite weak property, but a separation $(A,B)$ so that $A$ is well-linked and $B$ is $k$-well-linked turns out to be as strong as a doubly well-linked separation for the purpose of uncrossing separations of order $<k$.
Note that indeed the internal separations of $\Tc$ are now of this type.
Note also that if the $A$ is $k$-well-linked and $\bdc(A) < k$, then $A$ is in fact well-linked.

\paragraph{Small adhesions.}
Now, the main remaining issue is that the adhesions of $\Tc$ are still too large, having size at most $2k$ instead of less than $k$.
For reducing the adhesion size, our main tool will be a graph-theoretical concept of \emph{$k$-tangle-unbreakability}.
\emph{Tangles} were introduced by Robertson and Seymour~\cite{RobertsonS91}, and they are a certain way of defining ``well-connected'' regions in a (hyper)graph.
For readers familiar with tangles, a hypergraph is $k$-tangle-unbreakable if there is no separation $(A,B)$ of order $\bdc(A) < k$, so that $(A,B)$ is oriented in a different way by two tangles, in particular, if there is a unique maximal tangle of order $\le k$.
For readers not familiar with tangles, a good intuition is that the $k$-tangle-unbreakability of a hypergraph implies that if $(A,B)$ is a separation of order $\bdc(A) < k$ (and satisfies an extra condition), then either the \emph{branchwidth} of $G \rescliqs A$ is at most $\bdc(A)$, or the branchwidth of $G \rescliqs B$ is at most $\bdc(A)$.
The branchwidth of a hypergraph is equal to the minimum \emph{width} of a branch decomposition of it, where a branch decomposition is simply a superbranch decomposition with maximum degree $3$, and the width of a branch decomposition is its maximum adhesion size.

We show that if the torsos of a superbranch decomposition $\Tc = (T,\lmap)$ are $k$-tangle-unbreakable, and its internal separations are well-linked on one side and $k$-well-linked on the other side, then if we simply contract all edges of $T$ corresponding to adhesions of size $\ge k$, the resulting superbranch decomposition will have $k$-tangle-unbreakable torsos (and obviously, all of its internal separations will have order $<k$ and be doubly well-linked).
This will be our technique for getting rid of too large adhesions.

To be able to exploit this technique, we need to show that (1) $(2^{\OO(k)}, k)$-unbreakable graphs can be efficiently decomposed into $k$-tangle-unbreakable graphs by doubly well-linked separations of order $<k$, and (2) $k$-tangle-unbreakable graphs can be efficiently decomposed into $(k^{\OO(k)}, k)$-unbreakable graphs by doubly well-linked separations of order $<k$.
If we manage to do (1) and (2), then by first decomposing each torso by (1), then contracting the too large adhesions, and then applying (2) to decompose each resulting torso, we are done with our goal of computing a superbranch decomposition with $(k^{\OO(k)}, k)$-unbreakable torsos and doubly well-linked internal separations of order $<k$.

It turns out that both (1) and (2) can be achieved by designing local-search type algorithms that ``chip'' away small parts of the hypergraph at a time.
In particular, for (1) we show that to make a hypergraph $k$-tangle-unbreakable, it suffices to decompose it enough by doubly well-linked separations of certain type that can be found efficiently by local search in unbreakable graphs.
For (2), we have that $k$-tangle-unbreakable hypergraphs have the nice property that if they are not unbreakable, then they have ``small'' witnesses of non-unbreakability, which can be found efficiently by local search.

\section{Preliminaries}
\label{sec:prelis}
In this section we present definitions and preliminary results.

\subsection{Basic definitions}
We start with basic definitions.

\paragraph{Model of computation and $\OO$-notation.}
We assume the standard model of computation in the context of graph algorithms, that is, the word RAM model with $\Theta(\log n)$-bit words, where $n$ is the input length (see e.g. \cite{DBLP:journals/cacm/HopcroftT73} for discussion).
However, we do not abuse this model in any way; we only use elementary operations without any bit-manipulation tricks.
Furthermore, all of our algorithms are deterministic.

We use $\OO$-notation in the standard manner, but let us clarify one convention.
When writing $k^{\OO(f(\cdot))}$ in the context where $k$ is a positive integer, we ignore the special case of $k=1$, and in fact mean $\max(k,2)^{\OO(f(\cdot))}$.
This implies, in particular, that $f(\cdot) \le 2^{\OO(f(\cdot))} \le k^{\OO(f(\cdot))}$ holds always.

\paragraph{Miscellaneous.}
For two integers $a \le b$, we denote by $[a,b]$ the set of integers $\{a,\ldots,b\}$.
When $a > b$, the set $[a,b]$ is empty.
For an integer $a$, we use $[a]$ to denote the set $[1,a]$.

\paragraph{Graphs.}
By a graph we mean an undirected simple graph.
The set of vertices of a graph $G$ are denoted by $V(G)$ and the set of edges by $E(G)$.
We denote the size of a graph by $\|G\| = |V(G)| + |E(G)|$.
We use the convention that an edge between vertices $u$ and $v$ is denoted by an unordered pair $uv \in E(G)$.
For a set of edges $X \subseteq E(G)$, we denote by $V(X) = \bigcup_{uv \in X} \{u,v\}$ the union of their endpoints.

For a set of vertices $X \subseteq V(G)$ we denote by $G[X]$ the subgraph of $G$ induced by $X$, and use $G \setminus X = G[V(G) \setminus X]$.
We may use the notation $G \setminus X$ also in the context where $X$ contains edges or both vertices and edges, in which case it denotes the subgraph of $G$ obtained by deleting those vertices and edges.
We define that a \emph{connected component} of a graph $G$ is an inclusion-wise maximal non-empty set $C \subseteq V(G)$ so that $G[C]$ is connected.
We denote the collection of connected components of $G$ by $\cc(G)$.

The set of neighbors of a vertex $v \in V(G)$ is denoted by $N_G(v)$, and the closed neighborhood of $v$ is denoted by $N_G[v] = N_G(v) \cup \{v\}$.
For a set $X \subseteq V(G)$ we denote $N_G(X) = \bigcup_{v \in X} N_G(v) \setminus \{X\}$ and $N_G[X] = N_G(X) \cup X$.
The set of edges incident to $v$ is denoted by $\ninc_G(v)$, and for a set $X \subseteq V(G)$ we may denote $\ninc_G(X) = \bigcup_{v \in X} \ninc_G(v)$.
We drop the subscript if $G$ is clear from the context.

The \emph{contraction} of an edge $uv \in E(G)$ means inserting a new vertex $w_{uv}$ with $N(w_{uv}) = (N(u) \cup N(v)) \setminus \{u,v\}$ and deleting the vertices $u$ and $v$.
We denote by $G \contr uv$ the graph resulting from contracting $uv$.
A \emph{matching} is a set of edges $M \subseteq E(G)$ with no common endpoints, and we denote by $G \contr M$ the graph resulting from successively contracting all edges in $M$.

\paragraph{Edge cuts.}
An \emph{edge cut} of a graph $G$ is a pair $(A,B)$ with $A,B \subseteq V(G)$, so that $A \cup B = V(G)$ and $A \cap B = \emptyset$.
The \emph{cutset} of the edge cut $(A,B)$, denoted by $E(A,B)$, consists of the set of edges with one endpoint in $A$ and another in $B$.
For two vertices $u,v \in V(G)$, we say that $C \subseteq E(G)$ is an $(u,v)$-cutset if $C$ is the cutset of some edge cut $(A,B)$ with $u \in A$ and $v \in B$.

\paragraph{Vertex cuts and separators.}
A \emph{vertex cut} of a graph $G$ is a pair $(A,B)$ with $A,B \subseteq V(G)$, so that $A \cup B = V(G)$ and there are no edges between $A \setminus B$ and $B \setminus A$.
The \emph{order} of a vertex cut $(A,B)$ is $|A \cap B|$.

Let $X_1,X_2 \subseteq V(G)$ be two sets of vertices.
An \emph{$(X_1,X_2)$-path} is a path in $G$ that intersects both $X_1$ and $X_2$, and may consist of only one vertex.
An \emph{$(X_1,X_2)$-separator} is a set $S \subseteq V(G)$ that intersects every $(X_1,X_2)$-path.
Equivalently, $S$ is an $(X_1,X_2)$-separator, if there exists a vertex cut $(A,B)$ of $G$ so that $X_1 \subseteq A$, $X_2 \subseteq B$, and $S = A \cap B$.
For two vertices $u,v \in V(G)$, an $(\{u\},\{v\})$-separator $S$ is a \emph{proper $(u,v)$-separator} if $\{u,v\} \cap S = \emptyset$.
A set $S \subseteq V(G)$ is a \emph{proper vertex separator} of $G$ if $S$ is a proper $(u,v)$-separator for some pair $u,v \in V(G)$.

For two sets $X_1,X_2 \subseteq V(G)$ we denote by $\flow_G(X_1,X_2)$ the maximum number of vertex-disjoint $(X_1,X_2)$-paths, which by Menger's theorem is equal to the minimum size of an $(X_1,X_2)$-separator.
For two vertices $u \neq v$, we denote by $\flowp_G(u,v)$ the maximum number of internally vertex-disjoint $(u,v)$-paths, including the path $uv$ if $u$ and $v$ are adjacent.
If $u$ and $v$ are non-adjacent, Menger's theorem implies that $\flowp_G(u,v)$ is the minimum size of a proper $(u,v)$-separator, and if they are adjacent, then $\flowp_G(u,v)$ is one plus the minimum size of a proper $(u,v)$-separator in $G \setminus \{uv\}$.

A vertex cut $(A',B')$ is an \emph{orientation} of a vertex cut $(A,B)$ if either (1) $(A',B') = (A,B)$ or (2) $(B',A') = (A,B)$.
Two vertex cuts $(A_1, B_1)$ and $(A_2, B_2)$ are \emph{parallel} if there is an orientation $(A_1',B_1')$ of $(A_1, B_1)$ and an orientation $(A_2',B_2')$ of $(A_2,B_2)$ so that $A_1' \subseteq A_2'$ and $B_2' \subseteq B_1'$.
Otherwise they are \emph{crossing}.

If $(A_1,B_1)$ and $(A_2, B_2)$ are vertex cuts, then also $(A_1 \cap A_2, B_1 \cup B_2)$ and $(A_1 \cup A_2, B_1 \cap B_2)$ are vertex cuts, both of which are parallel with both $(A_1,B_1)$ and $(A_2, B_2)$.
The \emph{submodularity} of vertex cuts means the following property.
\begin{lemma}[Submodularity]
If $(A_1, B_1)$ and $(A_2, B_2)$ are vertex cuts, then
\[|(A_1 \cap A_2) \cap (B_1 \cup B_2)| + |(A_1 \cup A_2) \cap (B_1 \cap B_2)| \le |A_1 \cap B_1| + |A_2 \cap B_2|.\]
\end{lemma}
\begin{proof}
Let $S_1 = A_1 \cap B_1$, $S_2 = A_2 \cap B_2$, $\hat{S} = (A_1 \cap A_2) \cap (B_1 \cup B_2)$, and $\check{S} = (A_1 \cup A_2) \cap (B_1 \cap B_2)$.

Note that the right side is $|S_1| + |S_2| = |S_1 \cup S_2| + |S_1 \cap S_2|$.
Observe also that $\hat{S} \subseteq S_1 \cup S_2$ and $\check{S} \subseteq S_1 \cup S_2$.
Therefore, the left side is $|\hat{S}|+|\check{S}| \le |S_1 \cup S_2| + |\hat{S} \cap \check{S}|$.
However, observe that if $v \in \hat{S} \cap \check{S}$, then $v \in A_1 \cap A_2 \cap B_1 \cap B_2$, and therefore $v \in S_1 \cap S_2$, showing that the right side is always at least the left side.
\end{proof}

\paragraph{Element cuts.}
Let $G$ be a graph and $U \subseteq V(G)$.
For two vertices $u,v \in U$, we say that a set $S \subseteq E(G) \cup V(G) \setminus U$ is an $U$-element $(u,v)$-cutset if $u$ and $v$ are in different connected components of the graph $G \setminus S$.

\paragraph{Trees.}
A tree is a connected acyclic graph.
To more easily distinguish trees from graphs, the vertices of a tree may be called \emph{nodes}.
The \emph{leaves} of a tree $T$ are its nodes of degree $\le 1$, and are denoted by $\leafs(T)$.
Note that if $|V(T)| = 1$, then the single node is a leaf.
The \emph{internal nodes} of a tree $T$ are the non-leaf nodes, and are denoted by $\vint(T) = V(T) \setminus \leafs(T)$.
The \emph{internal edges} of a tree are the edges whose both endpoints are internal nodes, and are denoted by $\eint(T) = \{uv \in E(T) : u,v \in \vint(T)\}$.

A \emph{rooted tree} is a tree $T$ with one node $\troot(T) \in V(T)$ selected as a root.
Unintuitively, we allow the root of a rooted tree to be a leaf.
If $T$ is a rooted tree and $t \in V(T)$, the neighbors of $t$ that are further away from the root than $t$ are called the \emph{children} of $t$ and denoted by $\children(t)$, and if $t \neq \troot(T)$, the neighbor of $t$ that is closer to the root is called the \emph{parent} of $t$ and denoted by $\parent(t)$.
A node $x$ is a \emph{descendant} of a node $y$ if the unique path from $x$ to the root contains $y$.
Note that each node is a descendant of itself.
By a \emph{strict descendant} we mean a descendant that is not the node itself.
A \emph{prefix} of a rooted tree is a set $P \subseteq V(T)$ so that $T[P]$ is connected and contains the root.

In the context of trees, in addition to edges we will work with \emph{orientations} of edges.
For an edge $st \in E(T)$, the ordered pairs $\vec{st}$ and $\vec{ts}$ are \emph{orientations} of $st$.
In particular, $\vec{st}$ should be seen as the edge $st$ but oriented to point towards $t$.
We denote by $\vec{E}(T)$ the set of all orientations of edges of $T$.

\paragraph{Hypergraphs.}
In most of the technical sections of this paper we will work with hypergraphs instead of graphs.

A \emph{hypergraph} $G$ consists of a set of vertices $V(G)$, a set of hyperedges $E(G)$, and a mapping $V \colon E \to 2^{V(G)}$ that maps each hyperedge $e \in E$ to a subset $V(e) \subseteq V(G)$.
We require that $V(G) = \bigcup_{e \in E(G)} V(e)$ always holds, i.e., every vertex of a hypergraph must be in the vertex set of at least one hyperedge.
For a set $A \subseteq E(G)$, we denote $V(A) = \bigcup_{e \in A} V(e)$.
We denote the size of a hypergraph by $\|G\| = |E(G)| + \sum_{e \in E(G)} |V(e)|$.

The \emph{rank} of a hyperedge $e$ is $|V(e)|$, and the \emph{rank} $\rank(G)$ of a hypergraph $G$ is the maximum rank of its hyperedges.
A hypergraph is \emph{normal} if its every vertex occurs in at least two distinct hyperedges.

We will work only with hypergraphs whose rank is bounded by $2k$, where $k \ge 1$ is the original input parameter $k$.
Furthermore, we can essentially always assume that the hypergraphs we work with are normal, but we write this out only when we actually use this fact.

Our definition of hypergraphs does not allow isolated vertices, but isolated vertices can be simulated by having hyperedges of rank $1$.
We allow hyperedges of rank $0$, but they will not play any role in this paper.
We also allow that $V(e_1) = V(e_2)$ for two distinct hyperedges $e_1$ and $e_2$.
The \emph{multiplicity} of a hyperedge $e \in E(G)$ is the number of hyperedges $e' \in E(G)$ so that $V(e') = V(e)$.
This counts also $e$ itself, so the multiplicity is always at least one.
The \emph{multiplicity} of a hypergraph is the maximum multiplicity of its hyperedges.

For an integer $p \ge 1$, a \emph{multiplicity-$p$-violator} in $G$ is a set $A \subseteq E(G)$ so that $|A| = p+1$ and $V(e_1) = V(e_2)$ for all $e_1,e_2 \in A$.
The \emph{rank} of a multiplicity-$p$-violator $A$ is the rank of the hyperedges in $A$.
By a \emph{multiplicity-violator} we mean a multiplicity-$p$-violator for some $p \ge 1$.

For a vertex $v \in V(G)$, we define $\ninc_G(v) \subseteq E(G)$ to be the set of hyperedges incident to $v$, i.e., $\ninc_G(v) = \{e \in E(G) : v \in V(e)\}$.
We drop the subscript $G$ if it is clear from the context.

The \emph{primal graph} $\primal(G)$ of a hypergraph $G$ is the graph $\primal(G)$ with $V(\primal(G)) = V(G)$ and having an edge between $u,v \in V(\primal(G))$ if there exists $e \in E(G)$ with $\{u,v\} \subseteq V(e)$.
Note that $\|\primal(G)\| = \OO(\rank(G) \cdot \|G\|)$.

Let $G$ be a graph.
The \emph{hypergraph} $\hyperg(G)$ of $G$ is the hypergraph $\hyperg(G)$ with $E(\hyperg(G)) = V(G) \cup E(G)$, and for each $v \in V(G)$, $V(v) = \{v\}$, and for each $uv \in E(G)$, $V(uv) = \{u,v\}$.
Note that $\primal(\hyperg(G)) = G$, $\rank(\hyperg(G)) \le 2$, $\hyperg(G)$ is normal, the multiplicity of $\hyperg(G)$ is one, and $\|\hyperg(G)\| = \OO(\|G\|)$.

\paragraph{Separations of hypergraphs.}
A \emph{separation} of a hypergraph $G$ is a pair $(A,B)$ with $A,B \subseteq E(G)$ so that $A \cup B = E(G)$ and $A \cap B = \emptyset$.
The \emph{order} of a separation $(A,B)$ is $|V(A) \cap V(B)|$.
A separation $(A',B')$ is an \emph{orientation} of a separation $(A,B)$ if either $(A,B) = (A',B')$ or $(A,B) = (B',A')$.

Note that a separation $(A,B)$ is uniquely defined just by the set $A \subseteq E(G)$, as it must be that $B = E(G) \setminus A$.
We will use the notation $\co{A} = E(G) \setminus A$.
We denote $\bd_G(A) = \bd_G(\co{A}) = V(A) \cap V(\co{A})$, and $\bdc_G(A) = |\bd_G(A)| = |V(A) \cap V(\co{A})|$, and drop the subscript $G$ if it is clear from the context.
For a hyperedge $e \in E(G)$, we may use $\bdc(e)$ and $\bd(e)$ as shorthands for $\bdc(\{e\})$ and $\bd(\{e\})$.
Note that if $G$ is normal, then $\bd(e) = V(e)$ for all $e \in E(G)$.

Note that the separations of a hypergraph $G$ correspond to the vertex cuts of $\primal(G)$ in the sense that if $(A,\co{A})$ is a separation of $G$, then $(V(A),V(\co{A}))$ is a vertex cut of $\primal(G)$, and if $(A,B)$ is a vertex cut of $\primal(G)$, then there exists a separation $(C,\co{C})$ of $G$ so that $V(C) \subseteq A$ and $V(\co{C}) \subseteq B$.

Two separations $(A,\co{A})$ and $(B,\co{B})$ are \emph{parallel} if there is an orientation $(A',\co{A'})$ of $(A,\co{A})$ and an orientation $(B',\co{B'})$ of $(B,\co{B})$ so that $A' \subseteq B'$.
If they are not parallel, then they are \emph{crossing}.

The function $\bdc \colon 2^{E(G)} \to \mathbb{Z}_{\ge 0}$ is \emph{symmetric} and \emph{submodular} (see e.g.~\cite{RobertsonS91}), meaning that it satisfies
\begin{enumerate}
\item $\bdc(A) = \bdc(\co{A})$ for all $A \subseteq E(G)$ (symmetry), and
\item $\bdc(A \cup B) + \bdc(A \cap B) \le \bdc(A) + \bdc(B)$ for all $A,B \subseteq E(G)$ (submodularity).
\end{enumerate}

\paragraph{Transformations of hypergraphs.}
Let $G$ be a hypergraph and $A \subsetneq E(G)$.
We denote by $G \rescliqs A$ the hypergraph $G \rescliqs A$ with
\begin{enumerate}
\item $V(G \rescliqs A) = V(\co{A})$, and
\item $E(G \rescliqs A) = \co{A} \cup \{e_A\}$, where $V(e_A) = \bd(A)$.
\end{enumerate}

If $B \subseteq E(G \rescliqs A)$ is a non-empty set, we define $B \orescliqs A = B$ if $e_A \notin B$ and $B \orescliqs A = (B \setminus \{e_A\}) \cup A$ if $e_A \in B$.
Note that now, $\bd_G(B \orescliqs A) = \bd_{G \rescliqs A}(B)$.
Also, if $(B,\co{B})$ is a separation of $G \rescliqs A$, then $(B \orescliqs A, \co{B} \orescliqs A)$ is a separation of $G$.

When $\compset$ is a collection of disjoint non-empty sets, we define $G \rescliqs \compset$ to be the hypergraph resulting from applying $G \rescliqs C$ successively for each $C \in \compset$.
In particular, 
\begin{enumerate}
\item $V(G \rescliqs \compset) = \bigcap_{C \in \compset} V(\co{C})$, and
\item $E(G \rescliqs \compset) = \left(\bigcap_{C \in \compset} \co{C}\right) \cup \{e_C : C \in \compset\}$, where $V(e_C) = \bd(C)$ for each $C \in \compset$.
\end{enumerate}

For $A \subseteq E(G \rescliqs \compset)$, we define $A \orescliqs \compset = (A \setminus \{e_C : C \in \compset\}) \cup \bigcup \{C \in \compset : e_C \in A\}$.
Again, we have that $\bd_G(A \orescliqs \compset) = \bd_{G \rescliqs \compset}(A)$.

\subsection{Decompositions}
We will work with \emph{tree decompositions} of graphs, and with \emph{superbranch decompositions} of hypergraphs.
These are essentially two different formulations of the same object.

\paragraph{Tree decompositions of graphs.}
A tree decomposition of a graph $G$ is a pair $\Tc = (T,\bag)$, where $T$ is a tree, and $\bag \colon V(T) \rightarrow 2^{V(G)}$ is a mapping that satisfies
\begin{enumerate}
\item for each $uv \in E(G)$, there exists $t \in V(T)$ with $\{u,v\} \subseteq \bag(t)$, and\label{tddef:cond1}
\item for each $v \in V(G)$, the set of nodes $\{t \in V(T) : v \in \bag(t)\}$ induces a non-empty connected subtree of $T$.\label{tddef:cond2}
\end{enumerate}

The condition of \Cref{tddef:cond1} is called the \emph{edge condition} of tree decompositions, and the condition of \Cref{tddef:cond2} the \emph{vertex condition} of tree decompositions.

The \emph{adhesion} at the edge $st \in E(T)$ is the set $\adh(st) = \bag(s) \cap \bag(t)$.
For two nodes $x,y \in V(T)$, a set $A \subseteq V(G)$ is an \emph{$(x,y)$-adhesion} if $A = \adh(st)$ for some edge $st$ on the unique $(x,y)$-path in $T$.
The \emph{adhesion size} of $\Tc$, denoted by $\adhsize(\Tc)$, is the maximum size of an adhesion of $\Tc$.
If $\Tc$ has only one bag, then we define that $\adhsize(\Tc) = -\infty$.

We denote the total size of a tree decomposition $\Tc$ by $\|\Tc\| = |V(T)| + \sum_{t \in V(T)} |\bag(t)|$.

A \emph{rooted tree decomposition} is a tree decomposition $(T,\bag)$ where $T$ is a rooted tree.

\paragraph{Superbranch decompositions of hypergraphs.}
A superbranch decomposition of a hypergraph is defined similarly to the classical definition of a branch decomposition of a hypergraph~\cite{RobertsonS91}, but its internal nodes can have arbitrary high degree.

We say that a tree $T$ is \emph{supercubic} if all of its internal nodes have degree at least $3$.
Let $G$ be a hypergraph.
A \emph{superbranch decomposition} of $G$ is a pair $\Tc = (T,\lmap)$, where $T$ is a supercubic tree and $\lmap \colon \leafs(T) \to E(G)$ is a bijection from the leaves of $T$ to $E(G)$.
We note that superbranch decompositions are defined even when $|E(G)| = 1$ or $|E(G)| = 0$, in these cases the tree consisting only of one node or being the empty tree, respectively.
We now introduce several definitions to work with superbranch decompositions.

Let $\vec{st} \in \vec{E}(T)$ be an oriented edge of $T$.
We denote by $\leafs_T(\vec{st}) \subseteq \leafs(T)$ the set of leaves of $T$ that are closer to $s$ than $t$.
Then, we denote by $\lmap(\vec{st}) = \{\lmap(t) : t \in \leafs_T(\vec{st})\}$ the set of hyperedges of $G$ associated with the leaves of $T$ that are closer to $s$ than $t$.

Note that for all $st \in E(T)$, the pair $(\lmap(\vec{st}), \lmap(\vec{ts}))$ is a separation of $G$, and these separations are parallel with each other.
We define the \emph{adhesion} at $st$ to be $\adh(st) = \bd(\lmap(\vec{st})) = \bd(\lmap(\vec{ts}))$.
The \emph{adhesion size} of a superbranch decomposition, denoted by $\adhsize(\Tc)$, is the maximum cardinality of an adhesion of it.
Note that if $G$ is normal, then $\adhsize(\Tc) \ge \rank(G)$.
If $T$ has no edges, then we define that $\adhsize(\Tc) = -\infty$.

We say that a separation of form $(\lmap(\vec{st}), \lmap(\vec{ts}))$ is a \emph{separation} of $\Tc$, and if $st \in \eint(T)$, then $(\lmap(\vec{st}), \lmap(\vec{ts}))$ is an \emph{internal separation} of $\Tc$.
The motivation for distinguishing internal separations of $\Tc$ from all separations of $\Tc$ is that if $(\lmap(\vec{st}), \lmap(\vec{ts}))$ is a separation of $\Tc$ but not an internal separation, then $(\lmap(\vec{st}), \lmap(\vec{ts}))$ is an orientation of a separation of form $(\{e\},E(G) \setminus \{e\})$ for some $e \in E(G)$, in which case $(\lmap(\vec{st}), \lmap(\vec{ts}))$ is in fact a separation of every superbranch decomposition of $\Tc$.

Let $t \in \vint(T)$ be an internal node of $T$.
We define the \emph{torso} of $t$ to be the hypergraph $\torso(t)$, whose hyperedges correspond to the adhesions of the edges of $T$ incident to $t$.
In other words, $E(\torso(t)) = \{e_s : s \in N_T(t)\}$ and $V(e_s) = \adh(st)$ for all $e_s \in E(\torso(t))$.
Note that $\torso(t)$ is normal and $|E(\torso(t))| \ge 3$.
For $t \in \vint(T)$, we use the shorthand $V(t) = V(\torso(t))$.
We observe also that $\torso(t) = G \rescliqs \{\lmap(\vec{st}) : s \in N_T(t)\}$.
In light of this, for a set $A \subseteq E(\torso(t))$ we denote by $A \orescliqs \Tc = A \orescliqs \{\lmap(\vec{st}) : s \in N_T(t)\}$ the corresponding set in $G$.
When speaking about the ``torsos of a superbranch decomposition'', we mean all torsos $\torso(t)$ for $t \in \vint(T)$, recalling that $\torso(t)$ is not defined when $t$ is a leaf.

In the context where we are working with multiple superbranch decompositions/tree decompositions, we may put $\Tc$ to the subscript in the notations $\adh(st)$, $\torso(t)$, and $V(t)$ to clarify which decomposition we mean.

We define the \emph{total size} of a superbranch decomposition $\Tc = (T,\lmap)$ to be 
\[\|\Tc\| = \sum_{t \in \vint(T)} \|\torso(t)\| + \sum_{\ell \in \leafs(T)} |V(\lmap(\ell))|.\]
Note that $\|\Tc\| \ge \|G\|$ and $\|\Tc\| = \OO(\|G\| + |E(G)| \cdot \adhsize(\Tc))$.
Also, if $G$ is normal, then $\|\Tc\| = \OO(|E(G)| \cdot \adhsize(\Tc))$.

We observe the following relation between superbranch decompositions and tree decompositions.

\begin{observation}
If $(T,\lmap)$ is a superbranch decomposition of a hypergraph $G$, then $(T,\bag)$, where $\bag(t) = V(t)$ for all internal nodes $t \in \vint(T)$ and $\bag(\ell) = V(\lmap(\ell))$ for all leaves $\ell \in \leafs(T)$, is a tree decomposition of $\primal(G)$.
\end{observation}

Note that the adhesions of $(T,\bag)$ are the same as the adhesions of $(T,\lmap)$.
When thinking of superbranch decompositions, it may be useful to think about this tree decomposition.

A \emph{rooted superbranch decomposition} is a superbranch decomposition $(T,\lmap)$ where $T$ is a rooted tree.

\subsection{Unbreakability and leannes}
\label{subsec:defs:ubandlean}
We then define our notions of unbreakability and leannes.

\paragraph{Unbreakability.}
Let us first define unbreakability for graphs.
Let $G$ be a graph and $s,k$ integers.
A set $X \subseteq V(G)$ is \emph{$(s,k)$-unbreakable} if there is no vertex cut $(A,B)$ of $G$ of order $<k$ so that $|A \cap X| \ge s$ and $|B \cap X| \ge s$.
A graph $G$ is $(s,k)$-unbreakable if $V(G)$ is $(s,k)$-unbreakable.

Note that if $k \le 0$, then every set is $(s,k)$-unbreakable.
Furthermore, if $s < k$, then no set of size $\ge s$ is $(s,k)$-unbreakable.
Therefore, these definitions make sense only for $s \ge k \ge 1$.

We then define unbreakability for hypergraphs.
Let $G$ be a hypergraph and $s,k$ integers.
A set $X \subseteq V(G)$ is $(s,k)$-unbreakable if there is no separation $(A,\co{A})$ of $G$ of order $<k$ so that $|V(A) \cap X| \ge s$ and $|V(\co{A}) \cap X| \ge s$.
A hypergraph $G$ is $(s,k)$-unbreakable if $V(G)$ is $(s,k)$-unbreakable.

Note that $(s,k)$-unbreakability implies $(s',k')$-unbreakability for all $s' \ge s$ and $k' \le k$.
Also, if $X \subseteq V(G)$ is $(s,k)$-unbreakable, then every subset of $X$ is also $(s,k)$-unbreakable.

If $G$ is a hypergraph, $X \subseteq V(G)$, and $X$ is $(s,k)$-unbreakable, then $X$ is $(s+k,k)$-unbreakable in the graph $\primal(G)$.
If $X$ is $(s,k)$-unbreakable in $\primal(G)$, then $X$ is also $(s,k)$-unbreakable in $G$.

\paragraph{Unbreakable decompositions.}
A tree decomposition $\Tc = (T,\bag)$ of a graph is $(s,k)$-unbreakable if $\bag(t)$ is $(s,k)$-unbreakable for every $t \in V(T)$.
A superbranch decomposition $\Tc = (T,\lmap)$ of a hypergraph is $(s,k)$-unbreakable if $V(t)$ is $(s,k)$-unbreakable for every $t \in \vint(T)$.

We note that we will work with both $(s,k)$-unbreakable superbranch decompositions, and superbranch decompositions whose torsos are $(s,k)$-unbreakable.
The former notion implies the latter, but the latter notion does not necessarily imply the former.

\paragraph{Lean tree decompositions.}
Following the definition of \cite{DBLP:journals/siamdm/CarmesinDHH14}, we define that a tree decomposition $\Tc = (T,\bag)$ is \emph{$k$-lean} if
\begin{enumerate}
\item $\adhsize(\Tc) < k$, and 
\item for all $t_1,t_2 \in V(T)$, $X_1 \subseteq \bag(t_1)$, and $X_2 \subseteq \bag(t_2)$, if there exists an $(X_1,X_2)$-separator of size $k' < \min(|X_1|, |X_2|, k)$, then there exists a $(t_1,t_2)$-adhesion of size $\le k'$.\label{klean:prop2}
\end{enumerate}

In particular, the adhesions of a $k$-lean tree decomposition capture all separators of size $<k$ between subsets of bags.
The case of $t_1 = t_2$ of the property of \Cref{klean:prop2} implies that for every $t \in V(T)$, if $X_1,X_2 \subseteq \bag(t)$, then $\flow(X_1,X_2) \ge \min(|X_1|,|X_2|,k)$.
In particular, this implies the following observation.

\begin{observation}
A $k$-lean tree decomposition is $(i,i)$-unbreakable for all $i \le k$.
\end{observation}

When $\Tc = (T,\bag)$ is a tree decomposition of a graph $G$, we define that a \emph{non-$k$-lean-witness}\footnote{Similar objects were called lean-witnesses in~\cite{DBLP:journals/talg/CyganKLPPSW21} and breakable witnesses in~\cite{DBLP:journals/corr/abs-2408-09368}.} for $\Tc$ is a quadruple $(A,B,t_1,t_2)$ so that

\begin{enumerate}
\item $(A,B)$ is a vertex cut of $G$ of order $|A \cap B|<k$,
\item $t_1,t_2 \in V(T)$, 
\item $|\bag(t_1) \cap A| > |A \cap B|$, $|\bag(t_2) \cap B| > |A \cap B|$, and
\item every $(t_1,t_2)$-adhesion has size $> |A \cap B|$.
\end{enumerate}

It is not hard to observe that if a tree decomposition violates \Cref{klean:prop2} of the definition of a $k$-lean tree decomposition, then there exists a non-$k$-lean-witness for it.
In particular, we have the following.

\begin{observation}
A tree decomposition $\Tc$ is $k$-lean if and only if $\adhsize(\Tc) < k$ and there are no non-$k$-lean-witnesses for $\Tc$.
\end{observation}

\paragraph{Gomory-Hu trees.}
When $G$ is a graph and $U \subseteq V(G)$, an \emph{element connectivity $k$-Gomory-Hu tree} of $(G,U)$ is a triple $(T,\gamma,\alpha)$, where $T$ is a tree, $\gamma \colon U \to V(T)$ maps vertices in $U$ to nodes of $T$, and $\alpha \colon E(T) \to 2^{(V(G) \setminus U) \cup E(G)}$ maps edges of $T$ to subsets $\alpha(e)$ of $(V(G) \setminus U) \cup E(G)$ with $|\alpha(e)|<k$ and so that
\begin{enumerate}
\item for all $u,v \in U$ and $e \in E(T)$ so that $e$ is on the unique $(\gamma(u),\gamma(v))$-path in $T$, the set $\alpha(e)$ is an $U$-element $(u,v)$-cutset, and
\item for all $u,v \in U$, if there exists an $U$-element $(u,v)$-cutset of size $k'<k$, then there exists $e \in E(T)$ on the unique $(\gamma(u),\gamma(v))$-path in $T$ so that $\alpha(e)$ is an $U$-element $(u,v)$-cutset of size $|\alpha(e)| \le k'$.
\end{enumerate}

A \emph{$k$-Gomory-Hu} tree of a graph $G$ is an element connectivity $k$-Gomory-Hu tree of $(G,V(G))$.

\subsection{Graph-theoretical concepts}
\label{subsec:defs:graphtheory}
We then define some key graph-theoretical concepts for hypergraphs.

\paragraph{Well-linked sets.}
Let $G$ be a hypergraph and $A \subseteq E(G)$.
A \emph{bipartition} of $A$ is a pair $(B_1,B_2)$ so that $B_1 \cup B_2 = A$ and $B_1 \cap B_2 = \emptyset$.
We allow $B_1$ or $B_2$ to be empty (and when $A$ is empty, they both are empty).
A set $A \subseteq E(G)$ is \emph{well-linked} if for all bipartitions $(B_1,B_2)$ of $A$, either $\bdc(B_1) \ge \bdc(A)$ or $\bdc(B_2) \ge \bdc(A)$.
A separation $(A,\co{A})$ of $G$ is \emph{doubly well-linked} if both $A$ and $\co{A}$ are well-linked.

Let $k$ be an integer.
A set $A \subseteq E(G)$ is \emph{$k$-well-linked} if for all bipartitions $(B_1,B_2)$ of $A$ either (1) $\bdc(B_1) \ge \bdc(A)$, (2) $\bdc(B_2) \ge \bdc(A)$, or (3) $\bdc(B_1) \ge k$ and $\bdc(B_2) \ge k$.
Note that if $A$ is well-linked then $A$ is $k$-well-linked for every $k$.
Also, if $\bdc(A) \le k$, then $k$-well-linkedness of $A$ is equivalent to well-linkedness of $A$.

A separation $(A,\co{A})$ is \emph{mixed-$k$-well-linked} if one of $A$ and $\co{A}$ is well-linked, and other is $k$-well-linked.
A set $A \subseteq E(G)$ is \emph{$k$-better-linked} if either (1) $A$ is well-linked, or (2) $A$ is $k$-well-linked and $\co{A}$ is well-linked.
In particular, if $(A,\co{A})$ is mixed-$k$-well-linked separation, then both $A$ and $\co{A}$ are $k$-better-linked.

The point of the definition of $k$-better-linked sets is that they turn out to be as useful as well-linked sets for the purpose of uncrossing separations of order $<k$.

A \emph{tripartition} of a set $A \subseteq E(G)$ is a triple $(B_1,B_2,B_3)$ so that $B_1 \cup B_2 \cup B_3 = A$ and $B_1$, $B_2$, and $B_3$ are disjoint.
We again allow one or more of the sets $B_i$ to be empty.
A set $A \subseteq E(G)$ is \emph{tri-well-linked} if for all tripartitions $(B_1,B_2,B_3)$ of $A$, there is $i \in [3]$ so that $\bdc(B_i) \ge \bdc(A)$.
A separation $(A,\co{A})$ is \emph{doubly tri-well-linked} if both $A$ and $\co{A}$ are tri-well-linked.

\paragraph{Internal components.}
Let $G$ be a hypergraph and $A \subseteq E(G)$ a non-empty set.
We define that $A$ is \emph{internally connected} if there is no bipartition $(B_1,B_2)$ of $A$ so that both $B_1$ and $B_2$ are non-empty and $\bd(B_i) \subseteq \bd(A)$ for both $i \in [2]$.
Note that if $|A| = 1$, then $A$ is internally connected.
An \emph{internal component} of $A$ is an inclusion-wise maximal subset $C \subseteq A$ so that $C$ is internally connected.

An intuitive way of thinking of the internal components of $A$ is that each internal component $C$ consist of either a single hyperedge $C = \{e\}$ with $V(e) \subseteq \bd(A)$, or corresponds to a connected component $C'$ of the graph $\primal(G)[V(A) \setminus \bd(A)]$, in the sense that $C$ contains all hyperedges $e$ so that $V(e)$ intersects $C'$.

With this in mind, it is easy to observe that if $C$ is an internal component of $A$, then $\bd(C) \subseteq \bd(A)$.
We also observe that the internal components of $A$ form a partition of $A$, i.e., there is an unique division of $A$ into its internal components.

We say that a set $A \subseteq E(G)$ is \emph{semi-internally connected} if for every internal component $C$ of $A$ it holds that $\bd(C) = \bd(A)$.
Note that the maximal semi-internally connected subsets of $A$ form a partition of $A$ into at most $2^{\bdc(A)}$ semi-internally connected sets $C$ with $\bd(C) \subseteq \bd(A)$.

\paragraph{Tangles.}
Tangles were introduced by Robertson and Seymour~\cite{RobertsonS91}.
Let $G$ be a hypergraph and $k$ be an integer.
A \emph{tangle} of order $k$ of $G$ is a family $\tang$ of subsets of $E(G)$ that satisfies the following four conditions, called the four \emph{tangle axioms}:

\begin{enumerate}
\item $\bdc(A) < k$ for all $A \in \tang$,
\item if $A \subseteq E(G)$ and $\bdc(A) < k$, then either $A \in \tang$ or $\co{A} \in \tang$,
\item if $A,B,C \in \tang$, then $A \cup B \cup C \neq E(G)$, and
\item for all $e \in E(G)$, it holds that $E(G) \setminus \{e\} \notin \tang$.
\end{enumerate}

A \emph{tangle} of $G$ is a tangle of order $k$ of $G$ for some $k$.
One should think of a tangle of order $k$ as an orientation of all separations of order $<k$ that satisfies certain consistency conditions.
In particular, if $A \in \tang$, then we think that $\tang$ orients the separation $(A,\co{A})$ so that $A$ is the ``small side'' and $\co{A}$ is the ``big side''.

For example, if $G$ is a hypergraph and $C \subseteq V(G)$ is a connected component of $\primal(G)$, then $C$ defines a tangle of order~$1$ by letting
\[\tang = \{A \subseteq E(G) : \bdc(A) = 0 \text{ and } V(A) \cap C = \emptyset\}.\]
Also, for every hyperedge $e \in E(G)$ we can construct a tangle of order $\bdc(e)$ by setting 
\[\tang = \{A \subseteq E(G) : e \notin A \text{ and } \bdc(A)<\bdc(e)\}.\]
This is called the \emph{$e$-tangle} of $G$.

A tangle $\tang_1$ is a \emph{truncation} of a tangle $\tang_2$ if $\tang_1 \subseteq \tang_2$.
If a tangle $\tang$ has order $k$, then for every $k' \le k$, the \emph{truncation of $\tang$ of order $k'$}, defined as
\[\tang_{k'} = \{A \subseteq E(G) : \bdc(A) < k' \text{ and } A \in \tang\}\]
is a tangle of order $k'$.
Note that the example of a tangle of order $1$ defined by a connected component $C$ of $\primal(G)$ can be obtained as a truncation of any $e$-tangle for $e \in E(G)$ with $V(e) \subseteq C$ and $\bdc(e) \ge 1$.

It follows from these definitions that a hypergraph with $|E(G)| \le 1$ has no tangles of positive order.
All hypergraphs, including ones with $|E(G)| \le 1$, have the tangle $\tang = \emptyset$ of order $0$.

Let us then observe that a normal unbreakable hypergraph contains a tangle defined by vertex-cardinalities.

\begin{lemma}
\label{lem:vertexcardtangle}
Let $G$ be a normal hypergraph, and $s \ge k \ge 1$ integers so that $G$ is $(s,k)$-unbreakable and has $|V(G)| \ge 3s$.
Then, the family $\tang = \{A \subseteq E(G) : |V(A)| \le |V(\co{A})| \text{ and } \bdc(A) < k\}$ is a tangle of order $k$ of $G$.
\end{lemma}
\begin{proof}
By definition, we have that $\bdc(A) < k$ for all $A \in \tang$.
Also, clearly either $|V(A)| \le |V(\co{A})|$ or $|V(\co{A})| \le |V(A)|$ for all $A \subseteq E(G)$, so also the second tangle axiom holds.

For the third axiom, assume there exists $B_1,B_2,B_3 \in \tang$ with $B_1 \cup B_2 \cup B_3 = E(G)$.
For any $i \in [3]$, it cannot be that $|V(B_i)| \ge s$, as then also $|V(\co{B_i})| \ge s$, which would contradict the $(s,k)$-unbreakability of $G$.
Therefore $|V(B_i)| < s$ for all $i \in [3]$.
However, this is a contradiction, because every vertex of $G$ is in at least one of $V(B_i)$ and $|V(G)| \ge 3s$.

For the fourth axiom, assume there exists $e \in E(G)$ so that $|V(E(G) \setminus \{e\})| \le |V(e)|$.
Because $G$ is normal, this implies $V(e) = V(G)$.
However, then $\bdc(e) = |V(G)| > s \ge k$, so $E(G) \setminus \{e\} \notin \tang$.
\end{proof}

If the tangle of order $k$ defined in \Cref{lem:vertexcardtangle} exists, we call it the \emph{vertex-cardinality-tangle of order $k$}.

A separation $(A,\co{A})$ of a hypergraph $G$ \emph{distinguishes} a tangle $\tang_1$ from a tangle $\tang_2$ if $A \in \tang_1$ and $\co{A} \in \tang_2$.
If $\tang_1$ and $\tang_2$ are two tangles, then either (1) $\tang_1$ is a truncation of $\tang_2$, (2) $\tang_2$ is a truncation of $\tang_1$, or (3) there exists a separation $(A,B)$ that distinguishes $\tang_1$ from $\tang_2$.
We say that a hypergraph $G$ is \emph{$k$-tangle-unbreakable} if there is no separation of order $<k$ that distinguishes two tangles.
In particular, then for all $k' \le k$ there is at most one tangle of order $k'$, and the tangles of order $\le k$ can be totally ordered by the truncation relation.

\paragraph{Branchwidth.}
A tree $T$ is \emph{cubic} if all of its internal nodes have degree $3$.
A \emph{branch decomposition} of a hypergraph $G$ is a superbranch decomposition $\Tc = (T,\lmap)$ of $G$ where $T$ is a cubic tree.
When $|E(G)| \ge 2$, the \emph{width} of a branch decomposition $\Tc$ of $G$ is $\adhsize(\Tc)$, and when $|E(G)| \le 1$, we define the width of $\Tc$ to be $0$.
The \emph{branchwidth} of $G$, denoted by $\bw(G)$, is the minimum width of a branch decomposition of $G$.

The following lemma by Robertson and Seymour~\cite{RobertsonS91}\footnote{The definition of tangles in~\cite{RobertsonS91} is superficially different from our definition. See for example~\cite{OumS07} for a statement of \Cref{lem:tanglebwduality} that matches our definitions.} gives a min-max relation between tangles and branchwidth.

\begin{lemma}[\cite{RobertsonS91}]
\label{lem:tanglebwduality}
The branchwidth of a hypergraph is equal to the maximum order of a tangle of it.
\end{lemma}

\subsection{The Nagamochi-Ibaraki sparsifier}
\label{subsec:nagamochiibaraki}
Let $G$ be a graph and $k \ge 1$ an integer.
We define that a \emph{$k$-sparsifier} of $G$ is a subgraph $G'$ of $G$ so that $V(G') = V(G)$ and for any set $S \subseteq V(G)$ of size $|S| < k$, it holds that $\cc(G - S) = \cc(G' - S)$.
Equivalently, if $(A,B)$ is a vertex cut of $G'$ of order $<k$, then $(A,B)$ is a vertex cut of $G$.
We use the following result of Nagamochi and Ibaraki for computing $k$-sparsifiers with small number of edges.

\begin{lemma}[\cite{DBLP:journals/algorithmica/NagamochiI92,DBLP:books/cu/NI2008}]
\label{the:nispars}
There is an algorithm that, given a graph $G$, in time $\OO(\|G\|)$ returns a $k$-sparsifier $G'$ of $G$ with $|E(G')| \le k \cdot |V(G)|$.
\end{lemma}

The proof of the formulation of \Cref{the:nispars} presented here is given explicitly in~\cite[Section~2.2]{DBLP:books/cu/NI2008}, and is implicit in~\cite{DBLP:journals/algorithmica/NagamochiI92}.

We observe that if a $k$-sparsifier removes an edge $uv$, then it must retain a collection of $k$ internally disjoint $(u,v)$-paths.

\begin{lemma}
\label{lem:delhighconn}
Let $G'$ be a $k$-sparsifier of $G$.
If $uv \in E(G) \setminus E(G')$, then $\flowp_{G'}(u, v) \ge k$.
\end{lemma}
\begin{proof}
Suppose $\flowp_{G'}(u,v) < k$, and let $S \subseteq V(G) \setminus \{u,v\}$ be a proper $(u,v)$-separator of size $|S| < k$ in $G'$.
Now $u$ and $v$ are in different connected components of $G' \setminus S$, but in the same connected components of $G \setminus S$, which contradicts the definition of a $k$-sparsifier.
\end{proof}

\subsection{Basic data structures}
We then discuss the data structures that we use to represent graphs, hypergraphs, tree decompositions, and superbranch decompositions.
In particular, in this subsection we define what it means to give these objects as inputs to algorithms, and what it means to manipulate these objects when they are already stored.
We do not discuss any advanced techniques, but rather the purpose of this subsection is to fix definitions and recall what basic operations are available in the context of deterministic linear-time graph algorithms.

By a \emph{linked list} we mean a doubly linked list, with each entry storing pointers to the preceeding and following entries.
In particular, linked lists support insertion to the end in $\OO(1)$ time, and deleting an entry, given a pointer to it, in $\OO(1)$ time.

\paragraph{Representation of hypergraphs.}
We define that a representation of a hypergraph $G$ consists of
\begin{enumerate}
\item a linked list storing $V(G)$, where each entry corresponding to $v \in V(G)$ stores a pointer to the \emph{incidence list} of $v$,
\item for each $v \in V(G)$, the incidence list of $v$, storing $\ninc(v)$ by pointers to the hyperedges in $\ninc(v)$,
\item a linked list storing $E(G)$, where each entry corresponding to $e \in E(G)$ stores $V(e)$ as a linked list, with each entry of the linked list containing a pointer to $v \in V(e)$ and a pointer to the entry of the incidence list of $v$ that stores a pointer to $e$.
\end{enumerate}

We assume that the linked lists storing $V(G)$ and $E(G)$ have the capacity to store some extra information associated with each vertex and hyperedge of $G$.
Note that a representation of $G$ takes $\OO(\|G\|)$ space.
We then note the basic operations that can be efficiently implemented with this representation of hypergraphs.

\begin{lemma}
\label{lem:hypergraph_impl}
The representation of hypergraphs supports the following operations
\begin{enumerate}
\item add a hyperedge $e$ with $V(e) = \emptyset$ and return a pointer to $e$, in $\OO(1)$ time,\label{lem:hypergraph_impl:item1}
\item given a pointer to a vertex $v \in V(G)$ and a hyperedge $e \in E(G)$ with $v \notin V(e)$, add $v$ to $V(e)$, in $\OO(1)$ time,\label{lem:hypergraph_impl:item2}
\item given a pointer to a hyperedge $e \in E(G)$, add a new vertex $v$, add $v$ to $V(e)$, and return a pointer to $v$, in $\OO(1)$ time, and\label{lem:hypergraph_impl:item3}
\item given a pointer to a hyperedge $e \in E(G)$, remove $e$ and every vertex $v$ that occurs only in $V(e)$ in $\OO(|V(e)|)$ time.\label{lem:hypergraph_impl:item4}
\end{enumerate}
\end{lemma}
\begin{proof}
The operation of \Cref{lem:hypergraph_impl:item1} is simply adding one more entry to the linked list of hyperedges.

The operation of \Cref{lem:hypergraph_impl:item2} first adds the pointer to $e$ to the incidence list of $v$.
Then, it adds to the linked list representing $V(e)$ a pair of pointers pointing to $v$ and to the just added entry in the incidence list of $v$.

The operation of \Cref{lem:hypergraph_impl:item3} adds a new entry to the linked list representing $V(G)$, and then works as the operation of \Cref{lem:hypergraph_impl:item2}.

The operation of \Cref{lem:hypergraph_impl:item4} scans the linked list representing $V(e)$ in order.
For each $v \in V(e)$, it uses the pointer to the entry of the incidence list of $v$ to remove the occurrence of $e$ from the incidence list of $v$, and then inspects whether the incidence list of $v$ became empty, in which case it removes $v$ from the linked list representing $V(G)$.
\end{proof}

One important deficiency of our hypergraph representation is that given a hyperedge $e \in E(G)$, we cannot efficiently decide whether there exists another hyperedge $e' \in E(G)$ with $V(e) = V(e')$.
This would require an efficient deterministic dynamic dictionary, which is not known.

Let us then state how basic primities for manipulating hypergraphs are implemented.
We first consider computing $V(C)$ and $\bd(C)$ when given $C$.

\begin{lemma}
\label{lem:hypergraph_impl2}
Let $G$ be a hypergraph whose representation is stored.
There is an algorithm that given a set $C \subseteq E(G)$, in time $\OO(|C| \cdot \rank(G))$ returns the sets $V(C)$ and $\bd(C)$.
\end{lemma}
\begin{proof}
We compute a linked list representing the set $V(C)$ by iterating over $V(e)$ for each $e \in C$, for each vertex $v \in V(e)$ incrementing a counter stored at the vertex (which starts at $0$), and every time a vertex is seen the first time (i.e. the counter is increased from $0$ to $1$), adding the vertex to the linked list representing $V(C)$.
This runs in $\OO(\sum_{e \in C} |V(e)|) = \OO(|C| \cdot \rank(G))$ time, as each list $V(e)$ can be iterated over in $\OO(|V(e)|)$ time.

This process computed for each $v \in V(C)$ the number of hyperedges $e \in C$ with $v \in V(e)$.
Now, we have that $v \in \bd(C)$ if this number is less than the number of hyperedges incident to $v$.
This can be decided for each $v \in V(C)$ in time linear in the stored number, as we need to iterate $\ninc(v)$ for only that many (plus one) elements.
As the sum of the stored numbers is $\OO(|C| \cdot \rank(G))$, this part runs also in $\OO(|C| \cdot \rank(G))$ time.
\end{proof}

We then consider the operation $G \rescliqs C$.
We first give an algorithm for computing $G \rescliqs \co{C}$ when given $C$.
Note that this setting is perhaps the more natural than computing $G \rescliqs C$, as we have that $|E(G \rescliqs \co{C})| = |C|+1$.

\begin{lemma}
\label{lem:hypergraph_impl3}
Let $G$ be a hypergraph whose representation is stored.
There is an algorithm that, given $C \subseteq E(G)$, in time $\OO(|C| \cdot \rank(G))$ returns a representation of the hypergraph $G \rescliqs \co{C}$, with a mapping that maps vertices in $V(G \rescliqs \co{C})$ to corresponding vertices in $G$ and hyperedges in $E(G \rescliqs \co{C}) \setminus \{e_{\co{C}}\}$ to corresponding hyperedges in $G$.
\end{lemma}
\begin{proof}
We first use \Cref{lem:hypergraph_impl2} to construct $V(C)$ and $\bd(C)$ in $\OO(|C| \cdot \rank(G))$ time.
Then we simply copy the hyperedges in $C$ and vertices in $V(C)$ from the representation of $G$ to a representation of $G \rescliqs \co{C}$, which can be done in $\OO(|C| \cdot \rank(G))$ time.
Finally, we add the hyperedge $e_{\co{C}}$, which can be done in $\OO(\bdc(C)) = \OO(|C| \cdot \rank(G))$ time.
\end{proof}

Then we give an algorithm for computing $G \rescliqs C$ given $C$.

\begin{lemma}
\label{lem:hypergraph_impl4}
Let $G$ be a hypergraph whose representation is stored.
There is an algorithm that, given $C \subseteq E(G)$, in time $\OO(|C| \cdot \rank(G))$ turns the representation of $G$ into a representation of $G \rescliqs C$, preserving all pointers to hyperedges in $E(G \rescliqs C) \setminus \{e_C\} = \co{C}$ and to vertices in $V(G \rescliqs C) = V(\co{C})$, and returning a pointer to the hyperedge $e_C \in E(G \rescliqs C)$
\end{lemma}
\begin{proof}
We first use \Cref{lem:hypergraph_impl2} to compute $\bd(C)$ in time $\OO(|C| \cdot \rank(G))$.
Then, we delete all hyperedges in $C$ in time $\OO(|C| \cdot \rank(G))$, and insert a hyperedge $e_C$ with $V(e_C) = \bd(C)$ in time $\OO(|\bd(C)|) = \OO(|C| \cdot \rank(G))$.
\end{proof}

We then consider computing the internal components of a given set $A \subseteq E(G)$.

\begin{lemma}
\label{lem:alg_internal_comps}
Let $G$ be a hypergraph whose representation is already stored.
There is an algorithm that, given a set $A \subseteq E(G)$, in time $\OO(|A| \cdot \rank(G))$ returns the partition of $A$ into the internal components of $A$.
\end{lemma}
\begin{proof}
This can be done by first using \Cref{lem:hypergraph_impl2} to construct $V(A)$ and $\bd(A)$, and then computing the connected components of the bipartite graph whose vertices are the union of $A$ and $V(A) \setminus \bd(A)$, and whose edges are constructed by connecting $e \in A$ to $v \in V(A) \setminus \bd(A)$ whenever $v \in V(e)$.
\end{proof}

\paragraph{Representation of graphs.}
We define that a representation of a graph $G$ is a representation of the hypergraph $\hyperg(G)$.
By using \Cref{lem:hypergraph_impl}, we observe that the representation of graphs supports the following operations.

\begin{observation}
\label{lem:graph_impl}
The representation of graphs supports the following operations
\begin{enumerate}
\item add a new vertex and return a pointer to it, in $\OO(1)$ time,
\item given $u,v \in V(G)$ with $u \neq v$ and $uv \notin E(G)$, add $uv$ to $E(G)$ and return a pointer to $uv$, in $\OO(1)$ time,
\item given a pointer to an edge $uv \in E(G)$, remove $uv$ in $\OO(1)$ time, and
\item given a pointer to a vertex $v \in V(G)$, remove $v$ and all edges incident to $v$, in time $\OO(|N(v)|)$.
\end{enumerate}
\end{observation}

\paragraph{Representation of trees.}
The representation of a tree is its representation as a graph.
However, for rooted trees we store extra information.
When $T$ is a rooted tree, we store also a global pointer pointing to the root node $\troot(T)$ of $T$, and for each non-root node $t \in V(T) \setminus \{\troot(T)\}$ a pointer to the edge $tp$, where $p$ is the parent of $t$.
Note that changing the root to an arbitrary node can be implemented in time $\OO(\|T\|)$, but not more efficiently.

\paragraph{Representation of tree decompositions.}
A representation of a tree decomposition $\Tc = (T,\bag)$ consists of a representation of the tree $T$, and the function $\bag$ stored explicitly, so that each $t \in V(T)$ contains a pointer to a linked list storing $\bag(t)$.
Note that this representation takes $\OO(\|\Tc\|)$ space.

For the ease of manipulating tree decompositions we do not store adhesions explicitly, but let us note that they can be easily computed.

\begin{lemma}
\label{lem:computeadhesions}
There is an algorithm that, given a tree decomposition $\Tc = (T,\bag)$ of a graph $G$, in time $\OO(\|\Tc\|)$ returns the pairs $(e,\adh(e))$ for all $e \in E(T)$.
\end{lemma}
\begin{proof}
We root $T$ at an arbitrary node, and use depth-first search to store for each $v \in V(G)$ the node $t_v$ with $v \in \bag(t_v)$ that is the closest to the root.
Now, if $t \in V(T)$ and $p = \parent(t)$, we have that $v \in \bag(t)$ is in $\adh(tp)$ if and only if $t_v \neq t$.
This allows to construct the pairs $(e,\adh(e))$ in $\OO(\|\Tc\|)$ time after the initial depth-first search.
\end{proof}

We will also need the following auxiliary data structure for working with tree decompositions.

\begin{lemma}
\label{lem:cliqtdds}
There is a data structure that is initialized with a tree decomposition $\Tc = (T,\bag)$ of a graph $G$ in time $\OO(\|\Tc\|)$, and supports a query, where given a set of vertices $W \subseteq V(G)$ such that there exists $t \in V(T)$ with $W \subseteq \bag(t)$, returns such $t$ in $\OO(|W|)$ time.
\end{lemma}
\begin{proof}
We root $T$ at an arbitrary node, and use depth-first search to compute and store for each $t \in V(T)$ the distance from $t$ to the root.
We also compute for each $v \in V(G)$ the closest node $t_v$ to the root so that $v \in \bag(t_v)$.
These can be done in time $\OO(\|\Tc\|)$.

Now, given $W$, we return the node $t_v$ for $v \in W$ that maximizes the distance to the root.
This is correct because the subtrees corresponding to other vertices in $W$ must intersect the subtree rooted at $t_v$, but their highest node is closer to root than $t_v$ (or is also $t_v$), so they must contain $t_v$.
\end{proof}

\paragraph{Representation of superbranch decompositions.}
Let $G$ be a hypergraph.
We define that a representation of a superbranch decomposition $\Tc = (T, \lmap)$ of $G$ consists of
\begin{enumerate}
\item a representation of $T$,
\item for each $\ell \in \leafs(T)$, a pointer from $\ell$ to $\lmap(\ell)$,
\item for each edge $st \in E(T)$ a linked list storing $\adh(st)$, and
\item for each internal node $t \in \vint(T)$, a representation of $\torso(t)$, where additionally
\begin{enumerate}
\item each hyperedge $e_s \in E(\torso(t))$ corresponding to an edge $st \in E(T)$ stores a pointer to the edge $st$ of $T$, and
\item each edge $st \in E(T)$ stores a pointer to the hyperedge $e_s \in E(\torso(t))$ and to the hyperedge $e_t \in E(\torso(s))$, if $t$ (resp. $s$) is an internal node.
\end{enumerate}
\end{enumerate}

\Cref{subsec:transsuperbds} will be dedicated to efficient algorithms for manipulating superbranch decompositions.

\section{Roadmap}
\label{sec:roadmap}
Our main technical result, which implies the other results, is \Cref{the:mainalgklean}, i.e., the algorithm to compute a $k$-lean tree decomposition in time $k^{\OO(k^2)} n + m$.
In this section we outline the algorithm of \Cref{the:mainalgklean} and break it down to 9 ``main lemmas'' that will be proven one by one in the later sections.
At the end of this section, we prove \Cref{the:mainalgklean} assuming the 9 main lemmas, and then prove \Cref{the:mainkgomoryhu} and \Cref{the:mainelementgomoryhu,the:mainglobalvertexcut} by using \Cref{the:mainalgklean}.

\subsection{Generalized Bodlaender's technique}
The algorithm of \Cref{the:mainalgklean} consists of two conceptually separate ingredients.
The first ingredient is to show that if one can prove \Cref{the:mainalgklean} under the assumption that the input already contains an unbreakable tree decomposition with certain unbreakability parameters, then this ``improvement algorithm'' in fact can be used to construct an algorithm that returns a $k$-lean tree decomposition from scratch.
The second ingredient is, of course, such an improvement algorithm.

The first ingredient is heavily inspired by Bodlaender's linear-time parameterized algorithm for treewidth~\cite{Bodlaender96}, and can in fact be regarded as a generalized version of the matching contraction technique introduced by Bodlaender.
The first ingredient is captured by the following main lemma.

\begin{restatable}{mainlemma}{thehighlevelbodl}
\label{the:highlevel:bodl}
Suppose there is an algorithm that, given 
\begin{itemize}
\item a graph $G$, 
\item an integer $k \ge 1$, and 
\item a $(2k, k)$-unbreakable tree decomposition $\Tc$ of $G$ with $\adhsize(\Tc) \le 2k$,
\end{itemize}
in time $\runtime(k) \cdot (\|G\| + \|\Tc\|)$ returns a $k$-lean tree decomposition of $G$.
Then, there is an algorithm that, given a graph $G$ and an integer $k \ge 1$, in time $\runtime(k) \cdot k^{\OO(1)} \cdot |V(G)| + \OO(\|G\|)$ returns a $k$-lean tree decomposition of $G$.
\end{restatable}

In addition to ideas from~\cite{Bodlaender96}, a crucial ingredient that is used in multiple steps in the proof of \Cref{the:highlevel:bodl} is the sparsifier of Nagamochi and Ibaraki~\cite{DBLP:journals/algorithmica/NagamochiI92}.
The proof of \Cref{the:highlevel:bodl} will be given in \Cref{sec:bodl}.

\subsection{Refining a rooted superbranch decomposition}
The rest of the main lemmas are dedicated to providing the improvement algorithm, i.e., an algorithm that given a graph $G$, an integer $k \ge 1$, and an $(2k,k)$-unbreakable tree decomposition $\Tc$ of $G$ with $\adhsize(\Tc) \le 2k$, returns a $k$-lean tree decomposition of $G$.
The general plan for this improvement algorithm will be to refine $\Tc$ gradually in multiple steps, which will correspond to the main lemmas.
Here, we will mostly work with hypergraphs, so the first step is to turn $G$ into the hypergraph $\hyperg(G)$ and $\Tc$ into a superbranch decomposition of $\hyperg(G)$.
This corresponds to the following main lemma.

\begin{restatable}{mainlemma}{thehighlevelfromgraphstohypergraphs}
\label{the:highlevel:fromgraphstohypergraphs}
There is an algorithm that, given 
\begin{itemize}
\item a graph $G$,
\item an integer $k \ge 1$,
\item a $(2k, k)$-unbreakable tree decomposition $\Tc$ of $G$ with $\adhsize(\Tc) \le 2k$,
\end{itemize}
in time $\OO(k \cdot (\|G\| + \|\Tc\|))$ returns a superbranch decomposition $\Tc'$ of $\hyperg(G)$ so that
\begin{itemize}
\item $\Tc'$ is $(2k,k)$-unbreakable, and
\item $\adhsize(\Tc') \le 2k$.
\end{itemize}
\end{restatable}

The proof of \Cref{the:highlevel:fromgraphstohypergraphs} is not very difficult.
It will be given in \Cref{sec:manipulatingsuperbranch}.

Then we start improving the superbranch decomposition $\Tc$ of $\hyperg(G)$.
We would like to make the internal separations of $\Tc$ doubly well-linked, while maintaining the adhesion size bound and unbreakability.
It turns out that this is too difficult goal to achieve directly, but instead we can settle for making the internal separations mixed-$k$-well-linked, which will be good enough.

The first step in making the internal separations of $\Tc$ mixed-$k$-well-linked is to make $\Tc$ ``downwards well-linked''.
We define that a rooted superbranch decomposition $\Tc = (T,\lmap)$ is \emph{downwards well-linked} if for every edge $tp \in E(T)$, where $p$ is the parent of $t$, it holds that $\lmap(\vec{tp})$ is well-linked.

\begin{restatable}{mainlemma}{thehighleveldownwlalgo}
\label{the:highlevel:downwlalgo}
There is an algorithm that, given a superbranch decomposition $\Tc$ of a hypergraph $G$, and an integer $k \ge 1$ so that
\begin{itemize}
\item $\Tc$ is $(2k,k)$-unbreakable and 
\item $\adhsize(\Tc) \le 2k$,
\end{itemize}
in time $2^{\OO(k)} \cdot \|G\|$ outputs a rooted superbranch decomposition $\Tc'$ of $G$ so that
\begin{itemize}
\item $\Tc'$ is downwards well-linked,
\item $\Tc'$ is $(2^{\OO(k)}, k)$-unbreakable, and
\item $\adhsize(\Tc') \le 2k$.
\end{itemize}
\end{restatable}

The idea of the algorithm of \Cref{the:highlevel:downwlalgo} is to root the given superbranch decomposition at an arbitrary node, process edges of form $tp \in E(T)$ in a bottom-up order, and for each $\lmap(\vec{tp})$ that is not well-linked, ``collapse'' the adhesion at $tp$ into at most $2^{\bdc(\lmap(\vec{tp}))}$ adhesions that are downwards well-linked.
The proof of \Cref{the:highlevel:downwlalgo} will be given in \Cref{sec:downwl}.

Then, we would like to make $\Tc$ ``upwards $k$-well-linked'' in order to make its internal separations mixed-$k$-well-linked.
This turns out to be a too strong goal, but we are able to achieve the weaker goal of just making its internal separations mixed-$k$-well-linked, in the process ``scrambling'' $\Tc$ and losing the property of downwards well-linkedness, but in fact retaining an extra property that implies that in spirit, the resulting superbranch decomposition is downwards well-linked and upwards $k$-well-linked.

\begin{restatable}{mainlemma}{thehighlevelupwlalgo}
\label{the:highlevel:upwlalgo}
There is an algorithm that, given a rooted superbranch decomposition $\Tc$ of a hypergraph $G$, and an integer $k \ge 1$ so that
\begin{itemize}
\item $\Tc$ is downwards well-linked,
\item $\Tc$ is $(2^{\OO(k)},k)$-unbreakable, and
\item $\adhsize(\Tc) \le 2k$,
\end{itemize}
in time $2^{\OO(k^2)} \cdot \|G\|$ outputs a superbranch decomposition $\Tc' = (T', \lmap')$ of $G$ so that
\begin{itemize}
\item the internal separations of $\Tc'$ are mixed-$k$-well-linked,
\item for all $t \in V(T')$, there is at most one $v \in N_{T'}(t)$ so that $\lmap'(\vec{vt})$ is not well-linked,
\item the torsos of $\Tc'$ are $(2^{\OO(k)}, k)$-unbreakable, and
\item $\adhsize(\Tc') \le 2k$.
\end{itemize}
\end{restatable}

Note that the properties of $\Tc'$ imply that it can be partitioned into connected subtrees, so that each of these subtrees can be rooted so that it is downwards well-linked and upwards $k$-well-linked, and moreover, the separations between these subtrees are doubly well-linked.
Note also that at this point we switch from maintaining the unbreakability of the superbranch decomposition to maintaining the unbreakability of the torsos of the superbranch decomposition, essentially because the internal separations have strong enough uncrossing properties that the unbreakability of the torsos in fact implies the unbreakability of the decomposition.

The algorithm of \Cref{the:highlevel:upwlalgo} proceeds in an opposite order to the algorithm of \Cref{the:highlevel:downwlalgo}, processing the superbranch decomposition from the root to the leaves.
The idea is to ``collapse'' an adhesion towards the root if the corresponding set is not $k$-well-linked, but in this direction it is more technical than in the downwards direction.
The algorithm of \Cref{the:highlevel:upwlalgo} is perhaps conceptually the most difficult of all of the steps in our algorithm, although its proof comes quite naturally after finding the right invariants to maintain while refining the superbranch decomposition.
The proof of \Cref{the:highlevel:upwlalgo} will be given in \Cref{sec:mixedkwldecomp}.

\subsection{Decreasing adhesion size}
After making the internal separations of $\Tc$ mixed-$k$-well-linked, our main objective will be to decrease the sizes of the internal adhesions to $<k$.
Here, the plan is to proceed with the help of the concept of $k$-tangle-unbreakability.
Recall that a hypergraph is $k$-tangle-unbreakable if it does not have a separation $(A,\co{A})$ of order $\bdc(A) < k$ and two tangles $\tang_1,\tang_2$ so that $A \in \tang_1$ and $\co{A} \in \tang_2$.
The first step is to make the torsos of $\Tc$ $k$-tangle-unbreakable.

\begin{restatable}{mainlemma}{thehighlevelmakektangleunbreakable}
\label{the:highlevel:makektangleunbreakable}
There is an algorithm that, given a superbranch decomposition $\Tc = (T,\lmap)$ of a hypergraph $G$, and integers $s \ge k \ge 1$ so that
\begin{itemize}
\item the internal separations of $\Tc$ are mixed-$k$-well-linked,
\item for all $t \in V(T)$, there is at most one $v \in N_{T}(t)$ so that $\lmap(\vec{vt})$ is not well-linked,
\item the torsos of $\Tc$ are $(s,k)$-unbreakable, and
\item $\adhsize(\Tc) \le 2k$,
\end{itemize}
in time $s^{\OO(k)} \cdot \|G\|$ outputs a superbranch decomposition $\Tc'$ of $G$ so that
\begin{itemize}
\item the internal separations of $\Tc'$ are mixed-$k$-well-linked,
\item the torsos of $\Tc'$ are $k$-tangle-unbreakable, and
\item $\adhsize(\Tc') \le 2k$.
\end{itemize}
\end{restatable}

The idea of the algorithm of \Cref{the:highlevel:makektangleunbreakable} is to show that if a hypergraph is not $k$-tangle-unbreakable, then it has a doubly well-linked separation $(A,\co{A})$ so that $|V(A)| \le |V(\co{A})|$ and $A$ is internally connected.
This means that we can use a local-search type algorithm to further decompose an $(s,k)$-unbreakable torso into $k$-tangle-unbreakable torsos by doubly well-linked separations.
The proof of \Cref{the:highlevel:makektangleunbreakable} will be given in \Cref{sec:tangubrtorsos}.

At this point we do not state anymore that the resulting torsos are unbreakable, because we will lose that property anyway in the next step.

Then we get rid of internal separations of order $\ge k$ simply by contracting all edges of $\Tc$ corresponding to such internal separations.
Here, the main property to prove is that combining two $k$-tangle-unbreakable torsos separated by a mixed-$k$-well-linked separation of order $\ge k$ preserves the $k$-tangle-unbreakability.

\begin{restatable}{mainlemma}{thehighleveldecreaseadhesion}
\label{the:highlevel:decreaseadhesion}
There is an algorithm that, given a superbranch decomposition $\Tc$ of a hypergraph $G$, and an integer $k \ge 1$ so that
\begin{itemize}
\item the internal separations of $\Tc$ are mixed-$k$-well-linked,
\item the torsos of $\Tc$ are $k$-tangle-unbreakable, and
\item $\adhsize(\Tc) \le 2k$,
\end{itemize}
in time $\OO(k \cdot \|G\|)$ outputs a superbranch decomposition $\Tc'$ of $G$ so that
\begin{itemize}
\item the internal separations of $\Tc'$ have order $<k$ and are doubly well-linked, and
\item the torsos of $\Tc'$ are $k$-tangle-unbreakable.
\end{itemize}
\end{restatable}

The idea of the proof of \Cref{the:highlevel:decreaseadhesion} is to show that separations of order $<k$ that distinguish tangles can be ``uncrossed'' with mixed-$k$-well-linked separations of order $\ge k$.
The proof of \Cref{the:highlevel:decreaseadhesion} will be given in \Cref{sec:smalleradhesions}.

Then, the target is to re-obtain the property that the torsos are unbreakable.
At this point, the unbreakability bound unfortunately increases to $(k^{\OO(k)}, k)$, and along with it, the dependence on $k$ in the running time to $k^{\OO(k^2)}$.

\begin{restatable}{mainlemma}{thehighleveltunbrtounbr}
\label{the:highlevel:tunbrtounbr}
There is an algorithm that, given a superbranch decomposition $\Tc = (T, \lmap)$ of a hypergraph $G$ with multiplicity $1$ and $\rank(G) \le 2$, and an integer $k \ge 1$ so that
\begin{itemize}
\item the internal separations of $\Tc$ have order $<k$ and are doubly well-linked, and
\item the torsos of $\Tc$ are $k$-tangle-unbreakable,
\end{itemize}
in time $k^{\OO(k^2)} \cdot \|G\|$ outputs a superbranch decomposition $\Tc'$ of $G$ so that
\begin{itemize}
\item the internal separations of $\Tc'$ have order $<k$ and are doubly well-linked, and
\item the torsos of $\Tc'$ are $(k^{\OO(k)}, k)$-unbreakable.
\end{itemize}
\end{restatable}

The idea of the proof of \Cref{the:highlevel:tunbrtounbr} is that if $G$ is $k$-tangle-unbreakable and $(A,\co{A})$ is a doubly tri-well-linked separation of $G$ of order $<k$, then both $G \rescliqs A$ and $G \rescliqs \co{A}$ are $k$-tangle-unbreakable, and one of them has branchwidth at most $\bdc(A)$.
This makes it possible to decompose $k$-tangle-unbreakable torsos to unbreakable torsos by a local-search type algorithm that decomposes along doubly tri-well-linked separations.
The proof of \Cref{the:highlevel:tunbrtounbr} will be given in \Cref{sec:fromtubrtoubr}.

\subsection{From unbreakable to lean}
The final part is to improve $\Tc$ from being an unbreakable superbranch decomposition of $\hyperg(G)$ into being a $k$-lean tree decomposition of $G$.
Here, the first main insight is that we can simply compute $k$-lean tree decompositions of the primal graphs of the torsos of $\Tc$, and then glue them together along $\Tc$, to obtain a $k$-lean tree decomposition of $G$.

\begin{restatable}{mainlemma}{thehighlevelleanimplication}
\label{the:highlevel:leanimplication}
There is an algorithm that, given a superbranch decomposition $\Tc = (T,\lmap)$ of a normal hypergraph $G$, an integer $k \ge 1$, and a collection of tree decompositions $\{\Tc_t\}_{t \in \vint(T)}$ so that
\begin{itemize}
\item for each $t \in \vint(T)$, $\Tc_t$ is a $k$-lean tree decomposition of $\primal(\torso(t))$, and
\item the internal separations of $\Tc$ have order $<k$ and are doubly well-linked,
\end{itemize}
in time $\OO((k+\rank(G)) \cdot \|G\| + \sum_{t \in \vint(T)} \|\Tc^*_t\|)$ returns a $k$-lean tree decomposition of $\primal(G)$.
\end{restatable}

The idea of the proof of \Cref{the:highlevel:leanimplication} is that vertex cuts between subsets of torsos of $\Tc$ can be uncrossed with the internal separations of $\Tc$ because they are doubly well-linked.
The proof of \Cref{the:highlevel:leanimplication} will be given in \Cref{sec:combin}.

Then, it remains to give an algorithm to compute a $k$-lean tree decomposition of an $(k^{\OO(k)},k)$-unbreakable graph, which is the last main lemma.

\begin{restatable}{mainlemma}{thehighlevelunbrtolean}
\label{the:highlevel:unbrtolean}
There is an algorithm that, given a graph $G$ and integers $s \ge k \ge 1$ so that $G$ is $(s,k)$-unbreakable, in time $s^{\OO(k)} \cdot \|G\|$ returns a $k$-lean tree decomposition of $G$.
\end{restatable}

The idea of the proof of \Cref{the:highlevel:unbrtolean} is to implement the Bellenbaum-Diestel improvement procedure~\cite{bellenbaum2002two,DBLP:journals/talg/CyganKLPPSW21}.
To simplify this strategy, we show that computing a $k$-lean tree decomposition of an $(s,k)$-unbreakable graph reduces to the problem of computing only one ``central'' lean bag, and to the problem of computing $k$-lean tree decompositions of graphs whose size is bounded by $s^{\OO(k)}$.
In the latter problem we do not mind $\|G\|^{\OO(1)}$ factors, so the complications of a linear-time implementation need to be addressed only in the problem of computing the ``central'' lean bag.
The proof of \Cref{the:highlevel:unbrtolean} will be given in \Cref{sec:lean}.

\subsection{Putting things together}
Let us then prove \Cref{the:mainalgklean} assuming the main lemmas stated above.

\themainalgklean*
\begin{proof}
By \Cref{the:highlevel:bodl}, it suffices to give an algorithm that, given a graph $G$, an integer $k \ge 1$, and a $(2k,k)$-unbreakable tree decomposition $\Tc$ of $G$ with $\adhsize(\Tc) \le 2k$, in time $k^{\OO(k^2)} \cdot (\|G\| + \|\Tc\|)$ returns a $k$-lean tree decomposition of $G$.

To give this algorithm, we first combine \Cref{the:highlevel:fromgraphstohypergraphs,the:highlevel:downwlalgo,the:highlevel:upwlalgo,the:highlevel:makektangleunbreakable,the:highlevel:decreaseadhesion,the:highlevel:tunbrtounbr} to in time $k^{\OO(k^2)} (\|G\|+\|\hyperg(G)\|+\|\Tc\|) = k^{\OO(k^2)} (\|G\|+\|\Tc\|)$ compute a superbranch decomposition $\Tc' = (T',\lmap')$ of $\hyperg(G)$ so that the internal separations of $\Tc'$ have order $<k$ and are doubly well-linked and the torsos of $\Tc'$ are $(k^{\OO(k)}, k)$-unbreakable.

Then, for each $t \in \vint(T')$, we compute $\primal(\torso(t))$ from $\torso(t)$ in time $\OO(\|\torso(t)\| \cdot \rank(\torso(t))) = \OO(\|\torso(t)\| \cdot k)$.
\begin{claim}
\label{themainalgklean:claimunb}
For each $t \in \vint(T')$, the graph $\primal(\torso(t))$ is $(k^{\OO(k)}, k)$-unbreakable.
\end{claim}
\begin{claimproof}
Let $s = k^{\OO(k)}$ be the bound so that the torsos of $\Tc'$ are $(s,k)$-unbreakable.
We show that $\primal(\torso(t))$ is $(s+k, k)$-unbreakable.

Suppose otherwise, and let $(A,B)$ be a vertex cut of $\primal(\torso(t))$ of order $<k$ with $|A| \ge s+k$ and $|B| \ge s+k$.
Let $(C,\co{C})$ be a separation of $\torso(t)$ constructed by assigning every hyperedge $e \in E(\torso(t))$ so that $V(e)$ intersects $A \setminus B$ into $C$, every hyperedge $e \in E(\torso(t))$ so that $V(e)$ intersects $B \setminus A$ into $\co{C}$, and other hyperedges arbitrarily.
Now, $|V(C)| \ge |A \setminus B| \ge s$ and $|V(\co{C})| \ge |B \setminus A| \ge s$, and the order of $(C,\co{C})$ is at most $|A \cap B|$, so $(C,\co{C})$ contradicts the $(s, k)$-unbreakability of $\torso(t)$.
\end{claimproof}

By \Cref{themainalgklean:claimunb}, we can use \Cref{the:highlevel:unbrtolean} to compute a $k$-lean tree decomposition $\Tc_t^*$ of $\primal(\torso(t))$ in time $k^{\OO(k^2)} \cdot \|\primal(\torso(t))\| = k^{\OO(k^2)} \cdot \|\torso(t)\|$.
Doing this for all torsos of $\Tc'$ takes in total $k^{\OO(k^2)} \cdot \|\Tc'\| = k^{\OO(k^2)} \cdot \|G\|$ time.

Finally, we apply \Cref{the:highlevel:leanimplication} to combine $\Tc'$ and the $k$-lean tree decompositions $\Tc_t^*$ into a $k$-lean tree decomposition of $\primal(\hyperg(G)) = G$, in time $\OO(k \cdot \|G\| + \sum_{t \in \vint(T)} \|\Tc_t^*\|) = k^{\OO(k^2)} \cdot \|G\|$.
\end{proof}

We then prove \Cref{the:mainkgomoryhu} and \Cref{the:mainelementgomoryhu} using \Cref{the:mainalgklean}.
Note that \Cref{the:mainelementgomoryhu} implies \Cref{the:mainkgomoryhu} by setting $U = V(G)$, so it suffices to prove only \Cref{the:mainelementgomoryhu}.

\theelementkgomoryhu*
\begin{proof}
First, let $G'$ be a $k$-sparsifier of $G$, which can be computed in time $\OO(m)$ by the algorithm of \Cref{the:nispars}.

\begin{claim}
Any element connectivity $k$-Gomory-Hu tree of $(G',U)$ is also an element connectivity $k$-Gomory-Hu tree of $(G,U)$.
\end{claim}
\begin{claimproof}
For this, we need to argue that if $u,v \in U$ and $S$ is an $U$-element $(u,v)$-cutset of size $|S|<k$ in $G'$, then $S$ is also an $U$-element $(u,v)$-cutset in $G$.
Suppose not, and let $C_u \subseteq V(G)$ be the connected component of $G' \setminus S$ containing $u$.
There must be an edge $xy \in E(G) \setminus (E(G') \cup S)$ so that $x \in C_u$ and $y \in V(G) \setminus (C_u \cup S)$.
Therefore, by \Cref{lem:delhighconn} there exists a collection of $k$ internally vertex-disjoint $(x,y)$-paths in $G'$.
The set $S$ can hit at most $k-1$ of these paths, so this contradicts that $x$ and $y$ are in different connected components of $G' \setminus S$.
\end{claimproof}

Therefore, it remains to compute an element connectivity $k$-Gomory-Hu tree of $(G',U)$.

We construct a graph $G''$ from $G'$ by replacing all vertices in $U$ by cliques of size $k$, and replacing each edge $uv \in E(G')$ by a vertex $x_{uv}$, connecting it to the vertices corresponding to its endpoints.
More formally, we let
\[V(G'') = (V(G') \setminus U) \cup \{x_{uv} : uv \in E(G')\} \cup \{y_{v,i} : v \in U \text{ and } i \in [k]\},\]
and
\begin{align*}
E(G'') = &\{x_{uv}v : v \in V(G') \setminus U \text{ and } u \in N(v)\}\\
\cup &\{x_{uv}y_{v,i} : v \in U, i \in [k], \text{ and } u \in N(v)\}\\
\cup &\{y_{v,i}y_{v,j} : v \in U \text{ and } 1 \le i < j \le k\}.
\end{align*}
The graph $G''$ has $\|G''\| = \OO(k^2 \cdot |V(G')| + k \cdot |E(G')|) = \OO(k^2 n)$ and can be constructed in $\OO(k^2 n)$ time.

Then, we apply the algorithm of \Cref{the:mainalgklean} to in time $k^{\OO(k^2)} n$ compute a $k$-lean tree decomposition $\Tc = (T,\bag)$ of $G''$.
We then use the data structure of \Cref{lem:cliqtdds} to in time $\OO(\|\Tc\| + kn) = k^{\OO(k^2)} n$ compute for each vertex $v \in U$ a node $t_v \in V(T)$ so that $Y_v = \{y_{v,i} : i \in [k]\} \subseteq \bag(t_v)$.
Because $Y_v$ is a clique in $G''$, such a node $t_v$ must exist.
Now, we construct a tree decomposition $\Tc' = (T,\bag')$ by removing each vertex $y_{v,i}$ for $v \in U$ and $i \in [k]$ from the bags of all other nodes than $t_v$.
\begin{claim}
$\Tc'$ is a $k$-lean tree decomposition of $G''$.
\end{claim}
\begin{claimproof}
Let us first prove that $\Tc'$ is a tree decomposition of $G''$.
The vertex condition is clear from the construction.
For the edge condition, all other types of edges are clear except the edges of form $x_{uv} y_{v,i}$.
For this, we need to prove that $x_{uv} \in \bag(t_v)$.
Suppose not, and let $t_{uv}$ be a node of $T$ so that $x_{uv} \in \bag(t_{uv})$, and let $s_{uv}$ be the neighbor of $t_v$ on the $(t_{uv}, t_v)$-path.
Because $x_{uv}$ is not in $\bag(t_v)$ but is adjacent to all vertices in $Y_v$, it follows that $Y_v \subseteq \bag(s_{uv})$.
However, then $Y_v \subseteq \adh(s_{uv} t_v)$, contradicting that $\Tc$ has adhesion size $<k$.

To prove that $\Tc'$ is $k$-lean, observe that $\adhsize(\Tc') \le \adhsize(\Tc)$, and any non-$k$-lean-witness for $\Tc'$ would also be a non-$k$-lean-witness for $\Tc$.
\end{claimproof}

Now, we construct the mapping $\gamma \colon U \to V(T)$ by letting $\gamma(v) = t_v$, and the mapping $\alpha \colon E(T) \to 2^{(V(G') \setminus U) \cup E(G')}$ by taking for each $e \in E(T)$ the set $\adh_{\Tc'}(e)$ and replacing each vertex of type $x_{uv}$ by the edge $uv \in E(G')$.
By using \Cref{lem:computeadhesions} for constructing the adhesions of $\Tc'$, the mappings $\gamma$ and $\alpha$ can be constructed in total $\OO(|U| + \|\Tc'\|) = k^{\OO(k^2)} n$ time.

\begin{claim}
$(T,\gamma,\alpha)$ is an element connectivity $k$-Gomory-Hu tree of $(G',U)$.
\end{claim}
\begin{claimproof}
First, we need to prove that for each $e \in E(T)$, and all $u,v \in U$ so that $e$ is on the $(\gamma(u),\gamma(v))$-path in $T$, the set $\alpha(e)$ is an $U$-element $(u,v)$-cutset of $G'$.
For this, we recall that $\adh_{\Tc'}(e)$ is a $(\bag'(\gamma(u)), \bag'(\gamma(v)))$-separator not containing vertices of type $y_{w,i}$ for any $w \in U$, which by $Y_u \subseteq \bag'(\gamma(u))$ and $Y_v \subseteq \bag'(\gamma(v))$ implies that $\alpha(e)$ is an $U$-element $(u,v)$-cutset of $G'$.

Then, we need to prove that if $u,v \in U$ and there exists an $U$-element $(u,v)$-cutset $S$ of size $|S|<k$ of $G'$, then there exists $e \in E(T)$ on the $(\gamma(u), \gamma(v))$-path so that $|\alpha(e)| \le |S|$, i.e., $|\adh_{\Tc'}(e)| \le |S|$.
Let $S$ be an $U$-element $(u,v)$-cutset of size $|S|<k$.
Now, there exists a vertex cut $(A,B)$ of $G''$ so that $Y_u \subseteq A \setminus B$, $Y_v \subseteq B \setminus A$, and $A \cap B = (S \cap V(G')) \cup \{x_{uv} : uv \in S \cap E(G')\}$.
If all $(\gamma(u),\gamma(v))$-adhesions of $\Tc'$ would have size $>|S|$ then $(A,B,\gamma(u),\gamma(v))$ would be a non-$k$-lean-witness, so there exists an $(\gamma(u),\gamma(v))$-adhesion of size $\le |S|$.
\end{claimproof}
This claim finishes the proof.
\end{proof}

We then prove \Cref{the:mainglobalvertexcut}.

\theglobalvertexcut*
\begin{proof}
We apply the algorithm of \Cref{the:mainalgklean} to compute a $k$-lean tree decomposition $\Tc = (T,\bag)$ of $G$, and then use \Cref{lem:computeadhesions} to explicitly construct the adhesions of it.

First, note that if there exists $t \in V(T)$ so that $\bag(t) = V(G)$, then no proper vertex separators of size $<k$ exist, because any separation $(A,B)$ of order $<k$ with $A \setminus B$ and $B \setminus A$ non-empty would give a non-$k$-lean-witness $(A,B,t,t)$.

We claim that if no such node $t$ exists, then we can find a proper vertex separator of size $<k$.

Root $T$ at an arbitrary node.
We can in $\OO(k \cdot \|\Tc\|)$ time find a non-root node $t$ with a parent $p$ so that $\bag(t) \neq \adh(tp)$, so that among such nodes $t$ maximizes the distance to the root.
Such a node $t$ must exist, because otherwise the root bag would contain all vertices.
We claim that $\adh(tp)$ is a proper vertex separator of $G$ of size $<k$.
Let $A$ be the union of the bags of the descendants of $t$ and $B$ the union of the other bags, and note that $(A,B)$ is a vertex cut with $A \cap B = \adh(tp)$.
We have that $A \setminus B$ is non-empty because $\bag(t) \neq \adh(tp)$, so it suffices to prove that $B \setminus A$ is non-empty.
The selection of $t$, maximizing the distance to the root, implies that $A = \bag(t)$, implying that $A \neq V(G)$, which in turn implies that $B \setminus A$ is non-empty.

This algorithm runs in time $k^{\OO(k^2)} n + \OO(m) + \OO(k \cdot \|\Tc\|) = k^{\OO(k^2)} n + \OO(m)$.
\end{proof}

\subsection{Organization of the paper}
The rest of this paper is organized as follows.
First, in \Cref{sec:bodl} we give the proof of \Cref{the:highlevel:bodl}.
Then, \Cref{sec:gt,sec:localsearch,sec:manipulatingsuperbranch} will be about introducing tools for working with hypergraphs and superbranch decompositions that will be used in the later sections.
More precisely, in \Cref{sec:gt} we prove various graph-theoretic lemmas about hypergraphs, in \Cref{sec:localsearch} we give algorithms for finding separations of hypergraphs, and in \Cref{sec:manipulatingsuperbranch} we introduce tools for manipulating superbranch decompositions.
\Cref{sec:manipulatingsuperbranch} also includes the proof of \Cref{the:highlevel:fromgraphstohypergraphs}.
Then, \Cref{sec:downwl,sec:mixedkwldecomp,sec:tangubrtorsos,sec:smalleradhesions,sec:fromtubrtoubr,sec:combin,sec:lean}, respectively, will be dedicated to the proofs of \Cref{the:highlevel:downwlalgo,the:highlevel:upwlalgo,the:highlevel:makektangleunbreakable,the:highlevel:decreaseadhesion,the:highlevel:tunbrtounbr,the:highlevel:leanimplication,the:highlevel:unbrtolean}, respectively.

\section{Generalized Bodlaender's technique}
\label{sec:bodl}
In this section, we provide the first ingredient of our algorithm.
This is heavily based on Bodlaender's technique~\cite{Bodlaender96} that he introduced for obtaining a linear-time parameterized algorithm for computing treewidth, and can be regarded as a generalization of it.

In particular, our goal is to prove the following.

\thehighlevelbodl*

\subsection{Sparsification}
Recall the Nagamochi-Ibaraki sparsifier from \Cref{subsec:nagamochiibaraki}.
Our first step is to use it reduce the number of edges of $G$.
For this, we need to prove that if $\Tc$ is a $k$-lean tree decomposition of a $k$-sparsifier of $G$, then $\Tc$ is also a $k$-lean tree decomposition of $G$.

The following lemma combined with \Cref{lem:delhighconn} implies that if $G'$ is a $k$-sparsifier of $G$ and $\Tc$ is a tree decomposition of $G'$ with adhesion size $<k$, then $\Tc$ is a tree decomposition of $G$.

\begin{lemma}
\label{lem:highconnsamebag}
Let $G$ be a graph and $(T,\bag)$ a tree decomposition of $G$ with adhesion size $<k$.
Let also $u,v \in V(G)$ be distinct vertices of $G$.
If $\flowp(u,v) \ge k$, then there is $t \in V(T)$ with $\{u,v\} \subseteq \bag(t)$.
\end{lemma}
\begin{proof}
Suppose there is no bag containing both $u$ and $v$.
Now, the vertex condition of tree decompositions implies that there exists $st \in E(T)$ so that $\adh(st)$ is a proper $(u,v)$-separator, contradicting that $\flowp(u,v) \ge k$.
\end{proof}

We then apply the above fact to show that a $k$-lean tree decomposition of a $k$-sparsifier of $G$ is also a $k$-lean tree decomposition of $G$.

\begin{lemma}
\label{lem:sparsmaintain}
Let $G'$ be a $k$-sparsifier of $G$.
If $\Tc$ is a $k$-lean tree decomposition of $G'$, then $\Tc$ is also a $k$-lean tree decomposition of $G$.
\end{lemma}
\begin{proof}
First, to prove that $\Tc$ is a tree decomposition of $G$, we must show that if $uv \in E(G) \setminus E(G')$, then there is a bag of $\Tc$ that contains both $u$ and $v$.
By \Cref{lem:delhighconn} we have that in this case $\flowp_{G'}(u,v) \ge k$.
Then, because $\Tc$ has adhesion size $<k$, \Cref{lem:highconnsamebag} implies that there is a bag containing both $u$ and $v$.

It remains to prove that $\Tc$ is a $k$-lean tree decomposition of $G$.
We have that $\adhsize(\Tc) < k$ by assumption.
Because $G'$ is a subgraph of $G$, any vertex cut of $G$ is also a vertex cut of $G'$, and therefore any non-$k$-lean-witness for $\Tc$ with $G$ would also be a non-$k$-lean-witness for $\Tc$ with $G'$, contradicting that $\Tc$ is a $k$-lean tree decomposition of $G'$.
\end{proof}

\subsection{Compression}
We then move to the main part of the algorithm behind \Cref{the:highlevel:bodl}.

Let $G$ be a graph and $k$ an integer.
The \emph{$k$-improved graph} of $G$ is the graph obtained from $G$ by adding an edge between any two non-adjacent vertices $u$ and $v$ with $\flowp(u,v) \ge k$.
We denote the $k$-improved graph of $G$ by $I_k(G)$.

\begin{lemma}
\label{lem:improvedclique}
If $W \subseteq V(G)$ is a clique in $I_k(G)$, then for every vertex cut $(A,B)$ of $G$ of order $|A \cap B|<k$ it holds that either $W \subseteq A$ or $W \subseteq B$.
\end{lemma}
\begin{proof}
Suppose not, and let $u \in W \cap (A \setminus B)$ and $v \in W \cap (B \setminus A)$.
Now, if $uv \in E(G)$, then $(A,B)$ is not a vertex cut of $G$, and if $uv \notin E(G)$, then $A \cap B$ is a proper $(u,v)$-separator of size $<k$, contradicting that $\flowp(u,v) \ge k$.
\end{proof}

A vertex $w \in V(G)$ is a \emph{simplicial vertex} if $N(w)$ is a clique.
A vertex $w \in V(G)$ is an \emph{$I_k$-simplicial vertex} if $N(w)$ is a clique in the $k$-improved graph of $G$, or equivalently, if for every pair $u,v \in N(w)$ it holds that either $uv \in E(G)$ or $\flowp_G(u,v) \ge k$.

For a vertex $w \in V(G)$, we denote by $G \elim w$ the graph obtained from $G$ by \emph{eliminating} $w$, that is, by inserting edges between all pairs of distinct vertices $u,v \in N(w)$ and then removing $w$.
For a set of vertices $X \subseteq V(G)$, we denote by $G \elim X$ the graph obtained from $G$ by eliminating all vertices in $X$.
It can be shown that the order of eliminations does not matter, and therefore $G \elim X$ is well-defined.
In fact, we will use this notation only when $X$ is an independent set, in which case it is even easier to observe that the order of eliminations does not matter.

One case of our algorithm will identify a large independent set $I \subseteq V(G)$ of $I_k$-simplicial vertices of degree $<k$.
We now show that in that case, we can recursively call our algorithm with $G \elim I$, and lift the resulting $k$-lean tree decomposition of $G \elim I$ into an unbreakable tree decomposition of $G$.

\begin{lemma}
\label{lem:bodlputbacksimplsmalldeg}
There is an algorithm that, given a graph $G$, integer $k \ge 1$, an independent set $I \subseteq V(G)$ of $I_k$-simplicial vertices of degree $<k$, and a $k$-lean tree decomposition $\Tc$ of $G \elim I$, in time $\OO(\|G\|+\|\Tc\|)$ computes a $(k,k)$-unbreakable tree decomposition of $G$ with adhesion size~$<k$.
\end{lemma}
\begin{proof}
Denote $\Tc = (T,\bag)$.
For every $v \in I$, $N(v)$ is a clique in $G \elim I$, so $(T,\bag)$ contains a bag containing $N(v)$.
We construct a tree decomposition $(T^*, \bag^*)$ of $G$ from $(T,\bag)$ as follows.
For each $v \in I$, we select an arbitrary node $t_v \in V(T)$ with $N(v) \subseteq \bag(t_v)$.
Then, we insert a new node $t'_v$ adjacent to $t_v$ with $\bag^*(t'_v) = N[v]$.

This construction of $(T^*, \bag^*)$ can be implemented with the help of the data structure of \Cref{lem:cliqtdds} in $\OO(\|G\|+\|\Tc\|)$ time.
It remains to prove that $(T^*,\bag^*)$ is indeed a $(k,k)$-unbreakable tree decomposition of $G$ with adhesion size $<k$.

It is easy to see that $(T^*,\bag^*)$ is a tree decomposition of $G$.
The adhesion size of $(T^*,\bag^*)$ is $<k$, because all the new adhesions (compared to $(T,\bag)$) are equal to sets $N(v)$ for $v \in I$, which by assumption have size $|N(v)| < k$.

It remains to prove that $(T^*, \bag^*)$ is $(k,k)$-unbreakable.

First, assume that for some $v \in I$ the set $N[v]$ is not $(k,k)$-unbreakable, implying that there is a vertex cut $(A,B)$ of $G$ so that $|A \cap N[v]| \ge k$, $|B \cap N[v]| \ge k$, and $|A \cap B|<k$.
However, as $N[v]$ is a clique in $I_k(G)$, this yields a contradiction by \Cref{lem:improvedclique}.

Then, assume that for some $t \in V(T)$, the set $\bag(t)$ is not $(k,k)$-unbreakable in $G$, implying that there is a vertex cut $(A,B)$ of $G$ so that $|A \cap \bag(t)| \ge k$, $|B \cap \bag(t)| \ge k$, and $|A \cap B|<k$.
By \Cref{lem:improvedclique}, we have for all $v \in I$ that either $N[v] \subseteq A$ or $N[v] \subseteq B$.
Therefore, $(A \setminus I, B \setminus I)$ is a vertex cut of $G \elim I$, and because $\bag(t)$ is disjoint with $I$, it contradicts that $\bag(t)$ is $(k,k)$-unbreakable in $G \elim I$.
\end{proof}

Another case of our algorithm will be to identify a large independent set $I \subseteq V(G)$ of $I_k$-simplicial vertices of degree $\ge k$.
We show that also in this case, a $k$-lean tree decomposition of $G \elim I$ can be turned into an unbreakable tree decomposition of $G$.

\begin{lemma}
\label{lem:bodlputbacksimplhighdeg}
There is an algorithm that, given a graph $G$, integer $k$, an independent set $I \subseteq V(G)$ of $I_k$-simplicial vertices of degree $\ge k$, and a $k$-lean tree decomposition $\Tc$ of $G \elim I$, in time $\OO(\|G\|+\|\Tc\|)$ computes a $(k,k)$-unbreakable tree decomposition of $G$ with adhesion size~$<k$.
\end{lemma}
\begin{proof}
Denote $\Tc = (T,\bag)$.
For every $v \in I$, $N(v)$ is a clique in $G \elim I$, so $(T,\bag)$ contains a bag containing $N(v)$.
Because $|N(v)| \ge k$ and $(T,\bag)$ has adhesion size $<k$, this bag is in fact unique.
We construct a tree decomposition $(T,\bag^*)$ from $(T,\bag)$ by inserting each $v \in I$ to this uniquely determined bag.
With the data structure of \Cref{lem:cliqtdds}, this can be implemented in $\OO(\|G\|+\|\Tc\|)$ time.

It is easy to observe that $(T,\bag^*)$ is a tree decomposition of $G$.
Also, because each $v \in I$ is inserted to only one bag, the adhesions of $(T,\bag^*)$ are the same as the adhesions of $(T,\bag)$, so the adhesion size of $(T,\bag^*)$ is $<k$.

Suppose that $(T,\bag^*)$ is not $(k,k)$-unbreakable, and let $t \in V(T)$ be a node with $W = \bag^*(t)$ so that there is a vertex cut $(A,B)$ of $G$ of order $|A \cap B| < k$ so that $|A \cap W| > |A \cap B|$ and $|B \cap W| > |A \cap B|$.
Moreover, select $(A,B)$ so that it minimizes $|A \cap B|$.

\begin{claim}
$I$ is disjoint with $A \cap B$.
\end{claim}
\begin{claimproof}
Suppose there is $v \in I \cap A \cap B$.
By \Cref{lem:improvedclique} we have that either $N[v] \subseteq A$ or $N[v] \subseteq B$.
In the former case, $(A, B \setminus \{v\})$ would be a vertex cut that contradicts the choice of $(A,B)$, and in the latter case, $(A \setminus \{v\}, B)$ would contradict the choice of $(A,B)$.
\end{claimproof}

We have that $(A \setminus I, B \setminus I)$ is a vertex cut of $G \elim I$.
We claim that $(A \setminus I, B \setminus I)$ contradicts the fact that $\bag(t) = W \setminus I$ is $(i,i)$-unbreakable for $i = |A \cap B|+1$.
To prove that $|(A \setminus I) \cap W| > |A \cap B|$, suppose first that $I$ does not intersect $(A \setminus B) \cap W$.
In this case, $|(A \setminus I) \cap W| = |A \cap W| > |A \cap B|$.
Then suppose that $I$ intersects $(A \setminus B) \cap W$.
In this case, there is $v \in I \setminus B$ so that $v \in \bag^*(t)$, implying that $N(v) \subseteq W$ and $N(v) \subseteq A \setminus I$, implying that $|(A \setminus I) \cap W| \ge |N(v)| \ge k > |A \cap B|$.
The same argument, after exchanging the roles of $A$ and $B$, shows that $|(B \setminus I) \cap W| > |A \cap B|$.
\end{proof}

Another case of our algorithm will be to identify a large matching $M \subseteq E(G)$.
We show that in this case, a $k$-lean tree decomposition of $G \contr M$ can be turned into an unbreakable tree decomposition of $G$.

\begin{lemma}
\label{lem:matchinguncontract}
There is an algorithm that, given a graph $G$, an integer $k$, a matching $M \subseteq E(G)$ in $G$, and a $k$-lean tree decomposition $\Tc$ of $G \contr M$, in time $\OO(\|G\|+\|\Tc\|)$ computes a $(2k,k)$-unbreakable tree decomposition of $G$ with adhesion size $\le 2k-2$.
\end{lemma}
\begin{proof}
For an edge $uv \in M$, let us denote by $w_{uv} \in V(G \contr M)$ the vertex of $G \contr M$ corresponding to $uv$.
We construct a tree decomposition $(T, \bag^*)$ of $G$ by constructing $\bag^*$ from $\bag$ by replacing each occurrence of a vertex $w_{uv}$ corresponding to a contracted edge by the two vertices $u,v \in V(G)$.
It is easy to observe that $(T, \bag^*)$ is a tree decomposition of $G$, and each adhesion (and bag) of $(T, \bag^*)$ has size at most twice the corresponding adhesion (or bag) of $(T,\bag)$.
Also, this construction can be implemented in $\OO(\|\Tc\|)$ time.

It remains to prove that $(T, \bag^*)$ is $(2k,k)$-unbreakable.
Suppose not, and let $t \in V(T)$ and $W = \bag^*(t)$ so that there exists a vertex cut $(A,B)$ of $G$ of order $|A \cap B|<k$ so that $|A \cap W| \ge 2k$ and $|B \cap W| \ge 2k$.
Now, for a vertex set $X \subseteq V(G)$, we define $X \contr M \subseteq V(G \contr M)$ as
\[X \contr M = (X \setminus V(M)) \cup \{w_{uv} : uv \in M \text{ and } \{u,v\} \cap X \neq \emptyset\}.\]
We observe that 
\[|X|/2 \le |X \contr M| \le |X|\]
and that $(A \contr M, B \contr M)$ is a vertex cut of $G \contr M$.
Moreover, because $(A,B)$ is vertex cut of $G$, it holds that $(A \contr M)\cap(B \contr M) = (A \cap B) \contr M$, so we obtain that the order of the vertex cut $(A \contr M, B \contr M)$ of $G \contr M$ is 
\[|(A \contr M)\cap(B \contr M)| = |(A \cap B) \contr M| \le |A \cap B| < k.\]

Then, we observe that $W \contr M = \bag(t)$ and 
\[|(W \contr M) \cap (A \contr M)| \ge |(W \cap A) \contr M| \ge |W \cap A|/2 \ge k,\]
and with a similar argument that $|(W \contr M) \cap (B \contr M)| \ge k$.
Therefore, the vertex cut $(A \contr M,B \contr M)$ witnesses that $W \contr M = \bag(t)$ is not $(k,k)$-unbreakable in $G \contr M$, which is a contradiction.
\end{proof}

Then, we show that we can find either a large matching or a large set of $I_k$-simplicial vertices.

\begin{lemma}
\label{lem:matchorsimpl}
There is an algorithm that, given an integer $k$ and a graph $G$ with $|E(G)| \le k \cdot |V(G)|$, in time $\OO(k \cdot \|G\|)$ returns either
\begin{itemize}
\item a matching $M$ in $G$ of size $|M| \ge |V(G)|/\OO(k^3)$, or
\item an independent set $I$ of $|I| \ge |V(G)|/\OO(1)$ vertices that are $I_k$-simplicial in $G$ and have degree at most $4k$.
\end{itemize}
\end{lemma}
\begin{proof}
First, we in time $\OO(\|G\|)$ compute an inclusion-wise maximal matching $M$ in $G$.
If $|M| \ge |V(G)|/(32 k^3)$, we conclude with the first case.
Otherwise, $V(M)$ forms a vertex cover of $G$ of size $|V(M)| < |V(G)|/(16 k^3)$ and our goal will be to conclude with the second case.

Let us denote $X = V(G) \setminus V(M)$.
Note that $X$ is an independent set in $G$.
Let $X' \subseteq X$ be the vertices in $X$ with degree at most $4k$.
We have that $|X| \ge |V(G)| - |V(G)|/(16 k^3)$, and because $|E(G)| \le k \cdot |V(G)|$, we have $|X \setminus X'| \le |V(G)|/4$, implying that $|X'| \ge |V(G)|/2$.

We define the graph $G^{\star}$ to be the graph with the vertex set $V(G^{\star}) = V(M)$, and edge set obtained by inserting an edge between $u,v \in V(G^{\star})$ if either $uv \in E(G)$, or if there exists a vertex $w \in X'$ with $\{u,v\} \subseteq N(w)$.
For $uv \in E(G^{\star}) \setminus E(G)$, we say that a vertex $w \in X'$ is \emph{$uv$-useful} if $\{u,v\} \subseteq N(w)$.
We assign for each $uv \in E(G^{\star}) \setminus E(G)$ a set $X'(uv)$ of $uv$-useful vertices, so that if the number of $uv$-useful vertices is at most $k$, then $X'(uv)$ contains all $uv$-useful vertices, and if the number of $uv$-useful vertices is more than $k$, then $X'(uv)$ contains an arbitrarily selected set of $k$ $uv$-useful vertices.
In particular, $|X'(uv)| \le k$.
Note that $\flowp_{G}(u,v) \ge |X'(uv)|$.

We construct the graph $G^{\star}$ and the sets $X'(uv)$ explicitly.
This can be done in $\OO(k \cdot \|G\|)$ time, because vertices in $X'$ have degree at most $4k$.
In particular, this is done by first constructing the list of all triples $(u,v,w)$ so that $w \in X'$ and $\{u,v\} \subseteq N(w)$, then sorting this list by $(u,v)$ with radix sort in linear time, and then inspecting the list in the sorted order.

We then apply the algorithm of \Cref{the:nispars} to compute a $4k^2$-sparsifier $G^{\star\star}$ of $G^{\star}$ with $|E(G^{\star\star})| \le 4k^2 \cdot |V(M)|$.

\begin{claim}
\label{lem:matchorsimpl:claim:highflow}
If $uv \in E(G^{\star}) \setminus (E(G^{\star\star}) \cup E(G))$, then $\flowp_{G}(u,v) \ge k$.
\end{claim}
\begin{claimproof}
Suppose that $uv \in E(G^{\star}) \setminus (E(G^{\star\star}) \cup E(G))$, but $\flowp_{G}(u,v) < k$.
Let $S \subseteq V(G) \setminus \{u,v\}$ be an proper $(u, v)$-separator of size $<k$ in $G$.
We construct a set $S^{\star}$ as
\[S^{\star} = (S \cap V(M)) \cup \{N(w) \setminus \{u,v\} : w \in S \cap X'\}.\]
We observe that $S^{\star}$ is a proper $(u,v)$-separator in the graph obtained from $G^{\star}$ after the removal of the edge $uv$, and therefore because $G^{\star\star}$ is a subgraph of $G^{\star}$ not containing $uv$, $S^{\star}$ is a proper $(u,v)$-separator in $G^{\star\star}$.

Furthermore, $|S^{\star}| < 4k^2$, because $|N(w)| \le 4k$ for all $w \in X'$.
This implies that $\flowp_{G^{\star\star}}(u,v) < 4k^2$.
However, by \Cref{lem:delhighconn} this contradicts that $G^{\star\star}$ is a $4k^2$-sparsifier of $G^{\star}$.
\end{claimproof}

We say that a vertex $w \in X'$ is \emph{meaningful} if $w \in X'(uv)$ for some $uv \in E(G^{\star\star}) \setminus E(G)$, and otherwise \emph{meaningless}.
Because $|E(G^{\star\star})| \le |V(M)| \cdot 4k^2$ and $|X'(uv)| \le k$, there are at most $|V(M)| \cdot 4k^3$ meaningful vertices.
As $|V(M)| \le |V(G)|/(16 k^3)$, the number of meaningful vertices is at most $|V(G)|/4$.
As $|X'| \ge |V(G)|/2$, this means that there are at least $|V(G)|/4$ meaningless vertices.

We return the set of meaningless vertices.
By the construction of $X'$ they are an independent set and have degree at most $4k$.
It remains to prove the following claim.
\begin{claim}
\label{lem:matchorsimpl:claim:meaningless}
Every meaningless vertex is $I_k$-simplicial in $G$.
\end{claim}
\begin{claimproof}
Let $w \in X'$ be a meaningless vertex.
Now, for every pair of distinct vertices $u,v \in N(w)$, one of the following holds.

\begin{enumerate}
\item $uv \in E(G)$,\label{lem:matchorsimpl:claim:meaningless:case1}
\item $uv \in E(G^{\star\star}) \setminus E(G)$ and $|X'(uv)| \ge k$, or \label{lem:matchorsimpl:claim:meaningless:case2}
\item $uv \in E(G^{\star}) \setminus (E(G^{\star\star}) \cup E(G))$.\label{lem:matchorsimpl:claim:meaningless:case3}
\end{enumerate}

Observe that $|X'(uv)| \ge k$ implies that $\flowp_{G}(u,v) \ge k$.
By \Cref{lem:matchorsimpl:claim:highflow}, the case of \Cref{lem:matchorsimpl:claim:meaningless:case3} also implies $\flowp_{G}(u,v) \ge k$.

Therefore, either $uv \in E(G)$, or $\flowp_{G}(u,v) \ge k$, implying that $w$ is $I_k$-simplicial.
\end{claimproof}
\Cref{lem:matchorsimpl:claim:meaningless} finishes the proof.
\end{proof}

\subsection{Putting things together}
It remains to put the tools provided in the previous subsection together to show \Cref{the:highlevel:bodl}.

As one more auxiliary lemma, we need the following algorithm about making $k$-lean tree decompositions reasonably small.

\begin{lemma}
\label{lem:leantdreducetotalsize}
There is an algorithm that, given a graph $G$ and a $k$-lean tree decomposition $\Tc$ of $G$, in time $\OO(\|G\| + \|\Tc\|)$ returns a $k$-lean tree decomposition $\Tc'$ of $G$ with $\|\Tc'\| = \OO(k \cdot |V(G)|)$.
\end{lemma}
\begin{proof}
Denote $\Tc = (T,\bag)$.
Let us choose an arbitrary node $r \in V(T)$ with $\bag(r)$ non-empty as the root of $T$.
Then, say that a non-root node $t \in V(T)$ is \emph{unnecessary} if $\bag(t) \subseteq \bag(p)$, where $p$ is the parent of $t$.
Note that in particular, in this case it holds that $\bag(t) = \adh(tp)$, implying that $|\bag(t)| < k$.

By contracting an unnecessary node $t$ into its parent $p$, we mean removing $t$, and adding all children of $t$ as children of $p$.
Our algorithm will contract unnecessary nodes into their parents as long as there exists any of them.
Let us first prove the correctness of this in the following two claims, and then discuss how to obtain the running time.

We first prove that the result tree decomposition is $k$-lean.

\begin{claim}
If $\Tc = (T,\bag)$ is $k$-lean, and $\Tc' = (T',\bag')$ results from contracting an unnecessary node $t$ into its parent $p$, then $\Tc'$ is also $k$-lean.
\end{claim}
\begin{claimproof}
First, this does not increase adhesion size because every new adhesion is a subset of $\adh(tp)$.
Then, we need to argue that for any two distinct nodes $t_1,t_2 \in V(T')$, the minimum adhesion size at edges on the $(t_1, t_2)$-path in $\Tc'$ is no larger than on the $(t_1, t_2)$-path in $\Tc$.
The only adhesions that can be on the $(t_1, t_2)$-path in $\Tc$ but not in $\Tc'$ are of form $\adh_{\Tc}(ct)$ for a child $c$ of $t$, and $\adh_{\Tc}(tp)$.
In the former case, the adhesion $\adh_{\Tc'}(cp) = \adh_{\Tc}(ct)$ is on the $(t_1, t_2)$-path in $\Tc'$.
In the latter case, the $(t_1, t_2)$-path in $\Tc'$ includes $p$ and a child $c$ of $t$, and therefore the adhesion $\adh_{\Tc'}(cp) \subseteq \adh_{\Tc}(tp)$ is on the $(t_1, t_2)$-path in $\Tc'$.
\end{claimproof}

We then prove that $\|\Tc\| = \OO(k \cdot |V(G)|)$.

\begin{claim}
If $\Tc = (T,\bag)$ has no unnecessary nodes and has adhesion size $<k$, then $\|\Tc\| = \OO(k \cdot |V(G)|)$.
\end{claim}
\begin{claimproof}
For a vertex $v \in V(G)$, let us say that the \emph{forget-node} of $v$ is the closest node $t \in V(T)$ to the root with $v \in \bag(t)$.
Because the root bag is non-empty and there are no unnecessary nodes, we have that every node is the forget-node of at least one vertex, implying that $|V(T)| \le |V(G)|$.
Then, we note that if $t \in V(T)$ has parent $p$, then $t$ is the forget-node for $|\bag(t)|-|\adh(tp)|$ vertices.
Therefore as the adhesion size is $<k$, we have that $\sum_{t \in V(T)} |\bag(t)| \le |V(G)|+(k-1) \cdot |V(T)| \le k \cdot |V(G)|$.
\end{claimproof}

A naive implementation of the process of iteratively contracting unnecessary nodes into their parents is too slow.
However, we can find a partition of $T$ into disjoint connected subtrees $T_1,\ldots,T_h$, so that for each $T_i$, the node of $T_i$ closest to the root is not unnecessary, but all other nodes in $T_i$ are unnecessary.
Now, the process of iteratively contracting unnecessary nodes gives the same result as contracting each $T_i$ into its root node.
This can be implemented in $\OO(\|\Tc\|)$ time with the help of a data structure that stores for each vertex $v \in V(G)$ the node of $\Tc$ closest to the root that contains $v$.
\end{proof}

Finally, we are ready to prove \Cref{the:highlevel:bodl}.

\thehighlevelbodl*
\begin{proof}
Let $\mathcal{A}$ be the hypothetical given algorithm.
We will describe a recursive algorithm that, with the help of $\mathcal{A}$, given an $n$-vertex graph $G$ with at most $kn$ edges, in time $\runtime(k) \cdot k^{\OO(1)} \cdot n$ returns a $k$-lean tree decomposition $\Tc$ of $G$ with $\|\Tc\| = \OO(kn)$.
By applying this with the sparsification procedure of \Cref{the:nispars} (which can be applied by \Cref{lem:sparsmaintain}), this implies the desired conclusion.

We first apply \Cref{lem:matchorsimpl} to in time $\OO(kn)$ compute either (1) a matching $M$ of size $|M| \ge n/\OO(k^3)$, or (2) an independent set $I$ with $|I| \ge n/\OO(1)$ vertices that are $I_k$-simplicial and have degree at most $4k$.

\paragraph{Case 1: A matching.}
If we get a matching, then let $G' = G \contr M$ be graph obtained from $G$ by contracting the matching.
We compute $G'$ in $\OO(\|G\|) = \OO(kn)$ time with the help of radix sort.
We also compute a $k$-sparsifier of $G'$ with at most $k \cdot |V(G')|$ edges with the algorithm of \Cref{the:nispars}, in time $\OO(\|G'\|) = \OO(kn)$.
Let $G''$ denote this $k$-sparsifier.

We have that $|V(G'')| \le n \cdot (1-1/\OO(k^3))$.
We call our algorithm recursively to compute a $k$-lean tree decomposition $\Tc'$ of $G''$.
By \Cref{lem:sparsmaintain}, $\Tc'$ is also a $k$-lean tree decomposition of $G'$.
We then apply \Cref{lem:matchinguncontract} with $G'$, $M$, and $\Tc'$ to compute in time $\OO(\|G'\| + \|\Tc'\|) = \OO(kn)$ a $(2k,k)$-unbreakable tree decomposition $\Tc$ of $G$ with adhesion size $\le 2k$.

Then, we apply $\mathcal{A}$ with $G$ and $\Tc$ to compute in time $\runtime(k) \cdot (\|G\|+\|\Tc\|) = \runtime(k) \cdot \OO(kn)$ a $k$-lean tree decomposition $\Tc^*$ of $G$.
Finally, we apply \Cref{lem:leantdreducetotalsize} to ensure that $\|\Tc^*\| \le \OO(kn)$ in time $\OO(\|G\| + \|\Tc^*\|) = \runtime(k) \cdot \OO(k n)$, and then return $\Tc^*$.

\paragraph{Case 2: A set of simplicial vertices.}
Suppose then that we have an independent set $I \subseteq V(G)$ of $|I| \ge n/\OO(1)$ vertices that are $I_k$-simplicial and have degree at most $4k$.
First, we compute in time $\OO(\|G\|)$ whether the majority of the vertices in $I$ have degree $<k$.
If the majority have degree $<k$, let $I'$ be those, and if the majority have degree between $k$ and $4k$, let $I'$ be those.

We compute $G' = G \elim I'$ in time $\OO(k^2 n)$ with the help of radix sort, and compute a $k$-sparsifier $G''$ of $G'$ with at most $k \cdot |V(G')|$ edges in time $\OO(\|G'\|) = \OO(k^2 n)$ with \Cref{the:nispars}.

We have that $|V(G'')| \le n \cdot (1-1/\OO(1))$.
We call our algorithm recursively to compute a $k$-lean tree decomposition $\Tc'$ of $G''$.
By \Cref{lem:sparsmaintain}, $\Tc'$ is also a $k$-lean tree decomposition of $G'$.

Depending on whether $I'$ consists of vertices of degree $<k$ or between $k$ and $4k$, we use either \Cref{lem:bodlputbacksimplsmalldeg} or \Cref{lem:bodlputbacksimplhighdeg} to compute a $(k,k)$-unbreakable tree decomposition $\Tc$ of $G$ with adhesion size $<k$ in time $\OO(\|G\| + \|\Tc'\|) = \OO(kn)$.
Then, we apply $\mathcal{A}$ with $G$ and $\Tc$ to compute in time $\runtime(k) \cdot (\|G\|+\|\Tc\|) = \runtime(k) \cdot \OO(kn)$ a $k$-lean tree decomposition $\Tc^*$ of $G$.
Then, we apply \Cref{lem:leantdreducetotalsize} to ensure that $\|\Tc^*\| \le \OO(kn)$ in time $\OO(\|G\| + \|\Tc^*\|) = \runtime(k) \cdot \OO(kn)$, and finally return $\Tc^*$.

\paragraph{Overall analysis.}
In both of the cases, the overall running time of the algorithm can be bounded by the recurrence
\[\oruntime(n,k) \le \runtime(k) \cdot \OO(k n) + \OO(k^2 n) + \oruntime(n \cdot (1-1/\OO(k^3)), k),\]
which can be bounded by $\oruntime(n,k) \le \runtime(k) \cdot k^{\OO(1)} \cdot n$.
\end{proof}

\section{Graph-theoretic lemmas}
\label{sec:gt}
In this section we prove various graph-theoretic lemmas about hypergraphs, which will then be used in later sections.
Most of these lemmas use only the symmetry and submodularity of the $\bdc$ function, and could directly be generalized to the setting of arbitrary symmetric submodular functions.

\subsection{Transitivity of well-linkedness}
We start by proving lemmas about the concepts of well-linkedness, tri-well-linkedness, $k$-well-linkedness, and $k$-better-linkedness, which were defined in \Cref{subsec:defs:graphtheory}.

Let $\prop$ be a property of hyperedge sets of hypergraphs, i.e., for every hypergraph $G$ and $A \subseteq E(G)$, either $\prop$ holds or does not hold for $A$ in $G$.
In this paper, $\prop$ will always be one of well-linkedness, tri-well-linkedness, $k$-well-linkedness, $k$-better-linkedness, or semi-internally connectedness.
We define that $\prop$ is \emph{transitive} if for every hypergraph $G$, a set $A \subseteq E(G)$ for which $\prop$ holds, and a set $B \subseteq E(G \rescliqs A)$ for which $\prop$ holds in $G \rescliqs A$, we have that $\prop$ holds for $B \orescliqs A$ in $G$.
In this subsection, we will prove that the properties of well-linkedness and tri-well-linkedness are transitive.

Let $\prop_1$ and $\prop_2$ be properties of hyperedge sets.
We define that $\prop_1$ is \emph{transitive along} sets satisfying $\prop_2$ if for every hypergraph $G$, a set $A \subseteq E(G)$ satisfying $\prop_2$, and a set $B \subseteq E(G \rescliqs A)$ satisfying $\prop_1$ in $G \rescliqs A$, we have that $\prop_1$ holds for $B \orescliqs A$ in $G$.
In this subsection, we will prove that $k$-well-linkedness is transitive along $k$-better-linked sets.

\subsubsection{Well-linkedness}
Let $G$ be a hypergraph.
Recall that a set $A \subseteq E(G)$ is well-linked if for all bipartitions $(B_1,B_2)$ of $A$, either $\bdc(B_1) \ge \bdc(A)$ or $\bdc(B_2) \ge \bdc(A)$.

The next lemma, sometimes called the ``posimodularity'' of symmetric submodular functions, is useful for providing an alternative characterization of well-linkedness.

\begin{lemma}
\label{lem:uncrosswl}
Let $B_1,B_2 \subseteq E(G)$.
Then, either $\bdc(B_1 \setminus B_2) \le \bdc(B_1)$ or $\bdc(B_2 \setminus B_1) \le \bdc(B_2)$.
\end{lemma}
\begin{proof}
If $\bdc(B_1 \setminus B_2) \le \bdc(B_1)$ we are done, so assume that $\bdc(B_1 \setminus B_2) = \bdc(B_1 \cap \co{B_2}) > \bdc(B_1)$.
By submodularity we get that $\bdc(B_1 \cup \co{B_2}) < \bdc(\co{B_2})$, and by symmetry this implies $\bdc(B_2 \setminus B_1) = \bdc(\co{B_1} \cap B_2) < \bdc(B_2)$.
\end{proof}

\Cref{lem:uncrosswl} gives the following alternative characterization of well-linkedness.

\begin{lemma}
\label{lem:uncrosswlappl}
A set $A \subseteq E(G)$ is well-linked if and only if for all pairs $B_1,B_2 \subseteq A$ with $B_1 \cup B_2 = A$ it holds that $\bdc(B_1) \ge \bdc(A)$ or $\bdc(B_2) \ge \bdc(A)$.
\end{lemma}
\begin{proof}
The if-direction is clear from the definition of well-linkedness.
For the only-if-direction, suppose that there exists a pair $B_1,B_2 \subseteq A$ so that $B_1 \cup B_2 = A$ and $\bdc(B_i) < \bdc(A)$ for both $i \in [2]$.
Now, by \Cref{lem:uncrosswl}, either the bipartition $(B_1 \setminus B_2, B_2)$ or the bipartition $(B_1, B_2 \setminus B_1)$ witnesses that $A$ is not well-linked.
\end{proof}

We then show that the property of well-linkedness is transitive.

\begin{lemma}
\label{lem:linkedcliq}
Let $A \subseteq E(G)$ be a well-linked set.
Let also $B \subseteq E(G \rescliqs A)$ be a well-linked set in $G \rescliqs A$.
Then, $B \orescliqs A$ is well-linked in $G$.
\end{lemma}
\begin{proof}
Denote $B' = B \orescliqs A$.
First, if $B$ does not contain the hyperedge $e_A$, then $B' = B$ is well-linked in $G$ because $\bdc_G(X) = \bdc_{G \rescliqs A}(X)$ for all $X \subseteq B$.
Then suppose that $e_A \in B$, implying in particular that $A \subseteq B'$.

Suppose $B'$ is not well-linked, and let $(C_1,C_2)$ be a bipartition of $B'$ so that $\bdc(C_i) < \bdc(B')$ for both $i \in [2]$.
By well-linkedness of $A$, we have that either $\bdc(A \cap C_1) \ge \bdc(A)$ or $\bdc(A \cap C_2) \ge \bdc(A)$.
By symmetry, assume that $\bdc(A \cap C_1) \ge \bdc(A)$.

We claim that now, the bipartition $(\{e_A\} \cup C_1 \setminus A, C_2 \setminus A)$ of $B$ contradicts that $B$ is well-linked in $G \rescliqs A$.

First, we have that $\bdc_{G \rescliqs A}(\{e_A\} \cup C_1 \setminus A) = \bdc_G(A \cup C_1)$, which by the assumption $\bdc(A \cap C_1) \ge \bdc(A)$ and submodularity is 
\[\bdc(A \cup C_1) \le \bdc(C_1) < \bdc(B') = \bdc(B).\]

Then, we have that $\bdc_{G \rescliqs A}(C_2 \setminus A) = \bdc_G(C_2 \setminus A)$, which by submodularity is at most 
\begin{align*}
\bdc(C_2 \cap \co{A}) &\le \bdc(C_2) + \bdc(A) - \bdc(C_2 \cup \co{A}) && \\
&= \bdc(C_2) + \bdc(A) - \bdc(\co{C_2} \cap A) && \text{(symmetry)}\\
&= \bdc(C_2) + \bdc(A) - \bdc(C_1 \cap A) && \text{(by $C_1 \cap A$ = $\co{C_2} \cap A$)}\\
&\le \bdc(C_2) < \bdc(B') = \bdc(B). && \text{(by assumption $\bdc(A \cap C_1) \ge \bdc(A)$)}
\end{align*}
\end{proof}

\subsubsection{Tri-well-linkedness}
Recall that a set $A \subseteq E(G)$ is tri-well-linked if for all tripartitions $(B_1,B_2,B_3)$ of $E(G)$ it holds that $\bdc(B_i) \ge \bdc(A)$ for some $i \in [3]$.

Again, we start by providing an alternative characterization of tri-well-linkedness.

\begin{lemma}
\label{lem:uncrosswlappltri}
A set $A \subseteq E(G)$ is tri-well-linked if and only if for all triples $B_1,B_2,B_3 \subseteq A$ with $B_1 \cup B_2 \cup B_3 = A$, for at least one $i \in [3]$ it holds that $\bdc(B_i) \ge \bdc(A)$.
\end{lemma}
\begin{proof}
The if-direction again comes directly from the definition of tri-well-linkedness.
For the only-if-direction, suppose there exists a triple $B_1,B_2,B_3 \subseteq A$ so that $B_1 \cup B_2 \cup B_3 = A$, and $\bdc(B_i) < \bdc(A)$ for all $i \in [3]$.
Moreover, choose such a triple that minimizes $|B_1|+|B_2|+|B_3|$.

Now, if $B_1$, $B_2$, and $B_3$ are disjoint, the tripartition $(B_1,B_2,B_3)$ witnesses that $A$ is not tri-well-linked.
If they intersect, assume \wilog that $B_1$ and $B_2$ intersect.
However, now by \Cref{lem:uncrosswl} either the triple $B_1 \setminus B_2, B_2, B_3$, or the triple $B_1, B_2 \setminus B_1, B_3$ contradicts the choice of the triple $B_1, B_2, B_3$.
\end{proof}

We then show that the property of tri-well-linkedness is transitive.

\begin{lemma}
\label{lem:triwelllinkedcom}
Let $A \subseteq E(G)$ be a tri-well-linked set.
Let also $B \subseteq E(G \rescliqs A)$ be a tri-well-linked set in $G \rescliqs A$.
Then, $B \orescliqs A$ is well-linked in $G$.
\end{lemma}
\begin{proof}
Denote $B' = B \orescliqs A$.
If $B$ does not contain the hyperedge $e_A$, then $B' = B$ is tri-well-linked in $G$ because $\bdc_{G}(X) = \bdc_{G \rescliqs A}(X)$ for all $X \subseteq B$.
Then, suppose that $e_A \in B$ implying $A \subseteq B'$.

Suppose that $B'$ is not tri-well-linked and let $(C_1,C_2,C_3)$ be a tripartition of $B'$ so that $\bdc(C_i) < \bdc(B')$ for all $i \in [3]$.
As in the proof of \Cref{lem:linkedcliq}, our strategy will now be to ``uncross'' the tripartition $(C_1,C_2,C_3)$ with $A$ to contradict the well-linkedness of $B$ in $G \rescliqs A$.
The following claim gives the uncrossing argument when $\bdc(A \cap C_i) < \bdc(A)$.

\begin{claim}
\label{lem:triwelllinkedcom:auxclaim1}
For any $i \in \{1,2,3\}$, if $\bdc(A \cap C_i) < \bdc(A)$, then $\bdc(C_i \setminus A) < \bdc(B)$.
\end{claim}
\begin{claimproof}
The assumption implies $\bdc(A \cap \co{C_i}) \ge \bdc(A)$ because $A$ is well-linked.
Now, 
\begin{align*}
\bdc(C_i \setminus A) = \bdc(C_i \cap \co{A}) &\le \bdc(C_i) + \bdc(A) - \bdc(C_i \cup \co{A}) && \text{(submodularity)}\\
&\le \bdc(C_i) + \bdc(A) - \bdc(A \cap \co{C_i}) && \text{(symmetry)}\\
&\le \bdc(C_i) < \bdc(B') = \bdc(B). && \text{(by $\bdc(A \cap \co{C_i}) \ge \bdc(A)$)}
\end{align*}
\end{claimproof}

When $\bdc(A \cap C_i) \ge \bdc(A)$, we can uncross directly by submodularity, as it implies $\bdc(A \cup C_i) \le \bdc(C_i) < \bdc(B') = \bdc(B)$.

Now, for all $i \in [3]$, we define that if $\bdc(A \cap C_i) < \bdc(A)$, then $C_i' = C_i \setminus A$, and if $\bdc(A \cap C_i) \ge \bdc(A)$, then $C_i' = \{e_A\} \cup C_i \setminus A$.
By the preceeding uncrossing arguments, we have that $\bdc_{G \rescliqs A}(C_i') < \bdc_{G \rescliqs A}(B)$ for all $i \in [3]$.
Furthermore, as $A$ is tri-well-linked, for at least one $i \in [3]$ it must hold that $\bdc(A \cap C_i) \ge \bdc(A)$, and therefore $e_A$ must be in at least one of $C_i'$, implying $C_1' \cup C_2' \cup C_3' = B$.
By \Cref{lem:uncrosswlappltri}, this contradicts the tri-well-linkedness of $B$ in $G \rescliqs A$.
\end{proof}

\subsubsection{$k$-well-linkedness}
Let $G$ be a hypergraph and $k$ an integer.
Recall that a set $A \subseteq E(G)$ is $k$-well-linked if for all bipartitions $(B_1,B_2)$ of $A$, it holds that either (1) $\bdc(B_1) \ge \bdc(A)$, (2) $\bdc(B_2) \ge \bdc(A)$, or (3) $\bdc(B_1) \ge k$ and $\bdc(B_2) \ge k$.
Recall also that a set $A \subseteq E(G)$ is $k$-better-linked if either (1) $A$ is well-linked, or (2) $A$ is $k$-well-linked and $\co{A}$ is well-linked.

Let us prove that $k$-well-linkedness is transitive along $k$-better-linked sets.

\begin{lemma}
\label{lem:kwltransitive}
Let $A \subseteq E(G)$ be a $k$-better-linked set.
Let also $B \subseteq E(G \rescliqs A)$ be a $k$-well-linked set in $G \rescliqs A$.
Then, $B \orescliqs A$ is $k$-well-linked in $G$.
\end{lemma}
\begin{proof}
Denote $B' = B \orescliqs A$.
First, if $B$ does not contain the hyperedge $e_A$, then $B' = B$ is $k$-well-linked in $G$ because $\bdc_G(X) = \bdc_{G \rescliqs A}(X)$ for all $X \subseteq B$.
Then suppose that $e_A \in B$, implying that $A \subseteq B'$.

Suppose $B'$ is not $k$-well-linked, and let $(C_1,C_2)$ be a bipartition of $B'$ so that $\bdc(C_i) < \bdc(B')$ for both $i \in [2]$ and $\bdc(C_i) < k$ for at least one $i \in [2]$.

\begin{claim}
Either $\bdc(A \cap C_1) \ge \bdc(A)$ or $\bdc(A \cap C_2) \ge \bdc(A)$.
\end{claim}
\begin{claimproof}
If $A$ is well-linked this follows directly from the definition.
Suppose then that $A$ is $k$-well-linked and $\co{A}$ is well-linked.
We are again done unless $\bdc(A \cap C_i) \ge k$ for both $i \in [2]$, so assume that indeed $\bdc(A \cap C_i) \ge k$ for both $i \in [2]$.

Let $i \in [2]$ so that $\bdc(C_i) < k$, and let $j=3-i$.
By the well-linkedness of $\co{A}$, either $\bdc(\co{A} \cap C_i) \ge \bdc(A)$ or $\bdc(\co{A} \cap \co{C_i}) \ge \bdc(A)$.
In the former case we get by submodularity that $\bdc(\co{A} \cup C_i) \le \bdc(C_i)$, which by symmetry means $\bdc(A \cap \co{C_i}) \le \bdc(C_i)$.
This in turn means that $\bdc(A \cap C_j) < k$, contradicting our assumption.

In the latter case we get by submodularity that $\bdc(\co{A} \cup \co{C_i}) \le \bdc(C_i)$, which by symmetry means $\bdc(A \cap C_i) \le \bdc(C_i) < k$, contradicting our assumption.
\end{claimproof}

Now, assume \wilog that $\bdc(A \cap C_1) \ge \bdc(A)$.
We claim that the bipartition $(\{e_A\} \cup C_1 \setminus A, C_2 \setminus A)$ of $B$ contradicts that $B$ is $k$-well-linked in $G \rescliqs A$.

First, we have that $\bdc_{G \rescliqs A}(\{e_A\} \cup C_1 \setminus A) = \bdc_G(A \cup C_1)$, which by $\bdc(A \cap C_1) \ge \bdc(A)$ and submodularity is
\[\bdc(A \cup C_1) \le \bdc(C_1).\]

Then, we have that $\bdc_{G \rescliqs A}(C_2 \setminus A) = \bdc_G(C_2 \setminus A)$, which by submodularity is at most

\begin{align*}
\bdc(C_2 \cap \co{A}) &\le \bdc(C_2) + \bdc(A) - \bdc(C_2 \cup \co{A}) && \\
&= \bdc(C_2) + \bdc(A) - \bdc(\co{C_2} \cap A) && \text{(symmetry)}\\
&= \bdc(C_2) + \bdc(A) - \bdc(C_1 \cap A) && \text{(by $C_1 \cap A$ = $\co{C_2} \cap A$)}\\
&\le \bdc(C_2). && \text{(by assumption $\bdc(A \cap C_1) \ge \bdc(A)$)}
\end{align*}

Therefore, as $\bdc_{G \rescliqs A}(\{e_A\} \cup C_1 \setminus A) \le \bdc(C_1)$ and $\bdc_{G \rescliqs A}(C_2 \setminus A) \le \bdc(C_2)$, the bipartition $(\{e_A\} \cup C_1 \setminus A, C_2 \setminus A)$ contradicts that $B$ is $k$-well-linked in $G \rescliqs A$.
\end{proof}

\subsection{Linkedness}
\label{subsec:linkedness}
We then prove lemmas about the concept of \emph{linkedness}, which we define now.

Let $A \subseteq B \subseteq E(G)$.
We say that $B$ is \emph{linked into} $A$ if for all $X$ with $A \subseteq X \subseteq B$ it holds that $\bdc(X) \ge \bdc(B)$.
Note that $B$ can be linked into $A$ only if $\bdc(A) \ge \bdc(B)$.
The graph-theoretic way of thinking about this is simply that there exists a set of $\bdc(B)$ vertex-disjoint paths from $\bd(B)$ to $\bd(A)$.

\subsubsection{Transitivity of linkedness}
We need the following lemma, which is clearly true from the vertex-disjoint paths viewpoint of linkedness, but let us prove it from the definitions using only submodularity.

\begin{lemma}
\label{lem:transitivityoflinkedness}
Let $A \subseteq B \subseteq C \subseteq E(G)$ so that $B$ is linked into $A$ and $C$ is linked into $B$.
Then, $C$ is linked into $A$.
\end{lemma}
\begin{proof}
Assume otherwise, and let $X$ be a set so that $A \subseteq X \subseteq C$ and $\bdc(X) < \bdc(C)$.
By submodularity, either (1) $\bdc(X \cap B) < \bdc(B)$ or (2) $\bdc(X \cup B) \le \bdc(X)$.
However, (1) would contradict that $B$ is linked into $A$, and (2) would contradict that $C$ is linked into $B$.
\end{proof}

\subsubsection{Linking into a well-linked set}
The following lemma gives the main application of linkedness, in particular, to be a tool for arguing that some sets are well-linked.

\begin{lemma}
\label{lem:linkedcombwelllinked}
Let $A \subseteq B \subseteq E(G)$.
If $A$ is well-linked and $B$ is linked into $A$, then also $B$ is well-linked.
\end{lemma}
\begin{proof}
Suppose not, and let $(C_1,C_2)$ be a bipartition of $B$ so that $\bdc(C_i) < \bdc(B)$ for both $i \in [2]$.
By well-linkedness of $A$, we have either $\bdc(A \cap C_1) \ge \bdc(A)$ or $\bdc(A \cap C_2) \ge \bdc(A)$.
By symmetry, assume without loss of generality that $\bdc(A \cap C_1) \ge \bdc(A)$.
We get by submodularity that $\bdc(A \cup C_1) \le \bdc(C_1) < \bdc(B)$, which contradicts that $B$ is linked into $A$ because $A \subseteq A \cup C_1 \subseteq B$.
\end{proof}

The essentially same proof works for tri-well-linkedness.

\begin{lemma}
\label{lem:linkedcombtriwelllinked}
Let $A \subseteq B \subseteq E(G)$.
If $A$ is tri-well-linked and $B$ is linked into $A$, then $B$ is also tri-well-linked.
\end{lemma}
\begin{proof}
Suppose not, and let $(C_1,C_2,C_3)$ be a tripartition of $B$ so that $\bdc(B \cap C_i) < \bdc(B)$ for all $i \in [3]$.
By tri-well-linkedness of $A$, we have that $\bdc(A \cap C_i) \ge \bdc(A)$ for some $i$.
By symmetry, assume without loss of generality that $\bdc(A \cap C_1) \ge \bdc(A)$.
We get by submodularity that $\bdc(A \cup C_1) \le \bdc(C_1) < \bdc(B)$, which contradicts that $B$ is linked into $A$.
\end{proof}

\subsubsection{Obtaining linked sets}
We then present arguments for arguing that a set $A$ is linked into a set $B$.
We start with the following simple lemma.

\begin{lemma}
\label{lem:linkednesscornercase}
Let $G$ be a hypergraph and $A \subseteq B \subseteq E(G)$ so that $\bd(B) \subseteq \bd(A)$.
Then, $B$ is linked into $A$.
\end{lemma}
\begin{proof}
Because $\bd(B) \subseteq \bd(A)$, each vertex $v \in \bd(B)$ is incident to a hyperedge in $A$ and to a hyperedge in $\co{B}$.
Therefore if $A \subseteq X \subseteq B$, then each $v \in \bd(B)$ is incident to a hyperedge in $A \subseteq X$ and a hyperedge in $\co{B} \subseteq \co{X}$, implying $\bd(B) \subseteq \bd(X)$.
\end{proof}

Then we consider minimal witnesses for non-$k$-well-linkedness.
Let $k$ be an integer.
A \emph{non-$k$-well-linkedness witness} for $A$ is a bipartition $(B_1,B_2)$ of $A$ so that $\bdc(B_1) < \min(\bdc(A), k)$ and $\bdc(B_2) < \bdc(A)$.
A \emph{minimal non-$k$-well-linkedness witness} for $A$ is a non-$k$-well-linkedness witness $(B_1,B_2)$ for $A$, that among all non-$k$-well-linkedness witnesses $(B_1,B_2)$ minimizes $\bdc(B_1)+\bdc(B_2)$. 

\begin{lemma}
\label{lem:minnonkwlwitnesslinked}
Let $A \subseteq E(G)$, $k$ an integer, and $(B_1,B_2)$ a minimal non-$k$-well-linkedness witness for $A$.
Then, both $\co{B_1}$ and $\co{B_2}$ are linked into $\co{A}$.
\end{lemma}
\begin{proof}
Suppose there is $i \in [2]$ so that $\co{B_i}$ is not linked into $\co{A}$, and let $X$ be a set with $\co{A} \subseteq X \subseteq \co{B_i}$ and $\bdc(X) < \bdc(B_i)$.
We claim that either $(\co{X}, X \cap A)$ or $(X \cap A, \co{X})$ is a non-$k$-well-linkedness witness for $A$ that contradicts the minimality of $(B_1,B_2)$.

Let $j = 3-i$.
By assumption, $\bdc(\co{X}) < \bdc(B_i)$, so it suffices to prove that $\bdc(X \cap A) < \bdc(B_j)$.
We have that
\begin{align*}
\bdc(X \cap A) = \bdc(X \cap B_j) &\le \bdc(X) + \bdc(B_j) - \bdc(X \cup B_j) && \text{(by submodularity)}\\
&\le \bdc(X) + \bdc(B_j) - \bdc(\co{B_i}) && \text{(by $X \cup B_j = \co{B_i}$)} \\
&< \bdc(B_j) && \text{(by $\bdc(X) < \bdc(B_i)$)}.
\end{align*}
\end{proof}

Note that \Cref{lem:minnonkwlwitnesslinked} can also be used in the context of well-linkedness by setting $k=\bdc(A)$.

We will furthermore need the following lemma related to proving that a set is linked into another.

\begin{lemma}
\label{lem:linkednessexpand}
Let $G$ be a hypergraph and $C \subseteq E(G)$ a $k$-better-linked set.
Let also $A \subseteq B \subseteq E(G \rescliqs C)$ so that $B$ is linked into $A$ in $G \rescliqs C$.
Assume that $\bdc(B) < k$.
Then, $B \orescliqs C$ is linked into $A \orescliqs C$ in $G$.
\end{lemma}
\begin{proof}
Let $A' = A \orescliqs C$ and $B' = B \orescliqs C$.
Let $X \subseteq E(G)$ be a set with $A' \subseteq X \subseteq B'$ witnessing that $B'$ is not linked into $A'$, i.e., with $\bdc(X) < \bdc(B')$.
If $e_C \in A$, then $C \subseteq X$ and $(X \setminus C) \cup \{e_C\}$ would witness that $B$ is not linked into $A$ in $G \rescliqs C$.
Similarly, if $e_C \notin B$, then $X$ would witness that $B$ is not linked into $A$ in $G \rescliqs C$.
Therefore, assume that $e_C \in B \setminus A$, implying $C \subseteq B' \setminus A'$.

By submodularity, we have that
\begin{equation}
\label{lem:linkednessexpand:eq1}
\bdc(X \cup C) \le \bdc(X) + \bdc(C) - \bdc(X \cap C),
\end{equation}
and
\begin{equation}
\label{lem:linkednessexpand:eq2}
\bdc(X \setminus C) = \bdc(X \cap \co{C}) = \bdc(\co{X} \cup C) \le \bdc(X) + \bdc(C) - \bdc(\co{X} \cap C).
\end{equation}
If $C$ is well-linked, we have that either $\bdc(X \cap C) \ge \bdc(C)$ or $\bdc(\co{X} \cap C) \ge \bdc(C)$.
In the former case, \Cref{lem:linkednessexpand:eq1} implies that $\bdc(X \cup C) \le \bdc(X)$, and in the latter case, \Cref{lem:linkednessexpand:eq2} implies that $\bdc(X \setminus C) \le \bdc(X)$.
In either case, we obtain a set $X'$ with $\bdc(X') \le \bdc(X)$ and $A' \subseteq X' \subseteq B'$, so that either $C \subseteq X'$ or $C$ is disjoint with $X'$.
Therefore, either $X'$ or $(X' \setminus C) \cup \{e_C\}$ would contradict that $B$ is linked into $A$ in $G \rescliqs C$.

The other case is that $C$ is $k$-well-linked and $\co{C}$ is well-linked.
In this case, \Cref{lem:linkednessexpand:eq1,lem:linkednessexpand:eq2} still hold, and the above two cases of $\bdc(X \cap C) \ge \bdc(C)$ and $\bdc(\co{X} \cap C) \ge \bdc(C)$ work similarly as previosly, but we also have the case that $\bdc(X \cap C) \ge k$ and $\bdc(\co{X} \cap C) \ge k$.
Because $k > \bdc(B) > \bdc(X)$, by submodularity the former implies $\bdc(\co{X} \cap \co{C}) < \bdc(C)$, and the latter implies $\bdc(X \cap \co{C}) < \bdc(C)$.
However, this would contradict that $\co{C}$ is well-linked.
\end{proof}

\subsection{Tangles}
We then prove graph-theoretic lemmas about tangles, which were defined in \Cref{subsec:defs:graphtheory}.

\subsubsection{Separations distinguishing tangles}
Recall that a separation $(A,\co{A})$ distinguishes two tangles $\tang_1$ and $\tang_2$ if $A \in \tang_1$ and $\co{A} \in \tang_2$.
The following lemma is useful when working with separations distinguishing tangles.

\begin{lemma}
\label{lem:tangledistpushing}
Let $G$ be a hypergraph and $(A,\co{A})$ a separation of $G$ that distinguishes a tangle $\tang_1$ from a tangle $\tang_2$.
If there exists a bipartition $(B_1,B_2)$ of $A$ so that $\bdc(B_i) \le \bdc(A)$ for both $i \in [2]$, then either $(B_1,\co{B_1})$ or $(B_2,\co{B_2})$ distinguishes $\tang_1$ from $\tang_2$.
\end{lemma}
\begin{proof}
Note that $B_i \in \tang_1$ for both $i \in [2]$ because they are subsets of $A$.
Then, it cannot be that $B_1 \in \tang_2$ and $B_2 \in \tang_2$, as with $\co{A} \in \tang_2$ they would violate the third tangle axiom.
Therefore, either $\co{B_1} \in \tang_2$ or $\co{B_2} \in \tang_2$, implying the conclusion.
\end{proof}

We then prove that if there exists a separation distinguishing two tangles, then there exists a separation distinguishing two tangles with specific properties.
This will be used in \Cref{sec:tangubrtorsos}.

\begin{lemma}
\label{lem:tangdistinguisherprops}
Let $G$ be a hypergraph.
If $G$ contains a separation of order $<k$ that distinguishes two tangles, then $G$ contains a separation $(A,\co{A})$ so that
\begin{itemize}
\item $(A,\co{A})$ distinguishes two tangles,
\item $\bdc(A) < k$,
\item $|V(A)| \le |V(\co{A})|$,
\item $(A,\co{A})$ is doubly well-linked, and
\item $A$ is internally connected.
\end{itemize}
\end{lemma}
\begin{proof}
Let $(A,\co{A})$ be a separation so that (1) $\bdc(A) < k$, (2) $(A,\co{A})$ distinguishes a tangle $\tang_1$ from a tangle $\tang_2$ (with $A \in \tang_1$ and $\co{A} \in \tang_2$), (3) subject to (1) and (2), $\bdc(A)$ is minimized, (4) subject to (1),(2), and (3), $|V(A)|$ is minimized, and (5) subject to (1), (2), (3), and (4), $|A|$ is minimized.
By assumption, $G$ contains a separation satisfying (1) and (2), and therefore $G$ contains a separation $(A,\co{A})$ satisfying all (1), (2), (3), (4), and (5).
We claim that $(A,\co{A})$ satisfies the required properties.

By assumption, $\bdc(A) < k$.
Furthermore, because (1), (2), and (3) do not change if $A$ is swapped with $\co{A}$, and by (4), $|V(A)|$ is minimized, we have that $|V(A)| \le |V(\co{A})|$.

\begin{claim}
\label{lem:tangdistinguisherprops:claim1}
Both $A$ and $\co{A}$ is well-linked.
\end{claim}
\begin{claimproof}
Let us prove that $A$ is well-linked.
The same argument works to show that $\co{A}$ is well-linked.

Suppose for a contradiction that there exists a bipartition $(B_1,B_2)$ of $A$ so that $\bdc(B_i) < \bdc(A)$ for both $i \in [2]$.
By \Cref{lem:tangledistpushing}, either $(B_1, \co{B_1})$ or $(B_2, \co{B_2})$ distinguishes $\tang_1$ from $\tang_2$.
However, these separations have order smaller than $(A,\co{A})$, so this contradicts the third condition in the choice of $(A,\co{A})$.
\end{claimproof}

\begin{claim}
\label{lem:tangdistinguisherprops:claim2}
$A$ is internally connected.
\end{claim}
\begin{claimproof}
Suppose for a contradiction that there exists a bipartition $(B_1,B_2)$ of $A$ so that $B_i$ is non-empty and $\bd(B_i) \subseteq \bd(A)$ for both $i \in [2]$.
By \Cref{lem:tangledistpushing}, now either $(B_1, \co{B_1})$ or $(B_2, \co{B_2})$ distinguishes $\tang_1$ from $\tang_2$.
For both $i \in [2]$, we have $\bdc(B_i) \le \bdc(A)$, $|V(B_i)| \le |V(A)|$, and $|B_i| < |A|$, so this contradicts the fourth condition in the choice of $(A,\co{A})$.
\end{claimproof}
\Cref{lem:tangdistinguisherprops:claim1,lem:tangdistinguisherprops:claim2} finish the proof.
\end{proof}

Then we show that separations distinguishing tangles can be uncrossed with mixed-$k$-well-linked separations.
This will be used in \Cref{sec:smalleradhesions}.

\begin{lemma}
\label{lem:tangledisuncross}
Let $(A,\co{A})$ be a separation of a hypergraph $G$ of order $<k$ that distinguishes two tangles.
Let also $(B,\co{B})$ be a mixed-$k$-well-linked separation of $G$.
Then there are orientations $(A',\co{A'})$ of $(A,\co{A})$ and $(B',\co{B'})$ of $(B,\co{B})$ so that $(A' \cup B', \co{A'} \cap \co{B'})$ distinguishes two tangles.
\end{lemma}
\begin{proof}
Assume \wilog that $B$ is well-linked and $\co{B}$ is $k$-well-linked.

First, consider the case that $\bdc(A \cap \co{B}) > \bdc(A)$.
By submodularity and symmetry, we have that $\bdc(\co{A} \cap B) < \bdc(B)$.
By well-linkedness of $B$, this implies $\bdc(A \cap B) \ge \bdc(B)$, which in turn by submodularity and symmetry implies $\bdc(\co{A} \cap \co{B}) \le \bdc(A)$.
Then, because $\bdc(A) < k$ and $\co{B}$ is $k$-well-linked, the inequality $\bdc(\co{A} \cap \co{B}) \le \bdc(A)$ implies that $\bdc(A \cap \co{B}) \ge \bdc(B)$, which in turn by submodularity and symmetry implies $\bdc(\co{A} \cap B) \le \bdc(A)$.

So in the end, $\bdc(A \cap \co{B}) > \bdc(A)$ implies that $\bdc(\co{A} \cap \co{B}) \le \bdc(A)$ and $\bdc(\co{A} \cap B) \le \bdc(A)$.
Applying \Cref{lem:tangledistpushing} (substituting $\co{A}$ for $A$) tells that then either $(\co{A} \cap B, \co{\co{A} \cap B})$ or $(\co{A} \cap \co{B}, \co{\co{A} \cap \co{B}})$ distinguishes two tangles, so we are done in this case.

The remaining case is that $\bdc(A \cap \co{B}) \le \bdc(A)$.
Because $\bdc(A) < k$ and $\co{B}$ is $k$-well-linked, this implies $\bdc(\co{A} \cap \co{B}) \ge \bdc(B)$, implying by submodularity and symmetry that $\bdc(A \cap B) \le \bdc(A)$.
We are again done by applying \Cref{lem:tangledistpushing}.
\end{proof}

\subsubsection{Uncrossing tangles with tri-well-linked sets}
Then we show that tri-well-linked sets are useful for ``uncrossing'' tangles, in some sense.
These statements will be used later in this section to prove statements that will be used in \Cref{sec:fromtubrtoubr}.

First we show that if $A \subseteq E(G)$ is a tri-well-linked set, then tangles of $G \rescliqs A$ satisfying some conditions can be used to construct tangles of $G$.

\begin{lemma}
\label{lem:tanglelifting}
Let $G$ be a hypergraph and $A \subseteq E(G)$ a tri-well-linked set.
Suppose $G \rescliqs A$ contains a tangle $\tang$ of order $k$, so that there exists $C \subseteq E(G \rescliqs A)$ with $e_A \in C \in \tang$, where $e_A$ is the hyperedge of $G \rescliqs A$ corresponding to $A$.
Then, $G$ contains a tangle $\tang'$ of order $k$, so that for all $D \subseteq E(G \rescliqs A)$, $D \in \tang$ implies $D \orescliqs A \in \tang'$.
\end{lemma}
\begin{proof}
We construct a tangle $\tang'$ of order $k$ of $G$ as follows.
Let $B \subseteq E(G)$.
Let $B^{-} = B \setminus A$ and $B^{+} = (B \setminus A) \cup \{e_A\}$, and note that $B^{-},B^{+} \subseteq E(G \rescliqs A)$.
We define the tangle $\tang'$ as $\tang' = \{B \subseteq E(G) \mid \bdc(B) < k \text{ and } (B^{-} \in \tang \text{ or } B^{+} \in \tang)\}$.
From the definition we get that for all $D \subseteq E(G \rescliqs A)$, $D \in \tang$ implies $D \orescliqs A \in \tang'$.

We then need to prove that $\tang'$ in fact is a tangle of order $k$ of $G$.
The first tangle axiom is immediate from the construction of $\tang'$.
For the second tangle axiom, we start with the following claim.

\begin{claim}
\label{lem:tanglelifting:claim1}
For all $B \subseteq E(G)$, either $\bdc(B^{-}) \le \bdc(B)$ or $\bdc(B^{+}) \le \bdc(B)$.
\end{claim}
\begin{claimproof}
Suppose the contrary.
Then, $\bdc(B \cap \co{A}) > \bdc(B)$, implying by submodularity that $\bdc(\co{B} \cap A) = \bdc(B \cup \co{A}) < \bdc(A)$, and $\bdc(B \cup A) > \bdc(B)$, implying by submodularity that $\bdc(B \cap A) < \bdc(A)$.
However, this contradicts that $A$ is well-linked.
\end{claimproof}

We then complete the proof of the second tangle axiom in the following claim.

\begin{claim}
\label{lem:tanglelifting:claim2}
For all $B \subseteq E(G)$ with $\bdc(B) < k$, either $B \in \tang'$ or $\co{B} \in \tang'$.
\end{claim}
\begin{claimproof}
Suppose for a contradiction that there exists $B \subseteq E(G)$ so that $\bdc(B) < k$, but $B,\co{B} \notin \tang'$.
By \Cref{lem:tanglelifting:claim1}, either $\bdc(B^{-}) < k$ or $\bdc(B^{+}) < k$.

Suppose first that $\bdc(B^{-}) < k$.
By construction of $\tang'$ we have that $B^{-} \notin \tang$, so by the second tangle axiom for $\tang$ we have that $\co{B^{-}} \in \tang$.
However, $\co{B^{-}} = \co{B}^{+}$, so we get that $\co{B} \in \tang'$.

The argument is similar for the other case:
Suppose then that $\bdc(B^{+}) < k$.
By construction of $\tang'$ we have that $B^{+} \notin \tang$, so by the second tangle axiom for $\tang$ we have that $\co{B^{+}} \in \tang$.
However, $\co{B^{+}} = \co{B}^{-}$, so we get that $\co{B} \in \tang'$.
\end{claimproof}

We then prove the third tangle axiom.

\begin{claim}
For all $B_1,B_2,B_3 \in \tang'$, it holds that $B_1 \cup B_2 \cup B_3 \neq E(G)$.
\end{claim}
\begin{claimproof}
Suppose for a contradiction that there exists $B_1,B_2,B_3 \in \tang'$ so that $B_1 \cup B_2 \cup B_3 = E(G)$.
First, suppose that $B_1^{+} \in \tang$.
Now, as either $B_2^{-}$ or $B_2^{+}$ is in $\tang$, and either $B_3^{-}$ or $B_3^{+}$ is in $\tang$, we have that $\tang$ contradicts the third tangle axiom.
Therefore, $B_1^{+} \notin \tang$.
By repeating the same argument with $B_1^{+}$ replaced by $B_i^{+}$, we get that $B_i^{+} \notin \tang$ for all $i \in [3]$.
This implies in particular that $B_i^{-} \in \tang$ for all $i \in [3]$.

Suppose then that $\bdc(B_2^{-} \cup B_3^{-}) < k$.
Let $C \in \tang$ so that $e_A \in C$.
Now, if $B_2^{-} \cup B_3^{-} \in \tang$, the triple of $C$, $B_1^{-}$, and $B_2^{-} \cup B_3^{-}$ would show that $\tang$ contradicts the third tangle axiom.
However, if $\co{B_2^{-} \cup B_3^{-}} \in \tang$, then the triple of $\co{B_2^{-} \cup B_3^{-}}$, $B_2^{-}$, and $B_3^{-}$ would show that $\tang$ contradicts the third tangle axiom.
Therefore, we have that $\bdc(B_2^{-} \cup B_3^{-}) \ge k$.
Now, by submodularity,
\begin{align*}
&&\bdc(A \cap B_1) &\le \bdc(A) + \bdc(B_1) - \bdc(A \cup B_1)&&\\
\Rightarrow&&\bdc(A \cap B_1) &\le \bdc(A) + \bdc(B_1) - \bdc(B_2^{-} \cup B_3^{-})&& \text{(by $\co{A \cup B_1} = B_2^{-} \cup B_3^{-}$)}\\
\Rightarrow&&\bdc(A \cap B_1) &< \bdc(A) && \text{(by $\bdc(B_2^{-} \cup B_3^{-}) \ge k > \bdc(B_1)$)}.
\end{align*}
By repeating the same argument with $B_1$ replaced by $B_i$, we get that $\bdc(A \cap B_i) < \bdc(A)$ for all $i \in [3]$.
By \Cref{lem:uncrosswlappltri}, this contradicts the fact that $A$ is tri-well-linked.
\end{claimproof}

Finally, we prove the fourth tangle axiom.

\begin{claim}
\label{lem:tanglelifting:claim4}
For all $e \in E(G)$, it holds that $E(G) \setminus \{e\} \notin \tang'$.
\end{claim}
\begin{claimproof}
Let $B = E(G) \setminus \{e\}$ for some $e \in E(G)$.
We can assume that $\bdc(B) = \bdc(e) < k$, as otherwise we are done.
First, if $e \in A$, then $B^{-} = E(G \rescliqs A) \setminus \{e_A\}$ and $B^{+} = E(G \rescliqs A)$, so by the third and the fourth tangle axioms for $\tang$, neither of them is in $\tang$.

Second, if $e \notin A$, then $\{e\} \in \tang$, $B^{-} = E(G \rescliqs A) \setminus \{e, e_A\}$, and $B^{+} = E(G \rescliqs A) \setminus \{e\}$.
The third tangle axiom implies directly that $B^{+} \notin \tang$.
By the assumption that there exists $C \subseteq E(G \rescliqs A)$ with $e_A \in C \in \tang$, we also get from the third tangle axiom that $B^{-} \notin \tang$.
\end{claimproof}
\Cref{lem:tanglelifting:claim4} finishes the proof.
\end{proof}

Then we show that a doubly tri-well-linked separation can be used to construct a tangle.

\begin{lemma}
\label{lem:doublytriwltangle}
Let $G$ be a hypergraph and $(A,\co{A})$ a doubly tri-well-linked separation of $G$.
Then $G$ contains a tangle $\tang$ of order $\bdc(A)$, so that for all $B \subseteq A$ and $B \subseteq \co{A}$ with $\bdc(B) < \bdc(A)$, we have that $B \in \tang$.
\end{lemma}
\begin{proof}
We define the tangle $\tang$ as follows.
For every $B \subseteq E(G)$ with $\bdc(B) < \bdc(A)$, we let $B \in \tang$ if $\bdc(B \cap A) \le \bdc(B)$.

Now, if $B \subseteq A$, then $B \in \tang$ because $\bdc(B \cap A) = \bdc(B) \le \bdc(B)$.
If $B \subseteq \co{A}$, then $B \in \tang$ because $\bdc(B \cap A) = \bdc(\emptyset) = 0 \le \bdc(B)$.

The first tangle axiom clearly holds.
We then prove that $\tang$ satisfies the second tangle axiom.

\begin{claim}
If $B \subseteq E(G)$ and $\bdc(B) < \bdc(A)$, then either $B \in \tang$ or $\co{B} \in \tang$.
\end{claim}
\begin{claimproof}
Suppose $B \notin \tang$, implying $\bdc(B \cap A) > \bdc(B)$.
By submodularity, we get $\bdc(\co{B} \cap \co{A}) < \bdc(A)$.
Because $\co{A}$ is well-linked, this implies that $\bdc(B \cap \co{A}) \ge \bdc(A)$.
This in turn, by submodularity, implies that $\bdc(\co{B} \cap A) \le \bdc(B)$, implying that $\co{B} \in \tang$.
\end{claimproof}

We then prove that $\tang$ satisfies the third tangle axiom.

\begin{claim}
If $B_1,B_2,B_3 \in \tang$, then $B_1 \cup B_2 \cup B_3 \neq E(G)$.
\end{claim}
\begin{claimproof}
Suppose there are $B_1, B_2, B_3 \in \tang$ so that $B_1 \cup B_2 \cup B_3 = E(G)$.
But then, $\bdc(B_i \cap A) \le \bdc(B_i) < \bdc(A)$ for all $i \in [3]$, and thus by \Cref{lem:uncrosswlappltri}, this would contradict that $A$ is tri-well-linked.
\end{claimproof}

We then prove that $\tang$ satisfies the fourth tangle axiom.

\begin{claim}
\label{lem:doublytriwltangle:claim3}
For all $e \in E(G)$, it holds that $E(G) \setminus \{e\} \notin \tang$.
\end{claim}
\begin{claimproof}
Suppose there exists $e \in E(G)$ so that $E(G) \setminus \{e\} \in \tang$.
Denote $B = E(G) \setminus \{e\}$.
Now, $\bdc(B)<\bdc(A)$ and $\bdc(B \cap A) \le \bdc(B)$.
This implies that $B \cap A \neq A$, implying $e \in A$.
By well-linkedness of $A$, this yields $\bdc(e) \ge \bdc(A)$, implying that $\bdc(E(G) \setminus \{e\}) \ge \bdc(A)$, contradicting that $E(G) \setminus \{e\} \in \tang$.
\end{claimproof}
\Cref{lem:doublytriwltangle:claim3} finishes the proof.
\end{proof}

\subsubsection{Tangle-unbreakability}
Recall that a hypergraph $G$ is $k$-tangle-unbreakable if no separation of $G$ of order $<k$ distinguishes two tangles.
We now apply \Cref{lem:tanglelifting,lem:doublytriwltangle} to prove two structural results about doubly tri-well-linked separations in $k$-tangle-unbreakable graphs.
These will be used in \Cref{sec:fromtubrtoubr}.

First, we show a powerful result asserting that, informally speaking, the branchwidth of one side of a doubly tri-well-linked separation of order $k$ in a $k$-tangle-unbreakable graph is bounded.

\begin{lemma}
\label{lem:tangleunbreakablebwsmall}
Let $G$ be a hypergraph that is $k$-tangle-unbreakable, and $(A,\co{A})$ a doubly tri-well-linked separation of order $\bdc(A) < k$.
Then, $\bw(G \rescliqs A) \le \bdc(A)$ or $\bw(G \rescliqs \co{A}) \le \bdc(A)$.
\end{lemma}
\begin{proof}
If $\bw(G \rescliqs A) > \bdc(A)$, then by \Cref{lem:tanglebwduality}, $G \rescliqs A$ contains a tangle $\tang_1$ of order $\bdc(A)+1$.
Let $e_A$ be the hyperedge of $G \rescliqs A$ corresponding to $A$, and note that $\{e_A\} \in \tang_1$, and therefore we can apply \Cref{lem:tanglelifting}.
It yields that $G$ contains a tangle $\tang_1'$ of order $\bdc(A)+1$ so that $A \in \tang_1'$.

By repeating the same argument with $A$ replaced by $\co{A}$, we get that if $\bw(G \rescliqs \co{A}) > \bdc(A)$, then $G$ contains a tangle $\tang_2'$ of order $\bdc(A)+1$ so that $\co{A} \in \tang_2'$.

However, then the combination of $\bw(G \rescliqs A) > \bdc(A)$ and $\bw(G \rescliqs \co{A}) > \bdc(A)$ implies that $(A,\co{A})$ distinguishes the tangles $\tang_1'$ and $\tang_2'$.
\end{proof}

Then, we show that decomposing $k$-tangle-unbreakable graphs by doubly tri-well-linked separations keeps them $k$-tangle-unbreakable.

\begin{lemma}
\label{lem:tungunbhered}
Let $G$ be a hypergraph that is $k$-tangle-unbreakable, and $(A,\co{A})$ a doubly tri-well-linked separation.
Then, both $G \rescliqs A$ and $G \rescliqs \co{A}$ are $k$-tangle-unbreakable.
\end{lemma}
\begin{proof}
By symmetry, it suffices to prove that $G \rescliqs A$ is $k$-tangle-unbreakable.
Assume for a contradiction that $G \rescliqs A$ has a separation $(B,\co{B})$ of order $<k$ that distinguishes two tangles $\tang_1$ and $\tang_2$ of order $\bdc(B)+1$ (with $B \in \tang_1$ and $\co{B} \in \tang_2$).
Furthermore, assume \wilog that the hyperedge $e_A$ corresponding to $A$ is in the set $B$.

Because $e_A \in B \in \tang_1$, we can apply \Cref{lem:tanglelifting} with $\tang_1$ to obtain a tangle $\tang_1'$ of order $\bdc(B)+1$ in $G$, with $B \orescliqs A = (B \setminus \{e_A\}) \cup A \in \tang_1'$.

Then we consider cases.
First, suppose that $\bdc(e_A) = \bdc(A) \le \bdc(B)$.
Then, it must be that $\{e_A\} \in \tang_2$, so we can apply \Cref{lem:tanglelifting} with $\tang_2$ to obtain a tangle $\tang_2'$ of order $\bdc(B)+1$ in $G$, with $\co{B} \orescliqs A = \co{B} \in \tang'_2$.
Note that $(B \orescliqs A, \co{B})$ is a separation of $G$ of order $<k$, so this contradicts that $G$ is $k$-tangle-unbreakable.

Otherwise, $\bdc(e_A) = \bdc(A) > \bdc(B)$.
In that case, consider the tangle $\tang^*$ of $G$ of order $\bdc(A)$ given by \Cref{lem:doublytriwltangle}, having $C \in \tang^*$ for all $C \subseteq E(G)$ with $\bdc(C) < \bdc(A)$ and either $C \subseteq A$ or $C \subseteq \co{A}$.
Now, because $\co{B} \subseteq \co{A}$, we have that $\co{B} \in \tang^*$.
Again, $B \orescliqs A \in \tang_1'$ and $(B \orescliqs A, \co{B})$ is a separation of $G$ of order $<k$, so this contradicts that $G$ is $k$-tangle-unbreakable.
\end{proof}

\section{Algorithms for finding separations}
\label{sec:localsearch}
The main algorithmic subroutines in our algorithm are about finding separations of hypergraphs.
This section is dedicated to giving these subroutines.

\subsection{Well-linkedness}
We start by showing that non-$k$-well-linkedness witnesses can be found efficiently.
We first give a technical lemma and then wrap it in a nicer statement.

\begin{lemma}
\label{lem:algffwlpre}
Let $G$ be a hypergraph whose representation is already stored.
There is an algorithm that, given a set $A \subseteq E(G)$ and integers $b_1$, $b_2$, in time $2^{\OO(\bdc(A))} \cdot \rank(G)^2 \cdot |A|$ returns a bipartition $(B_1,B_2)$ of $A$ so that $\bdc(B_i) \le b_i$ for each $i \in [2]$, if such a bipartition exists.
\end{lemma}
\begin{proof}
Let $G^*$ be the graph constructed by taking $V(G^*) = V(A)$ and inserting an edge between $u,v \in V(G^*)$ if there is $e \in A$ with $u,v \in V(e)$.
We have that such a bipartition $(B_1,B_2)$ of $A$ corresponds to a vertex cut $(C_1, C_2)$ of $G^*$ so that $|(C_i \cap \bd(A)) \cup (C_1 \cap C_2)| \le b_i$ for each $i \in [2]$.
In particular, given such a bipartition $(B_1,B_2)$ of $A$, such a vertex cut is given by $(V(B_1), V(B_2))$, and given such a vertex cut $(C_1,C_2)$ of $G^*$, such a bipartition $(B_1,B_2)$ can be constructed by assigning a hyperedge $e \in A$ to $B_1$ if $V(e) \subseteq C_1$, and to $B_2$ if $V(e) \subseteq C_2$ (and arbitrarily if $V(e) \subseteq C_1 \cap C_2$).

We can construct the graph $G^*$ in time $\OO(|A| \cdot \rank(G)^2)$.
Then, the problem of finding such a vertex cut $(C_1,C_2)$ boils down to trying all $2^{\bdc(A)}$ possibilities of $X_1 = C_1 \cap \bd(A)$ and $X_2 = C_2 \cap \bd(A)$ and for each computing a minimum order vertex cut $(C_1,C_2)$ with $X_i \subseteq C_i$ for both $i \in [2]$, by using the Ford-Fulkerson algorithm, which runs in time $\OO(\bdc(A) \cdot \|G^*\|)$.
Therefore, in total the running time is $\OO(|A| \cdot \rank(G)^2) + 2^{\OO(\bdc(A))} \cdot \|G^*\| = 2^{\OO(\bdc(A))} \cdot \rank(G)^2 \cdot |A|$.
\end{proof}

We then give the nicer statement.
Note that even though this is stated in terms of $k$-well-linkedness, it can be used for well-linkedness by setting $k=\bdc(A)$.

\begin{lemma}
\label{lem:algffwl}
Let $G$ be a hypergraph whose representation is already stored.
There is an algorithm that, given a set $A \subseteq E(G)$ and an integer $k$, in time $2^{\OO(\bdc(A))} \cdot \rank(G)^2 \cdot |A|$ either
\begin{itemize}
\item correctly concludes that $A$ is $k$-well-linked, or
\item returns a bipartition $(B_1,B_2)$ of $A$ so that 
\begin{itemize}
\item $\bdc(B_1) < \min(\bdc(A), k)$,
\item $\bdc(B_2) < \bdc(A)$, and
\item both $\co{B_1}$ and $\co{B_2}$ are linked into $\co{A}$.
\end{itemize}
\end{itemize}
\end{lemma}
\begin{proof}
It suffices to apply the algorithm of \Cref{lem:algffwlpre} for all $b_1 < \min(\bdc(A), k)$ and $b_2 < \bdc(A)$.
If it provides no output for any such choice, then $A$ is $k$-well-linked, and otherwise the bipartition with minimal $b_1+b_2$ guarantees the properties (in particular, linkedness by \Cref{lem:minnonkwlwitnesslinked}).
\end{proof}

\subsection{Local search in hypergraphs}
The basic ingredient for giving algorithms for finding separations in unbreakable hypergraphs is a local search routine that given a hyperedge $e \in E(G)$, enumerates all small internally connected sets $A$ containing $e$.
In fact, we need a slight generalization of this statement, also addressing the technicality that the multiplicity and the rank of $G$ should be taken into account.

\begin{lemma}
\label{lem:localchipalgo}
Let $G$ be a hypergraph whose representation is already stored.
There is an algorithm that, given a set $I \subseteq E(G)$, and integers $p,s,k \ge 0$, in time $\OO(|I| \cdot \rank(G)) + p \cdot s^{\OO(k+\rank(G))}$ either (1) lists all sets $A \subseteq E(G)$ so that
\begin{itemize}
\item $I \subseteq A$,
\item $|V(A)| \le s$,
\item $\bdc(A) \le k$, and
\item every internal component of $A$ intersects $I$,
\end{itemize}
or (2) returns a multiplicity-$p$-violator in $G$.
Furthermore, in the case (1), the number of such sets $A$ is at most $s^{\OO(k)}$.
\end{lemma}
\begin{proof}
We will describe a recursive algorithm, that works otherwise as the algorithm described in the statement of the lemma, but additionally takes as input a set $X \subseteq V(I)$, and outputs only such sets $A$ with $X \subseteq \bd(A)$.
Calling this algorithm with $X = \emptyset$ gives the algorithm of the statement.

First, we compute in time $\OO(|I| \cdot \rank(G))$ the set $V(I)$.
If $|V(I)| > s$, we can immediately return.
Then, if $|I| > p \cdot (s+1)^{\rank(G)}$, there must be a multiplicity-$p$-violator that is a subset of $I$, and by taking an arbitrary subset $I'$ of $I$ of size $|I'| = p \cdot (s+1)^{\rank(G)} + 1$ and using radix sort for finding duplicates, we can find and return such a subset in $p \cdot s^{\OO(\rank(G))}$ time.

Now we can assume that $|V(I)| \le s$ and $|I| \le p \cdot (s+1)^{\rank(G)}$.
If $|X| > k$, we can also return immediately.
Then we compute $\bd(I)$ in time $\OO(|I| \cdot \rank(G)) = p \cdot s^{\OO(k+\rank(G))}$.
If $X$ intersects $V(I) \setminus \bd(I)$, we can again return the empty list.
If $X = \bd(I)$, then there is exactly one set $A$ satisfying the conditions, namely $I$ itself, so we can output it.

Then assume that $X \subsetneq \bd(I)$, and choose an arbitrary $v \in \bd(I) \setminus X$.
For any set $A$ satisfying the conditions it holds that either $v \in \bd(A)$ or $v \in V(A) \setminus \bd(A)$.
We now make two recursive calls, one enumerating the sets $A$ of the first type and other enumerating the sets $A$ of the latter type.
For the first type, it suffices to simply add $v$ to $X$.

For the second type, we need to add all hyperedges in $\ninc(v)$ to $I$.
If $|\ninc(v)| \le p \cdot (s+1)^{\rank(G)}$, we simply do this, recursing with $I \cup \ninc(v)$.
If $|\ninc(v)| > p \cdot (s+1)^{\rank(G)}$, we take a subset $Q \subseteq \ninc(v)$ of size $|Q| = p \cdot (s+1)^{\rank(G)}+1$.
We first check in $p \cdot s^{\OO(\rank(G))}$ time whether $Q$ contains a multiplicity-$p$-violator, and if yes, return that.
If not, then it must be that $|V(Q)| > s$, implying that $|V(\ninc(v))| > s$, so the recursive call with $I \cup \ninc(v)$ would immediately return the empty list anyway, so we can omit it.

This concludes the description of the algorithm.
We observe that the first recursive call takes $\OO(|I| \cdot \rank(G)) + p \cdot s^{\OO(\rank(G))}$ time, and the subsequent ones take $p \cdot s^{\OO(\rank(G))}$ time.
Therefore, if the total size of the recursion tree is $R$, then the running time can be bounded by $\OO(|I| \cdot \rank(G)) + R^{\OO(1)} \cdot p \cdot s^{\OO(\rank(G))}$.
Furthermore, the output-size can be bounded by $R$.
It remains to prove that $R \le s^{\OO(k)}$.

To bound $R$, we observe that in a recursive call of the first type, the size of $X$ increases, and in a recursive call of the second type, $|V(I) \setminus \bd(I)|$ increases.
As $|X|$ is bounded by $k$ and $|V(I) \setminus \bd(I)|$ by $s$, we get that $R \le \binom{s+k}{k} \le s^{\OO(k)}$.
\end{proof}

We then apply \Cref{lem:localchipalgo} to give an algorithm for verifying if a set is $k$-well-linked or returning a non-$k$-well-linkedness witness of it.
For generality, we phrase the algorithm in terms of $k$-well-linkedness, but note that by setting $k = \bdc(A)$, it can be also used in the context of well-linkedness.

\begin{lemma}
\label{lem:kwlubalg}
Let $G$ be a hypergraph whose representation is already stored.
Let also $s, k \ge 1$ be integers so that $G$ is $(s,k)$-unbreakable.
There is an algorithm that, given $s$, $k$, a set $A \subseteq E(G)$, and an integer $p$, in time $|A| \cdot p \cdot s^{\OO(k+\rank(G))} \cdot 2^{\OO(\bdc(A))}$ either
\begin{itemize}
\item outputs a multiplicity-$p$-violator in $G$,
\item correctly concludes that $\co{A}$ is $k$-well-linked, or
\item outputs a set $B_1 \subseteq \co{A}$ so that for $B_1$ and $B_2 = \co{A} \setminus B_1$ it holds that
\begin{itemize}
\item $|V(B_1)| < s$, 
\item $\bdc(B_1) < \bdc(A)$ and $\bdc(B_2) < \bdc(A)$,
\item either $\bdc(B_1) < k$ or $\bdc(B_2) < k$, and
\item both $\co{B_1}$ and $\co{B_2}$ are linked into $A$.
\end{itemize}
\end{itemize}
\end{lemma}
\begin{proof}
Let $(B_1,B_2)$ be a bipartition of $\co{A}$ witnessing that it is not $k$-well-linked.
Moreover, choose $(B_1,B_2)$ so that it (1) minimizes $\bdc(B_1)+\bdc(B_2)$, (2) subject to (1), minimizes $|V(B_1)|$, and (3) subject (1) and (2) minimizes $|B_1|$.
Now, if $|V(B_i)| \ge s$ for both $i \in [2]$, then one of the separations $(B_1, \co{B_1})$ or $(B_2, \co{B_2})$, whichever has order $<k$, would witness that $G$ is not $(s,k)$-unbreakable.
Therefore, we have that $|V(B_1)| < s$.
By \Cref{lem:minnonkwlwitnesslinked}, the first minimization implies that both $\co{B_1}$ and $\co{B_2}$ are linked into $A$.

\begin{claim}
\label{lem:kwlubalg:claim1}
If $C$ is an internal component of $B_1$, then $\bd(C)$ intersects with $\bd(A) \setminus \bd(B_2)$.
\end{claim}
\begin{claimproof}
Suppose that $B_1$ contains an internal component $C$ so that $\bd(C) \cap \bd(A) \subseteq \bd(B_2)$.
We claim that then, $(B_1',B_2') = (B_1 \setminus C, B_2 \cup C)$ contradicts the choice of $(B_1,B_2)$.
Clearly, $|V(B_1')| \le |V(B_1)|$ and $|B_1'| < |B_1|$, so it remains to prove that $\bdc(B_i') \le \bdc(B_i)$ for both $i \in [2]$.

Because $C$ is an internal component of $B_1$, we have that $\bd(C) \subseteq \bd(B_1) = \bd(B_2 \cup A) \subseteq \bd(A) \cup \bd(B_2)$.
However, by $\bd(C) \cap \bd(A) \subseteq \bd(B_2)$, we get that $\bd(C) \subseteq \bd(B_2)$.
This implies that $\bd(B'_i) \subseteq \bd(B_i)$ for both $i \in [2]$.
\end{claimproof}

Let us show that we can enumerate all candidates for the set $B_1$ in time $|A| \cdot p \cdot s^{\OO(k+\rank(G))}$ (or find a multiplicity-$p$-violator).
First, the set $\bd(A)$ can be computed in $\OO(\rank(G) \cdot |A|)$ time by \Cref{lem:hypergraph_impl2}.
We try all candidates for the set $X = \bd(A) \setminus \bd(B_2)$, of which there are $2^{\bdc(A)}$.

If $v \in X$, then it must be that $(\ninc(v) \setminus A) \subseteq B_1$.
This implies that if $X$ is correct, then $|V(\ninc(v) \setminus A)| < s$.
Now, if $|\ninc(v) \setminus A| > p \cdot (s+1)^{\rank(G)}$, then either a subset of $\ninc(v) \setminus A$ is a multiplicity-$p$-violator, or $|V(\ninc(v) \setminus A)| \ge s$.
For each $v \in X$, we test if $|\ninc(v) \setminus A| > p \cdot (s+1)^{\rank(G)}$ in time $|A| \cdot p \cdot s^{\OO(\rank(G))}$, and if yes, we obtain in the same running time either a multiplicity-$p$-violator, in which case we can immediately terminate the algorithm, or the conclusion that $|V(\ninc(v) \setminus A)| \ge s$, in which case $X$ is not guessed correctly and we can discard it.

Now, assume that $|\ninc(v) \setminus A| \le p \cdot (s+1)^{\rank(G)}$ for all $v \in X$.
Therefore we have that $|\ninc(X) \setminus A| \le \bdc(A) \cdot p \cdot (s+1)^{\rank(G)}$.
It must be that $(\ninc(X) \setminus A) \subseteq B_1$, and by \Cref{lem:kwlubalg:claim1}, each internal component of $B_1$ intersects with $\ninc(X) \setminus A$.
Therefore, we can use \Cref{lem:localchipalgo} to enumerate all candidates for $B_1$ in time $\bdc(A) \cdot p \cdot s^{\OO(k+\rank(G))}$ (or obtain a multiplicity-$p$-violator).
For each we can compute $\bdc(B_1)$ in $\bdc(A) \cdot p \cdot s^{\OO(k+\rank(G))}$ time by \Cref{lem:hypergraph_impl2}, and $\bdc(B_2) = \bdc(\co{A} \setminus B_1) = \bdc(B_1 \cup A)$ in time $|A| \cdot p \cdot s^{\OO(k+\rank(G))}$ by \Cref{lem:hypergraph_impl2}.
Finally, we return a $B_1$ that satisfies the conditions and minimizes $\bdc(B_1)+\bdc(B_2)$, or if none exist we conclude that $\co{A}$ is $k$-well-linked.
\end{proof}

\section{Manipulating superbranch decompositions}
\label{sec:manipulatingsuperbranch}
In this section we prove several auxiliary algorithmic lemmas about manipulating superbranch decompositions efficiently.

\subsection{From tree decompositions to superbranch decompositions}
We start by proving \Cref{the:highlevel:fromgraphstohypergraphs}.

\thehighlevelfromgraphstohypergraphs*
\begin{proof}
Denote $\Tc = (T,\bag)$.
We first construct a tree decomposition $\Tc^* = (T^*,\bag^*)$ as follows.
For each $uv \in E(G)$, we add a leaf node $\ell_{uv}$ with $\bag^*(\ell_{uv}) = \{u,v\}$ adjacent to a node $t \in V(T)$ so that $\{u,v\} \subseteq \bag(t)$.
Also, for each $v \in V(G)$, we add a leaf node $\ell_{v}$ with $\bag^*(\ell_{v}) = \{v\}$ adjacent to a node $t \in V(T)$ so that $v \in \bag(t)$.
With the help of the data structure of \Cref{lem:cliqtdds}, $\Tc^*$ can be constructed in time $\OO(\|G\| + \|\Tc\|)$.

As $k \ge 1$ and the newly added bags have size at most $2$, it follows that $\Tc^*$ is $(2k,k)$-unbreakable and has $\adhsize(\Tc^*) \le 2k$.

Now, let $T'$ be the tree obtained from $T^*$ by successively removing leaves that are not in $\{\ell_{uv} : uv \in E(G)\} \cup \{\ell_{v} : v \in V(G)\}$ and contracting edges incident to degree-$2$ nodes.
Note that $T'$ can be computed in $\OO(|V(T^*|) = \OO(\|G\| + \|\Tc\|)$ time.
Let also $\lmap' \colon \leafs(T') \to E(\hyperg(G))$ be the function that maps each $\ell_{v}$ to the hyperedge $e$ with $V(e) = \{v\}$, and each $\ell_{uv}$ to the hyperedge $e$ with $V(e) = \{u,v\}$.
We claim that $\Tc' = (T',\lmap')$ is a superbranch decomposition of $\hyperg(G)$ that satisfies the required properties.

Let us first prove that $\adhsize(\Tc') \le 2k$.
Consider $xy \in E(T')$.
By the construction of $T'$, there exists $x^*y^* \in E(T^*)$, so that for each $e_{v} \in \lmap'(\vec{xy})$, the leaf $\ell_{v}$ is closer to $x^*$ than $y^*$ in $T^*$, and for each $e_{uv} \in \lmap'(\vec{xy})$, the leaf $\ell_{uv}$ is closer to $x^*$ than $y^*$ in $T^*$, and similarly while swapping the roles of $x$ and $y$.
By the vertex condition of $T^*$, this implies that if $v \in \adh_{\Tc'}(xy)$, then also $v \in \adh_{\Tc^*}(x^*y^*)$, implying that $\adhsize(\Tc') \le \adhsize(\Tc^*) \le 2k$.

Let us then prove that for every $t \in \vint(T')$, $V(t)$ is $(2k,k)$-unbreakable in $\hyperg(G)$.
By the way $T'$ was constructed, we have that $t$ corresponds to a connected subtree $T_t^*$ of internal nodes of $T^*$, so that $T_t^*$ all nodes of $T_t^*$ have degree $2$ in $T^*$, except exactly one which has degree $\ge 3$.
Suppose that $v \in V(t)$.
This means that there exists two neighbors $t_1,t_2 \in N_{T'}(t)$ so that $v \in V(\lmap(\vec{t_1 t}))$ and $v \in V(\lmap(\vec{t_2 t}))$.
In particular, $v \in \adh_{\Tc'}(t_1 t) \cap \adh_{\Tc'}(t_2 t)$.
Now, by our previous argument about adhesion size, the adhesions of $\Tc'$ at $t_1 t$ and $t_2 t$ are subsets of adhesions of $\Tc^*$ that are between the subtree $T_t^*$ and its neighbors.
Because all nodes in $T_t^*$ have degree $2$ in $T^*$ except one, it follows that $v$ must be in the bag of the node of degree $\ge 3$ in $T_t^*$, and therefore $V(t)$ is a subset of that bag.

We should also argue that the representation of $\Tc' = (T',\lmap')$ can be computed in $\OO(k \cdot (\|G\|+\|\Tc\|))$ time.
Let us first compute the adhesions.
We first compute in total time $\OO(\|G\|)$ for each $v \in V(G)$ the number of hyperedges $e \in E(\hyperg(G))$ so that $v \in V(e)$.
We then root $T'$ at an arbitrary node, and compute by bottom-up dynamic programming for each edge of form $tp \in E(T')$, where $p = \parent(t)$, the adhesion $\adh(tp)$, and for each $v \in \adh(tp)$ the number of hyperedges $e \in \lmap(\vec{tp})$ with $v \in V(e)$.
In particular, if $t \in V(T)$ and we have computed this information for all edges of form $ct$, where $c \in \children(t)$, then this information can be computed for the edge $tp$, where $p = \parent(t)$, by simply summing up the numbers of hyperedges for all vertices occuring in the adhesions $\adh(ct)$, and observing that if this number is less than the total number of hyperedges in which the vertex occurs, then the vertex is in $\adh(tp)$.
This takes in total $\OO(k \cdot |V(T')|) = \OO(k \cdot \|G\|)$ time.

After computing the adhesions, the other information stored in the representation can be easily computed in $\OO(k \cdot |V(T')|) = \OO(k \cdot \|G\|)$ time.
\end{proof}

\subsection{Transformations of superbranch decompositions}
\label{subsec:transsuperbds}
In our algorithm we will manipulate superbranch decompositions in various ways.
We now define notation for these manipulations and show that they can be implemented efficiently.

\subsubsection{Contraction}
The first transformation is the contraction of an edge of the tree.

Let $\Tc = (T,\lmap)$ be a superbranch decomposition of a hypergraph $G$.
Let also $uv \in \eint(T)$ be an internal edge of $T$.
We define that the \emph{contraction} of $\Tc$ along $uv$ is the superbranch decomposition $\Tc \contr uv = \Tc' = (T',\lmap')$ with
\begin{itemize}
\item $V(T') = (V(T) \setminus \{u, v\}) \cup \{w\}$,
\item $E(T') = (E(T) \setminus \ninc(\{u,v\})) \cup \{wx : x \in N_T(\{u,v\})\}$, and
\item $\lmap'(\ell) = \lmap(\ell)$ for all $\ell \in \leafs(T) = \leafs(T')$.
\end{itemize}

In other words, $T' = T \contr uv$, while the leaf mapping is preserved.

When $\Tc$ is rooted, we define that the contracton preserves the root in the natural way, i.e., if neither $u$ or $v$ is the root, then the root does not change, and if one of $u$ or $v$ is the root, then the new node $w$ becomes the root.

We then give an efficient algorithm for contracting an edge of a superbranch decomposition.

\begin{lemma}
\label{lem:superbdcontractalgo}
Let $\Tc = (T,\lmap)$ be a (rooted) superbranch decomposition whose representation is stored.
There is an algorithm that, given $uv \in \eint(T)$, in time $\OO(\min(\|\torso(u)\|, \|\torso(v)\|))$ transforms $\Tc$ into $\Tc \contr uv$.
The pointers to all nodes other than $u$ and $v$ are preserved, and a pointer to the new node $w$ is returned.
\end{lemma}
\begin{proof}
We denote $\Tc \contr uv = (T', \lmap')$.
First, we can in time $\OO(\min(\|\torso(u)\|, \|\torso(v)\|))$ determine which one of $\|\torso(u)\|$ and $\|\torso(v)\|$ is smaller.
Assume then \wilog that $\|\torso(u)\| = \min(\|\torso(u)\|, \|\torso(v)\|)$.

We first transform $T$ into $T'$, so that $v$ will correspond to the new node $w$.
We can in $\OO(\|\torso(u)\|)$ time compute $N_T(u)$, and then in time $\OO(\|\torso(u)\|)$ delete $u$ and all of its incident edges from $T$.
Then, we insert edge $vx$ for all $x \in N_T(u) \setminus \{v\}$ to $T$, in time $\OO(\|\torso(u)\|)$.
This transforms $T$ into $T'$.

Then, to transform $\torso_{\Tc}(v)$ into $\torso_{\Tc'}(v)$, we first remove the hyperedge $e_u$ corresponding to $\adh_{\Tc}(uv)$.
A pointer to this hyperedge was stored in $uv$, so we can delete it in $\OO(|V(e_u)|) = \OO(\|\torso(u)\|)$ time.
Then, for each $e \in \torso_{\Tc}(u)$, except for the hyperedge $e_v$ corresponding to $\adh(uv)$, we insert $e$ to $\torso(v)$, inserting also the corresponding vertices if needed.
Note that whether a vertex $x \in V(\torso_{\Tc}(u))$ is in $V(\torso_{\Tc}(v))$ is determined by whether $x \in \adh_{\Tc}(uv)$, and therefore this can be determined for all vertices in $V(\torso(u))$ in time $\OO(\|\torso(u)\|)$.
At the same time when we insert the hyperedges, we can also add the pointers from them to the corresponding edges of $T'$ and from the edges of $T'$ to the hyperedges.

Lastly, if $\Tc$ is rooted, then for every $x \in N_T(u)$, except when $x = v$ or $x$ is the parent of $u$, we need to change the parent-edge pointer of $x$ to point at the edge $xv$.
Because such edge was already processed at this point and there are at most $\|\torso(u)\|$ of them, this can be done in $\OO(\|\torso(u)\|)$ time.
Also, if $u$ is the parent of $v$, we need to change the parent-edge pointer of $v$ to point to either $vp$, where $p$ is the parent of $u$, or to indicate that $v$ is the root in the case when $u$ was the root.
\end{proof}

\subsubsection{Splitting}
\label{subsec:plainsplitting}
The second transformation is the opposite operation of contraction, in particular, it is the splitting of a node.

Let $\Tc = (T,\lmap)$ be a superbranch decomposition, $t \in \vint(T)$ an internal node of $\Tc$, and $C \subseteq E(\torso(t))$ so that $|C| \ge 2$ and $|\co{C}| \ge 2$.
We define $\Tc \rescliqs (t,C) = (T',\lmap')$ to be the superbranch decomposition with
\begin{itemize}
\item $V(T') = (V(T) \setminus \{t\}) \cup \{t', t_C\}$,
\item $E(T') = (E(T) \setminus \ninc_{T}(t)) \cup \{t's : s \in N_T(t) \text{ and } e_s \in \co{C}\} \cup \{t_C s : s \in N_T(t) \text{ and } e_s \in C\} \cup \{t't_C\}$,
\item $\lmap'(\ell) = \lmap(\ell)$ for all $\ell \in \leafs(T) = \leafs(T')$.
\end{itemize}

In other words, $\Tc \rescliqs (t,C)$ is obtained from $\Tc$ by ``separating'' the node $t$ along the separation $(C,\co{C})$ of $\torso(t)$ into two new nodes $t'$ and $t_C$, with the node $t'$ of $\Tc \rescliqs (t,C)$ corresponding to the side of $\co{C}$ of the separation, i.e., having $\torso(t') = \torso(t) \rescliqs C$, and the node $t_C$ of $\Tc \rescliqs (t,C)$ corresponding to the side of $C$ of the separation, i.e., having $\torso(t_C) = \torso(t) \rescliqs \co{C}$.

We observe that this operation leaves the decomposition completely intact everywhere except for the splitted node $t$.
In particular, every internal separation of $\Tc \rescliqs (t,C)$ is either an internal separation of $\Tc$, or an orientation of the separation $(C \orescliqs \Tc, \co{C} \orescliqs \Tc)$, where the latter corresponds to the edge $t't_C$ of $\Tc \rescliqs (t,C)$.
Also, all torsos of $\Tc \rescliqs (t,C)$ are either torsos of $\Tc$, or are $\torso(t') = \torso(t) \rescliqs C$ or $\torso(t_C) = \torso(t) \rescliqs \co{C}$.
We will frequently use these simple observations.

If $\Tc$ is rooted, this operation is defined similarly, but if the root is $t$, then $t'$ becomes the new root.

\begin{lemma}
\label{lem:sbdchippingalgo}
Let $\Tc = (T,\lmap)$ be a (rooted) superbranch decomposition whose representation is already stored.
There is an algorithm that, given an internal node $t \in \vint(T)$ and a set $C \subseteq E(\torso(t))$ with $|C| \ge 2$ and $|\co{C}| \ge 2$, in time $\OO(|C| \cdot \adhsize(\Tc))$ transforms the representation of $\Tc$ into a representation of $(T',\lmap') = \Tc' = \Tc \rescliqs (t,C)$.
All pointers to the other nodes and torsos of other nodes are preserved, and furthermore pointers to the hyperedges of $\torso(t)$ in $\co{C}$ now point to the corresponding hyperedges of $\torso(t')$ in $\co{C}$.
Also, pointers to the nodes $t'$ and $t_C$, and to the edge $t't_C \in E(T')$ are returned.
\end{lemma}
\begin{proof}
First, by iterating over $C$ in time $\OO(|C|)$, we can enumerate the set $X = \{s \in N_{T}(t) \mid e_s \in C\}$, and in particular obtain pointers to the edges of form $ts \in E(T)$ with $s \in X$.
We edit $T$ by (1) deleting each such edge $ts$ with $s \in X$, (2) inserting the node $t_C$, (3) renaming $t$ into $t'$, (4) inserting the edge $t' t_C$, and (4) inserting edges of form $t_C s$ for all $s \in X$.
This turns $T$ into $T'$, and can be done in total $\OO(|X|) = \OO(|C|)$ time.

For all $s \in X$, we set $\adh_{\Tc'}(st_C) = \adh_{\Tc}(st)$.
We compute $\bd(C)$ in time $\OO(|C| \cdot \rank(\torso(t))) = \OO(|C| \cdot \adhsize(\Tc))$ by \Cref{lem:hypergraph_impl2} and set $\adh_{\Tc'}(t't_C) = \bd(C)$.
This correctly sets the adhesions of $\Tc'$.

We compute $\torso_{\Tc'}(t_C) = \torso_{\Tc}(t) \rescliqs \co{C}$ in time $\OO(|C| \cdot \adhsize(\Tc))$ by using \Cref{lem:hypergraph_impl3}.
We store a pointer to $t_C s$ to the hyperedge $e_s \in E(\torso_{\Tc'}(t_C))$ for all $s \in X$, and likewise we store a pointer to $e_s$ in each edge $t_C s$.
We also need to store a pointer to $e_{t_C} \in E(\torso_{\Tc'})(s)$ to each $t_C s$, which we can obtain from the previous pointers stored in corresponding edges.

We turn $\torso_{\Tc}(t)$ into $\torso_{\Tc'}(t') = \torso_{\Tc}(t) \rescliqs C$ in time $\OO(|C| \cdot \adhsize(\Tc))$ by using \Cref{lem:hypergraph_impl4}.
This preserves the pointers stored in the hyperedges not in $C$, and returns a pointer to the hyperedge corresponding to $C$, to which we store a pointer to $t' t_C$.
Finally, we need to also store a pointer to $t' t_C$ in the hyperedge corresponding to $\co{C}$ in $\torso_{\Tc'}(t_C)$, and correspondingly store a pointer to the hyperedge corresponding to $\co{C}$ in $\torso_{\Tc'}(t_C)$ to the edge $t' t_C$.

This turns the representation of $\Tc$ into a representation of $\Tc'$.

Finally, if $\Tc$ is rooted, we have to update the pointers to the parent edges.
First, if $t$ was the root, it suffices to set the root pointer to $t'$, the parent-edge pointer of $t_C$ to point to $t'$, and for each $s \in X$, set the parent-edge pointer to point to $t_C$.
Then, if the parent of $t$ was $p$ and $p \in X$, we set the parent-edge pointer of $t'$ to point to $t' t_C$, for each $s \in X \setminus \{p\}$ the parent-edge pointer to point to $s t_C$, and for $t_C$ the parent-edge pointer to point to $t_C p$.
If $p \notin X$, the parent-edge pointer of $t'$ already points to the correct edge, i.e., $t'p$, and for $t_C$ we set the parent-edge pointer to point to $t'$, and for each $s \in X$ we set the parent-edge pointer to point to $t_C$.
All of this can be implemented in time $\OO(|X|) = \OO(|C|)$.
\end{proof}

We observe that $\Tc \rescliqs (t,C)$ and $\Tc \rescliqs (t,\co{C})$ are identical up to renaming nodes, with the node $t'$ of the first one corresponding to the node $t_{\co{C}}$ of the latter one and the node $t_C$ of the first one corresponding to the node $t'$ of the latter one.
Due to this, \Cref{lem:sbdchippingalgo} can easily be modified into one where the set $\co{C}$ is given as input instead of the set $C$, which can be useful when $|\co{C}|$ is smaller than $|C|$.

\subsubsection{Splitting with a partition}
We then define an operation that generalizes the previous splitting operation, splitting with many sets simultaneosly instead of just one.

Let $\Tc = (T,\lmap)$ be a superbranch decomposition and $t \in \vint(T)$ an internal node of $\Tc$.
We say that $\compset$ is a \emph{splitting family} for $t$ if $\compset$ is a collection of disjoint non-empty subsets of $E(\torso(t))$, so that for each $C \in \compset$, $|C| \ge 2$, $|\co{C}| \ge 2$, and $\co{C} \notin \compset$.
When $\compset$ is a splitting family for $t$, we define that $\Tc \rescliqs (t,\compset) = \Tc' = (T',\lmap')$ is the superbranch decomposition with
\begin{itemize}
\item $V(T') = (V(T) \setminus \{t\}) \cup \{t'\} \cup \{t_C : C \in \compset\}$,
\item $E(T') = (E(T) \setminus \ninc_T(t)) \cup \{t' t_C : C \in \compset\} \cup \{t' x : e_x \in E(\torso(t)) \setminus \bigcup_{C \in \compset} C\} \cup \{t_C x : e_x \in C\}$, and
\item $\lmap'(\ell) = \lmap(\ell)$ for all $\ell \in \leafs(T) = \leafs(T')$.
\end{itemize}

In other words, if we enumerate $\compset = \{C_1, \ldots, C_h\}$, then $\Tc \rescliqs (t, \compset)$ can be defined in terms of the plain splitting operation of \Cref{subsec:plainsplitting} as $\Tc \rescliqs (t,\compset) = \Tc \rescliqs (t,C_1) \rescliqs (t',C_2) \rescliqs \ldots \rescliqs (t',C_h)$, where $\rescliqs$ is applied from left to right.

Again, the operation is defined for rooted superbranch decompositions so that the root is preserved if it is not $t$, and if $t$ is the root then $t'$ becomes the new root.

\begin{lemma}
\label{lem:superbdsplittingalgo}
Let $\Tc = (T,\lmap)$ be a (rooted) superbranch decomposition of a hypergraph $G$, whose representation is already stored.
There is an algorithm that, given an internal node $t \in \vint(T)$ and a splitting family $\compset$ for $t$, in time $\OO(|E(\torso(t))| \cdot (\adhsize(\Tc)+\adhsize(\Tc \rescliqs (t,\compset))))$ transforms the representation of $\Tc$ into a representation of $\Tc \rescliqs (t,\compset)$.
\end{lemma}
\begin{proof}
This can be done by using \Cref{lem:sbdchippingalgo} repeatedly, in particular, by the observation that $\Tc \rescliqs (t,\compset) = \Tc \rescliqs (t,C_1) \rescliqs (t',C_2) \rescliqs \ldots \rescliqs (t',C_h)$, where $\compset = \{C_1,\ldots, C_h\}$.
\end{proof}

We say that a \emph{splitting partition} for $t \in \vint(T)$ is a partition $\compset$ of $E(\torso(t))$ so that $|\compset| \ge 3$.
We define that the splitting family $\compset'$ corresponding to a splitting partition $\compset$ is $\compset' = \{C \in \compset : |C| \ge 2\}$, and denote $\Tc \rescliqs (t,\compset) = \Tc \rescliqs (t,\compset')$.

Note that sometimes $\compset$ can be both a splitting partition and a splitting family for $t$, in which case this definition is superficially ambiguous, but it turns out that in this case $\compset' = \compset$.

Let us say that an \emph{implicit representation} of a splitting partition $\compset = \{C_1, C_2, \ldots, C_h\}$ is the collection $\compset' = \{C_1, C_2, \ldots, C_{h-1}\}$.
These definitions are geared towards the following alternative version of \Cref{lem:superbdsplittingalgo}.

\begin{lemma}
\label{lem:superbdsplittingalgoimplicit}
Let $\Tc = (T,\lmap)$ be a (rooted) superbranch decomposition of a hypergraph $G$, whose representation is already stored.
There is an algorithm that, given an internal node $t \in \vint(T)$ and an implicit representation $\compset'$ of a splitting partition $\compset$ for $t$, in time $\OO((\adhsize(\Tc)+\adhsize(\Tc \rescliqs (t,\compset))) \cdot \sum_{C \in \compset'} |C|)$ transforms the representation of $\Tc$ into a representation of $\Tc \rescliqs (t,\compset)$.
\end{lemma}
\begin{proof}
Let $\compset' = \{C_1,C_2,\ldots,C_{h-1}\}$ and $\compset = \{C_1,C_2,\ldots,C_h\}$.
We first use \Cref{lem:sbdchippingalgo} to transform $\Tc$ into $\Tc' = \Tc \rescliqs (t,\co{C_h}) = \Tc \rescliqs (t,\bigcup_{C \in \compset'} C)$ in time $\OO(|\co{C_h}| \cdot \adhsize(\Tc)) = \OO(\adhsize(\Tc) \cdot \sum_{C \in \compset'} |C|)$.
Note that the adhesion size of $\Tc'$ is at most the maximum of the adhesion sizes of $\Tc$ and $\bdc(C_h) \le \adhsize(\Tc \rescliqs (t,\compset))$.

Then, let $\compset'' = \{C \in \compset' : |C| \ge 2\}$, and note that $\compset''$ is a splitting family for $t_{\co{C_h}}$.
Then, we use \Cref{lem:superbdsplittingalgo} to turn $\Tc'$ into $\Tc' \rescliqs (t_{\co{C_h}}, \compset'') = \Tc \rescliqs (t,\compset)$.
\end{proof}

\subsubsection{Splitting with a superbranch decomposition}
Let $\Tc = (T,\lmap)$ be a superbranch decomposition, $t \in \vint(T)$ an internal node of $T$, and $\Tc_t = (T_t,\lmap_t)$ a superbranch decomposition of $\torso(t)$.
We define $\Tc \rescliqs (t, \Tc_t) = (T',\lmap')$ to be the superbranch decomposition with

\begin{itemize}
\item $V(T') = (V(T) \setminus \{t\}) \cup \vint(T_t)$,
\item $E(T') = (E(T) \setminus \ninc(t)) \cup \eint(T_t) \cup \{uv : u \in N_T(t) \text{ and } \lmap_t^{-1}(e_u)v \in E(T_t)\}$, and
\item $\lmap'(\ell) = \lmap(\ell)$ for all $\ell \in \leafs(T') = \leafs(T)$.
\end{itemize}

In other words, we place the superbranch decomposition $\Tc_t$ of $\torso(t)$ in place of the node $t$.
Note that this can be interpreted as successive applications of the splitting operation $\Tc \rescliqs (t,C)$.

We observe that each torso of $\Tc'$ is either a torso of a node $x \in \vint(T) \setminus \{t\}$ in $\Tc$, or a torso of an internal node of $\Tc_t$.
Furthermore, each internal separation of $\Tc'$ is either an internal separation of $\Tc$, or is of form $(A \orescliqs \Tc, \co{A} \orescliqs \Tc)$, where $(A,\co{A})$ is an internal separation of $\Tc_t$ (and therefore a separation of $\torso(t)$).

We then give an efficient algorithm for this operation.

\begin{lemma}
\label{lem:superb_impl_breakbydecomp}
Let $G$ be a hypergraph and $\Tc = (T,\lmap)$ a superbranch decomposition of $G$ whose representation is already stored.
There is an algorithm, that given an internal node $t \in \vint(T)$ and a superbranch decomposition $\Tc_t$ of $\torso(t)$, in time $\OO(\adhsize(\Tc_t) \cdot \|\torso(t)\|)$ turns the representation of $\Tc$ into a representation of $\Tc \rescliqs (t, \Tc_t)$.
Moreover, this preserves all pointers to the torsos of $\Tc$ and $\Tc_t$ (except for $\torso(t)$, which no longer exists in $\Tc \rescliqs (t,\Tc_t)$).
\end{lemma}
\begin{proof}
Denote $\Tc_t = (T_t, \lmap_t)$.
First, we add $T_t$ to the representation of $T$ so that it turns into the disjoint union of $T$ and $T_t$.
We also directly copy the representations of the torsos and adhesions of $\Tc_t$, along with the various pointers stored therein.

Then consider $e_u \in \torso(t)$ that corresponds to a hyperedge of $T$ between $t$ and $u$.
Note that we can find this $u$ in $\OO(1)$ time by using the pointer stored in $e_u$, finding at the same time a pointer to the edge $ut \in E(T)$.
We compute for each $e_u \in \torso(t)$ the leaf $\ell_u = \lmap_t^{-1}(e_u) \in \leafs(T_t)$ in total $\OO(\|\torso(t)\|)$ time by iterating over $\leafs(T_t)$.
At the same time, we compute the unique neighbor $v_{e_u} \in \vint(T_t)$ of $\ell_u$.
By using the pointers stored in $\Tc_t$, we also find the hyperedge $e_{\ell_u} \in \torso_{\Tc_t}(v_{e_u})$ corresponding to $\ell_u$.

After computing all this information for all $e_u \in \torso(t)$ in total time $\OO(\|\torso(t)\|)$, we do the following for each $e_u \in \torso(t)$.
We delete the node $\ell_u$ and the edge $\ell_u v_{e_u}$ from $T_t$, and insert the edge $u v_{e_u}$ between $T$ and $T_t$.
We set $\adh(u v_{e_u}) = \adh(ut)$, and redirect the pointer stored in a hyperedge of $\torso(u)$ pointing at $ut$ to point at $u v_{e_u}$.
We also store at $u v_{e_u}$ a pointer to that hyperedge in $\torso(u)$ and a pointer to the hyperedge $e_{\ell_u} \in \torso_{\Tc_t}(v_{e_u})$.
We set the pointer stored at $e_{\ell_u}$ to also point at $u v_{e_u}$.
Then we delete the edge $ut$ from $T$.
This can be done in $\OO(|\adh(u v_{e_u})|) = \OO(\adhsize(\Tc_t))$ time for each $e_u \in \torso(t)$, so in total this runs in time $\OO(\adhsize(\Tc_t) \cdot \|\torso(t)\|)$.

Finally, we delete the node $t$ from $T$, and the construction is complete.
\end{proof}

We will often do the operation of \Cref{lem:superb_impl_breakbydecomp} for a large number of nodes of $\Tc$ simultaneosly, so let us introduce notation for that.
Let $Y \subseteq \vint(T)$ be a set of internal nodes of $\Tc$, and $\refiset = \{\Tc_t\}_{t \in Y}$ be a set indexed by $t \in Y$, which contains for each $t \in Y$ a superbranch decomposition $\Tc_t$ of $\torso(t)$.
We call such a set a \emph{refinement set} for $\Tc$.
Let us index $Y = \{t_1, \ldots, t_h\}$.
We define $\Tc \rescliqs \refiset = \Tc \rescliqs (t_1, \Tc_{t_1}) \rescliqs \ldots \rescliqs (t_h, \Tc_{t_h})$, where $\rescliqs$ is applied from left to right.
It can be seen that the order of $t_1,\ldots,t_h$ does not matter, and all torsos of $\Tc \rescliqs \refiset$ are either torsos of nodes of $\Tc$ not in $Y$, or torsos of $\Tc_{t}$ for $t \in Y$.
Also, each internal separation of $\Tc \rescliqs \refiset$ is either an internal separation of $\Tc$ or corresponds to an internal separation of some $\Tc_{t}$.
Let us observe that $\Tc \rescliqs \refiset$ can be implemented via \Cref{lem:superb_impl_breakbydecomp}. 

\begin{lemma}
\label{lem:superbdalgrefinebyrefiset}
There is an algorithm that, given a superbranch decomposition $\Tc = (T,\lmap)$, and a refinement set $\refiset$ for $\Tc$, so that $\adhsize(\Tc_i) < k$ for all $\Tc_i \in \refiset$, in time $\OO((k+\adhsize(\Tc)) \cdot \|G\|)$ returns the superbranch decomposition $\Tc \rescliqs \refiset$.
\end{lemma}
\begin{proof}
We call the algorithm of \Cref{lem:superb_impl_breakbydecomp} successively for all $t \in Y$, and then return the resulting superbranch decomposition.
\end{proof}

\section{Downwards well-linked decomposition}
\label{sec:downwl}
Recall that a rooted superbranch decomposition $\Tc = (T,\lmap)$ is downwards well-linked if for each oriented edge $\vec{tp} \in \vec{E}(T)$, where $p$ is the parent of $t$, it holds that $\lmap(\vec{tp})$ is well-linked.

The goal of this section is to prove the following main lemma.

\thehighleveldownwlalgo*

The first step to prove \Cref{the:highlevel:downwlalgo} is to turn the given superbranch decomposition into \emph{downwards semi-internally connected}, whose definition we give now.
Let $\Tc = (T,\lmap)$ be a rooted superbranch decomposition.
We say that a node $t \in V(T)$ is \emph{downwards semi-internally connected} if either
\begin{itemize}
\item $t$ has parent $p \in V(T)$, and $\lmap(\vec{tp})$ is semi-internally connected, or
\item $t = \troot(T)$.
\end{itemize}

We say that $\Tc$ is downwards semi-internally connected if all nodes of $\Tc$ are downwards semi-internally connected.
In other words, $\Tc$ is downwards semi-internally connected if for all oriented edges $\vec{tp} \in \vec{E}(T)$ where $p = \parent(t)$ it holds that $\lmap(\vec{tp})$ is semi-internally connected.

For making $\Tc$ downwards semi-internally connected it is useful that being semi-internally connected is a transitive property of sets of hyperedges, which we prove now.

\begin{lemma}
\label{lem:conddwsic}
Let $A \subseteq E(G)$ be a semi-internally connected set.
Let also $B \subseteq E(G \rescliqs A)$ be a semi-internally connected set in $G \rescliqs A$.
Then, $B \orescliqs A$ is semi-internally connected in $G$.
\end{lemma}
\begin{proof}
Let $e_A$ be the hyperedge of $G \rescliqs A$ corresponding to $A$.
The lemma is clear if $e_A \notin B$, so assume that $e_A \in B$.

Note that $\bd_G(B \orescliqs A) = \bd_{G \rescliqs A}(B)$.
Note also that if $C \subseteq B$ is an internal component of $B$ in $G \rescliqs A$ and $e_A \notin C$, then $C$ is also an internal component of $B \orescliqs A$ in $G$ and $\bd_G(C) = \bd_{G \rescliqs A}(C) = \bd_{G \rescliqs A}(B) = \bd_G(B \orescliqs A)$.

First, if the internal component of $B$ in $G \rescliqs A$ containing $e_A$ contains only $e_A$, then $\bd_G(B \orescliqs A) = \bd(A)$, and the internal components of $B \orescliqs A$ are simply the union of the internal components of $A$ in $G$, and the internal components of $B$ in $G \rescliqs A$, except for $\{e_A\}$.

Otherwise, let $C$ be an internal component of $B$ in $G \rescliqs A$ with $e_A \in C$ and $|C| \ge 2$.
Now, there exists $v \in V(e_A) \setminus \bd(C)$, in particular, $v \in \bd_G(A) \setminus \bd_G(B \orescliqs A)$.
As $A$ is semi-internally connected, this implies that $A \cup C \setminus \{e_A\}$ is an internal component of $B \orescliqs A$ in $G$, and has $\bd_G(A \cup C \setminus \{e_A\}) = \bd_G(B \orescliqs A)$.
Therefore, the internal components of $B \orescliqs A$ in $G$ are $A \cup C \setminus \{e_A\}$, and the internal components of $B$ in $G \rescliqs A$ that do not contain $e_A$.
\end{proof}

We then give an algorithm for transforming the initial input superbranch decomposition for \Cref{the:highlevel:downwlalgo} into a downwards semi-internally connected superbranch decomposition.

\begin{lemma}
\label{lem:algdownwlsemicon}
There is an algorithm that, given a superbranch decomposition $\Tc$ of a hypergraph $G$ and an integer $k \ge 1$ so that
\begin{itemize}
\item $\Tc$ is $(2k,k)$-unbreakable, 
\item $\adhsize(\Tc) \le 2k$, and
\end{itemize}
in time $2^{\OO(k)} \cdot \|\Tc\|$ outputs a rooted superbranch decomposition $\Tc'$ of $G$ so that
\begin{itemize}
\item $\Tc'$ is downwards semi-internally connected,
\item $\Tc'$ is $(2k,k)$-unbreakable, and
\item $\adhsize(\Tc') \le 2k$.
\end{itemize}
\end{lemma}
\begin{proof}
We first turn $\Tc = (T,\lmap)$ into a rooted superbranch decomposition by picking an arbitrary internal node as the root.
We then edit $\Tc$, while maintaining a prefix $P \subseteq \vint(T)$ of $T$ and the following invariants.

\begin{itemize}
\item If $t \in V(T) \setminus P$ and $p = \parent(t)$, then $\lmap(\vec{tp})$ is semi-internally connected,
\item $\Tc$ is $(2k,k)$-unbreakable, and
\item $\adhsize(\Tc) \le 2k$.
\end{itemize}

Note that if $P = \{\troot(\Tc)\}$ and $\Tc$ satisfies the invariants, then we can output $\Tc$.
We initially set $P = \vint(T)$, which satisfies the invariants.
We repeat the following process as long as $|P| \ge 2$.

Let $t \in P$ be a node in $P$ so that no child of $t$ is in $P$.
Let $p$ be the parent of $t$, and $e_p$ the hyperedge of $\torso(t)$ corresponding to the adhesion at $tp$.
We use the algorithm of \Cref{lem:alg_internal_comps} to compute the partition of $E(\torso(t)) \setminus \{e_p\}$ into its internal components in $\OO(\|\torso(t)\| \cdot k)$ time.
We then use radix sort to group the internal components $C$ by $\bd(C)$, yielding a partition $\compset$ of $E(\torso(t)) \setminus \{e_p\}$ so that (1) each $C \in \compset$ is semi-internally connected (2) $\bd(C) \subseteq \adh(tp)$ for all $C \in \compset$, and (3) $|\compset| \le 2^{|\adh(tp)|} \le 2^{2k}$.
This can be done in $2^{\OO(k)} \cdot \|\torso(t)\|$ time.

If $|\compset| = 1$, then $E(\torso(t)) \setminus \{e_p\}$ is semi-internally connected in $\torso(t)$, and by the transitivity of semi-internally connectedness (\Cref{lem:conddwsic}), $\lmap(\vec{tp})$ is semi-internally connected, so we can remove $t$ from $P$ without editing $\Tc$.

Otherwise, when $|\compset| \ge 2$, let $\compset' = \{C \in \compset : |C| \ge 2\}$ and observe that $\compset'$ is a splitting family for $t$.
We transform $\Tc$ into $\Tc^* = \Tc \rescliqs (t, \compset')$ in time $\OO(\|\torso(t)\| \cdot k)$ by \Cref{lem:superbdsplittingalgo}.
Let us denote by $t_C$ for $C \in \compset'$ the resulting nodes corresponding to the sets $C$ in $\compset'$, and by $t'$ their parent.

\begin{claim}
The rooted superbranch decomposition $\Tc \rescliqs (t,\compset') = \Tc^* = (T^*,\lmap^*)$ satisfies the properties that
\begin{itemize}
\item if $x \in V(T^*) \setminus (P \cup \{t'\})$, and $y$ is the parent of $x$, then $\lmap(\vec{xy})$ is semi-internally connected,
\item $V_{\Tc^*}(t') = \adh_{\Tc}(tp)$,
\item $\Tc^*$ is $(2k,k)$-unbreakable, and
\item $\adhsize(\Tc^*) \le 2k$.
\end{itemize}
\end{claim}
\begin{claimproof}
The first property follows for all nodes other than $t_C$ for $C \in \compset'$ by the corresponding property of $\Tc$.
Thus, we only need to prove that for each $C \in \compset'$, it holds that $C \orescliqs \Tc$ is semi-internally connected.
This follows from the transitivity of semi-internally connectedness (\Cref{lem:conddwsic}).
The fact that $V_{\Tc^*}(t') = \adh_{\Tc}(tp)$ follows because (1) $\bd(C) \subseteq \adh_{\Tc}(tp)$ for all $C \in \compset$, and (2) if $C \in \compset$ and $|C| = 1$, then $V(C) = \bd(C)$.

Then, for all $x \in \vint(T^*)$, there exists $y \in \vint(T)$ so that $V_{\Tc^*}(x) \subseteq V_{\Tc}(y)$, implying that $\Tc^*$ is $(2k,k)$-unbreakable.
Finally, the fact that $\bd(C) \subseteq \adh(tp)$ for all $C \in \compset$ implies that $\adhsize(\Tc^*) \le \adhsize(\Tc) \le 2k$.
\end{claimproof}

For the node $t'$ we have that $|E(\torso(t'))| \le 2^{2k}$ and $V(t') = \adh(tp)$.
We then transform $\Tc^*$ into the contraction $\Tc^* \contr t'p$ with the algorithm of \Cref{lem:superbdcontractalgo}, which takes $2^{\OO(k)}$ time.
This maintains the adhesion sizes and the unbreakability properties because $V(t') \subseteq V(p)$, and now we can set $P$ to $P \setminus \{t\}$.
This concludes the description of the process.

Now, each iteration of the process runs in $2^{\OO(k)} \cdot \|\torso(t)\|$ time.
However, it is not immediately clear that this sums up to $2^{\OO(k)} \cdot \|\Tc\|$, because the torsos of nodes in $P$ can grow.
However, in each iteration the torso of the parent of $p$ grows only by $2^{\OO(k)}$ hyperedges, so by considering a potential function of form $\Phi(P) = 2^{c k} \cdot |P| + \sum_{t \in P} \|\torso(t)\|$, for a large enough constant $c > 1$, we can show that the overall running time is indeed $2^{\OO(k)} \cdot \|\Tc\|$.
\end{proof}

Then, we further transform the rooted branch decomposition outputted by \Cref{lem:algdownwlsemicon} so that its torsos have multiplicity at most $3$.

\begin{lemma}
\label{lem:algdownwlsemiconmult}
There is an algorithm that, given a rooted superbranch decomposition $\Tc$ of a hypergraph $G$ and an integer $k \ge 1$ so that
\begin{itemize}
\item $\Tc$ is downwards semi-interally connected,
\item $\Tc$ is $(2k,k)$-unbreakable, and
\item $\adhsize(\Tc) \le 2k$,
\end{itemize}
in time $\OO(k \cdot \|\Tc\|)$ outputs a rooted superbranch decomposition $\Tc'$ of $G$ so that
\begin{itemize}
\item the multiplicity of each torso of $\Tc'$ is at most $3$.
\item $\Tc'$ is downwards semi-internally connected,
\item $\Tc'$ is $(2k,k)$-unbreakable, and
\item $\adhsize(\Tc') \le 2k$.
\end{itemize}
\end{lemma}
\begin{proof}
Assume that the root of $\Tc$ is an internal node, because if it is a leaf, then we can move it to the adjacent internal node while keeping downwards semi-internally connectedness.
We will edit $\Tc = (T,\lmap)$, while maintaining a prefix $P \subseteq \vint(T)$ of $T$ (here, we allow $P = \emptyset$) and the following invariants.
\begin{itemize}
\item If $t \in \vint(T) \setminus P$, then the multiplicity of $\torso(t)$ is at most $3$,
\item $\Tc$ is downwards semi-internally connected,
\item $\Tc$ is $(2k,k)$-unbreakable, and
\item $\adhsize(\Tc) < 2k$.
\end{itemize}

Note that if $P = \emptyset$ and $\Tc$ satisfies the invariants, then we can return $\Tc$.
We initially set $P = \vint(T)$, which clearly satisfies the invariants.
We repeat the following process as long as $|P| \ge 1$.

Let $t \in P$ be a node in $P$ so that no child of $t$ is in $P$.
Let $p$ be the parent of $t$, and $e_p$ the hyperedge of $\torso(t)$ corresponding to the adhesion at $tp$.
We use radix sort to partition the hyperedges $E(\torso(t)) \setminus \{e_p\}$ into equivalence classes based on $V(e)$, in time $\OO(\|\torso(t)\|)$.
Let $\compset$ be this partition of $E(\torso(t)) \setminus \{e_p\}$.
If $t$ is not the root and $|\compset| = 1$, i.e., all hyperedges in $E(\torso(t)) \setminus \{e_p\}$ contain the same set of vertices, then let $\compset$ be an arbitrary partition of $E(\torso(t)) \setminus \{e_p\}$ with $|\compset| = 2$, which exists by the fact that $T$ is supercubic.
If $t$ is the root and $|\compset| \le 2$, then let $\compset$ be an arbitrary partition of $E(\torso(t)) \setminus \{e_p\}$ so that $|\compset| = 3$ and for each $C \in \compset$, all $e \in C$ have equal $V(e)$.
This again exists by the fact that $T$ is supercubic.
Now let $\compset' = \{C \in \compset : |C| \ge 2\}$.

First, we transform $\Tc$ into $\Tc \rescliqs (t,\compset')$ in time $\OO(k \cdot \|\torso(t)\|)$ by \Cref{lem:superbdsplittingalgo}.
For each child node $t_C$ of $t$ created like this, corresponding to a set $C \in \compset'$, we have that for all $e,e' \in E(\torso(t_C))$ it holds that $V(e) = V(e')$.
Then, if $t_C$ has degree more than $3$, we replace $t_C$ by an arbitrary cubic tree of nodes.
This construction can be done in total $\OO(k \cdot \|\torso(t)\|)$ time.
We then remove $t$ from $P$.

\begin{claim}
The resulting rooted superbranch decomposition and $P$ satisfy the invariants.
\end{claim}
\begin{claimproof}
Let $\Tc = (T,\lmap)$ be the old rooted superbranch decomposition and $\Tc^* = (T^*,\lmap^*)$ the resulting rooted superbranch decomposition.
Let $t'$ be the node of $\Tc^*$ corresponding to $t$.

We first argue the multiplicity.
All nodes in the binary trees created have degree at most $3$, so the multiplicity of their torsos is at most $3$.
We have that $\torso_{\Tc^*}(t') = \torso_{\Tc}(t) \rescliqs \compset' = \torso_{\Tc}(t) \rescliqs \compset$, implying that either the multiplicity of $\torso_{\Tc}(t')$ is $1$, or the degree of $t'$ is at most $3$, implying that the multiplicity of $\torso_{\Tc^*}(t')$ is at most $3$.

We then argue the downwards semi-internally connectedness.
Consider $\vec{uv} \in \vec{E}(T^*)$ with $v = \parent(u)$.
We observe that then, there exists $\vec{u_1 v_1}, \ldots, \vec{u_h, v_h} \in \vec{E}(T)$, with $v_i = \parent(u_i)$, so that $\lmap'(\vec{uv}) = \bigcup_{i=1}^h \lmap(\vec{u_i v_i})$ and $\bd(\lmap(\vec{u_i v_i})) = \bd(\lmap'(\vec{uv}))$.
As $\Tc$ is downwards semi-internally connected, this implies that $\Tc^*$ is downwards semi-internally connected.

For each $x \in V(T^*)$, there exists $y \in V(T)$ so that $V_{\Tc^*}(x) \subseteq V_{\Tc}(y)$, implying that $V(x)$ is $(2k,k)$-unbreakable in $G$.
Also, the operation does not increase the adhesion size, so $\adhsize(\Tc^*) \le \adhsize(\Tc) < 2k$.
\end{claimproof}

As each iteration runs in time $\OO(k \cdot \|\torso(t)\|)$, removes $t$ from $P$, and does not alter the torsos of nodes remaining in $P$, the overall algorithm runs in time $\OO(k \cdot \|\Tc\|)$.
\end{proof}

We then start preparing for the main step, making the superbranch decomposition outputted by the algorithm of \Cref{lem:algdownwlsemiconmult} downwards well-linked.
Here, the idea will be to work bottom-up in the superbranch decomposition, and for each oriented edge $\vec{tp}$, where $p=\parent(t)$, so that $\lmap(\vec{tp})$ is not well-linked but $\lmap(\vec{ct})$ is well-linked for all $c \in \children(t)$, find a partition of $E(\torso(t)) \setminus \{e_p\}$ into a family $\compset$ of $2^{\OO(k)}$ well-linked sets $C \in \compset$, and ``collapse'' the adhesion at $tp$ to $|\compset|$ adhesions corresponding to the separations $(C \orescliqs \Tc, \co{C} \orescliqs \Tc)$ for $C \in \compset$.
The next lemma is about finding this partition of $E(\torso(t)) \setminus \{e_p\}$.

\begin{lemma}
\label{lem:partitionintowlsets}
There is an algorithm that, given a hypergraph $G$ and a non-empty set $A \subseteq E(G)$, in time $2^{\OO(\bdc(A))} \cdot \rank(G)^2 \cdot \|G\|$ returns a partition $\{C_1,\ldots,C_h\}$ of $A$ into at least one and at most $h \le 2^{\bdc(A)}$ non-empty parts, so that for each $i \in [h]$,
\begin{itemize}
\item $C_i$ is well-linked, and
\item $\bdc(C_i) \le \bdc(A)$.
\end{itemize}
\end{lemma}
\begin{proof}
We do an iterative process that maintains a partition $\compset = \{C_1,\ldots,C_h\}$ of $A$ satisfying that $\sum_{i=1}^h 2^{\bdc(C_i)} \le \bdc(A)$.
Initially, we set $\compset = \{A\}$.

In each iteration, we first check by using the algorithm of \Cref{lem:algffwl} if there exists $C_i \in \compset$ so that $C_i$ is not well-linked.
If all $C_i \in \compset$ are well-linked, we can return $\compset$.

Otherwise, let $C_i$ be not well-linked, and let $(B_1,B_2)$ be a bipartition of $C_i$ returned by the algorithm of \Cref{lem:algffwl} so that $\bdc(B_1),\bdc(B_2) < \bdc(C_i)$.
We replace $C_i$ in $\compset$ by $B_1$ and $B_2$.
Note that this maintains the invariant that $\sum_{i=1}^h 2^{\bdc(C_i)} \le \bdc(A)$.

This procedure runs for at most $2^{\bdc(A)}$ iterations, and each iteration runs in $2^{\OO(\bdc(A))} \cdot \rank(G)^2 \cdot \|G\|$ time.
\end{proof}

Next, we need to argue that if $V(t)$ was $(2k,k)$-unbreakable, then after ``collapsing'' the adhesions below $t$, the resulting $V(t)$ is $(2^{\OO(k)}, k)$-unbreakable.
The next lemma is about this argument.

\begin{lemma}
\label{lem:expdownwardsunbrk}
Let $\Tc = (T,\lmap)$ be a rooted superbranch decomposition of a hypergraph $G$, $t \in \vint(T)$, and $k \ge 1$ an integer.
Let also $Q \subseteq \children(t) \cap \vint(T)$, and assume that
\begin{itemize}
\item for each $c \in Q$, $\lmap(\vec{ct})$ is semi-internally connected and $|V(c)| \le 2^{\OO(k)}$,
\item the multiplicity of $\torso(t)$ is at most $3$,
\item $V(t)$ is $(2k,k)$-unbreakable in $G$, and
\item $\adhsize(\Tc) \le 2k$.
\end{itemize}
Then, the set $V(t) \cup \bigcup_{c \in Q} V(c)$ is $(2^{\OO(k)}, k)$-unbreakable in $G$.
\end{lemma}
\begin{proof}
Let $h = 2^{\OO(k)}$ denote the maximum of $|V(c)|$ for $c \in Q$.
Denote $X = V(t) \cup \bigcup_{c \in Q} V(c)$.
For the sake of contradiction, suppose that $X$ is not $(h \cdot 3 \cdot 2^{3k}, k)$-unbreakable in $G$, and let $(A,\co{A})$ be a separation of $G$ of order $<k$ so that $|V(A) \cap X| \ge h \cdot 3 \cdot 2^{3k}$ and $|V(\co{A}) \cap X| \ge h \cdot 3 \cdot 2^{3k}$.
By the $(2k,k)$-unbreakability of $V(t)$, we have that either $|V(A) \cap V(t)| < 2k$ or $|V(\co{A}) \cap V(t)| < 2k$.
Assume \wilog that $|V(A) \cap V(t)| < 2k$.

\begin{claim}
\label{lem:expdownwardsunbrk:claim2}
There are at most $k-1$ children $c \in Q$ so that $\adh(ct)$ intersects $V(\co{A}) \setminus V(A)$ and $V(c) \setminus \adh(ct)$ intersects $V(A)$.
\end{claim}
\begin{claimproof}
Suppose that $\adh(ct)$ intersects $V(\co{A}) \setminus V(A)$, and $V(c) \setminus \adh(ct)$ intersects $V(A)$.
Because $\lmap(\vec{ct})$ is semi-internally connected, there then exists a set $C_c \subseteq V(\lmap(\vec{ct})) \setminus \adh(ct)$ so that $C_c$ is a connected set in $\primal(G)$, $N_{\primal(G)}(C_c) = \adh(ct)$, and $C_c$ intersects $V(A)$.
Let $v \in \adh(ct) \setminus V(A)$.
Because $C_c \cup \{v\}$ is a connected set in $\primal(G)$, and intersects both $V(A)$ and $V(\co{A})$, it must be that $\bd(A)$ intersects $C_c \cup \{v\}$ and therefore intersects $C_c$ because $v \notin V(A)$.

Because the sets $C_c$ over different children $c \in Q$ are disjoint, there can be at most $k-1$ of such children.
\end{claimproof}

Let $P$ be the set of children $c \in Q$ so that $V(c) \setminus \adh(ct)$ intersects $V(A)$.
At most $k-1$ of the children in $P$ satisfy \Cref{lem:expdownwardsunbrk:claim2}.
For the rest it holds that $\adh(ct) \subseteq V(A) \cap V(t)$, so there are at most $3 \cdot 2^{2k}$ of them because the multiplicity of $\torso(t)$ is $3$ and $|V(A) \cap V(t)| \le 2k$.
Therefore, we have that 
\[|P| \le k-1 + 3 \cdot 2^{2k}.\]
This implies that $|X \cap V(A)| \le 2k + h \cdot (k-1 + 3 \cdot 2^{2k}) < h \cdot 3 \cdot 2^{3k}$, contradicting our assumption.
\end{proof}

Then, we give the main algorithm for proving \Cref{the:highlevel:downwlalgo}, in particular, an algorithm that takes as an input the decomposition outputted by the algorithm of \Cref{lem:algdownwlsemiconmult}, and makes it downwards well-linked, with the cost of increasing the unbreakability guarantee.

\begin{lemma}
\label{lem:downwlalgfinallem}
There is an algorithm that, given a rooted superbranch decomposition $\Tc$ of a hypergraph $G$ and an integer $k \ge 1$ so that
\begin{itemize}
\item the multiplicity of each torso of $\Tc$ is at most $3$,
\item $\Tc$ is downwards semi-internally connected,
\item $\Tc$ is $(2k, k)$-unbreakable, and
\item $\adhsize(\Tc) \le 2k$,
\end{itemize}
in time $2^{\OO(k)} \cdot \|\Tc\|$ outputs a rooted superbranch decomposition $\Tc'$ of $G$ so that
\begin{itemize}
\item $\Tc'$ is downwards well-linked,
\item $\Tc'$ is $(2^{\OO(k)}, k)$-unbreakable, and
\item $\adhsize(\Tc') \le 2k$.
\end{itemize}
\end{lemma}
\begin{proof}
Denote $\Tc = (T,\lmap)$, and assume that $\troot(T) \in \vint(T)$, because if the root is a leaf, we can move it to the adjacent internal node.
We will edit $\Tc = (T,\lmap)$, while maintaining a prefix $P \subseteq \vint(T)$ of $T$, and two disjoint subsets $P_s, P_b \subseteq P$, satisfying the following invariants

\begin{enumerate}
\item if $t \in P \setminus (P_s \cup P_b)$, then the multiplicity of $\torso(t)$ is at most $3$,\label{lem:downwlalgfinallem:inv8}
\item if $t \in P \setminus (P_s \cup P_b)$, then $V(t)$ is $(2k,k)$-unbreakable,\label{lem:downwlalgfinallem:inv6}
\item if $t \in P \setminus \{\troot(T)\}$, then $\lmap(\vec{tp})$ is semi-internally connected, where $p = \parent(t)$,\label{lem:downwlalgfinallem:inv2}
\item $\adhsize(\Tc) \le 2k$\label{lem:downwlalgfinallem:inv5},
\item if $t \in V(T) \setminus P$, then $\lmap(\vec{tp})$ is well-linked, where $p = \parent(t)$,\label{lem:downwlalgfinallem:inv1}
\item if $t \in P_s \cup P_b$, then all children of $t$ are in $V(T) \setminus P$,\label{lem:downwlalgfinallem:inv3}
\item if $t \in P_s$, then $|E(\torso(t))| \le 2^{2k}+1$, and\label{lem:downwlalgfinallem:inv4}
\item $\Tc$ is $(2^{\OO(k)}, k)$-unbreakable (where the $2^{\OO(k)}$ is from \Cref{lem:expdownwardsunbrk}).\label{lem:downwlalgfinallem:inv7}
\end{enumerate}

Initially, these invariants are satisfied by picking $P = \vint(T)$ and $P_s,P_b = \emptyset$.
Once they become satisfied with $P = \{\troot(T)\}$, we can output $\Tc$.

Until that, we will repeat the following process.
Note that $P \setminus P_s$ is a prefix of $T$, and pick a node $t \in P \setminus P_s$ so that none of the children of $t$ are in $P \setminus P_s$.
We consider the following cases.

\paragraph{Case 1: $t \in P \setminus P_b$ and all children of $t$ are in $V(T) \setminus P$.}
In this case, we simply insert $t$ into $P_b$.
This satisfies the invariants, as there are no other extra requirements for $P_b$ compared to $P$ except that all children are required to be in $V(T) \setminus P$.
This runs in time $\OO(1)$.

\paragraph{Case 2: $t \in P \setminus P_b$ and at least one child of $t$ is in $P$.}
Now, $t$ has some children in $P_s$, but by the choice of $t$ no children in $P \setminus P_s$.
We transform $\Tc$ by contracting all edges between $t$ and its children that are in $P_s$, and inserting the resulting node to $P$ and $P_b$.
By using \Cref{lem:superbdcontractalgo} for the contractions, and the fact that $|E(\torso(x))| \le 2^{2k}+1$ for all $x \in P_s$, we implement this transformation in time $2^{\OO(k)} \cdot \|\torso(t)\|$.

\begin{claim}
This transformation maintains the invariants.
\end{claim}
\begin{claimproof}
Let us denote by $\Tc = (T, \lmap)$, $P$, $P_s$, and $P_b$ the decomposition and the sets before the contractions, and by $\Tc' = (T',\lmap')$, $P'$, $P_s'$, $P_b'$ after the contractions.
Denote by $t$ both the originally picked node in $P \setminus P_b$ and the resulting contracted node in $P_b'$.

The most interesting of the invariants in this case is \Cref{lem:downwlalgfinallem:inv7}, so let us prove it first by applying \Cref{lem:expdownwardsunbrk}.
We observe that for all $x \in \vint(T') \setminus \{t\}$ we have that $\torso_{\Tc'}(x) = \torso_{\Tc}(x)$, in which case \Cref{lem:downwlalgfinallem:inv7} clearly holds, so it remains to prove that $\torso_{\Tc'}(t)$ is $(2^{\OO(k)}, k)$-unbreakable.
Let $Q = \children_{T}(t) \cap P_s$ be the children of $t$ in $P_s$.
Now we have that 
\[V_{\Tc'}(t) = V_{\Tc}(t) \cup \bigcup_{c \in Q} V_{\Tc}(c),\]
so by \Cref{lem:expdownwardsunbrk}, $V_{\Tc'}(t)$ is $(2^{\OO(k)}, k)$-unbreakable.
Note that for applying \Cref{lem:expdownwardsunbrk}, we used that $\Tc$, $P$, $P_b$, and $P_s$ satisfy the invariants of \Cref{lem:downwlalgfinallem:inv2,lem:downwlalgfinallem:inv5,lem:downwlalgfinallem:inv6,lem:downwlalgfinallem:inv8,lem:downwlalgfinallem:inv4}

To prove \Cref{lem:downwlalgfinallem:inv6,lem:downwlalgfinallem:inv8} we observe that if $x \in P' \setminus (P_s' \cup P_b')$, then $x \in P \setminus (P_s \cup P_b)$ and $\torso_{\Tc'}(x) = \torso_{\Tc}(x)$.
For \Cref{lem:downwlalgfinallem:inv2} we observe that $\lmap'(\vec{xp'}) = \lmap(\vec{xp})$ for all $x \in P'$, where $p' = \parent_{T'}(x)$ and $p = \parent_{T}(x)$.
For \Cref{lem:downwlalgfinallem:inv5} we recall that contraction does not increase adhesion size.
For \Cref{lem:downwlalgfinallem:inv1}, it suffices to observe that if $x \in V(T') \setminus P'$, then $x \in V(T) \setminus P$, and $\lmap'(\vec{xp'}) = \lmap(\vec{xp})$, where $p' = \parent_{T'}(x)$ and $p = \parent_{T}(x)$.
For \Cref{lem:downwlalgfinallem:inv3}, the only node in $P'_s \cup P'_b$ for which the children in $T'$ are different than those in $T$ is the node $t$, which by construction does not have any children in $P$.
For \Cref{lem:downwlalgfinallem:inv4} we observe that if $x \in P_s'$, then $\torso_{\Tc'}(x) = \torso_{\Tc}(x)$.
\end{claimproof}

\paragraph{Case 3: $t \in P_b$.}
Now all children of $t$ are in $V(T) \setminus P$.
Let $p = \parent(t)$, and denote by $e_p$ the hyperedge of $\torso(t)$ corresponding to the adhesion at $pt$.
Note that $t$ has a parent because we assumed that $P \neq \{\troot(T)\}$.
We apply \Cref{lem:partitionintowlsets} to compute a partition $\compset = \{C_1,\ldots,C_h\}$ of $E(\torso(t)) \setminus \{e_p\}$ so that $h \le 2^{2k}$, and for each $i \in [h]$, $C_i$ is well-linked and has $\bdc(C_i) \le \bdc(e_p) \le 2k$.
This runs in time $2^{\OO(k)} \cdot \rank(\torso(t))^2 \cdot \|\torso(t)\| = 2^{\OO(k)} \cdot \|\torso(t)\|$ (note that $\rank(\torso(t)) \le \adhsize(\Tc) \le 2k$).

First, if $h=1$, then $E(\torso(t)) \setminus \{e_p\}$ is well-linked in $\torso(t)$, and therefore the invariant of \Cref{lem:downwlalgfinallem:inv1} and the fact that well-linkedness is transitive (\Cref{lem:linkedcliq}) implies that $\lmap(\vec{tp})$ is well-linked.
In this case, we simply remove $t$ from $P$ and $P_b$.
This maintains the invariants because the only extra property required for nodes not in $P$ compared to nodes in $P_b$ is \Cref{lem:downwlalgfinallem:inv1}, i.e., that $\lmap(\vec{tp})$ is well-linked.

Then assume $h \ge 2$, and let $\compset' = \{C \in \compset : |C| \ge 2\}$.
We have that for all $C \in \compset'$ it holds that $|C| \ge 2$, $|\co{C}| \ge 2$, and $\co{C} \notin \compset'$, and therefore $\compset'$ is a splitting family for $t$.
We transform $\Tc$ into $\Tc \rescliqs (t,\compset')$ in time $\OO(k \cdot \|\torso(t)\|)$ by \Cref{lem:superbdsplittingalgo}.
We insert the resulting node $t'$ into $P$ and $P_s$, and the resulting nodes $t_C$ for $C \in \compset'$ will not be in $P$.

\begin{claim}
This transformation maintains the invariants.
\end{claim}
\begin{claimproof}
Let us denote by $\Tc = (T, \lmap)$, $P$, $P_s$, and $P_b$ the decomposition and the sets before the transformation, and by $\Tc' = (T',\lmap')$, $P'$, $P_s'$, $P_b'$ after.
Denote by $t_C$ the nodes corresponding to $C \in \compset'$ resulting from the splitting, and denote by $t$ both the originally picked node $t \in P_b$ and the node resulting from the splitting operation.

Let us start proving the invariants with \Cref{lem:downwlalgfinallem:inv1}, as it is the most interesting here.
Consider $x \in V(T') \setminus P'$.
If $x \notin \{t_C : C \in \compset'\}$, then $x \in V(T) \setminus P$ and $\lmap'(\vec{xp}) = \lmap(\vec{xp})$, where $p$ is the parent of $x$ in $T$ (and $T'$).
If $x = t_C$ for some $C \in \compset'$, then $\lmap'(\vec{t_C t}) = C \orescliqs \Tc$.
Because $C$ is well-linked in $\torso_{\Tc}(t)$, and all children of $t$ in $\Tc$ are in $V(T) \setminus P$, \Cref{lem:downwlalgfinallem:inv1} and the transitivity of well-linkedness imply that $C \orescliqs \Tc$ is well-linked in $G$.

For \Cref{lem:downwlalgfinallem:inv6,lem:downwlalgfinallem:inv8} we have that if $x \in P' \setminus (P_s' \cup P_b')$, then $x \in P \setminus (P_s \cup P_b)$ and $\torso_{\Tc'}(x) = \torso_{\Tc}(x)$.
For \Cref{lem:downwlalgfinallem:inv2}, we observe that if $x \in P'$, then $x \in P$ and $\lmap'(\vec{xp}) = \lmap(\vec{xp})$, where $p$ is the parent of $x$ (in both $T$ and $T'$).
For \Cref{lem:downwlalgfinallem:inv5}, we have that all internal separations of $\Tc'$ are also internal separations of $\Tc$ or of form $(C \orescliqs \Tc, \co{C} \orescliqs \Tc)$ for $C \in \compset$.
Therefore, they have order at most $2k$.

For \Cref{lem:downwlalgfinallem:inv3}, let $x \in P_s' \cup P_b'$, and assume first that $x \neq t$.
Then, $x \in P_s \cup P_b$ and the status of the children of $x$ has not changed.
If $x = t$, then all children of $t$ are either children of $t$ in $T$, which were in $V(T) \setminus P$ and did not get inserted into $P$, or of form $t_C$, which also were not inserted into $P$.
For \Cref{lem:downwlalgfinallem:inv4}, we first observe that if $x \in P_s' \setminus \{t\}$, then $x \in P_s$ and $\torso_{\Tc'}(x) = \torso_{\Tc}(x)$.
Then, if $x = t$, the fact that $|\compset| \le 2^{2k}$ implies that $|E(\torso_{\Tc'}(t))| \le 2^{2k}+1$.
For \Cref{lem:downwlalgfinallem:inv7} we have that if $x \in P_b'$, then $x \in P_b$ and $\torso_{\Tc'}(x) = \torso_{\Tc}(x)$.
\end{claimproof}

\paragraph{Overall analysis.}
All of Cases 1, 2, and 3 run in time at most $2^{\OO(k)} \cdot \|\torso(t)\|$.
However, it is not immediately clear that the running times sum up to at most $2^{\OO(k)} \cdot \|\Tc_{\mathsf{init}}\|$, where $\Tc_{\mathsf{init}}$ is the initial superbranch decomposition, because we are editing the decomposition along the way.

To prove that the overall running time is $2^{\OO(k)} \cdot \|\Tc_{\mathsf{init}}\|$, we consider a potential function of form 
\[\Phi(\Tc) = (2^{2k}+2) \cdot \sum_{t \in P \setminus (P_b \cup P_s)} |E(\torso(t))| + \sum_{t \in P_b} |E(\torso(t))|.\]

In Case 1, the first sum decreases by $|E(\torso(t))|$ and the second sum increases by $|E(\torso(t))|$.
As the first sum has larger multiplier, $\Phi(\Tc)$ decreases by at least $|E(\torso(t))|$.
In Case 2, the first sum again decreases by $|E(\torso(t))|$, and the second sum increases by at most $(2^{2k}+1) \cdot |E(\torso(t))|$.
Therefore, as the first sum has multiplier $(2^{2k}+2)$, again $\Phi(\Tc)$ decreases by at least $|E(\torso(t))|$.
In Case 3, the first sum does not change, while the second sum decreases by $|E(\torso(t))|$, implying that $\Phi(\Tc)$ decreases by at least $|E(\torso(t))|$.

As this potential is initially bounded by $\Phi(\Tc) \le 2^{\OO(k)} \cdot \|\Tc\|$, the overall running time is $2^{\OO(k)} \cdot \|\Tc\|$.
\end{proof}

Now, to conclude with \Cref{the:highlevel:downwlalgo}, it remains to combine the above algorithms for transforming superbranch decompositions.

\thehighleveldownwlalgo*
\begin{proof}
By combining the algorithms of \Cref{lem:algdownwlsemicon,lem:algdownwlsemiconmult,lem:downwlalgfinallem}.
\end{proof}

\section{Mixed-$k$-well-linked decomposition}
\label{sec:mixedkwldecomp}
The goal of this section is to prove \Cref{the:highlevel:upwlalgo}, which we re-state now.

\thehighlevelupwlalgo*

Informally speaking, the goal after making the decomposition downwards well-linked would have been to make it also ``upwards well-linked''.
However, we have to settle for making it only upwards $k$-well-linked, and moreover, in the process the decomposition becomes ``scrambled'', so we cannot even say that it is upwards $k$-well-linked, but we still get mixed-$k$-well-linked internal separations and the extra property of the second bullet point above.
Note that in particular, the first two bullet points mean that $\Tc'$ can be partitioned into connected subtrees, so that each subtree can be assigned a root so that it is downwards well-linked and upwards $k$-well-linked with respect to that root, and the separations between the subtrees are doubly well-linked.

We first give an algorithm that will be used for partitioning the torsos of the input superbranch decomposition.
This is analogous to \Cref{lem:partitionintowlsets}, but targeting $k$-well-linkedness instead of well-linkedness, and having to work in a local manner instead of accessing the whole hypergraph.
This algorithm works by applying the algorithm of \Cref{lem:kwlubalg} repeatedly.

\begin{lemma}
\label{lem:algpartitiontokwlsets}
Let $G$ be a hypergraph whose representation is already stored, and let $s \ge k \ge 1$ be integers so that $G$ is $(s,k)$-unbreakable.
There is an algorithm that, given a set $A \subseteq E(G)$, in time $|A| \cdot s^{\OO(k + \rank(G) + \bdc(A))}$ returns either a multiplicity-$3$-violator in $G$, or a collection of $h \le 2^{\bdc(A)}-1$ sets $C_1, \ldots, C_h \subseteq \co{A}$, so that
\begin{itemize}
\item $|V(C_i)| < s$ and $|C_i| \le s^{\OO(\rank(G))}$ for $i \in [h]$,
\item $(C_1, \ldots, C_h, C_{h+1} = \co{A} \setminus (\bigcup_{i=1}^h C_i))$ is a partition of $\co{A}$ into $h+1$ non-empty sets,
\item $C_i$ is $k$-well-linked for all $i \in [h+1]$,
\item $\co{C_i}$ is linked into $A$ for all $i \in [h+1]$,
\item $\bdc(C_i) \le \bdc(A)$ for all $i \in [h+1]$, and
\item there is at most one $i \in [h+1]$ so that $\bdc(C_i) \ge k$.
\end{itemize}
\end{lemma}
\begin{proof}
We will maintain a collection of sets $C_1,\ldots,C_h \subseteq \co{A}$ and a marking $M \subseteq [h+1]$ so that
\begin{itemize}
\item $|V(C_i)| < s$ and $|C_i| \le 3 \cdot s^{\rank(G)}$ for $i \in [h]$,
\item $(C_1, \ldots, C_h, C_{h+1} = \co{A} \setminus (\bigcup_{i=1}^h C_i))$ is a partition of $\co{A}$ into $h+1$ non-empty sets,
\item $\sum_{i=1}^{h+1} 2^{\bdc(C_i)} \le 2^{\bdc(A)}$,
\item if $i \in M$, then $C_i$ is $k$-well-linked,
\item $\co{C_i}$ is linked into $A$ for all $i \in [h+1]$, and
\item there is at most one $i \in [h+1]$ so that $\bdc(C_i) \ge k$.
\end{itemize}

Note that the condition $\sum_{i=1}^{h+1} 2^{\bdc(C_i)} \le 2^{\bdc(A)}$ implies that $\bdc(C_i) \le \bdc(A)$ for all $i \in [h+1]$, and that $h+1 \le 2^{\bdc(A)}$.

Initially, the empty collection with $h=0$ and $M = \emptyset$ satisfies these requirements.
If $M = [h+1]$, then we can return the sets $C_1,\ldots,C_h$.

Therefore, assume then that $M \neq [h+1]$, and let $i \in [h+1] \setminus M$.

\paragraph{Case 1: $i = h+1$.}
We apply the algorithm of \Cref{lem:kwlubalg} with $\co{C_{h+1}} = A \cup \bigcup_{i=1}^h C_i$ and the parameter $p=4$.
It runs in time $(|A|+h \cdot s^{\OO(\rank(G))}) \cdot s^{\OO(k+\rank(G))} \cdot 2^{\OO(\bdc(A))} = |A| \cdot s^{\OO(k+\rank(G))} \cdot 2^{\OO(\bdc(A))}$.
If it concludes that $C_{h+1}$ is $k$-well-linked, we may insert $h+1$ to $M$, and we make progress by increasing $|M|$.
If it returns a multiplicity-$3$-violator, then we may immediately return it.

Otherwise, it returns a set $B \subseteq C_{h+1}$ so that $|V(B)| < s$, $\bdc(B) < \bdc(C_{h+1})$, $\bdc(C_{h+1} \setminus B) < \bdc(C_{h+1})$, either $\bdc(B) < k$ or $\bdc(C_{h+1} \setminus B) < k$, and both $\co{B}$ and $\co{C_{h+1} \setminus B}$ are linked into $\co{C_{h+1}}$.
If $|B| > 3 \cdot s^{\rank(G)}$, we find in time $s^{\OO(\rank(G))}$ a multiplicity-$3$-violator in $B$ and return it.
Then, we insert $B$ to the collection $C_1,\ldots,C_h$ of sets, adjusting $h$ and $M$ accordingly.
The first two invariants are clearly satisfied, and the third is implied by the fact that $\bdc(B) < \bdc(C_{h+1})$ and $\bdc(C_{h+1} \setminus B) < \bdc(C_{h+1})$.
The fourth invariant is clearly maintained.
For the fifth, we have that $\co{B}$ and $\co{C_{h+1} \setminus B}$ are linked into $\co{C_{h+1}}$ as guaranteed by \Cref{lem:kwlubalg}, and therefore by \Cref{lem:transitivityoflinkedness} they are linked into $A$.
The sixth holds because if $\bdc(C_{h+1}) \ge k$, then at most one of $\bdc(B) \ge k$ and $\bdc(C_{h+1} \setminus B) \ge k$ holds, and if $\bdc(C_{h+1}) < k$, then neither holds.
We make progress by increasing $h$.

\paragraph{Case 2: $i \le h$.}
We apply the algorithm of \Cref{lem:algffwl} to either conclude that $C_i$ is $k$-well-linked, or to return a bipartition $(B_1,B_2)$ of $C_i$ so that $\bdc(B_j) < \bdc(C_i)$ for both $j \in [2]$, $\bdc(B_1) < k$, and $\co{B_j}$ is linked into $\co{C_i}$ for both $j \in [2]$.
If $C_i$ is $k$-well-linked, we simply insert $i$ into $M$, making progress by increasing $|M|$, and in the latter case we replace $C_i$ by $B_1$ and $B_2$, adjusting $h$ and $M$ accordingly.
The invariants can be proven to hold similarly as in Case 1.
In the latter case we make progress by increasing $h$.

As each of the cases increases $h$ or $|M|$, and they never decrease and are bounded by $2^{\bdc(A)}$, this process terminates after $\OO(2^{\bdc(A)})$ iterations, so also the total running time is $|A| \cdot s^{\OO(k+\rank(G)+\bdc(A))}$.
\end{proof}

We then need a lemma that asserts that the partition given by \Cref{lem:algpartitiontokwlsets} satisfies certain properties when \Cref{lem:algpartitiontokwlsets} is applied to a torso of a node of a superbranch decomposition satisfying certain properties.

\begin{lemma}
\label{lem:colsetslinkedargument}
Let $\Tc = (T,\lmap)$ be a superbranch decomposition of a hypergraph $G$, $k \ge 1$ an integer, and $t \in \vint(T)$ so that for all $u \in N_T(t)$ the set $\lmap(\vec{ut})$ is $k$-better-linked.
Let also $x \in N_T(t)$ so that $\lmap(\vec{xt})$ is well-linked.
Then, let $(C_1, \ldots, C_h)$ be a partition of $E(\torso(t)) \setminus \{e_x\}$ into non-empty sets so that for all $i \in [h]$,
\begin{itemize}
\item $C_i$ is $k$-well-linked in $\torso(t)$,
\item $\co{C_i}$ is linked into $\{e_x\}$ in $\torso(t)$,
\end{itemize}
and $\bdc(C_i) < k$ for all $i < h$.
Then, for all $i \in [h]$,
\begin{itemize}
\item $C_i \orescliqs \Tc$ is $k$-well-linked in $G$,
\item $\co{C_i} \orescliqs \Tc$ is well-linked in $G$.
\end{itemize}
\end{lemma}
\begin{proof}
First, because being $k$-well-linked is transitive along $k$-better-linked sets (\Cref{lem:kwltransitive}), and each set $\lmap(\vec{ut})$ for $u \in N_T(t)$ is $k$-better-linked, it follows that each $C_i \orescliqs \Tc$ is $k$-well-linked.

Then, because $\bdc(\co{C_i}) < k$ for all $i < h$, and for each $u \in N_T(t)$ we have that $\lmap(\vec{ut})$ is $k$-better-linked, we get by \Cref{lem:linkednessexpand} that $\co{C_i} \orescliqs \Tc$ is linked into $\lmap(\vec{xt})$ in $G$.
Now, \Cref{lem:linkedcombwelllinked} implies that $\co{C_i} \orescliqs \Tc$ is well-linked in $G$ for all $i<h$.

It remains to argue that $\co{C_h} \orescliqs \Tc$ is well-linked in $G$.
For this, the challenge is that it can be that $\bdc(\co{C_h}) \ge k$, so we cannot directly use \Cref{lem:linkednessexpand} to argue that $\co{C_h} \orescliqs \Tc$ is linked into $\lmap(\vec{xt})$ in $G$.
However, we can do this indirectly by using the fact that each $(C_i \orescliqs \Tc, \co{C_i} \orescliqs \Tc)$ for $i<h$ has order $<k$ and is mixed-$k$-well-linked, implying that it is in fact doubly well-linked.

In particular, let $G' = G \rescliqs (C_1 \orescliqs \Tc, \ldots, C_{h-1} \orescliqs \Tc)$.
Let also $U = C_h \orescliqs \Tc$ and $D = (E(G) \setminus U) \rescliqs (C_1 \orescliqs \Tc, \ldots, C_{h-1} \orescliqs \Tc)$.
Note that $(U,D)$ is a separation of $G'$ corresponding to the separation $(C_h, \co{C_h})$ of $\torso(t)$.
Because $\co{C_h}$ is linked into $\{e_x\}$ in $\torso(t)$, we have that $D$ is linked into $\lmap(\vec{xt})$ in $G'$.

Because $\lmap(\vec{xt})$ is well-linked in $G$ and thus in $G'$, by \Cref{lem:linkedcombwelllinked} this implies that $D$ is well-linked in $G'$.
Now, each $C_i \orescliqs \Tc$ for $i \in [h-1]$ is well-linked in $G$ (because it is $k$-well-linked and has $\bdc(C_i)<k$), and therefore because well-linkedness is transitive (\Cref{lem:linkedcliq}), this implies that $D \orescliqs (C_1 \orescliqs \Tc, \ldots, C_{h-1} \orescliqs \Tc) = \co{C_h} \orescliqs \Tc$ is well-linked in $G$.
\end{proof}

We also need to handle the case when \Cref{lem:algpartitiontokwlsets} outputs a multiplicity-$3$-violator, and for this we also need a lemma analogous to \Cref{lem:colsetslinkedargument} but simpler, which we prove next.

\begin{lemma}
\label{lem:multipvioldwl}
Let $\Tc = (T,\lmap)$ be a superbranch decomposition of a graph $G$.
Let $t \in \vint(T)$ so that for all $v \in N_T(t)$, except at most one, $\lmap(\vec{vt})$ is well-linked in $G$.
Let $X$ be a multiplicity-$3$-violator in $\torso(t)$.
If $C \subseteq X$ and $|C| = 2$, then $(C \orescliqs \Tc, \co{C} \orescliqs \Tc)$ is a doubly-well-linked separation of $G$.
\end{lemma}
\begin{proof}
Consider $C' = C$ or $C' = \co{C}$.
In both cases, $|C'| \ge 2$ implies that there exists $v \in N_T(t)$ so that $e_v \in C' \cap X$ and $\lmap(\vec{vt})$ is well-linked in $G$.
Because $\lmap(\vec{vt}) \subseteq C' \orescliqs \Tc$ and $\bd(\lmap(\vec{vt})) = \bd(C' \orescliqs \Tc)$, by \Cref{lem:linkednesscornercase}, $C' \orescliqs \Tc$ is linked into $\lmap(\vec{vt})$, and therefore by \Cref{lem:linkedcombwelllinked}, $C'$ is well-linked.
\end{proof}

Now we are ready to prove \Cref{the:highlevel:upwlalgo}.

\thehighlevelupwlalgo*
\begin{proof}
Let $s = 2^{\OO(k)}$ so that the input superbranch decomposition $\Tc = (T,\lmap)$ is $(s,k)$-unbreakable.
Let us also assume that $\troot(T) \in \vint(T)$, because if the root is a leaf, we can move it to an adjacent internal node.

We will do a process that maintains a rooted superbranch decomposition $\Tc = (T,\lmap)$ of $G$ and a prefix $P \subseteq \vint(T)$ of $T$, satisfying the following invariants.

\begin{enumerate}
\item if $u \in \vint(T) \setminus P$ and $p = \parent(u)$, then $\lmap(\vec{up})$ is well-linked,\label{the:highlevelupwlalgo:inv2}
\item if $u \in \vint(T) \setminus P$, then $V(u)$ is $(s, k)$-unbreakable.\label{the:highlevelupwlalgo:inv6}
\item $\adhsize(\Tc) \le 2k$,\label{the:highlevelupwlalgo:inv4}
\item if $uv \in E(T)$ and $u,v \in P$, then the separation $(\lmap(\vec{uv}), \lmap(\vec{vu}))$ of $G$ is mixed-$k$-well-linked,\label{the:highlevelupwlalgo:inv1}
\item if $u \in P$, then all but at most one of the sets $\lmap(\vec{vu})$ for $v \in N_T(u)$ are well-linked,\label{the:highlevelupwlalgo:inv3}
\item if $u \in P$, then $V(u)$ is $(s + 2k \cdot (2^{2k}+1), k)$-unbreakable, and\label{the:highlevelupwlalgo:inv5}
\end{enumerate}

Note that the initial input rooted superbranch decomposition $\Tc$ and the set $P = \{\troot(T)\}$ satisfy the invariants.
Also, if $P = \vint(T)$ and the invariants are satisfied, we can return $\Tc$.

Let us denote $s' = s + 2k \cdot (2^{2k}+1) \le 2^{\OO(k)}$.

Let us observe that our invariants imply that if $u \in P$ and $v \in N_T(u)$, then $\lmap(\vec{vu})$ is $k$-better-linked.

Now, as long as $P \neq \vint(T)$, we repeatedly do the following process.
Choose $t \in \vint(T)$ so that $t \notin P$ but $p = \parent(t) \in P$.

\paragraph{Case 1: Small parent.}
First, assume that $|E(\torso(p))| \le 2^{2k}+1$.
Note that we can test this in $2^{\OO(k)}$ time by iterating over $E(\torso(p))$.
In this case, we simply contract $p$ and $t$ together, i.e., transform $\Tc$ into $\Tc \contr tp$.
This can be done in $\OO(\|\torso(p)\|) = 2^{\OO(k)}$ time by \Cref{lem:superbdcontractalgo}.
We insert the node introduced by the contraction to the set $P$.
This decreases $|V(T)|$ while keeping $|P|$ the same.

\begin{claim}
This transformation maintains the invariants.
\end{claim}
\begin{claimproof}
Denote by $\Tc = (T,\lmap)$ the superbranch decomposition before the transformation and by $\Tc' = (T',\lmap') = \Tc \contr tp$ the superbranch decomposition after.
Use similarly $P$ and $P'$.
Denote by $t'$ the node of $\Tc'$ corresponding to the contraction of $tp$.

To prove \Cref{the:highlevelupwlalgo:inv2}, suppose $u \in \vint(T') \setminus P'$ and $g = \parent_{T'}(u)$.
Note that then, $u \in \vint(T) \setminus P$.
If $g \neq t'$, then $\lmap'(\vec{ug}) = \lmap(\vec{ug})$ and therefore is well-linked.
If $g = t'$, then $u$ was either a child of $p$ or a child of $t$ in $\Tc$, and therefore either $\lmap'(\vec{ug}) = \lmap(\vec{up})$ or $\lmap'(\vec{ug}) = \lmap(\vec{ut})$, both of which are well-linked by the invariant.

For \Cref{the:highlevelupwlalgo:inv5,the:highlevelupwlalgo:inv6}, we first observe that all other torsos of $\Tc'$ except $\torso(t')$ are also torsos of $\Tc$.
Because $|E(\torso(p))| \le 2^{2k}+1$ and $\adhsize(\Tc) \le 2k$, we have that $|V(\torso(p))| \le 2k \cdot (2^{2k}+1)$.
Now, as $V(\torso(t))$ is $(s,k)$-unbreakable, we get that $V(\torso(t')) = V(\torso(t)) \cup V(\torso(p))$ is $(s+2k \cdot (2^{2k}+1), k)$-unbreakable.

\Cref{the:highlevelupwlalgo:inv4} follows trivially from the fact that contraction does not increase adhesion size.
For \Cref{the:highlevelupwlalgo:inv1}, we observe that all internal separations corresponding to adhesions between two nodes of $\Tc'$ in $P'$ correspond to internal separations between two nodes of $\Tc$ in $P$, and therefore are mixed-$k$-well-linked.

To prove \Cref{the:highlevelupwlalgo:inv3}, consider first $u \in P' \setminus \{t'\}$.
We have that the collection of sets $\{\lmap'(\vec{vu})\}_{v \in N_{T'}(u)}$ is the same as the collection of sets $\{\lmap(\vec{vu})\}_{v \in N_{T}(u)}$, and therefore the invariant continues to hold.
Then consider $u = t'$.
We have that $\{\lmap'(\vec{vt'})\}_{v \in N_{T'}(t')} = \{\lmap(\vec{vt})\}_{v \in N_{T}(t) \setminus \{p\}} \cup \{\lmap(\vec{vp})\}_{v \in N_{T}(p) \setminus \{t\}}$, and therefore as $\lmap(\vec{vt})$ for the first set is always well-linked by \Cref{the:highlevelupwlalgo:inv2}, and for the second set always except at most once by \Cref{the:highlevelupwlalgo:inv3}, we get that the invariant is maintained.
\end{claimproof}

\paragraph{Case 2: Large parent.}
Now assume $|E(\torso(p))| \ge 2^{2k}+2$.
We apply the algorithm of \Cref{lem:algpartitiontokwlsets} with the hypergraph $\torso(p)$, the set $A = \{e_t\}$, and $p = 4$, to either return a multiplicity-$3$-violator in $\torso(p)$, or a collection of $h \le 2^{2k}-1$ sets $C_1,\ldots,C_h \subseteq E(\torso(p))$ satisfying certain properties.
This runs in time ${s'}^{\OO(k+\rank(\torso(t)))} \cdot 2^{\OO(k)} = 2^{\OO(k^2)}$.

\paragraph{Case 2.1: Empty collection of sets.}
Assume first that the algorithm of \Cref{lem:algpartitiontokwlsets} returned the empty collection of sets, i.e., $h=0$.
This implies that $E(\torso(p)) \setminus \{e_t\}$ is $k$-well-linked in $\torso(p)$.
Because $k$-well-linkedness is transitive along $k$-better-linked sets (\Cref{lem:kwltransitive}) and $p \in P$, this implies that $\lmap(\vec{pt})$ is $k$-well-linked in $G$.
In this case, we simply insert $t$ into $P$, without editing~$\Tc$.

\begin{claim}
This transformation maintains the invariants.
\end{claim}
\begin{claimproof}
\Cref{the:highlevelupwlalgo:inv2,the:highlevelupwlalgo:inv6} is only weaker for the new $P$, and \Cref{the:highlevelupwlalgo:inv4} holds trivially.
To prove \Cref{the:highlevelupwlalgo:inv1}, the only new separation to consider is $(\lmap(\vec{tp}), \lmap(\vec{pt}))$, and for it we just argued that $\lmap(\vec{pt})$ is $k$-well-linked in $G$, and by \Cref{the:highlevelupwlalgo:inv2} $\lmap(\vec{tp})$ is well-linked.
\Cref{the:highlevelupwlalgo:inv3} for $t$ follows from \Cref{the:highlevelupwlalgo:inv2} and \Cref{the:highlevelupwlalgo:inv5} follows from \Cref{the:highlevelupwlalgo:inv6}.
\end{claimproof}

\paragraph{Case 2.2: Non-empty collection of sets.}
Assume then that \Cref{lem:algpartitiontokwlsets} returned the collection of sets $C_1,\ldots,C_h \subseteq E(\torso(p))$ and $h \ge 1$.
By combining the guarantees given by \Cref{lem:algpartitiontokwlsets,lem:colsetslinkedargument}, we get that these sets satisfy the following properties:
\begin{itemize}
\item $|V(C_i)| < s'$ and $|C_i| \le {s'}^{\OO(k)}$ for all $i \in [h]$,
\item $\compset = \{C_1, \ldots, C_h, C_{h+1} = E(\torso(p)) \setminus (\{e_t\} \cup \bigcup_{i=1}^h C_i)\}$ is a partition of $E(\torso(p)) \setminus \{e_t\}$ into $h+1 \le 2^{2k}$ non-empty sets,
\item $\bdc(C_i) \le \bdc(\{e_t\})$ for all $i \in [h+1]$ and there is at most one $i \in [h+1]$ so that $\bdc(C_i) \ge k$,
\item $C_i \orescliqs \Tc$ is $k$-well-linked in $G$ for all $i \in [h+1]$, and
\item $\co{C_i} \orescliqs \Tc$ is well-linked in $G$ for all $i \in [h+1]$.
\end{itemize}

Denote $\compset' = \compset \cup \{\{e_t\}\}$.
Now, $\compset'$ is a partition of $E(\torso(p))$ into at least $3$ and at most $2^{2k}+1$ non-empty sets, with $\bdc(C) \le 2k$ for all $C \in \compset'$.
As $|E(\torso(p))| \ge 2^{2k}+2$, there exists at least one $C \in \compset'$ with $|C| \ge 2$.
We have that $\compset'' = \{C_1, \ldots, C_h, \{e_t\}\}$ is an implicit representation of the splitting partition $\compset'$.

We use \Cref{lem:superbdsplittingalgoimplicit} with $\compset''$ to transform $\Tc$ into $\Tc \rescliqs (p,\compset')$, and insert all the newly created nodes into $P$.
This runs in time $2^{\OO(k)} \cdot {s'}^{\OO(k)} \cdot \adhsize(\Tc) = 2^{\OO(k^2)}$.

Because there exists at least one $C \in \compset'$ with $|C| \ge 2$, this increases the number of nodes of $T$, and also increases the number of nodes in $P$ by the same number.

\begin{claim}
This transformation maintains the invariants.
\end{claim}
\begin{claimproof}
We denote by $\Tc = (T,\lmap)$ the superbranch decomposition before the transformation and by $\Tc' = (T',\lmap') = \Tc \rescliqs (p,\compset')$ the superbranch decomposition after.
We use similarly $P$ and $P'$.
We denote by $p_C$ for $C \in \compset$ the newly created nodes corresponding to $\torso(p) \rescliqs \co{C}$ for $C \in \compset$ with $|C| \ge 2$, and by $p'$ the newly created node corresponding to $\torso(p) \rescliqs \compset'$.

To prove \Cref{the:highlevelupwlalgo:inv2}, let $u \in \vint(T') \setminus P'$ and $g = \parent_{T'}(u)$.
Then, $u \in \vint(T) \setminus P$ and $g' = \parent_{T}(u) \in P$, and we have that $\lmap'(\vec{ug}) = \lmap(\vec{ug'})$, and therefore it is well-linked by \Cref{the:highlevelupwlalgo:inv2}.
For \Cref{the:highlevelupwlalgo:inv6} we observe that torsos of all nodes of $\Tc'$ in $\vint(T') \setminus P'$ are the same as the torsos of the corresponding nodes in $\Tc$.
For \Cref{the:highlevelupwlalgo:inv4}, we note that all new internal separations of $\Tc'$ have the form $(C \orescliqs \Tc, \co{C} \orescliqs \Tc)$ for $C \in \compset$, and therefore have order $\le \bdc(\{e_t\}) \le 2k$.

To prove \Cref{the:highlevelupwlalgo:inv1}, we note that every separation of form $(\lmap'(\vec{uv}), \lmap'(\vec{vu}))$ for $u,v \in P'$ either is a separation of form $(\lmap(\vec{u'v'}), \lmap(\vec{v'u'}))$ for $u',v' \in P$, or is a separation of form $(C_i \orescliqs \Tc, \co{C_i} \orescliqs \Tc)$.
The separations of the first form are mixed-$k$-well-linked by \Cref{the:highlevelupwlalgo:inv1}, and the separations of the second form are mixed-$k$-well-linked by their already stated properties.

For \Cref{the:highlevelupwlalgo:inv3}, assume first that $u$ is not a newly created node.
Then, the collection of sets $\{\lmap'(\vec{vu})\}_{v \in N_{T'}(u)}$ is the same as the collection of sets $\{\lmap(\vec{vu})\}_{v \in N_{T}(u)}$, and therefore the invariant continues to hold.
Then, assume $u = p_C$ for some $C \in \compset$ with $|C| \ge 2$.
We have that $\lmap'(\vec{p'u}) = \co{C} \orescliqs \Tc$ is well-linked, and for $v \in N_{T'}(u) \setminus \{p'\}$, we have that $\lmap'(\vec{vu}) = \lmap(\vec{vp})$, so at most one of them is not well-linked.
Finally, assume $u = p'$.
We have that $\{\lmap'(\vec{vp'}) : v \in N_{T'}(p')\} = \{\lmap(\vec{tp})\} \cup \{C \orescliqs \Tc : C \in \compset\}$.
By \Cref{the:highlevelupwlalgo:inv2}, $\lmap(\vec{tp})$ is well-linked, and because each $C \orescliqs \Tc$ is $k$-well-linked and there is at most one $C \in \compset$ so that $\bdc(C) \ge k$, all but at most one of the sets $C \orescliqs \Tc$ is well-linked.

For \Cref{the:highlevelupwlalgo:inv5} we observe that $V(p_C) \subseteq V(p)$ for all $C \in \compset$ and $V(p') \subseteq V(p)$, and the torsos of all other nodes in $P'$ are the same in $\Tc'$ as in $\Tc$.
\end{claimproof}

\paragraph{Case 2.3: Multiplicity-$3$-violator.}
Assume then that \Cref{lem:algpartitiontokwlsets} returned a multiplicity-$3$-violator $X \subseteq E(\torso(p))$.
Select $C \subseteq X$ so that $|C| = 2$ and $e_t \notin C$.
By \Cref{lem:multipvioldwl} we have that $C \orescliqs \Tc$ and $\co{C} \orescliqs \Tc$ are well-linked in $G$.
Now, we use \Cref{lem:sbdchippingalgo} to transform $\Tc$ into $\Tc \rescliqs (p,C)$ in time $\OO(k)$, and let both of the resulting nodes, $p'$ and $p_C$, be in $P$.
This increases both $|V(T)|$ and $|P|$ by one.

\begin{claim}
This transformation maintains the invariants.
\end{claim}
\begin{claimproof}
Let us denote by $\Tc = (T,\lmap)$ the superbranch decomposition before the transformation and by $\Tc' = (T',\lmap') = \Tc \rescliqs (p,C)$ the superbranch decomposition after.
We use similarly $P$ and $P'$.
Denote by $p_C$ the newly created node corresponding to $\torso(p) \rescliqs \co{C}$ and by $p'$ the newly created node corresponding to $\torso(p) \rescliqs C$.

For \Cref{the:highlevelupwlalgo:inv2}, let $u \in \vint(T') \setminus P'$ and $g = \parent_{T'}(u)$.
We have that then, $u \in \vint(T) \setminus P$ and $\lmap'(\vec{ug}) = \lmap(\vec{ug'})$, where $g' = \parent_T(u)$, and therefore $\lmap'(\vec{ug})$ is well-linked by \Cref{the:highlevelupwlalgo:inv2}.

For \Cref{the:highlevelupwlalgo:inv5,the:highlevelupwlalgo:inv6} we have that all vertex sets of torsos of $\Tc'$ in $P'$ are subsets of vertex sets of torsos of $\Tc$ in $P$, and all torsos of $\Tc'$ in $\vint(T') \setminus P'$ are also torsos of $\Tc$ in $\vint(T) \setminus P$.

For \Cref{the:highlevelupwlalgo:inv4} we note that all internal separation of $\Tc'$ except $(C \orescliqs \Tc, \co{C} \orescliqs \Tc)$ are internal separations of $\Tc$, and $\bdc(C) \le \rank(\torso(p)) \le \adhsize(\Tc) \le 2k$.

To prove \Cref{the:highlevelupwlalgo:inv1}, we note that if $uv \in E(T')$ and $u,v \in P'$, then either $(\lmap'(\vec{uv}), \lmap'(\vec{vu}))$ is equal to $(\lmap(\vec{u'v'}),\lmap(\vec{v'u'}))$ for some $u',v' \in P$, or to $(C \orescliqs \Tc, \co{C} \orescliqs \Tc)$.
The former types are mixed-$k$-well-linked by \Cref{the:highlevelupwlalgo:inv1}, and the latter type was argued above to be doubly-well-linked.

For \Cref{the:highlevelupwlalgo:inv3}, first assume that $u \in P' \setminus \{p',p_C\}$.
Then, $\{\lmap'(\vec{vu}) : v \in N_{T'}(u)\} = \{\lmap(\vec{vu}) : v \in N_{T}(u)\}$, and \Cref{the:highlevelupwlalgo:inv3} holds by \Cref{the:highlevelupwlalgo:inv3}.
Then consider $u = p_C$.
We have that $\lmap'(\vec{p'p_C}) = \co{C} \orescliqs \Tc$ is well-linked, and for $v \in N_{T'}(p_C) \setminus \{p'\}$, we have that $\lmap'(\vec{v p_C}) = \lmap(\vec{v p})$, so at most one of them is not well-linked by \Cref{the:highlevelupwlalgo:inv3}.
Finally, assume $u = p'$.
We have that $\lmap'(p_C p') = C \orescliqs \Tc$ is well-linked, and for $v \in N_{T'}(p') \setminus \{p_C\}$, we have that $\lmap'(\vec{v p'}) = \lmap(\vec{v p})$, so at most one of them is not well-linked by \Cref{the:highlevelupwlalgo:inv3}.
\end{claimproof}

\paragraph{Overall analysis.}
Each of the cases runs in $2^{\OO(k^2)}$ time.
In all of the cases, the measure $2|P|-|V(T)|$ increases: In Case 1, $|V(T)|$ decreases while $|P|$ stays the same, in Case 2.1, $|P|$ increases while $|V(T)|$ stays the same, and in Cases 2.2 and 2.3, $|P|$ and $|V(T)|$ both increase by the same amount and by at least one.
Therefore, as both $|V(T)|$ and $|P|$ are non-negative and bounded by $\OO(\|G\|)$, this process runs for at most $\OO(\|G\|)$ iterations.

We also need to argue that in each iteration, we can efficiently find a node in $\vint(T) \setminus P$ that has a parent in $P$.
This is not difficult, as after a node in $\vint(T) \setminus P$ first has a parent in $P$, it will always have a parent in $P$.
However, this requires in Cases 1 and 2.1 to iterate over all the children of $t$, which increases the running time of these cases by $\OO(|E(\torso(t))|)$.
However, this can be charged from a potential of form $\sum_{v \in \vint(T) \setminus P} |E(\torso(v))|$, which is initially at most $\OO(\|G\|)$, does not increase in any of the cases, and in Cases 1 and 2.1 decreases by $|E(\torso(t))|$.
\end{proof}

\section{Tangle-unbreakable torsos}
\label{sec:tangubrtorsos}
In this section we prove \Cref{the:highlevel:makektangleunbreakable}, which we re-state now.

\thehighlevelmakektangleunbreakable*

Our first goal is to give an algorithm that, given an $(s,k)$-unbreakable hypergraph $G$, computes a superbranch decomposition of $G$ whose all internal separations are doubly well-linked, and which has $k$-tangle-unbreakable torsos.
Then, this algorithm can simply be applied to decompose each torso to yield the algorithm of \Cref{the:highlevel:makektangleunbreakable}.

We then define sets that we will ``chip'' from $G$ to decompose it.
Let $s \ge k \ge 1$ be integers.
We say that a \emph{$(s,k)$-connchip} of a hypergraph $G$ is a set $C \subseteq E(G)$ so that
\begin{itemize}
\item $\bdc(C) < k$,
\item $|C| \ge 2$ and $|\co{C}| \ge 2$,
\item $|V(C)| \le s$,
\item $C$ is internally connected, and
\item both $C$ and $\co{C}$ are well-linked.
\end{itemize}

We show that unbreakable hypergraphs that are not tangle-unbreakable always have connchips.

\begin{lemma}
\label{lem:nokchipsimpliesktangunbrk}
If a hypergraph $G$ is $(s,k)$-unbreakable and does not have any $(s,k)$-connchips, then $G$ is $k$-tangle-unbreakable.
\end{lemma}
\begin{proof}
Suppose that $G$ contains a separation of order $<k$ that distinguishes two tangles.
Then, by \Cref{lem:tangdistinguisherprops}, $G$ contains a separation $(C,\co{C})$ of order $<k$ so that $(C,\co{C})$ distinguishes two tangles, is doubly well-linked, has $|V(C)| \le |V(\co{C})|$, and $C$ is internally connected.
The fact that $(C,\co{C})$ distinguishes two tangles implies that $|C|,|\co{C}| \ge 2$.
The facts that $G$ is $(s,k)$-unbreakable and $|V(C)| \le |V(\co{C})|$ imply that $|V(C)| < s$.
It follows that then, $C$ would be an $(s,k)$-connchip of $G$.
\end{proof}

In our algorithm, we need the following lemma about verifying whether a given set is a connchip.

\begin{lemma}
\label{lem:algcheckifconchip1}
Let $G$ be a hypergraph whose representation is already stored.
Let also $s \ge k \ge 1$ be integers so that $G$ is $(s,k)$-unbreakable.
There is an algorithm that, given a set $C \subseteq E(G)$, in time $|C| \cdot s^{\OO(k+\rank(G))}$ either
\begin{itemize}
\item returns a multiplicity-$3$-violator in $G$, or
\item returns whether $C$ is a $(s,k)$-connchip of $G$.
\end{itemize}
\end{lemma}
\begin{proof}
We can test if $\bdc(C) < k$ in time $\OO(|C| \cdot \rank(G))$ by \Cref{lem:hypergraph_impl2}, and whether $|C| \ge 2$ and $|\co{C}| \ge 2$ simply by iterating first over $C$ and then over $E(G)$.
We can test whether $C$ is internally connected in $\OO(|C| \cdot \rank(G))$ time by \Cref{lem:alg_internal_comps}, and whether $|V(C)| \le s$ in $\OO(|C| \cdot \rank(G))$ time by \Cref{lem:hypergraph_impl2}.
Then, we use \Cref{lem:kwlubalg} to in time $|C| \cdot s^{\OO(k+\rank(G))}$ either (1) return a multiplicity-$3$-violator, or (2) test whether $\co{C}$ is well-linked.
Finally, we use \Cref{lem:algffwl} to in time $2^{\OO(k)} \cdot \rank(G)^2 \cdot |C|$ test whether $C$ is well-linked.
\end{proof}

We then show that decomposing by doubly well-linked separations does not, in some sense, create new connchips.

\begin{lemma}
\label{lem:conchipdwlinduction12}
Let $G$ be a hypergraph, $s \ge k \ge 1$ integers, and $(A,\co{A})$ a doubly well-linked separation of $G$.
If $C \subseteq E(G \rescliqs A)$ is an $(s,k)$-connchip of $G \rescliqs A$ and $e_A \notin C$, then $C$ is an $(s,k)$-connchip of $G$.
\end{lemma}
\begin{proof}
The set $C$ trivially satifies all properties to be a $(s,k)$-connchip of $G$ except that $\co{C} \orescliqs A$ is well-linked.
However, being well-linked is a transitive property (\Cref{lem:linkedcliq}), which implies that $\co{C} \orescliqs A$ is well-linked.
\end{proof}

We then give our algorithm for decomposing $(s,k)$-unbreakable hypergraphs into $k$-tangle-unbreakable hypergraphs.

\begin{lemma}
\label{lem:algfromunbrtotangunbr}
There is an algorithm that, given integers $s \ge k \ge 1$ and a hypergraph $G$ that is $(s,k)$-unbreakable, in time $s^{\OO(k+\rank(G))} \cdot \|G\|$ returns a superbranch decomposition $\Tc$ of $G$ so that
\begin{itemize}
\item the internal separations of $\Tc$ have order $\le \max(k-1,\rank(G))$ and are doubly well-linked, and
\item the torsos of $\Tc$ are $k$-tangle-unbreakable.
\end{itemize}
\end{lemma}
\begin{proof}
Let us define that a \emph{chip} of a hypergraph $G$ is a set $C$ so that $(C,\co{C})$ is doubly well-linked, $|C| \ge 2$, $|\co{C}| \ge 2$, and $\bdc(C) \le \max(k-1,\rank(G))$.

We initially let $\Tc = (T,\lmap)$ be the unique superbranch decomposition of $G$ with one internal node.
We will do a process where we repeatedly discover a chip $C$ of $\torso(t)$, where $t \in \vint(T)$, and replace $\Tc$ by $\Tc \rescliqs (t,C)$.

\begin{claim}
\label{lem:algfromunbrtotangunbr:claim1}
This process maintains the invariants that
\begin{itemize}
\item the internal separations of $\Tc$ have order $\le \max(k-1,\rank(G))$ and are doubly well-linked, and
\item each torso of $\Tc$ is $(s,k)$-unbreakable and has rank at most $\max(k-1,\rank(G))$.
\end{itemize}
\end{claim}
\begin{claimproof}
The order of the internal separations is clear as by definition chips $C$ have $\bdc(C) \le \max(k-1,\rank(G))$.
The doubly well-linkedness of the internal separations follows from the transitivity of well-linkedness (\Cref{lem:linkedcliq}).
The $(s,k)$-unbreakability of torsos follows from the fact that if a hypergraph $G$ is $(s,k)$-unbreakable, then for every $C \subseteq E(G)$ the hypergraph $G \rescliqs C$ is $(s,k)$-unbreakable.
The rank of the torsos follows again by the fact that a chip $C$ has $\bdc(C) \le \max(k-1,\rank(G))$.
\end{claimproof}

We will continue this process until no torso of $\Tc$ has $(s,k)$-connchips, at which point each torso of $\Tc$ is $k$-tangle-unbreakable by \Cref{lem:nokchipsimpliesktangunbrk}, and we can output $\Tc$.

This process is implemented with a help of a queue $Q$.
The queue $Q$ stores pairs $(t,e)$, where $t \in \vint(T)$ and $e \in E(\torso(t))$.
We allow $Q$ to also contain such pairs so that $t$ is no longer in $\vint(T)$ or $e$ is no longer in $E(\torso(t))$ because $\Tc$ was edited after inserting them, but such pairs will be discarded when encountered.
The queue $Q$ will satisfy the invariant that if there exists $t \in \vint(T)$ and a $(s,k)$-connchip $C$ of $\torso(t)$, then $Q$ contains a pair $(t,e)$, where $e \in C$.
Initially, when $\vint(T) = \{t\}$, we let $Q$ contain all pairs $(t,e)$ with $e \in \torso(t)$, so it trivially satisfies the invariant.

\begin{claim}
\label{lem:algfromunbrtotangunbr:claim2}
Assume $Q$ is non-empty.
We can in $s^{\OO(k+\rank(G))}$ time either (1) obtain a pair $(t,C)$ so that $C$ is a chip of $\torso(t)$ and $|C| \le s^{\OO(k+\rank(G))}$, or (2) remove an element from $Q$ while maintaining its invariant.
\end{claim}
\begin{claimproof}
Let $(t,e) \in Q$ be the front element from $Q$.
If $t$ is no longer an internal node of $\Tc$ or $e$ is no longer a hyperedge of $\torso(t)$, we can remove $(t,e)$ from $Q$.
Otherwise, we use \Cref{lem:localchipalgo} to either (1) enumerate all sets $C \subseteq \torso(t)$ so that $e \in C$, $|V(C)| \le s$, $\bdc(C) < k$, and $C$ is internally connected, or (2) return a multiplicity-$3$-violator in $\torso(t)$.
This algorithm runs in time $s^{\OO(k+\rank(\torso(t)))} = s^{\OO(k+\rank(G))}$.

First, if it returns a multiplicity-$3$-violator $X \subseteq E(\torso(t))$, then any subset $X' \subseteq X$ with $|X'| = 2$ is a chip, and we can return $(t,X')$.

Otherwise, we use \Cref{lem:algcheckifconchip1} to check for each enumerated set $C$ if it is a $(s,k)$-connchip, or obtain a multiplicity-$3$-violator.
This runs in time $s^{\OO(k+\rank(G))}$ in total.
If we obtain a multiplicity-$3$-violator, we proceed as we did before in that case.
If we find an $(s,k)$-connchip $C$, we return $(t,C)$, as any $k$-connchip is also a chip.
If no enumerated set $C$ is a $k$-connchip, then $\torso(t)$ has no $(s,k)$-connchips $C$ with $e \in C$, and we can remove $(t,e)$ from $Q$.
\end{claimproof}

Suppose that applying \Cref{lem:algfromunbrtotangunbr:claim2} returned a chip $C \subseteq E(G)$.
We use \Cref{lem:sbdchippingalgo} to transform $\Tc$ into $\Tc \rescliqs (t,C)$ in time $\OO(|C| \cdot k) = s^{\OO(k+\rank(G))}$.
Now, $|E(\torso(t_C))| = |C|+1$, so we can within the same time add all pairs of form $(t_C, e)$, where $e \in E(\torso(t_C))$, to the queue $Q$.
We also add the pair $(t, e_C)$ to $Q$.
This maintains the invariant of $Q$ by \Cref{lem:conchipdwlinduction12}.

As each chip $C$ has $|C| \ge 2$ and $|\co{C}| \ge 2$, each iteration increases $|V(T)|$, so this process runs for at most $\|G\|$ iterations.
Therefore, the total running time is $\|G\| \cdot s^{\OO(k+\rank(G))}$.
\end{proof}

Now we can prove \Cref{the:highlevel:makektangleunbreakable}.

\thehighlevelmakektangleunbreakable*
\begin{proof}
For each $t \in \vint(T)$, we use \Cref{lem:algfromunbrtotangunbr} to compute a superbranch decomposition $\Tc_t = (T_t,\lmap_t)$ of $\torso(t)$ so that the internal separations of $\Tc_t$ have order $\le 2k$ and are doubly well-linked, and the torsos of $\Tc_t$ are $k$-tangle-unbreakable.
This runs in total $s^{\OO(k)} \cdot \sum_{t \in \vint(T)} \|\torso(t)\| = s^{\OO(k)} \cdot \|G\|$ time.

Then, we use \Cref{lem:superbdalgrefinebyrefiset} to compute the superbranch decomposition $\Tc' = (T',\bag') = \Tc \rescliqs \{\Tc_t\}_{t \in \vint(T)}$ in time $\OO(k \cdot \|G\|)$.
Because $\adhsize(\Tc) \le 2k$ and $\adhsize(\Tc_t) \le 2k$, we have that $\adhsize(\Tc') \le 2k$.
Because the torsos of all $\Tc_t$ are $k$-tangle-unbreakable, also the torsos of $\Tc'$ are $k$-tangle-unbreakable.
It remains to argue that the internal separations of $\Tc'$ are mixed-$k$-well-linked.

Let $T_t'$ be a connected subtree of $T'$ corresponding to a node $t \in \vint(T)$.
Let $v \in N_T(t)$ be the node in $N_T(t)$ so that $\lmap(\vec{vt})$ is not well-linked (if no such node exist, then let $v$ be an arbitrary neighbor of $t$).
Then let $r \in V(T_t')$ be the node so that there is $v' \in N_{T'}(r)$ so that $\lmap'(\vec{v'r}) = \lmap(\vec{vt})$.
Consider $T_t'$ as rooted on $r$.
Now, if $x \in V(T_t')$ and $y$ is the parent of $x$ in $T_t'$, then $\lmap'(\vec{xy})$ is well-linked by the transitivity of well-linkedness (\Cref{lem:linkedcliq}).
Also, $\lmap'(\vec{yx})$ is $k$-well-linked by the transitivity of $k$-well-linkedness along $k$-better-linked sets (\Cref{lem:kwltransitive}).
Therefore, all internal separations of $\Tc'$ corresponding to edges of such connected subtrees $T_t'$ are mixed-$k$-well-linked.
All other internal separations of $\Tc'$ are internal separations of $\Tc$, and therefore are mixed-$k$-well-linked.
\end{proof}

\section{Small adhesions}
\label{sec:smalleradhesions}
In this section we prove \Cref{the:highlevel:decreaseadhesion}, which we re-state now.

\thehighleveldecreaseadhesion*

We have in fact already proven the main ingredient of \Cref{the:highlevel:decreaseadhesion}, which is that separations of order $<k$ distinguishing tangles can be uncrossed with mixed-$k$-well-linked separations, which was proven as \Cref{lem:tangledisuncross}.
To prove \Cref{the:highlevel:decreaseadhesion}, let us start by applying \Cref{lem:tangledisuncross} to show that hypergraphs that admit certain superbranch decompositions are $k$-tangle-unbreakable.

\begin{lemma}
\label{lem:decomptangleunbrimpliestangunbr}
Let $G$ be a hypergraph, $k \ge 1$ an integer, and $\Tc = (T,\lmap)$ a superbranch decomposition of $G$, so that every internal separation of $\Tc$ has order $\ge k$ and is mixed-$k$-well-linked, and every torso of $\Tc$ is $k$-tangle-unbreakable.
Then, $G$ is $k$-tangle-unbreakable.
\end{lemma}
\begin{proof}
Assume for a contradiction that such a hypergraph $G$ has a separation $(A,\co{A})$ of order $<k$ that distinguishes two tangles.
Moreover, select such $(A,\co{A})$ so that it crosses the minimum number of edges of $\Tc$, where we say that $(A,\co{A})$ crosses an edge $uv \in E(T)$ if the separations $(A,\co{A})$ and $(\lmap(\vec{uv}), \lmap(\vec{vu}))$ cross.

\begin{claim}
$(A,\co{A})$ crosses no edges of $\Tc$.
\end{claim}
\begin{claimproof}
Suppose $(A,\co{A})$ crosses at least one edge $uv \in E(T)$.
Now, by \Cref{lem:tangledisuncross}, there exists an orientation $(A',\co{A'})$ of $(A,\co{A})$ and an orientation $\vec{uv}$ of $uv$ so that the separation $(A' \cup \lmap(\vec{uv}), \co{A'} \cap \lmap(\vec{vu}))$ distinguishes two tangles.
Let us furthermore assume (by permuting $A$ with $\co{A}$ if necessary) that $A' = A$, and in particular, that $(B,\co{B}) = (A \cup \lmap(\vec{uv}), \co{A} \cap \lmap(\vec{vu}))$ distinguishes two tangles.

We claim that $(B,\co{B})$ crosses less edges of $\Tc$ than $(A,\co{A})$, contradicting the choice of $(A,\co{A})$.
First, it does not cross the edge $uv$.
Second, it does not cross any edge $xy \neq uv$ that is closer to $u$ than $v$ in $T$, as for such edge there exists an orientation $\vec{xy}$ so that $\lmap(\vec{xy}) \subseteq \lmap(\vec{uv}) \subseteq B$.

Third, we claim that $(B,\co{B})$ crosses an edge $xy \neq uv$ that is closer to $v$ than $u$ only if $(A,\co{A})$ crosses the edge.
Let $\vec{xy}$ be an orientation of such an edge so that $\lmap(\vec{xy}) \subseteq \lmap(\vec{vu})$.
We have that $\co{B} \cap \lmap(\vec{xy}) = \co{A} \cap \lmap(\vec{xy})$ and $B \cap \lmap(\vec{xy}) = A \cap \lmap(\vec{xy})$.
We also have that $A$ and $\co{A}$ both intersect $\lmap(\vec{yx})$, because $\lmap(\vec{uv}) \subseteq \lmap(\vec{yx})$, and $(A,\co{A})$ crosses $uv$.
\end{claimproof}

Now we can assign an orientation $\vec{uv}$ for every edge $uv$ of $T$ so that $\lmap(\vec{uv})$ intersects only one of $A$ and $\co{A}$, and $\lmap(\vec{vu})$ intersects both $A$ and $\co{A}$.
Let $t \in V(T)$ be a node so that all of the incident edges are oriented towards $t$.
The fourth tangle axiom implies that $|A|,|\co{A}| \ge 2$, implying that $t$ is not a leaf of $T$.

Without loss of generality, let the tangles distinguished be $\tang_1$ and $\tang_2$, so that $A \in \tang_1$, $\co{A} \in \tang_2$, and both $\tang_1$ and $\tang_2$ have order $\bdc(A)+1 \le k$.
From each $\tang_i$ ($i \in \{1,2\}$), we construct a tangle $\tang'_i$ of $\torso(t)$ as follows.
For a set $B \subseteq E(\torso(t))$, we let $B \in \tang'_i$ if $B \orescliqs \Tc \in \tang_i$.

\begin{claim}
Both $\tang'_1$ and $\tang'_2$ are tangles of $\torso(t)$ of order $\bdc(A)+1$.
\end{claim}
\begin{claimproof}
First, $\bdc(B) = \bdc(B \orescliqs \Tc)$ and $\co{B \orescliqs \Tc} = \co{B} \orescliqs \Tc$, implying the first two tangle axioms.
For the third axiom, assume that there are $B_1,B_2,B_3 \in \tang'_i$ so that $B_1 \cup B_2 \cup B_3 = E(\torso(t))$.
Then, it would be that $(B_1  \orescliqs \Tc) \cup (B_2  \orescliqs \Tc) \cup (B_3  \orescliqs \Tc) = E(G)$, contradicting that $\tang_i$ is a tangle of $G$.
For the fourth axiom, assume there is $e_s \in E(\torso(t))$ so that $E(\torso(t)) \setminus \{e_s\} \in \tang'_i$.
Because internal adhesions of $\Tc$ have order $\ge k$ and $\tang_i$ has order $k$, it must be that $s$ is a leaf of $T$.
However, then $(E(\torso(t)) \setminus \{e_s\}) \orescliqs \Tc = E(G) \setminus \{e\}$ for some $e \in E(G)$, contradicting that $\tang_i$ is a tangle of $G$.
\end{claimproof}

Now, because for each $s \in N(t)$ it holds that $\lmap(\vec{st}) \subseteq A$ or $\lmap(\vec{st}) \subseteq \co{A}$, there exists $B \subseteq E(\torso(t))$ so that $B \orescliqs \Tc = A$ and $\co{B} \orescliqs \Tc = \co{A}$.
Now, $(B,\co{B})$ distinguishes the tangles $\tang'_1$ and $\tang'_2$ of $\torso(t)$, contradicting that $\torso(t)$ is $k$-tangle-unbreakable.
\end{proof}

Now we are ready to give the algorithm of \Cref{the:highlevel:decreaseadhesion}, which simply contracts all edges of $\Tc$ whose adhesions have size $\ge k$.

\thehighleveldecreaseadhesion*
\begin{proof}
We transform $\Tc = (T,\lmap)$ into $\Tc' = (T', \lmap')$ by contracting all edges of $T$ corresponding to internal separations of order $\ge k$.
This can be implemented in time $\OO(\|\Tc\|) = \OO(k \cdot \|G\|)$.

By construction, the internal separations of $\Tc'$ have order $<k$, and because they are also internal separations of $\Tc$ they are mixed-$k$-well-linked and thus doubly well-linked.
It remains to argue that the torsos of $\Tc'$ are $k$-tangle-unbreakable.

Each node of $\Tc'$ corresponds to a connected subtree of $\Tc$.
We observe that by taking the corresponding subtree from $\Tc$, each torso $\torso_{\Tc'}(t)$ of an internal node $t$ of $\Tc'$ has a superbranch decomposition $\Tc_t$ whose internal separations have order $\ge k$ and are mixed-$k$-well-linked, and whose torsos are $k$-tangle-unbreakable.
Therefore, \Cref{lem:decomptangleunbrimpliestangunbr} implies that $\torso_{\Tc'}(t)$ is $k$-tangle-unbreakable.
\end{proof}

\section{From tangle-unbreakability to unbreakability}
\label{sec:fromtubrtoubr}
In this section we prove \Cref{the:highlevel:tunbrtounbr}, which we re-state now.

\thehighleveltunbrtounbr*



The proof of \Cref{the:highlevel:tunbrtounbr} will be divided to many lemmas across this section.
For it, the following definitions are central.

Let $G$ be a hypergraph and $h \ge s \ge k \ge 1$ integers.
An \emph{$(s,h,k)$-smallchip} of $G$ is a set $C \subseteq E(G)$ so that
\begin{itemize}
\item $\bdc(C) < k$,
\item $s \le |V(C)| \le h$, 
\item $C$ is semi-internally connected, and
\item for all $e \in C$, $\bdc(e) < k-1$.
\end{itemize}

An \emph{$(h,k)$-bigchip} of $G$ is a set $C \subseteq E(G)$ so that
\begin{itemize}
\item $\bdc(C) < k$,
\item $\bdc(C) < |V(C)| \le h$,
\item $C$ is internally connected, and
\item there exists $e \in C$ so that $\bdc(e) \ge k-1$.
\end{itemize}

We say that an $(s,h,k)$-chip is a set $C \subseteq E(G)$ that is either an $(s,h,k)$-smallchip or an $(h,k)$-bigchip.

We will prove that if $G$ is $k$-tangle-unbreakable and does not have $(s,h,k)$-chips for certain parameters of $s$, $h$, and $k$, then $G$ is $(k^{\OO(k)}, k)$-unbreakable.
For this, we first need the following lemma.

\begin{lemma}
\label{lem:bigenoughunbreak}
Let $G$ be a hypergraph that is $(s,k-1)$-unbreakable, and let $C \subseteq E(G)$ so that $\bdc(C) < k$, $|V(C)| \ge 3s$ and $|V(\co{C})| \ge s$.
Then, $C$ is tri-well-linked.
\end{lemma}
\begin{proof}
Suppose $(B_1,B_2,B_3)$ is a tripartition of $C$ that witnesses otherwise.
Then, there exists $i \in [3]$ so that $\bdc(B_i) < k-1$, $|V(B_i)| \ge s$ and $|V(\co{B_i})| \ge |V(\co{C})| \ge s$, contradicting that $G$ is $(s,k-1)$-unbreakable.
\end{proof}

Then we prove that in $k$-tangle-unbreakable graphs, there always exists ``small'' witnesses of non-unbreakability.

\begin{lemma}
\label{lem:ktubnsmallwitness}
Let $s_1$, $s_2$, and $k$ be integers so that $s_1 \ge 3 s_2$ and $s_1 \ge k$.
If $G$ is a normal hypergraph that is $k$-tangle-unbreakable and $(s_2, k-1)$-unbreakable, but not $(s_1, k)$-unbreakable, then there exists a set $A \subseteq E(G)$ so that $\bdc(A) < k$, $s_1 \le |V(A)| \le 2 s_1$, and $|V(\co{A})| \ge s_1$.
\end{lemma}
\begin{proof}
Let $(A,\co{A})$ be a separation of $G$ of order $<k$ so that $|V(A)|,|V(\co{A})| \ge s_1$.
If $|V(A)| \le 2 s_1$ or $|V(\co{A})| \le 2 s_1$ we are done, so assume that $|V(A)|,|V(\co{A})| > 2 s_1$.
By \Cref{lem:bigenoughunbreak}, $(A,\co{A})$ is doubly tri-well-linked.
Now, by \Cref{lem:tangleunbreakablebwsmall}, either $\bw(G \rescliqs A) \le \bdc(A)$ or $\bw(G \rescliqs \co{A}) \le \bdc(A)$.
Assume \wilog that $\bw(G \rescliqs A) \le \bdc(A)$.

Let $\Tc = (T,\lmap)$ be a branch decomposition of $G \rescliqs A$ of width at most $\bdc(A) < k$.
Consider $T$ as rooted at the leaf $\lmap^{-1}(e_A)$, where $e_A$ is the hyperedge of $G \rescliqs A$ corresponding to $A$.
Now, let $x \in V(T)$ be a node so that $|V(\lmap(xp))| \ge 2s_1$, where $p = \parent(x)$, selected so that it maximizes the distance to the root.
Note that such a node exists because $|V(\co{A})| > 2 s_1$.
It cannot be that $x$ is a leaf, because $2 s_1 \ge k$ and $G$ is normal.
Therefore, $x$ has exactly two children $c_1$ and $c_2$.
By the selection of $x$, we have that $|V(\lmap(c_i x))| < 2s_1$ for both $i \in [2]$, and because $V(\lmap(xp)) = V(\lmap(c_1 x)) \cup V(\lmap(c_2 x))$, we have that $|V(\lmap(c_i x))| \ge s_1$ for at least one $i \in [2]$.
Such an $i$ gives a set $A' = \lmap(c_i x)$ so that $s_1 \le |V(A')| \le 2 s_1$ and $|V(\co{A'})| \ge s_1$.
\end{proof}

Then we prove that $k$-tangle-unbreakability combined with the absence of chips implies unbreakability.
To avoid writing complicated expressions, let us define that for an integer $p \ge 0$, $s(p) = p^{2p-1}$, $u(p) = p \cdot s(p) = p^{2p}$, and $h(p) = 2 \cdot u(p) = 2 \cdot p^{2p}$ (here define $0^0 = 0$).
In the rest of this section, let us reserve the symbols $s$, $u$, and $h$ for these functions.
Let us immediately remark that $3 \cdot u(p-1) < s(p)$ for all $p \ge 2$.

\begin{lemma}
\label{lem:tubrimprunbrklem}
If a normal hypergraph $G$ is $k$-tangle-unbreakable and does not have $(s(p), h(p), p)$-chips for any $p \in [k]$, then $G$ is $(u(k), k)$-unbreakable.
\end{lemma}
\begin{proof}
Suppose that $G$ does not have $(s(p), h(p), p)$-chips for any $p \in [k]$, but is not $(u(k), k)$-unbreakable.
Let $p \le k$ be the smallest integer so that $G$ is not $(u(p), p)$-unbreakable.
Note that if $p \ge 2$, we have that $u(p) > 3 \cdot u(p-1)$, and if $p=1$, we have that $G$ is $(y,p-1)$-unbreakable for all $y$.
By \Cref{lem:ktubnsmallwitness}, there exists a set $A \subseteq E(G)$ so that $\bdc(A) < p$, $u(p) \le |V(A)| \le h(p)$, and $|V(\co{A})| \ge u(p)$.

Now, let $C_1,\ldots,C_{\ell}$ be the unique partition of $A$ into its semi-internally connected components.
Furthermore, for each $d \in [0,p-1]$, let $\compset_d$ be the collection of the components $C_i$ with $\bdc(C_i) = d$.
There exists $d \in [0,p-1]$ so that $\sum_{C \in \compset_d} |V(C)| \ge u(p)/p \ge s(p)$.
Because $\compset_d$ contains at most $\binom{p-1}{d} \le p^{p-1-d}$ sets, there exists $C \in \compset_d$ so that 
\[|V(C)| \ge s(p)/p^{p-1-d} = p^{2p-1}/p^{p-1-d} \ge p^{p+d} \ge (d+1)^{2(d+1)-1} \ge s(d+1).\]
If $|V(C)| > h(d+1)$, then because $|V(A)| \le h(p)$ we have that $d+1<p$, and because $h(d+1) > u(d+1)$, this would contradict the $(u(d+1), d+1)$-unbreakability of $G$.
Therefore, $|V(C)| \le h(d+1)$.

We then have two cases.
First, if $C$ has an internal component $C'$ with $|V(C')| > d$ containing $e \in C'$ with $\bdc(e) \ge d$, then $C'$ would be a $(h(d+1), d+1)$-bigchip.
Otherwise, all internal components $C'$ of $C$ containing hyperedges $e \in C'$ with $\bdc(e) \ge d$ are single-hyperedge components with $V(C') = \bd(C)$.
In this case, let $C''$ be the union of the internal components of $C$ not containing hyperedges $e$ with $\bdc(e) \ge d$.
Because $|V(C)| \ge h(d+1) > d = \bdc(C)$, $C''$ is non-empty, implying in fact that $V(C'') = V(C)$, implying that $C''$ is a $(s(d+1), h(d+1), d+1)$-smallchip.
This is a contradiction as $d<k$.
\end{proof}

Then the goal will be to decompose a given hypergraph by chips.
We first need the following lemma about bigchips, which also explains the motivation behind the definition of bigchips.

\begin{lemma}
\label{lem:connchiptriwl}
Let $G$ be a normal hypergraph and $k \ge 1$ an integer so that $G$ is $(u(k-1),k-1)$-unbreakable and $k$-tangle-unbreakable.
Let also $|V(G)| \ge 2 \cdot h(k) + 4 \cdot u(k-1)$.
If $C \subseteq E(G)$ is an $(h(k), k)$-bigchip of $G$, then 
\begin{itemize}
\item $(C,\co{C})$ is doubly tri-well-linked,
\item $|C|,|\co{C}| \ge 2$, and
\item $\bdc(C) = k-1$.
\end{itemize}
\end{lemma}
\begin{proof}
The facts that $G$ is normal and $|V(C)| > \bdc(C)$ imply $|C| \ge 2$.
Also, the facts that $|V(G)| > |V(C)|$ and $G$ is normal imply $|\co{C}| \ge 2$.

Let $e \in C$ so that $\bdc(e) \ge k-1$.

\begin{claim}
\label{lem:connchiptriwl:claim1}
There is no set $C'$ with $e \in C'$, $|V(C')| \le h(k)+2 \cdot u(k-1)$, and $\bdc(C') < k-1$.
\end{claim}
\begin{claimproof}
Suppose such a $C'$ exists.
Let $\tang_1$ be the $e$-tangle of $G$, and $\tang_2$ the vertex-cardinality-tangle of order $k-1$, which exists by \Cref{lem:vertexcardtangle} because $G$ is $(u(k-1),k-1)$-unbreakable.
Now, $\co{C'} \in \tang_1$ and $C' \in \tang_2$, contradicting that $G$ is $k$-tangle-unbreakable.
\end{claimproof}

If $\bdc(C)<k-1$, then $C' = C$ would contradict \Cref{lem:connchiptriwl:claim1}.

Now suppose that $C$ is not tri-well-linked, and let $(B_1,B_2,B_3)$ be a tripartition of $C$ with $\bdc(B_i) < \bdc(C)$ for all $i \in [3]$.
Now, if $e \in B_i$, then $B_i$ contradicts \Cref{lem:connchiptriwl:claim1}, but it must be that $e \in B_i$ for some $i \in [3]$.

Then suppose that $\co{C}$ is not tri-well-linked, and let $(B_1,B_2,B_3)$ be a tripartition of $\co{C}$ with $\bdc(B_i) < \bdc(C)$ for all $i \in [3]$.
Because $|V(C)| \le h(k)$, there is $i \in [3]$ with $|V(B_i)| \ge u(k-1)$.
Assume \wilog that $|V(B_1)| \ge u(k-1)$.
Now, $|V(B_2)| < u(k-1)$ and $|V(B_3)| < u(k-1)$ because $G$ is $(u(k-1),k-1)$-unbreakable.
Therefore, $|V(C \cup B_2 \cup B_3)| \le h(k)+2 \cdot u(k-1)$, and $\bdc(C \cup B_2 \cup B_3) = \bdc(B_1) < k-1$.
Therefore, $C \cup B_2 \cup B_3$ contradicts \Cref{lem:connchiptriwl:claim1}.
\end{proof}

Then we show the statement of \Cref{lem:connchiptriwl} applies in fact to all chips, under certain conditions.

\begin{lemma}
\label{lem:turbchipproperties}
Let $G$ be a normal hypergraph and $k \ge 1$ an integer so that $G$ is $k$-tangle-unbreakable.
If $|V(G)| \ge 6 \cdot h(k)$ and $G$ has no $(s(p), h(p), p)$-chips for any $p<k$, and $C \subseteq E(G)$ is an $(s(k), h(k), k)$-chip, then 
\begin{itemize}
\item $(C,\co{C})$ is doubly tri-well-linked,
\item $|C|,|\co{C}| \ge 2$,
\item $\bdc(C) = k-1$, and
\item $G \rescliqs C$ has no $(s(p), h(p), p)$-chips for any $p<k$.
\end{itemize}
\end{lemma}
\begin{proof}
By \Cref{lem:tubrimprunbrklem}, we have that $G$ is $(u(k-1), k-1)$-unbreakable.
Now, because $|V(G)| \ge 6 \cdot h(k) \ge 2 \cdot h(k) + 4 \cdot u(k-1)$, \Cref{lem:connchiptriwl} immediately implies the first three bullet points for all $(s(k), k)$-bigchips.

For $(s(k), h(k), k)$-smallchips, we have that if $k=1$, all the properties are trivial (note that $|C|,|\co{C}| \ge 2$ follows from the fact that $G$ is normal).
For $k \ge 2$, we have that $s(k) \ge 3 \cdot u(k-1)$, and therefore by \Cref{lem:bigenoughunbreak} both $C$ and $\co{C}$ are tri-well-linked.
We also have $|V(C)|,|V(\co{C})| > \bdc(C)$, implying that $|C|,|\co{C}| \ge 2$.
The $(u(k-1), k-1)$-unbreakability of $G$ also implies that $\bdc(C) = k-1$.

It remains to prove the last bullet point.
Suppose that $D$ is a $(s(p), h(p), p)$-chip of $G \rescliqs C$ for $p<k$.
If $e_C \notin D$, then $D$ would directly also be a $(s(p), h(p), p)$-chip of $G$.
If $e_C \in D$, then $D$ is not a $(s(p), h(p), p)$-smallchip because $\bdc(e_C) = k-1$.

It remains to consider the case that $D$ is a $(h(p), p)$-bigchip of $G \rescliqs C$, and we have that $\bdc(e_C) = k-1 > p-1 \ge \bdc(D)$.
Because $(C,\co{C})$ is doubly tri-well-linked, we have that $G \rescliqs C$ is $k$-tangle-unbreakable by \Cref{lem:tungunbhered}.
We also have that $G \rescliqs C$ is $(u(k-1), k-1)$-unbreakable because $G$ is $(u(k-1), k-1)$-unbreakable.
Let $\tang_1$ be the $e_C$-tangle of $G \rescliqs C$ and $\tang_2$ the vertex-cardinality-tangle of $G \rescliqs C$ of order $k-1$, which exists by \Cref{lem:vertexcardtangle} because $|V(G \rescliqs C)| \ge 5 \cdot h(k) \ge 3 \cdot u(k-1)$.
Now, $\co{D} \in \tang_1$ and $D \in \tang_2$, contradicting that $G \rescliqs C$ is $k$-tangle-unbreakable.
\end{proof}

Now we are ready to give the algorithm to remove smallchips.

\begin{lemma}
\label{lem:tubrtoubrremovesmallchipsalg}
There is an algorithm that, given a normal hypergraph $G$ and an integer $k \ge 1$, so that
\begin{itemize}
\item $G$ is $k$-tangle-unbreakable,
\item $G$ has multiplicity at most $3$, and
\item $G$ has no $(s(p), h(p), p)$-chips for any $p < k$,
\end{itemize}
in time $k^{\OO(k \cdot (k+\rank(G)))} \cdot \|G\|$ returns a superbranch decomposition $\Tc = (T,\bag)$ of $G$ so that
\begin{itemize}
\item the internal separations of $\Tc$ are doubly tri-well-linked and have order $<k$, and
\item each torso of $\Tc$ either
\begin{itemize}
\item has at most $6 \cdot h(k)$ vertices, or
\item has no $(s(k), h(k), k)$-smallchips, no $(s(p), h(p), p)$-chips for any $p < k$, and no multiplicity-$3$-violators of rank $\neq k-1$.
\end{itemize}
\end{itemize}
\end{lemma}
\begin{proof}
First, if $|V(G)| \le 6 \cdot h(p)$, then we return the trivial superbranch decomposition with only one internal node.
Therefore assume that $|V(G)| > 6 \cdot h(p)$.

Our first goal is to find a maximal collection $\compset$ of disjoint $(s(k),h(k),k)$-smallchips of $G$.
By maximal, we mean that every $(s(k),h(k),k)$-smallchip of $G$ intersects a set in $\compset$.

We say that a \emph{prechip} is a set $C \subseteq E(G)$ so that
\begin{itemize}
\item $|V(C)| \le h(k)$,
\item $\bdc(C) < k$,
\item $C$ is internally connected, and
\item for each $e \in C$, $\bdc(e) < k-1$.
\end{itemize}

We observe that if $C$ is an $(s(k),h(k),k)$-smallchip, then all internal components of $C$ are prechips, and that if $C_1, \ldots, C_d$ are prechips, then $C_1 \cup C_2 \cup \ldots \cup C_d$ is an $(s(k),h(k),k)$-smallchip if $\bd(C_1) = \bd(C_2) = \ldots = \bd(C_d)$ and $h(k) \ge |V(C_1 \cup \ldots \cup C_d)| \ge s(k)$.

We enumerate all prechips of $G$ with the help of \Cref{lem:localchipalgo} in time $\|G\| \cdot h(k)^{\OO(k+\rank(G))}$.
Specifically, for each $e \in E(G)$ with $\bdc(e) < k$ we apply \Cref{lem:localchipalgo} to enumerate all prechips containing $e$ in time $h(k)^{\OO(k+\rank(G))}$, and finally, we use radix sort to deduplicate the resulting collection.
We then use radix sort to in time $\|G\| \cdot h(k)^{\OO(k+\rank(G))}$ partition these prechips to collections based on $\bd(C)$.

We then construct $\compset$ as follows.
We iterate over these collections one by one, maintaining for each $e \in E(G)$ the information whether a smallchip containing $e$ has already been added to $\compset$.

Let $X \subseteq V(G)$ and let $\compset_X$ be the collection of all prechips $C$ of $G$ with $\bd(C) = X$.
Note that the sets $C \in \compset_X$ are pairwise disjoint because they are internally connected and have the same border.
We first remove from $\compset_X$ all sets $C$ so that $C$ intersects a smallchip already in $\compset$.
Second, for all remaining $C \in \compset_X$ with $|V(C)| \ge s(k)$, we add $C$ to $\compset$ and remove $C$ from $\compset_X$.
Third, as long as $|\bigcup_{C \in \compset_X} V(C)| \ge s(k)$, we can identify $\compset_X' \subseteq \compset_X$ with $s(k) \le |\bigcup_{C \in \compset_X'} V(C)| \le 2 \cdot s(k) \le h(k)$, add $\bigcup_{C \in \compset_X'} C$ to $\compset$, and remove them from $\compset_X$.

This process for one $X$ can be implemented in time $\OO(h(k) \cdot |\bigcup_{C \in \compset_X} C|)$ by one iteration over $\compset_X$, and therefore the total running time of constructing $\compset$ is $\|G\| \cdot h(k)^{\OO(k+\rank(G))}$.

\begin{claim}
$\compset$ is a maximal collection of disjoint $(s(k),h(k),k)$-smallchips of $G$.
\end{claim}
\begin{claimproof}
By the construction, all $C \in \compset$ are $(s(k),h(k),k)$-smallchips of $G$.

Assume that $\compset$ is not maximal, i.e., there exists a $(s(k),h(k),k)$-smallchip $C$ of $G$ so that $C$ does not intersect any set in $\compset$.
Let $X = \bd(C)$.
Because $C$ is semi-internally connected and disjoint from $\compset$, in the end of the iteration considering $\compset_X$, all internal components of $C$ were still in $\compset_X$.
However, this contradicts that in the of the iteration $|\bigcup_{C \in \compset_X} V(C)| < s(k)$.
\end{claimproof}

By \Cref{lem:turbchipproperties}, for all $C \in \compset$ we have that $(C,\co{C})$ is doubly tri-well-linked, $|C|,|\co{C}| \ge 2$, and $\bdc(C) = k-1$.
Now, let $\Tc' = (T',\lmap')$ be the trivial superbranch decomposition of $G$, with $r = V(T')$ being its only internal node, and let $\Tc = \Tc' \rescliqs (r, \compset)$.
We construct $\Tc$ with \Cref{lem:superbdsplittingalgo} in time $\OO(\|G\| \cdot \rank(G))$, and return $\Tc$.

All internal separations of $\Tc$ are of form $(C,\co{C})$ for $C \in \compset$, so they have order $<k$ and are doubly tri-well-linked.
The torsos of nodes of form $r_C$ have $|V(\torso(r_C))| \le h(k)$, and the torso of the node $r'$ is $\torso(r') = G \rescliqs \compset$.
If $|V(G \rescliqs \compset)| \le 6 \cdot h(k)$ we are done, so assume that $|V(G \rescliqs \compset)| > 6 \cdot h(k)$, in which case we have to prove that $\torso(r') = G \rescliqs \compset$ satisfies the required properties.

As $\bdc(C) = k-1$ for all $C \in \compset$ and $G$ has no multiplicity-$3$-violators, all multiplicity-$3$-violators of $G \rescliqs \compset$ must have rank $k-1$.

By repeated application of \Cref{lem:turbchipproperties}, using the properties that $|V(G \rescliqs \compset)| \ge 6 \cdot h(k)$ and that if $G$ is $k$-tangle-unbreakable and $(C,\co{C})$ is doubly tri-well-linked, then $G \rescliqs C$ is $k$-tangle-unbreakable (\Cref{lem:tungunbhered}), we get that $G \rescliqs \compset$ has no $(s(p), h(p), p)$-chips for any $p < k$.

It remains to prove that $G \rescliqs \compset$ has no $(s(k), h(k), k)$-smallchips.
Suppose otherwise, and let $C'$ be an $(s(k), h(k), k)$-smallchip of $G \rescliqs \compset$.
For all $C \in \compset$, we have that $e_C \in G \rescliqs \compset$ has $\bdc(e_C) = k-1$, so $C'$ cannot contain $e_C$.
Therefore, $C'$ would be a $(s(k),h(k),k)$-smallchip of $G$ disjoint from all smallchips in $\compset$, but this contradicts the maximality of $\compset$.
\end{proof}

We need to also argue that we can reduce the multiplicity of $G$.

\begin{lemma}
\label{lem:tubr_to_ubr_chip_mult_ok}
Let $G$ be a normal hypergraph, $k \ge 1$ an integer, and $C \subseteq E(G)$ a multiplicity-violator.
Then,
\begin{itemize}
\item if $G$ does not have $(s(k), h(k), k)$-smallchips, then $G \rescliqs C$ does not have $(s(k), h(k), k)$-smallchips, and
\item if $G$ does not have $(h(k), k)$-bigchips, then $G \rescliqs C$ does not have $(h(k), k)$-bigchips.
\end{itemize}
\end{lemma}
\begin{proof}
Let $D \subseteq E(G \rescliqs C)$ be either a $(s(k), h(k), k)$-smallchip or a $(h(k), k)$-bigchip of $G \rescliqs C$.
We have that $\bdc(D \orescliqs C) = \bdc(D)$, $|V(D \orescliqs C)| = |V(D)|$, and the maximum of $\bdc(e)$ for $e \in D$ and $e \in D \orescliqs C$ is the same.
Therefore, it remains to prove that if $D$ is a $(s(k), h(k), k)$-smallchip, then $D \orescliqs C$ is semi-internally connected, and that if $D$ is $(h(k), k)$-bigchip, then $D \orescliqs C$ is internally connected.
The former follows from the transitivity of semi-internally connectedness (proven as \Cref{lem:conddwsic}), because multiplicity-violators are semi-internally connected.

For the latter, assume that $e_C \in D$ as otherwise it is trivial.
Because $D$ is a $(h(k), k)$-bigchip and $G$ is normal, we have that $|D| \ge 2$.
Because $D$ is internally connected, there must be then $v \in V(e_C) \setminus \bd(D)$.
Therefore, in $D \orescliqs C$ there is also $v \in V(D \orescliqs C) \setminus \bd(D)$ that intersects all $V(e)$ for $e \in C$, and therefore $D \orescliqs C$ is internally connected.
\end{proof}

We then give the algorithm for removing multiplicity-violators.

\begin{lemma}
\label{lem:algdecompmultihyperedges}
There is an algorithm that, given an integer $k \ge 1$ and a hypergraph $G$ that has no multiplicity-violators of rank $\ge k$, in time $k^{\OO(1)} \cdot \|G\|$ outputs a superbranch decomposition $\Tc = (T,\lmap)$ of $G$ so that
\begin{itemize}
\item internal separations of $\Tc$ have order $<k$ and are doubly tri-well-linked,
\item each torso of $\Tc$ either has $<k$ vertices or has multiplicity $\le 3$,
\item for all $p \ge 1$, if $G$ does not have $(s(p), h(p), p)$-smallchips, then no torso of $\Tc$ with $\ge k$ vertices has $(s(p), h(p), p)$-smallchips, and
\item for all $p \ge 1$, if $G$ does not have $(h(p), p)$-bigchips, then no torso of $\Tc$ with $\ge k$ vertices has $(h(p), p)$-smallchips.
\end{itemize}
\end{lemma}
\begin{proof}
If $|E(G)| \le 3$ or $|V(G)| < k$ we return the trivial superbranch decomposition with one internal node.

We partition $E(G)$ into a collection of sets $\compset$ based on $V(e)$ for each hyperedge $e \in E(G)$.
This can be done in $\OO(k \cdot \|G\|)$ time with the help of radix sort and the promise that there are no multiplicity-violators of rank $\ge k$.

Because $|V(G)| \ge k$, we have that $|\compset| \ge 2$.
We construct $\compset'$ from $\compset$ by removing all sets $C \in \compset$ with $|C| \le 2$, and removing one hyperedge from each set $C \in \compset'$ with $|C| \ge 3$.

We observe that now, for all $C \in \compset'$, the separation $(C, \co{C})$ is doubly tri-well-linked, simply because both $C$ and $\co{C}$ contain a hyperedge $e$ with $V(e) = \bd(C)$.
Also, for all $C \in \compset'$ it holds that $|C| \ge 2$, $|\co{C}| \ge 2$, and $\co{C} \notin \compset'$.

Let $\Tc' = (T',\lmap')$ be the trivial superbranch decomposition of $G$ with one internal node $r$, and let us use \Cref{lem:superbdsplittingalgo} to construct $\Tc = (T,\lmap) = \Tc' \rescliqs \compset'$ in time $k^{\OO(1)} \|G\|$.
As discussed before, the internal separations of $\Tc$ are now doubly tri-well-linked and have order $<k$.
All other torsos than $\torso(r')$ have $<k$ vertices, and $\torso(r')$ has multiplicity at most $3$.
Now, repeated application of \Cref{lem:tubr_to_ubr_chip_mult_ok} implies that if $G$ does not have $(s(p), h(p), p)$-smallchips, then also $\torso(r')$ has no $(s(p), h(p), p)$-smallchips, and if $G$ does not have $(h(p), p)$-bigchips, then also $\torso(r')$ does not have $(h(p), p)$-bigchips.
\end{proof}

As the last step, we then give the algorithm to remove bigchips.

\begin{lemma}
\label{lem:tubrtoubralgremovebigchipsalg}
There is an algorithm, that given a normal hypergraph $G$ and an integer $k \ge 1$, so that 
\begin{itemize}
\item $G$ has no $(s(p), h(p), p)$-chips for any $p < k$,
\item $G$ has no $(s(k), h(k), k)$-smallchips, and
\item $G$ has no multiplicity-$3$-violators of rank $\neq k-1$,
\end{itemize}
in time $k^{\OO(k \cdot (k+\rank(G)))} \cdot \|G\|$ returns a superbranch decomposition $\Tc = (T,\bag)$ of $G$ so that
\begin{itemize}
\item the internal separations of $\Tc$ are doubly tri-well-linked and have order $<k$,
\item each torso of $\Tc$ either
\begin{itemize}
\item has at most $6 \cdot h(k)$ vertices, or
\item has no $(s(p),h(p),p)$-chips for any $p \le k$.
\end{itemize}
\end{itemize}
\end{lemma}
\begin{proof}
Let us say that a \emph{chiplet} is a set $C \subseteq E(G)$ so that either
\begin{itemize}
\item $C$ is a $(h(k), k)$-bigchip, or
\item $C$ is a multiplicity-$1$-violator of rank $k-1$ (so $|C|=2$ and $|V(C)| = k-1$), so that there exists a multiplicity-$3$-violator $C' \supseteq C$.
\end{itemize}

We initially let $\Tc = (T,\lmap)$ be the unique superbranch decomposition of $G$ with one internal node $r \in V(T)$.
Note that $\torso(r) = G$, and from now identify $G$ with $\torso(r)$.
We will implement a process where we repeatedly discover a chiplet $C$, and transform $\Tc$ to $\Tc \rescliqs (r,C)$, transforming $G$ to $G \rescliqs C$, maintaining that $\torso(r) = G$.

This process stops either when $|V(G)| \le 6 \cdot h(k)$, or when we have concluded that $G$ has no $(h(k), k)$-bigchips.

Let us then state the invariants maintained by the process.
Denote by $G_{\mathsf{init}}$ the initial input hypergraph $G$.

\begin{claim}
\label{lem:algotunbtotunb:claim0}
The process maintains the invariants that
\begin{enumerate}
\item the internal separations of $\Tc$ have order $<k$ and are doubly tri-well-linked,\label{lem:algotunbtotunb:claim0:enu2}
\item if $t \in \vint(T) \setminus \{r\}$, then $|V(\torso(t))| \le h(k)$,\label{lem:algotunbtotunb:claim0:enu3}
\item $\rank(G) \le \rank(G_{\mathsf{init}})$,\label{lem:algotunbtotunb:claim0:enu4}
\item $G$ is $k$-tangle-unbreakable,\label{lem:algotunbtotunb:claim0:enu5}
\item $G$ has no $(s(p), h(p), p)$-chips for any $p < k$,\label{lem:algotunbtotunb:claim0:enu6}
\item $G$ has no $(s(k), h(k), k)$-smallchips, and\label{lem:algotunbtotunb:claim0:enu7}
\item $G$ has no multiplicity-$3$-violators of rank $\neq k-1$.\label{lem:algotunbtotunb:claim0:enu8}
\end{enumerate}
\end{claim}
\begin{claimproof}
Initially, these invariants are satisfied by the initial guarantees for $G$ and the construction of $\Tc$.

Then, suppose that $\Tc$ satisfies the invariants and let $\Tc' = (T',\lmap') = \Tc \rescliqs (r,C)$, where $C$ is a chiplet.
If $C$ is a $(h(k), k)$-bigchip, then $(C,\co{C})$ is doubly tri-well-linked by \Cref{lem:turbchipproperties}.
If $C$ is a multiplicity-violator with $\bdc(C) = k-1$ and $|C| = 2$, so that there exists a multiplicity-$3$-violator $C' \supseteq C$, then $(C,\co{C})$ is doubly tri-well-linked because both $C$ and $\co{C}$ contain a hyperedge $e$ with $\bd(e) = \bd(C)$.

All internal separations of $\Tc'$ either (1) are internal separations of $\Tc$ or (2) are an orientation of the separation $(C \orescliqs \Tc, \co{C} \orescliqs \Tc)$.
The order of this separation is $k-1$.
Furthermore, as $(C,\co{C})$ is doubly tri-well-linked, and all internal separations of $\Tc$ are doubly tri-well-linked, the transitivity of tri-well-linkedness (\Cref{lem:triwelllinkedcom}) implies that $(C \orescliqs \Tc, \co{C} \orescliqs \Tc)$ is doubly tri-well-linked, implying that $\Tc'$ satisfies \Cref{lem:algotunbtotunb:claim0:enu2}.
The torsos of $\Tc'$ are either torsos of $\Tc$ or are equal to $G \rescliqs C$ or $G \rescliqs \co{C}$.
Therefore, for \Cref{lem:algotunbtotunb:claim0:enu3} it suffices to observe that $|V(G \rescliqs \co{C})| = |V(C)| \le h(k)$.

For \Cref{lem:algotunbtotunb:claim0:enu4}, it suffices to observe that if $C$ is a chiplet, then there is $e \in C$ with $\bdc(e) \ge \bdc(C)$.
\Cref{lem:algotunbtotunb:claim0:enu5} follows from the fact that if $G$ is $k$-tangle-unbreakable and $(C,\co{C})$ is doubly tri-well-linked, then $G \rescliqs C$ is also $k$-tangle-unbreakable by \Cref{lem:tungunbhered}.
\Cref{lem:algotunbtotunb:claim0:enu6} follows from \Cref{lem:turbchipproperties} in the case of a $(h(k),k)$-bigchip and from \Cref{lem:tubr_to_ubr_chip_mult_ok} in the case of a multiplicity-violator.

For \Cref{lem:algotunbtotunb:claim0:enu7}, suppose that $G \rescliqs C$ contains a $(s(k), h(k), k)$-smallchip $D$.
Because $\bdc(e) < k-1$ for all $e \in D$, we have that $e_C \notin D$.
Therefore, $D$ would also be a $(s(k), h(k), k)$-smallchip of $G$, contradicting the invariant.

\Cref{lem:algotunbtotunb:claim0:enu8} follows from that $\bdc(C) = k-1$.
\end{claimproof}

If in the end of the process either $|V(G)| \le 6 \cdot h(k)$ or $G$ has no $(h(k), k)$-bigchips, then we can return $\Tc$.
It remains to describe how to implement this process efficiently.

We let $Q$ to be a queue storing hyperedges of $G$ (including hyperedges that are no longer in $G$, which will be discarded when popping the queue).
This queue will have the invariant that if $C \subseteq E(G)$ is an $(h(k), k)$-bigchip of $G$, then there exist $e \in C$ so that $e \in Q$.
We initialize $Q$ to contain all hyperedges of $G$, meaning that initially the invariant is trivially satisfied.

\begin{claim}
Assume that $Q$ is non-empty.
We can in time $h(k)^{\OO(k+\rank(G))}$ either obtain a chiplet or remove an element from $Q$.
\end{claim}
\begin{claimproof}
Let $e \in Q$ be the front element from $Q$.
If $e$ has been deleted from $E(G)$, we can remove $e$ from $Q$.
Otherwise, we use \Cref{lem:localchipalgo} to either (1) enumerate all sets $C \subseteq E(G)$ so that $e \in C$, $|V(C)| \le h(k)$, $\bdc(C) < k$, and $C$ is internally connected, or (2) return a multiplicity-$3$-violator in $G$.
This algorithm runs in time $h(k)^{\OO(k+\rank(G))}$.

If it returns a multiplicity-$3$-violator $X \subseteq E(G)$, then because $G$ has no multiplicity-$3$-violators of rank $\neq k-1$, we have that any subset $X' \subseteq X$ with $|X'| = 2$ is a chiplet.

Otherwise, if there is an $(h(k), k)$-bigchip $C$ with $e \in C$, such $C$ must be enumerated by the algorithm.
We check each $C$ enumerated by the algorithm in time $\OO(|C| \cdot (k+\rank(G)))$, and if there is an $(h(k), k)$-bigchip, we return it.
If there is no $(h(k), k)$-bigchip, we can remove $e$ from~$Q$.
\end{claimproof}

Now suppose we have a chiplet $C \subseteq E(G)$.
We use \Cref{lem:sbdchippingalgo} to transform $\Tc$ into $\Tc \rescliqs (r,C)$ in time $\OO(|C| \cdot \adhsize(\Tc)) = h(k)^{\OO(k+\rank(G))}$, preserving the pointers to hyperedges in $\co{C}$.
Then, we add the hyperedge $e_C$ to the queue $Q$, which maintains the invariant of the queue because if $C'$ is a $(h(k),k)$-bigchip of $G \rescliqs C$ with $e_C \notin C'$, then $C'$ is also a $(h(k),k)$-bigchip of $G$.

As each chiplet $C$ has size $|C| \ge 2$, this process runs in time $|E(G_{\mathsf{init}})| \cdot h(k)^{\OO(k+\rank(G_{\mathsf{init}}))} = \|G\| \cdot k^{\OO(k \cdot (k+\rank(G)))}$.
\end{proof}

We then put together the proof of \Cref{the:highlevel:tunbrtounbr} by combining the algorithms of \Cref{lem:algdecompmultihyperedges,lem:tubrtoubrremovesmallchipsalg,lem:tubrtoubralgremovebigchipsalg}.

\thehighleveltunbrtounbr*
\begin{proof}
We will compute a series of superbranch decompositions $\Tc_0 = (T_0, \lmap_0), \ldots, \Tc_k = (T_k,\lmap_k)$ so that for each $i \in [0,k]$,
\begin{enumerate}
\item the internal separations of $\Tc_i$ have order $<k$ and are doubly well-linked,
\item the torsos of $\Tc_i$ are $k$-tangle-unbreakable, and
\item each torso of $\Tc_i$ either has at most $6 \cdot h(k)$ vertices, or has no $(s(p), h(p), p)$-chips for any $p \le i$.
\end{enumerate}

By \Cref{lem:tubrimprunbrklem}, this implies that each torso of $\Tc_k$ is $(6 \cdot h(k), k)$-unbreakable, so we can output $\Tc' = \Tc_k$.
Note also that we can initially set $\Tc_0 = \Tc$.

Let us then show that $\Tc_i$ can be computed from $\Tc_{i-1}$, for $i \in [k]$, in time $k^{\OO(k^2)} \cdot \|G\|$.
This boils down to applying the algorithms of \Cref{lem:algdecompmultihyperedges,lem:tubrtoubrremovesmallchipsalg,lem:tubrtoubralgremovebigchipsalg} in sequence to decompose the torsos.

First, we compute a superbranch decomposition $\Tc_{i-1}' = (T_{i-1}', \lmap_{i-1}')$, which will have the same guarantees as $\Tc_{i-1}$, except that each torso with more than $6 \cdot h(k)$ vertices has multiplicity at most $3$.
For this, we apply for each $t \in \vint(T_{i-1})$ the algorithm of \Cref{lem:algdecompmultihyperedges} with the parameter $k$ to compute a superbranch decomposition $\Tc_{i-1,t}'$ of $\torso(t)$.
This runs in total time $k^{\OO(1)} \cdot \|\Tc_{i-1}\| = k^{\OO(1)} \cdot \|G\|$.
Note that $\torso(t)$ does not have multiplicity-violators of rank $\ge k$ because $G$ has multiplicity $1$ and the internal separations of $\Tc_{i-1}$ have order $<k$.
We have that the internal separations of $\Tc_{i-1,t}'$ are doubly tri-well-linked, each torso of $\Tc_{i-1,t}'$ has either (1) at most $6 \cdot h(k)$ vertices, or (2) has multiplicity at most $3$ and no $(s(p), h(p), p)$-chips for any $1 \le p \le i-1$.
Furthermore, \Cref{lem:tungunbhered} implies that the torsos of $\Tc_{i-1,t}'$ are $k$-tangle-unbreakable.
We then apply the algorithm of \Cref{lem:superbdalgrefinebyrefiset} to obtain $\Tc_{i-1}' = \Tc_{i-1} \rescliqs \{\Tc_{i-1,t}'\}_{t \in \vint(T_{i-1})}$ in time $k^{\OO(1)} \cdot \|G\|$.
All torsos of $\Tc_{i-1}'$ are torsos of $\Tc_{i-1,t}'$ so they satisfy the required properties, and all internal separations of $\Tc_{i-1}'$ are either internal separations of $\Tc_{i-1}$ or correspond to internal separations of $\Tc_{i-1,t}'$, so by the transitivity of well-linkedness (\Cref{lem:linkedcliq}) they also satisfy the required properties.

Then we compute a superbranch decomposition $\Tc_{i-1}'' = (T_{i-1}'', \lmap_{i-1}'')$, which will have the same guarantees as $\Tc_{i-1}$, except that each torso with more than $6 \cdot h(k)$ vertices has no $(s(i), h(i), i)$-smallchips and has no multiplicity-$3$-violators of rank $\neq i-1$.
Let $Y \subseteq \vint(T_{i-1}')$ be the set of nodes of $\Tc_{i-1}'$ whose torsos have more than $6 \cdot h(k)$ vertices.
We apply for each $t \in Y$ the algorithm of \Cref{lem:tubrtoubrremovesmallchipsalg} with $\torso(t)$ and the parameter $i$.
We obtain for each $t \in Y$ a superbranch decomposition $\Tc_{i-1,t}''$ of $\torso(t)$ so that the internal separations of $\Tc_{i-1,t}''$ are doubly tri-well-linked and have order $<i$, and each torso of $\Tc_{i-1,t}''$ has either at most $6 \cdot h(k)$ vertices or has no $(s(p), h(p), p)$-smallchips for any $p \le i$, no $(h(p), p)$-bigchips for any $p < i$, and no multiplicity-$3$-violators of rank $\neq i-1$.
Furthermore, the fact that the internal separations are doubly tri-well-linked implies by \Cref{lem:tungunbhered} that the torsos of $\Tc_{i-1,t}''$ are $k$-tangle-unbreakable.
This runs in total time $k^{\OO(k \cdot (k+\adhsize(\Tc_{i-1}')))} \cdot \|\Tc_{i-1}'\| = k^{\OO(k^2)} \cdot \|G\|$.
We then apply the algorithm of \Cref{lem:superbdalgrefinebyrefiset} to obtain $\Tc_{i-1}'' = \Tc_{i-1}' \rescliqs \{\Tc_{i-1,t}''\}_{t \in Y}$ in time $k^{\OO(1)} \cdot \|G\|$.
All torsos of $\Tc_{i-1}''$ are either torsos of $\Tc_{i-1}'$ with at most $6 \cdot h(k)$ vertices, or are torsos of $\Tc_{i-1,t}''$, so they satisfy the required properties.
The internal separations of $\Tc_{i-1}''$ are either internal separations of $\Tc_{i-1}'$ or correspond to internal separations of $\Tc_{i-1,t}''$, so by the transitivity of well-linkedness (\Cref{lem:linkedcliq}) they also satisfy the required properties.

Then we finally compute the superbranch decomposition $\Tc_i$ from $\Tc_{i-1}''$.
Let $Y \subseteq \vint(T_{i-1}'')$ be the set of nodes of $\Tc_{i-1}''$ whose torsos have more than $6 \cdot h(k)$ vertices.
We apply for each $t \in Y$ the algorithm of \Cref{lem:tubrtoubralgremovebigchipsalg} with $\torso(t)$ and the parameter $i$.
We obtain for each $t \in Y$ a superbranch decomposition $\Tc_{i,t}$ of $\torso(t)$ so that the internal separations of $\Tc_{i,t}$ are doubly tri-well-linked and have order $<i$, and each torso of $\Tc_{i,t}$ has either at most $6 \cdot h(k)$ vertices, or has no $(s(p), h(p), p)$-chips for any $p \le i$.
Furthermore, the fact that the internal separations are doubly tri-well-linked implies by \Cref{lem:tungunbhered} that the torsos of $\Tc_{i,t}$ are $k$-tangle-unbreakable.
This runs in total time $k^{\OO(k \cdot (k+\adhsize(\Tc_{i-1}'')))} \cdot \|\Tc_{i-1}''\| = k^{\OO(k^2)} \cdot \|G\|$.
We then apply the algorithm of \Cref{lem:superbdalgrefinebyrefiset} to obtain $\Tc_{i} = \Tc_{i-1}'' \rescliqs \{\Tc_{i,t}\}_{t \in Y}$ in time $k^{\OO(1)} \cdot \|G\|$.
All torsos of $\Tc_{i}$ are either torsos of $\Tc_{i-1}''$ with at most $6 \cdot h(k)$ vertices, or are torsos of $\Tc_{i,t}$, so they satisfy the required properties.
The internal separations of $\Tc_{i}$ are either internal separations of $\Tc_{i-1}''$ or correspond to internal separations of $\Tc_{i,t}$, so by the transitivity of well-linkedness (\Cref{lem:linkedcliq}) they also satisfy the required properties.

This completes the description of computing $\Tc_i$ from $\Tc_{i-1}$ in $k^{\OO(k^2)} \cdot \|G\|$ time.
\end{proof}

\section{Combining lean tree decompositions}
\label{sec:combin}
In this section, we prove \Cref{the:highlevel:leanimplication}, which we re-state now.

\thehighlevelleanimplication*

The algorithm of \Cref{the:highlevel:leanimplication} simply stitches together the tree decompositions $\Tc_t$ following the structure of $\Tc$.
For the correctness of this, the main graph-theoretical lemma to prove is that vertex cuts of $\primal(G)$ separating subsets of torsos of $\Tc$ can be uncrossed with vertex cuts corresponding to internal separations of $\Tc$, which we prove now.

\begin{lemma}
\label{lem:insideleanuncross}
Let $G$ be a hypergraph and $\Tc = (T,\lmap)$ a superbranch decomposition of $G$ whose all internal adhesions are doubly well-linked.
Let $t_1, t_2 \in \vint(T)$, $X_1 \subseteq V(t_1)$, and $X_2 \subseteq V(t_2)$.
For every vertex cut $(A,B)$ of $\primal(G)$ with $X_1 \subseteq A$ and $X_2 \subseteq B$, there is a vertex cut $(C,D)$ of $\primal(G)$ that 
\begin{itemize}
\item is parallel with every vertex cut of form $(V(\lmap(\vec{xy})), V(\lmap(\vec{yx})))$, where $xy$ is an edge of $T$,
\item has $X_1 \subseteq C$, $X_2 \subseteq D$, and 
\item has $|C \cap D| \le |A \cap B|$.
\end{itemize}
\end{lemma}
\begin{proof}
Let $G'$ be the hypergraph obtained from $G$ by adding for each $v \in X_1$ a hyperedge $e_{v,1}$ with $V(e_{v,1}) = \{v\}$, and for each $v \in X_2$ a hyperedge $e_{v,2}$ with $V(e_{v,2}) = \{v\}$.
Let also $\Tc' = (T',\lmap')$ be the superbranch decomposition of $G'$ obtained from $\Tc$ by adding leaves corresponding to the hyperedges of form $e_{v,1}$ adjacent to $t_1$, and adding leaves corresponding to the hyperedges of form $e_{v,2}$ adjacent to $t_2$.

We observe that the internal edges of $T'$ are the same as the internal edges of $T$.
Moreover, for every internal edge $xy$ of $T'$ we have that $(V(\lmap(\vec{xy})), V(\lmap(\vec{yx}))) = (V(\lmap'(\vec{xy})), V(\lmap'(\vec{yx})))$ and $(\lmap'(\vec{xy}), \lmap'(\vec{yx}))$ is doubly well-linked.
In particular, the internal separations of $\Tc'$ are doubly well-linked.

Let us denote $X_1^e = \{e_{v,1} : v \in X_1\}$ and $X_2^e = \{e_{v,2} : v \in X_2\}$.

Let $(Q,\co{Q})$ be a separation of $G'$ so that (1) $X_1^e \subseteq Q$, (2) $X_2^e \subseteq \co{Q}$, (3) $\bdc(Q) \le |A \cap B|$, and (4) subject to (1), (2), and (3), $(Q,\co{Q})$ crosses the minimum number of separations of $\Tc'$.
We have that such a separation $(Q,\co{Q})$ satisfying (1), (2), and (3) exists by assigning each $e \in E(G)$ with $V(e) \cap (A \setminus B) \neq \emptyset$ to $Q$, each $e \in E(G)$ with $V(e) \cap (B \setminus A) \neq \emptyset$ to $\co{Q}$, each $e \in X_1^e$ to $Q$, each $e \in X_2^e$ to $\co{Q}$, and the remaining hyperedges $e \in E(G')$, for all of which it holds that $V(e) \subseteq A \cap B$, arbitrarily.

\begin{claim}
\label{lem:insideleanuncross:claim1}
$(Q,\co{Q})$ crosses no separations of $\Tc'$.
\end{claim}
\begin{claimproof}
Suppose that $xy \in E(T')$ and $(Q,\co{Q})$ crosses $(\lmap'(\vec{xy}), \lmap'(\vec{yx}))$.
Note that for $(Q,\co{Q})$ to cross $(\lmap'(\vec{xy}), \lmap'(\vec{yx}))$, it must be that $|\lmap'(\vec{xy})|, |\lmap'(\vec{yx})| \ge 2$.
Therefore, neither $x$ nor $y$ is a leaf, i.e., $(\lmap'(\vec{xy}), \lmap'(\vec{yx}))$ is an internal separation of $\Tc'$.

First, assume that $xy$ is not on the unique $(t_1,t_2)$-path in $T'$.
Assume \wilog that both $t_1$ and $t_2$ are closer to $y$ than to $x$ in $T'$, in particular, $X_1^e \cup X_2^e \subseteq \lmap(\vec{yx})$.
Because $\lmap(\vec{xy})$ is well-linked, we have that either (a) $\bdc(Q \cap \lmap(\vec{xy})) \ge \bdc(\lmap(\vec{xy}))$, or (b) $\bdc(\co{Q} \cap \lmap(\vec{xy})) \ge \bdc(\lmap(\vec{xy}))$.

If (a), then by submodularity
\begin{align*}
\bdc(Q \cup \lmap(\vec{xy})) \le \bdc(Q).
\end{align*}
However, in this case the separation $(Q \cup \lmap(\vec{xy}), \co{Q} \cap \lmap(\vec{yx}))$ satisfies the required properties (1), (2), and (3), and crosses less internal separations of $\Tc'$ than $(Q,\co{Q})$, and therefore contradicts the choice of $(Q,\co{Q})$.

If (b), then by submodularity
\begin{align*}
\bdc(\co{Q} \cup \lmap(\vec{xy})) \le \bdc(Q).
\end{align*}
In this case, the separation $(Q \cap \lmap(\vec{yx}), \co{Q} \cup \lmap(\vec{xy}))$ satisfies the required properties (1), (2), and (3), and crosses less internal separations of $\Tc'$ than $(Q,\co{Q})$, and therefore contradicts the choice of $(Q,\co{Q})$.

Second, assume that $xy$ is on the unique $(t_1, t_2)$-path in $T$.
Assume \wilog that $t_1$ is closer to $x$ than to $y$, and $t_2$ is closer to $y$ than $x$.
In particular, $X_1^e \subseteq \lmap(\vec{xy})$ and $X_2^e \subseteq \lmap(\vec{yx})$.
First, if $\bdc(\lmap(\vec{xy})) \le \bdc(Q)$, then the separation $(\lmap(\vec{xy}), \lmap(\vec{yx}))$ would contradict the choice of $(Q,\co{Q})$, so assume that $\bdc(\lmap(\vec{xy})) > \bdc(Q)$.
If $\bdc(Q \cup \lmap(\vec{xy})) \le \bdc(Q)$, then the separation $(Q \cup \lmap(\vec{xy}), \co{Q} \cap \lmap(\vec{yx}))$ would contradict the choice of $(Q,\co{Q})$, so assume that $\bdc(Q \cup \lmap(\vec{xy})) > \bdc(Q)$.
Now, 
\begin{align*}
&& \bdc(Q \cup \lmap(\vec{xy})) > \bdc(Q) && \\
\Rightarrow && \bdc(Q \cap \lmap(\vec{xy})) < \bdc(\lmap(\vec{xy})) && \text{ (submodularity) }\\
\Rightarrow && \bdc(\co{Q} \cap \lmap(\vec{xy})) \ge \bdc(\lmap(\vec{xy})) > \bdc(Q) && \text{ (well-linkedness of $\lmap(\vec{xy})$) }\\
\Rightarrow && \bdc(\co{Q} \cup \lmap(\vec{xy})) < \bdc(\lmap(\vec{xy})) && \text{ (submodularity) }\\
\Rightarrow && \bdc(Q \cap \lmap(\vec{yx})) < \bdc(\lmap(\vec{yx})) && \text{ (symmetry) }\\
\Rightarrow && \bdc(\co{Q} \cap \lmap(\vec{yx})) \ge \bdc(\lmap(\vec{yx})) && \text{ (well-linkedness of $\lmap(\vec{yx})$) }\\
\Rightarrow && \bdc(\co{Q} \cup \lmap(\vec{yx})) \le \bdc(Q) && \text{ (submodularity) }\\
\Rightarrow && \bdc(Q \cap \lmap(\vec{xy})) \le \bdc(Q) && \text{ (symmetry) }
\end{align*}

Therefore, the separation $(Q \cap \lmap(\vec{xy}), \co{Q} \cup \lmap(\vec{yx}))$ would contradict the choice of $(Q,\co{Q})$.
\end{claimproof}

Now, let $(C,D) = (V(Q), V(\co{Q}))$.
By definition, $|C \cap D| = \bdc(Q) \le |A \cap B|$.
Furthermore, as $X_1^e \subseteq Q$ and $X_2^e \subseteq \co{Q}$, we have that $X_1 \subseteq C$ and $X_2 \subseteq D$.
It remains to prove that $(C,D)$ is parallel with every vertex cut of $\primal(G)$ of form $(V(\lmap(\vec{xy}), V(\lmap(\vec{yx}))))$, where $xy$ is an edge of $T$.
For this, it suffices to prove that $(C,D)$ is parallel with every vertex cut of $\primal(G)$ of form $(V(\lmap'(\vec{xy})), V(\lmap'(\vec{yx})))$, where $xy$ is an edge of $T'$.
By \Cref{lem:insideleanuncross:claim1}, for every $e \in E(T')$, there exists an orientation $\vec{e}$ so that either $Q \subseteq \lmap'(\vec{e})$ or $\co{Q} \subseteq \lmap'(\vec{e})$.
In the first case, we have that $C \subseteq V(\lmap'(\vec{e}))$, and in the second that $D \subseteq V(\lmap'(\vec{e}))$.
\end{proof}

Now we are ready to prove \Cref{the:highlevel:leanimplication}.

\thehighlevelleanimplication*
\begin{proof}
For each $t \in \vint(T)$, denote $\Tc_t = (T_t, \bag_t)$.
We define a tree decomposition $\Tc' = (T', \bag')$ as follows.
First, we set $V(T')$ to consist of two different kinds of nodes, the union of all nodes of the decompositions $\Tc_t$, and a set of nodes corresponding to internal adhesions of $\Tc$, denoted by $\alpha_{xy}$ for $xy \in \eint(T)$, i.e.,
\[V(T') = \bigcup_{t \in \vint(T)} V(T_t) \cup \bigcup_{xy \in \eint(T)} \{\alpha_{xy}\}.\]
Now, for $xy \in \eint(T)$, let $c_{\vec{xy}} \in V(T_x)$ be a node of $\Tc_x$ with $\adh(xy) \subseteq \bag_x(c_{\vec{xy}})$ (if there are multiple such nodes, choose an arbitrary one).
Note that such a node exists because $\adh(xy)$ is a clique in $\primal(\torso(x))$.
Moreover, it can be found in $|\adh(xy)|$ time for each $xy$ after $\OO(\|\Tc_t\|)$ time precomputation by \Cref{lem:cliqtdds}.

Then, we set $E(T')$ to connect $V(T')$ in the natural way, i.e.,
\[E(T') = \bigcup_{t \in \vint(T)} E(T_t) \cup \bigcup_{xy \in \eint(T)} \{\alpha_{xy} c_{\vec{xy}}, \alpha_{xy} c_{\vec{yx}}\}.\]

Finally, we define $\bag'$ in the natural way, i.e.,
\begin{itemize}
\item for $t \in \vint(T)$ and $x \in V(T_t)$, we set $\bag'(x) = \bag_t(x)$, and
\item for $xy \in \eint(T)$, we set $\bag'(\alpha_{xy}) = \adh(xy)$.
\end{itemize}

By using the above mentioned technique for finding the nodes $c_{\vec{xy}}$, the tree decomposition $\Tc' = (T',\bag')$ can be constructed in time $\OO((k+\rank(G)) \cdot |E(T)| +  \sum_{t \in \vint(T)} \|\Tc_t\|) = \OO((k+\rank(G)) \cdot \|G\| + \sum_{t \in \vint(T)} \|\Tc_t\|)$.

It remains to prove that $\Tc'$ is indeed a $k$-lean tree decomposition of $\primal(G)$.

\begin{claim}
\label{the:highlevel:leanimplication:claim1}
$\Tc'$ is a tree decomposition of $\primal(G)$.
\end{claim}
\begin{claimproof}
Let us first prove the edge condition.
Because $G$ is normal, for each $e \in E(G)$ there exists $t \in \vint(T)$ and $e' \in E(\torso(t))$ with $V(e') = V(e)$.
Now, $V(e') = V(e)$ is a clique in $\primal(\torso(t))$, implying that $\Tc_t$ contains a bag containing $V(e)$, and the same bag exists also in $\Tc'$.

Then we prove the vertex condition.
Because each vertex of $\primal(G)$ is in $V(e)$ for some $e \in E(G)$, by the above argument we have that each vertex of $\primal(G)$ is in at least one bag of $\Tc'$.
The decompositions $\Tc_t$ satisfy the vertex condition, and for each $xy \in \eint(T)$ we have that $\bag'(\alpha_{xy}) \subseteq \bag'(c_{\vec{xy}})$ and $\bag'(\alpha_{xy}) \subseteq \bag'(c_{\vec{yx}})$ by construction.
It remains to observe that if $t_1, t_2 \in \vint(T)$, and $v \in V(t_1) \cap V(t_2)$, then $v \in \adh(xy)$ for all $xy \in \eint(T)$ that are on the unique $(t_1,t_2)$-path in $T$.
\end{claimproof}

\begin{claim}
\label{the:highlevel:leanimplication:claim2}
$\Tc'$ is $k$-lean.
\end{claim}
\begin{claimproof}
To prove that $\adhsize(\Tc') < k$, we note that all adhesions of $\Tc'$ are either (1) adhesions of $\Tc_t$ for $t \in \vint(T)$, which by assumption have size $<k$, or (2) equal to $\adh(xy)$ for $xy \in \eint(T)$, which also by assumption have size $<k$.

Then it remains to prove that $\Tc'$ has no non-$k$-lean-witnesses.
For the sake of contradiction, assume that $(A,B,t_a,t_b)$ is a non-$k$-lean-witness of $\Tc'$, i.e., $(A,B)$ is a vertex cut of $\primal(G)$ of order $<k$, $t_a,t_b \in V(T')$, $|\bag'(t_a) \cap A|>|A \cap B|$, $|\bag'(t_b) \cap B| > |A \cap B|$, and all $(t_a,t_b)$-adhesions have size $>|A \cap B|$.

We can assume that neither $t_a$ nor $t_b$ is a node of type $\alpha_{xy}$ corresponding to an adhesion of $\Tc$, because if $t_a = \alpha_{xy}$, then in fact both $(A,B,c_{\vec{xy}}, t_b)$ and $(A,B,c_{\vec{yx}}, t_b)$ would also be non-$k$-lean-witnesses.

By applying \Cref{lem:insideleanuncross} with the vertex cut $(A,B)$, the superbranch decomposition $\Tc$, $X_1 = \bag'(t_a) \cap A$, and $X_2 = \bag'(t_b) \cap B$, we can assume that $(A,B)$ is parallel with every vertex cut of form $(V(\lmap(\vec{xy})), V(\lmap(\vec{yx})))$, where $xy \in E(T)$.
In particular, it holds that $\adh(xy) \subseteq A$ or $\adh(xy) \subseteq B$ for all $xy \in E(T)$.

Consider first the case that $t_a,t_b \in V(T_t)$ for some $t \in \vint(T)$.
Because the adhesions of $\Tc$ are subsets of $A$ or $B$, we have that $(A \cap V(t), B \cap V(t))$ is a vertex cut of $\primal(\torso(t))$.
However, now $(A \cap V(t), B \cap V(t), t_a, t_b)$ would be a non-$k$-lean-witness for $\Tc_t$.

Then suppose that $t_a \in V(T_{t_1})$ and $t_b \in V(T_{t_2})$, where $t_1,t_2 \in \vint(T)$ and $t_1 \neq t_2$.
Let $t_1 = h_1, h_2, \ldots, h_p = t_2$ be the unique $(t_1,t_2)$-path in $T$.
Recall that by the definition of non-$k$-lean-witness, we have that $|\adh(h_i h_{i+1})| > |A \cap B|$ for all $i \in [h-1]$.

Suppose first that $\adh(h_1 h_2) \subseteq B$.
Then, however, $(A,B,t_a, c_{\vec{h_1 h_2}})$ would be a non-$k$-lean-witness with $t_a, c_{\vec{h_1 h_2}} \in V(T_{t_1})$, but we already proved that such non-$k$-lean-witnesses do not exist.
Similary, if $\adh(h_{p-1} h_p) \subseteq A$, then $(A,B,c_{\vec{h_p h_{p-1}}}, t_b)$ would be a non-$k$-lean-witness with $c_{\vec{h_p h_{p-1}}}, t_b \in V(T_{t_2})$.

It remains to consider the case when $\adh(h_1 h_2) \subseteq A$ and $\adh(h_{p-1} h_p) \subseteq B$.
Now, there exists $i \in [2,h-1]$ so that $\adh(h_{i-1} h_i) \subseteq A$ and $\adh(h_i h_{i+1}) \subseteq B$.
However, then $(A,B,c_{\vec{h_i h_{i-1}}}, c_{\vec{h_i h_{i+1}}})$ would be a non-$k$-lean-witness with $c_{\vec{h_i h_{i-1}}}, c_{\vec{h_i h_{i+1}}} \in V(T_{h_i})$.
\end{claimproof}
\Cref{the:highlevel:leanimplication:claim1,the:highlevel:leanimplication:claim2} finish the proof.
\end{proof}

\section{Lean tree decompositions of unbreakable graphs}
\label{sec:lean}
In this section we prove \Cref{the:highlevel:unbrtolean}, which we re-state now.

\thehighlevelunbrtolean*

\subsection{Reduction to simpler cases}
In this subsection we reduce \Cref{the:highlevel:unbrtolean} to simpler ingredients, which then will be proven in the subsequent subsections.

We will exploit the fact that $k$-lean tree decompositions of $(s,k)$-unbreakable graphs have a lot of structure.
We start by observing that if $G$ is $(s,k)$-unbreakable, then tree decompositions of $G$ with adhesion size $<k$ have at most one large bag.

\begin{lemma}
\label{lem:atmostonelargebag}
Let $s \ge k \ge 1$ be integers and $G$ be an $(s,k)$-unbreakable graph.
Then, any tree decomposition $\Tc$ of $G$ with $\adhsize(\Tc)<k$ has at most one bag of size $\ge s$.
\end{lemma}
\begin{proof}
If $\Tc$ would have two nodes $t_1,t_2$ with bags of size $\ge s$, then any $(t_1,t_2)$-adhesion would give a vertex cut that would contradict the $(s,k)$-unbreakability of $G$.
\end{proof}

Then we show that if the number of vertices of $G$ is large enough, every tree decomposition of $G$ with adhesion size $<k$ has in fact exactly one large bag.

\begin{lemma}
\label{lem:bigbagexists}
Let $s \ge k \ge 1$ be integers and $G$ an $(s,k)$-unbreakable graph with $|V(G)| \ge (2s)^{k+2}$.
Any tree decomposition of $G$ with adhesion size $<k$ has exactly one bag of size $\ge 2s$.
\end{lemma}
\begin{proof}
Let $\Tc = (T,\bag)$ be a tree decomposition of $G$ with $\adhsize(\Tc) < k$.
By \Cref{lem:atmostonelargebag}, there is at most one bag of size $\ge 2s$.

Suppose all bags of $\Tc$ have size $<2s$.
Let $t \in V(T)$ be a central node of $\Tc$ in the sense that all connected components of $G \setminus \bag(t)$ have size $\le |V(G)|/2$.
Because $\adhsize(\Tc) < k$, for such connected components $C$ it holds that $|N(C)|<k$.
This in turn, by $|V(G)| \ge 2s$ and the $(s,k)$-unbreakability of $G$, implies that $|C| < s$.

Let us group vertices in $V(G) \setminus \bag(t)$ based on the neighborhood $N(C)$ of the component $C$ of $G \setminus \bag(t)$ containing the vertex.
Because $|N(C)|<k$, $N(C) \subseteq \bag(t)$, and $|\bag(t)|<2s$, there are at most $(2s)^k$ such groups, so one such group must contain at least $s$ vertices.
This implies that we can choose a union $C'$ of some connected components of $G \setminus \bag(t)$, so that $s \le |C'| \le 2s$, and $|N(C')|<k$.
As $|V(G)| \ge 4s$, this would contradict the $(s,k)$-unbreakability of $G$.
\end{proof}

We then show that such tree decompositions containing exactly one large bag must have non-$k$-lean-witnesses with certain structure.

\begin{lemma}
\label{lem:leanwitstruct}
Let $s \ge k \ge 1$ be integers, $G$ an $(s,k)$-unbreakable graph, and $\Tc = (T,\bag)$ a tree decomposition of $G$ with $\adhsize(\Tc) < k$ and containing a node $t \in V(T)$ with $|\bag(t)| \ge 2s$.
If there exists a non-$k$-lean-witness $(A,B,x,y)$ for $\Tc$, then there exists a non-$k$-lean-witness $(A',B',x',y')$ for $\Tc$ so that either 
\begin{itemize}
\item $x' = y' = t$, or
\item the unique $(x',y')$-path in $T$ is disjoint from $t$,
\end{itemize}
and furthermore
\begin{itemize}
\item the distance between $x'$ and $y'$ in $T$ is at most the distance between $x$ and $y$ in $T$.
\end{itemize}
\end{lemma}
\begin{proof}
First, because $G$ is $(s,k)$-unbreakable and the order of the vertex cut $(A,B)$ is $<k$, we can assume that either $|A|<s$ or $|B|<s$.
Then, noting that if $(A,B,x,y)$ is a non-$k$-lean-witness then $(B,A,y,x)$ is a non-$k$-lean-witness, assume \wilog that $|A| < s$.

Our goal is now to turn $(A,B,x,y)$ into a non-$k$-lean-witness that satisfies the conclusion of the lemma.
Note that if the unique $(x,y)$-path in $T$ does not contain $t$, then $(A,B,x,y)$ already satisfies the conclusion, so assume that the $(x,y)$-path contains $t$.

Because $|A|<s$ and $|\bag(t)| \ge 2s$, we have that $|\bag(t) \cap B| \ge s > |A \cap B|$.
Now, because the $(x,y)$-path in $T$ contains $t$, every $(x,t)$-adhesion has size $> |A \cap B|$, and therefore $(A,B,x,t)$ is a non-$k$-lean-witness for $(T,\bag)$.
Furthermore, in this case the distance between $x$ and $t$ is at most the distance between $x$ and $y$.
If $x=t$ we are done, and otherwise we reduced the general case to the case when $y=t$ and $x \neq t$, so from now assume these conditions.

Let $c \in V(T)$ be the node adjacent to $t$ on the $(x,t)$-path, and let $(C,D)$ be the vertex cut of $G$ corresponding to the adhesion at $ct$, particularly with $\bag(x) \subseteq C$, $\bag(t) \subseteq D$, and $\adh(ct) = C \cap D$.
We have that $(A \cap C, B \cup D)$ and $(A \cup C, B \cap D)$ are vertex cuts, and by submodularity that either
\begin{equation}
\label{lem:leanwitstruct:case1}
|(A \cap C) \cap (B \cup D)| \le |A \cap B|,
\end{equation}
or
\begin{equation}
\label{lem:leanwitstruct:case2}
|(A \cup C) \cap (B \cap D)| < |C \cap D|.
\end{equation}

\begin{claim}
\label{lem:leanwitstruct:claimcase1}
If \Cref{lem:leanwitstruct:case1} holds, then $(A \cap C, B \cup D, x, c)$ is a non-$k$-lean-witness.
\end{claim}
\begin{claimproof}
First, because $\bag(x) \subseteq C$, we have that $\bag(x) \cap A = \bag(x) \cap A \cap C$, implying that 
\[|\bag(x) \cap A \cap C| > |A \cap B| \ge |(A \cap C) \cap (B \cup D)|.\]

Then, because every $(x,t)$-adhesion has size $>|A \cap B|$, we have that $|\adh(ct)| > |A \cap B|$.
Because $\adh(ct) \subseteq D$, it follows that 
\[|\bag(c) \cap (B \cup D)| > |A \cap B| \ge |(A \cap C) \cap (B \cup D)|.\]

Finally, because every $(x,t)$-adhesion has size $>|A \cap B|$, clearly every $(x,c)$-adhesion has size $>|(A \cap C) \cap (B \cup D)|$.
\end{claimproof}

\begin{claim}
\label{lem:leanwitstruct:claimcase2}
If \Cref{lem:leanwitstruct:case2} holds, then $(A \cup C, B \cap D, t, t)$ is a non-$k$-lean-witness.
\end{claim}
\begin{claimproof}
Because $|A| < s$, $|C| < s$, $|V(G)| \ge 4s$, and $|(A \cup C) \cap (B \cap D)|<k$, the $(s,k)$-unbreakability of $G$ implies that $|A \cup C| < s$.
This in turn, with $|\bag(t)| \ge 2s$, implies that $|B \cap D \cap \bag(t)| \ge s > |C \cap D| > |(A \cup C) \cap (B \cap D)|$.

Then, we have that $\adh(ct) \subseteq C \cap \bag(t)$ and $|\adh(ct)| = |C \cap D| > |(A \cup C) \cap (B \cap D)|$, implying that $|\bag(t) \cap (A \cup C)| > |(A \cup C) \cap (B \cap D)|$.
\end{claimproof}

Because both the distance between $x$ and $c$ and the distance between $t$ and $t$ are at most the distance between $x$ and $y$, \Cref{lem:leanwitstruct:claimcase1,lem:leanwitstruct:claimcase2} finish the proof.
\end{proof}

We say that a \emph{star decomposition} of a graph is a rooted tree decomposition $(T,\bag)$ so that all other nodes of $T$ than the root have degree $1$.
Let $s \ge k \ge 1$ be integers, $G$ be an $(s,k)$-unbreakable graph, and $(T,\bag)$ a star decomposition of $G$ with $\adhsize(\Tc) < k$.
We say that $(T,\bag)$ is \emph{big-star-$k$-lean} if all non-root nodes have bags of size $<s$, and there is no non-$k$-lean-witness of form $(A,B,\troot(T),\troot(T))$ for $(T,\bag)$.

Next we show that in the setting when $G$ has a large number of vertices, the computation of a $k$-lean tree decomposition reduces to computing a big-star-$k$-lean tree decomposition of $G$, and many $k$-lean tree decompositions of graphs with less than $s$ vertices.
In particular, we will show that after computing a big-star-$k$-lean tree decomposition, one can greedily commit to the bag of the root node of this decomposition, independently compute $k$-lean tree decompositions of the small graphs hanging from it, and then combine them to obtain a $k$-lean tree decomposition.

To formalize this, we introduce the following notation.
Let $(A,B)$ be a vertex cut of $G$.
We denote by $G \resgcliqs (A,B)$ the graph obtained from $G[A]$ by adding edges between all pairs of vertices in $A \cap B$.

\begin{lemma}
\label{lem:mergeintolean}
Let $s \ge k \ge 1$ be integers, $G$ an $(s,k)$-unbreakable graph with $|V(G)| \ge (2s)^{k+2}$, and $\Tc = (T,\bag)$ a big-star-$k$-lean tree decomposition of $G$.
Let $t$ be the root of $(T,\bag)$, and enumerate the other nodes by $x_1,\ldots,x_\ell$.
For $i \in [\ell]$, let $(C_i, D_i)$ be the vertex cut of $G$ corresponding to the adhesion at $x_i t$, i.e., with $C_i = \bag(x_i)$ and $C_i \cap D_i = \adh(x_i t)$.
Let also $(T^i, \bag^i)$ be a $k$-lean tree decomposition of $G \resgcliqs (C_i,D_i)$, with a node $r^i \in V(T^i)$ with $C_i \cap D_i \subseteq \bag^i(r^i)$.
Then, the tree decomposition $(T',\bag')$, obtained by replacing each $x_i$ by the tree decomposition $(T^i,\bag^i)$ rooted at $r^i$ is a $k$-lean tree decomposition of $G$.
\end{lemma}
\begin{proof}
As the tree decompositions $(T,\bag)$ and $(T^i,\bag^i)$ for all $i \in [\ell]$ have adhesion size $<k$, the resulting tree decomposition $(T',\bag')$ has adhesion size $<k$.
Let $t \in V(T')$ be the root node of $(T',\bag')$.
We have that $\bag'(t) = \bag(t)$, and therefore, because $(T,\bag)$ is big-star-$k$-lean, there is no non-$k$-lean-witness of form $(A,B,t,t)$ for $(T',\bag')$.
By \Cref{lem:bigbagexists}, $|\bag'(t)| \ge 2s$.

By \Cref{lem:leanwitstruct}, the other possible form of a non-$k$-lean-witness for $(T',\bag')$ is $(A,B,x,y)$, where the $(x,y)$-path in $T'$ does not contain $t$.
For the sake of contradiction, assume that such a non-$k$-lean-witness exists.
Because the $(x,y)$-path does not contain $t$, both $x$ and $y$ come from the same tree decomposition $(T^i,\bag^i)$.

Assume \wilog that $|A| < s$.
Recall that $(A \cap C_i, B \cup D_i)$ is a vertex cut of $G$.

\begin{claim}
$|(A \cap C_i) \cap (B \cup D_i)| \le |A \cap B|$.
\end{claim}
\begin{claimproof}
By submodularity, we have that either
\begin{equation}
\label{lem:mergeintolean:eq1}
|(A \cap C_i) \cap (B \cup D_i)| \le |A \cap B|,
\end{equation}
or
\begin{equation}
\label{lem:mergeintolean:eq2}
|(A \cup C_i) \cap (B \cap D_i)| < |C_i \cap D_i|.
\end{equation}

In the case of \Cref{lem:mergeintolean:eq1} we are done, so assume \Cref{lem:mergeintolean:eq2} holds.
We show that in this case, $(A \cup C_i, B \cap D_i, t, t)$ is a non-$k$-lean-witness for $(T,\bag)$, which is a contradiction.

Because $|A| < s$, $|C_i| < s$, and $|V(G)| \ge 4s$, we have that $|A \cup C_i| < s$, implying that $|\bag(t) \cap (B \cup D_i)| \ge s > |C_i \cap D_i| > |(A \cup C_i) \cap (B \cap D_i)|$.
We also have that $C_i \cap D_i \subseteq \bag(t)$, and $C_i \cap D_i \subseteq A \cup C_i$, implying that $|\bag(t) \cap (A \cup C_i)| \ge |C_i \cap D_i| > |(A \cup C_i) \cap (B \cap D_i)|$.
\end{claimproof}

Because $\bag'(x) \subseteq C_i$, it follows that that $(A \cap C_i, B \cup D_i, x, y)$ is a non-$k$-lean-witness for $(T',\bag')$.
However, note that $(A \cap C_i, (B \cup D_i) \cap C_i)$ is a vertex cut of order $|(A \cap C_i) \cap (B \cup D_i)|$ in the graph $G \resgcliqs (C_i, D_i)$, which implies that $(A \cap C_i, (B \cup D_i) \cap C_i, x_i, y_i)$ is a non-$k$-lean-witness for $(T^i, \bag^i)$, where $x_i$ and $y_i$ are the nodes of $T^i$ corresponding to $x$ and $y$.
This contradicts the $k$-leannes of $(T^i,\bag^i)$.
\end{proof}

We now state two lemmas whose proofs are the goals of \Cref{subsec:leandecompsmallgraph,subsec:bignodelean}, respectively.
First, in \Cref{subsec:bignodelean} we give an algorithm for computing big-star-$k$-lean tree decompositions.

\begin{restatable}{lemma}{leanalglarge}
\label{lem:leanalglarge}
There is an algorithm that, given integers $s \ge k \ge 1$, a graph $G$ that is $(s,k)$-unbreakable and has $|V(G)| \ge (2s)^{k+2}$ vertices, in time $s^{\OO(k)} \cdot \|G\|$ returns a big-star-$k$-lean tree decomposition of $G$.
\end{restatable}

Then, in \Cref{subsec:leandecompsmallgraph} we show that $k$-lean tree decompositions of $(s,k)$-unbreakable graphs can be computed in time $s^{\OO(k)} \cdot \|G\|^{\OO(1)}$, essentially by a direct algorithmic adaptation of the proof of \cite{bellenbaum2002two} (as was done also in \cite{DBLP:journals/talg/CyganKLPPSW21} for computing unbreakable decompositions).

\begin{restatable}{lemma}{leanalgsmall}
\label{lem:leanalgsmall}
There is an algorithm that, given integers $s \ge k \ge 1$, and a graph $G$ that is $(s,k)$-unbreakable, in time $s^{\OO(k)} \cdot \|G\|^{\OO(1)}$ returns a $k$-lean tree decomposition of $G$.
\end{restatable}

Before proceeding to the proofs of \Cref{lem:leanalglarge,lem:leanalgsmall}, we write out how the combination of \Cref{lem:mergeintolean,lem:leanalglarge,lem:leanalgsmall} implies \Cref{the:highlevel:unbrtolean}.

\thehighlevelunbrtolean*
\begin{proof}
First, if $|V(G)| < (2s)^{k+2}$, we run the algorithm of \Cref{lem:leanalgsmall}, which runs in time $s^{\OO(k)}$ in this case.

Otherwise, $|V(G)| \ge (2s)^{k+2}$.
We apply the algorithm of \Cref{lem:leanalglarge} to compute a big-star-$k$-lean tree decomposition $(T,\bag)$ of $G$.
Let $t \in V(T)$ be the root of $T$, and $x_1, \ldots, x_{\ell}$ the leaves of $T$.
One can assume that $\ell \le |V(G)|$, because if $\bag(x_i) \setminus \bag(t)$ is empty, the leaf $x_i$ can be simply removed.

For each leaf $x_i$, we construct the graph $G_i = G \resgcliqs (\bag(x_i), (V(G) \setminus \bag(x_i)) \cup \adh(x_i t))$.
These graphs can be constructed explicitly in total time $k^{\OO(1)} \cdot \|G\|$ by exploiting the fact that $|\adh(x_i t)|<k$.
Because $|\bag(x_i)|<s$, each such graph has at most $s$ vertices, and is trivially $(s,k)$-unbreakable.

We then apply \Cref{lem:leanalgsmall} to compute $k$-lean tree decompositions of each $G_i$, which takes total $s^{\OO(k)} \cdot |V(G)|$ time.
We then combine these $k$-lean tree decompositions with $(T,\bag)$ as described in the statement of \Cref{lem:mergeintolean}, which by \Cref{lem:mergeintolean} results in a $k$-lean tree decomposition of $G$.
\end{proof}

\subsection{Structure of non-lean-witnesses}
Before going to the proofs of \Cref{lem:leanalglarge,lem:leanalgsmall}, we prove that non-$k$-lean-witnesses for bags of size $\ge s$ can be assumed to have several extra properties.
We start with an observation about vertices of degree $\ge s$.

\begin{lemma}
\label{lem:degnodecenter}
Let $s \ge k \ge 1$ be integers, $G$ an $(s,k)$-unbreakable graph, $\Tc = (T,\bag)$ a tree decomposition of $G$ with $\adhsize(\Tc)<k$, and $t \in V(T)$ a node with $|\bag(t)| \ge s$.
Then, $\bag(t)$ contains all vertices of degree $\ge s$ of $G$.
\end{lemma}
\begin{proof}
Suppose $v \in V(G)$ is a vertex of degree $\ge s$ that is not in $\bag(t)$.
Let $t_v \in V(T)$ be the node closest to $t$ so that $v \in \bag(t_v)$, and $p \in N_T(t_v)$ the neighbor of $t_v$ on the $(t_v,t)$-path.
Now, the adhesion at $t_v p$ gives a separation $(A,B)$ of order $<k$ with $v \in A \setminus B$ and $\bag(t) \subseteq B$.
Because $v \in A \setminus B$ we have $|A| \ge s$, and because $\bag(t) \subseteq B$ we have $|B| \ge s$, so $(A,B)$ would contradict the $(s,k)$-unbreakability of $G$.
\end{proof}

Then we show that if there is a non-$k$-lean-witness of form $(A,B,t,t)$, where $t$ is node with large bag, then there is such a non-$k$-lean-witness with many extra properties.
This will be the key for finding non-$k$-lean-witnesses for large bags.

\begin{lemma}
\label{lem:wellstructwitness}
Let $s \ge k \ge 1$ be integers, $G$ an $(s,k)$-unbreakable graph, $\Tc = (T,\bag)$ a tree decomposition of $G$ with $\adhsize(\Tc)<k$, and $t \in V(T)$ so that $|\bag(t)| \ge s$.
Let also $L \subseteq V(G)$ be the vertices of $G$ with degree $<s$.
If there is a non-$k$-lean-witness of form $(A,B,t,t)$ for $\Tc$, then there is such non-$k$-lean-witness so that
\begin{enumerate}
\item $|A| < s$, \label{lem:wellstructwitness:cond1}
\item $A \setminus B \subseteq L$, \label{lem:wellstructwitness:cond2}
\item $A \setminus B \neq \emptyset$, \label{lem:wellstructwitness:cond25}
\item $N(A \setminus B) = A \cap B$, \label{lem:wellstructwitness:cond3}
\item $G[A \cap L]$ is connected, and \label{lem:wellstructwitness:cond4}
\item $G[A \setminus B]$ has at most $k$ connected components. \label{lem:wellstructwitness:cond5}
\end{enumerate}
\end{lemma}
\begin{proof}
Let $(A,B,t,t)$ be a non-$k$-lean-witness for $\Tc$ that minimizes $|A|$ among all such non-$k$-lean-witnesses.
Because $A$ and $B$ can be swapped if $|A| > |B|$, we have that $|A| < s$ (\Cref{lem:wellstructwitness:cond1}).
We have $A \setminus B \subseteq L$ (\Cref{lem:wellstructwitness:cond2}), because if $A \setminus B$ would contain a vertex of degree $\ge s$, then $|A| \ge s$.
The set $A \setminus B$ is non-empty (\Cref{lem:wellstructwitness:cond25}) because otherwise $|\bag(t) \cap A| \le |A \cap B|$.

For \Cref{lem:wellstructwitness:cond3}, it must be that $N(A \setminus B) \subseteq A \cap B$, so assume that $N(A \setminus B) \subsetneq A \cap B$.
We observe that in this case $(N[A \setminus B], B,t,t)$ is also a non-$k$-lean-witness for $(T,\bag)$, which by $|N[A \setminus B]|<|A|$ contradicts the choice of $(A,B,t,t)$.

For \Cref{lem:wellstructwitness:cond4}, suppose that $G[A \cap L]$ is not connected.
Define the \emph{utility} of a subset $A' \subseteq A$ as $|A' \cap \bag(t)|-|A' \cap B|$.
Because $(A,B,t,t)$ is a non-$k$-lean-witness, the utility of $A$ is positive.
Then, by $A \setminus B \subseteq L$ and \Cref{lem:degnodecenter}, all vertices in $A \setminus L$ must be in $B$ and $\bag(t)$, implying that the utility of $A \setminus L$ is $0$.
Therefore, the utility of $A \cap L$ is positive.

Let $C \in \cc(G[A \cap L])$ be the connected component of $G[A \cap L]$ with the smallest utility, and observe that because the utilities of $A$ and $A \cap L$ are positive, the utility of $A \setminus C$ is positive.

\begin{claim}
$(A \setminus C, B \cup C)$ is a vertex cut of $G$.
\end{claim}
\begin{claimproof}
We must prove that there is no edge between $A \setminus (B \cup C)$ and $(B \setminus A) \cup C$.
Because $(A,B)$ is a vertex cut, it suffices to prove that there is no edge between $A \setminus (B \cup C)$ and $C$.
Because $A \setminus B \subseteq L$ and $C \subseteq L$, we have that $A \setminus (B \cup C) \subseteq A \cap L$ and $C \subseteq A \cap L$.
Thus, if such an edge existed, then $C$ would not be a connected component of $G[A \cap L]$.
\end{claimproof}

By the fact that the utility of $A \setminus C$ is positive, we deduce that
\begin{align*}
|(A \setminus C) \cap \bag(t)| - |(A \setminus C) \cap (B \cup C)| = |(A \setminus C) \cap \bag(t)| - |(A \setminus C) \cap B| > 0,
\end{align*}
implying that $|(A \setminus C) \cap \bag(t)| > |(A \setminus C) \cap (B \cup C)|$.
We also have that 
\[|(B \cup C) \cap \bag(t)| \ge |B \cap \bag(t)| > |A \cap B| \ge |(A \setminus C) \cap (B \cup C)|.\]
These imply that $(A \setminus C, B \cup C,t,t)$ is a non-$k$-lean-witness for $(T,\bag)$, which contradicts the choice of $(A,B,t,t)$.
Therefore $G[A \cap L]$ is connected.

Lastly, for \Cref{lem:wellstructwitness:cond5} we will prove that $G[A \setminus B]$ has at most $k$ connected components, so for the sake of contradiction suppose $G[A \setminus B]$ has more than $k$ connected components.
Let $C \in \cc(G[A \setminus B])$ be a component that minimizes $|C \cap \bag(t)|$.
Now, $(A \setminus C, B \cup C)$ is a vertex cut of $G$, and 
\[|(A \setminus C) \cap \bag(t)| \ge k > |A \cap B| = |(A \setminus C) \cap (B \cup C)|.\]
Therefore, $(A \setminus C, B \cup C, t, t)$ is a non-$k$-lean-witness for $(T,\bag)$, which contradicts the choice of $(A,B,t,t)$.
\end{proof}

We call a non-$k$-lean-witness of form $(A,B,t,t)$ that satisfies \Cref{lem:wellstructwitness:cond1,lem:wellstructwitness:cond2,lem:wellstructwitness:cond25,lem:wellstructwitness:cond3,lem:wellstructwitness:cond4,lem:wellstructwitness:cond5} of \Cref{lem:wellstructwitness} a \emph{well-structured} non-$k$-lean-witness.
We then give our algorithm for finding well-structured non-$k$-lean-witnesses for large bags.

\begin{lemma}
\label{lem:compwelstrucleanwit}
Let $s \ge k \ge 1$ be integers, $G$ an $(s,k)$-unbreakable graph, $\Tc = (T,\bag)$ a tree decomposition of $G$ with $\adhsize(\Tc)<k$, and $t \in V(T)$ so that $|\bag(t)| \ge s$.
Suppose a representation of $G$ is already stored, and moreover a representation of $(T,\bag)$, where one can in $\OO(1)$ time query whether a given vertex $v \in V(G)$ is in $\bag(t)$, is stored.
Then, there is an algorithm that, given a vertex $v \in V(G)$ with degree $< s$, in time $s^{\OO(k)}$ either
\begin{itemize}
\item returns a set $C \subseteq V(G)$ so that $v \in N[C]$ and $(N[C], V(G) \setminus C,t,t)$ is a well-structured non-$k$-lean-witness, or
\item correctly concludes that no well-structured non-$k$-lean-witness $(A,B,t,t)$ with $v \in A$ exist.
\end{itemize}
\end{lemma}
\begin{proof}
Let $L \subseteq V(G)$ be the set of vertices of $G$ of degree $<s$.
We will first enumerate a set of $s^{\OO(k)}$ candidates for the pair of sets $(A \cap L, A \cap B \cap L)$, where $A$ and $B$ are as in \Cref{lem:wellstructwitness}.

In particular, we will enumerate all pairs $D,R \subseteq V(G)$ that satisfy
\begin{enumerate}
\item $R \subseteq D \subseteq L$, \label{lem:compwelstrucleanwit:enumcond1}
\item $v \in D$, \label{lem:compwelstrucleanwit:enumcond2}
\item $|D|<s$, $|R|<k$, \label{lem:compwelstrucleanwit:enumcond3}
\item $G[D]$ is connected, \label{lem:compwelstrucleanwit:enumcond4}
\item $G[D \setminus R]$ has at most $k$ connected components, and \label{lem:compwelstrucleanwit:enumcond5}
\item there are no edges between $D \setminus R$ and $L \setminus D$. \label{lem:compwelstrucleanwit:enumcond6}
\end{enumerate}

Note that when setting $D = A \cap L$ and $R = A \cap B \cap L$, any well-structured non-$k$-lean-witness $(A,B,t,t)$ with $v \in A$ must satisfy \Cref{lem:compwelstrucleanwit:enumcond1,lem:compwelstrucleanwit:enumcond2,lem:compwelstrucleanwit:enumcond3,lem:compwelstrucleanwit:enumcond4,lem:compwelstrucleanwit:enumcond5,lem:compwelstrucleanwit:enumcond6}.

Let us then describe the enumeration algorithm.
The algorithm is a recursive branching algorithm.
A recursive call of the algorithm takes as inputs two sets $D',R' \subseteq V(G)$, satisfying $R' \subseteq D' \subseteq L$, $v \in D'$, $|D'| \le s$, $|R'| \le k$, $G[D']$ is connected, and $G[D' \setminus R']$ has at most $k$ connected components.
It outputs all pairs $D,R$ satisfying \Cref{lem:compwelstrucleanwit:enumcond1,lem:compwelstrucleanwit:enumcond2,lem:compwelstrucleanwit:enumcond3,lem:compwelstrucleanwit:enumcond4,lem:compwelstrucleanwit:enumcond5,lem:compwelstrucleanwit:enumcond6} so that $D \supseteq D'$ and $R \cap D' = R'$.

The branching works as follows.
First, if $|D'| \ge s$ or $|R'| \ge k$, we return the empty set.

Second, if there exists a vertex $u \in D' \setminus R'$ with a neighbor $w \in N(u) \cap (L \setminus D')$, we must extend $D'$ to include $w$.
Therefore, we call the algorithm recursively with the pairs $(D' \cup \{w\}, R')$ and $(D' \cup \{w\}, R' \cup \{w\})$, and return the union of the obtained lists.
We will call this ``branching of the first type''.
Note that both recursive calls indeed satisfy the preconditions, in particular, they do not increase the number of connected components of $G[D' \setminus R']$.

If the above two cases do not hold, the pair $D',R'$ itself satisfies \Cref{lem:compwelstrucleanwit:enumcond1,lem:compwelstrucleanwit:enumcond2,lem:compwelstrucleanwit:enumcond3,lem:compwelstrucleanwit:enumcond4,lem:compwelstrucleanwit:enumcond5,lem:compwelstrucleanwit:enumcond6}, and the output thus must include at least the pair $D',R'$.
However, it is not necessarily correct to return only the pair $D',R'$, and we must again branch.
Any desired pair $D,R$ with $D \supsetneq D'$ must contain a vertex $w \in N(D') \cap L$.
However, we note that if the number of connected components of $G[D' \setminus R']$ is $k$, then if $w \in D \setminus R$, the number of connected components of $G[D \setminus R]$ must be more than $k$.
Therefore, if the number of connected components of $G[D' \setminus R']$ is less than $k$, we branch to $2 \cdot |N(D') \cap L| \le 2 s^2$ directions, calling the algorithm recursively with $(D' \cup \{w\}, R')$ and $(D' \cup \{w\}, R' \cup \{w\})$ for each $w \in N(D') \cap L$, and returning the union of the obtained lists and the pair $D',R'$.
Note that in the first case, the number of connected components of $G[D' \cup \{w\} \setminus R']$ must be more than that of of $G[D' \setminus R']$.
If the number of connected components of $G[D' \setminus R']$ is $k$, we branch to $|N(D') \cap L| \le s^2$ directions, calling the algorithm recursively with $(D' \cup \{w\}, R' \cup \{w\})$ for each $w \in N(D') \cap L$, and returning the union of the obtained lists and the pair $D',R'$.
We will call this ``branching of the second type''.

We claim that the size of the resulting recursion tree is $s^{\OO(k)}$.
Consider a root-leaf path in the recursion tree.
It has at most $s$ edges, as each branch increases the size of $D'$.
It has at most $2k$ branchings of the second type, as each such branching increases either the number of connected components of $G[D' \setminus R']$ (which is never decreased), or the size of $R'$.
Of the branchings of the first type, at most $k$ of them increase the size of $R'$, while the rest just increase the size of $D'$.
Therefore, the root-leaf path is characterized by the choice of at most $2k$ edges of the second type, for each of them the choice of one of the at most $2 \cdot s^2$ possible directions, and then the choice of at most $k$ branchings of the first type that increase $R'$.
This means that the number of leaves of the branching tree is at most

\[(s+1)^{2k} \cdot (2 \cdot s^2)^{2k} \cdot (s+1)^{k} = s^{\OO(k)},\]
which is also a bound for the overall size of the branching tree and the size of the outputted list.
As $|D'| \le s$ and all vertices in $L$ have degree $<s$, it is easy to implement one recursive call in time $s^{\OO(k)}$, so the overall running time of this enumeration is also $s^{\OO(k)}$.

We then iterate through all of the produced candidates for the pair $(A \cap L, A \cap B \cap L)$, and for each determine if there exists a well-structured non-$k$-lean-witness $(A,B,t,t)$ corresponding to it.

We claim that the sets $A \cap L$ and $A \cap B \cap L$ uniquely determine the well-structured non-$k$-lean-witness $(A,B,t,t)$ if one exists.
In particular, because $B$ must contain all vertices in $V(G) \setminus L$ by $|A| < s$, we have that $B = (A \cap B \cap L) \cup (L \setminus A) \cup (V(G) \setminus L)$.
Then, we must set $A = (A \cap L) \cup N((A \cap L) \setminus (A \cap B \cap L))$.

Clearly, this set $A$, along with the set $A \cap B$, can be constructed in $s^{\OO(1)}$ time from $A \cap L$ and $A \cap B \cap L$.
Having constructed $A$ and $A \cap B$ explicitly, we can check the conditions of \Cref{lem:wellstructwitness} in $s^{\OO(1)}$ time, and if they are satisfied, return the set $C = A \setminus B$.
\end{proof}

\subsection{Big-star-lean tree decompositions}
\label{subsec:bignodelean}
This subsection is dedicated to the proof of \Cref{lem:leanalglarge}, which we now re-state.

\leanalglarge*

Let us first introduce some definitions.
Let $(T,\bag)$ be a star decomposition with the root $t = \troot(T)$.
We say that a non-$k$-lean-witness $(A,B,t,t)$ is \emph{$A$-linked} if $\flow(A \cap B, A \cap \bag(t)) = |A \cap B|$.
Let us show that we can easily turn a non-$k$-lean-witness outputted by the algorithm of \Cref{lem:compwelstrucleanwit} into an $A$-linked non-$k$-lean-witness.

\begin{lemma}
\label{lem:alinkedwitness}
Let $G$ be a graph whose representation is already stored, and $\Tc = (T,\bag)$ a star decomposition of $G$ with root $t = \troot(T)$, and assume a representation of $\Tc$, where one can in $\OO(1)$ time query whether a vertex $v \in V(G)$ is in $\bag(t)$, is stored.
Then, there is an algorithm that, given a set $C \subseteq V(G)$ so that $(N[C], V(G) \setminus C, t, t)$ is a non-$k$-lean-witness, in time $(|C|+k)^{\OO(1)}$ returns sets $S \subseteq A \subseteq N[C]$ so that $(A, V(G) \setminus (A \setminus S), t, t)$ is an $A$-linked non-$k$-lean-witness.
\end{lemma}
\begin{proof}
If $(X,Y)$ is a a vertex cut in $G[N[C]]$ of minimum order so that $\bag(t) \cap N[C] \subseteq X$ and $N(C) \subseteq Y$, then $(X, V(G) \setminus (X \setminus Y), t, t)$ is an $A$-linked non-$k$-lean-witness, because then $\flow(X \cap Y, X \cap \bag(t)) = |X \cap Y|$.
To find such minimum vertex cut, we first construct in time $(|C|+k)^{\OO(1)}$ the graph $G[N[C]]$ without the edges between vertices in $N(C)$, as they are irrelevant, and then use the Ford-Fulkerson algorithm.
\end{proof}

The idea of the proof of \Cref{lem:leanalglarge} is to repeatedly improve a star decomposition of $G$ by \emph{refinement operations}.
We now define this operation.

Let $\Tc = (T,\bag)$ be a star decomposition of a graph $G$ with root $t = \troot(T)$. 
We define that the \emph{refinement} of $(T,\bag)$ with a non-$k$-lean-witness $(A,B,t,t)$ is the star decomposition $(T', \bag')$, where $T'$ is obtained from $T$ by adding an additional leaf $y$ adjacent to $t$, and the $\bag'$ function is constructed as
\[\bag'(t) = (\bag(t) \setminus A) \cup (A \cap B),\]
for the root $t$,
\[\bag'(x) = \bag(x) \cap B,\]
for nodes $x \in V(T) \setminus \{t\}$ of $T$, and
\[\bag'(y) = A,\]
for the new leaf $y$.

\begin{lemma}
\label{lem:refinresultingdecomp2}
$(T',\bag')$ is a star decomposition of $G$.
\end{lemma}
\begin{proof}
Clearly, other nodes of $T$ than $t$ have degree $1$.
The edge condition of tree decompositions holds for edges between vertices in $B$, as no vertex in $B$ is removed from any bag.
The edge condition also holds for edges between vertices in $A$ as there is a bag containing $A$, and therefore as $(A,B)$ is a vertex cut, it holds for all edges.

The vertex condition holds for vertices in $B \setminus A$, because the occurrences of vertices in $B \setminus A$ are the same as in $(T,\bag)$.
We then observe that vertices in $A \setminus B$ are only in the new bag $\bag'(y)$, and vertices in $A \cap B$ are in $\bag'(y)$ and $\bag'(t)$, and thus the fact that they are in $\bag'(t)$ implies that the vertex condition must hold for them.
\end{proof}

Then we show that if we use an $A$-linked non-$k$-lean-witness, then the refinement operation does not increase adhesion size above $k-1$.
We also observe that the refinement operation makes progress by decreasing the number of vertices of the root bag.

\begin{lemma}
\label{lem:refinresultingdecomp}
Let $\Tc = (T,\bag)$ be a star decomposition rooted at $t = \troot(\Tc)$ so that $\adhsize(\Tc) < k$, and let $(A,B,t,t)$ be an $A$-linked non-$k$-lean-witness for $\Tc$.
If $\Tc' = (T',\bag')$ the refinement of $\Tc$ with $(A,B,t,t)$, then
\begin{itemize}
\item $|\bag'(t)|<|\bag(t)|$ and
\item $\adhsize(\Tc') < k$.
\end{itemize}
\end{lemma}
\begin{proof}
Let us first prove that $|\bag'(t)|<|\bag(t)|$.
We have that $\bag'(t) = (\bag(t) \setminus A) \cup (A \cap B)$, and by the definition of a non-$k$-lean-witness, we have that $|A \cap \bag(t)| > |A \cap B|$, implying $|\bag'(t)|<|\bag(t)|$.

Then we prove that the adhesion size of $(T',\bag')$ is $<k$.
First, because $\bag'(y) = A$ and $\bag'(t) \cap A = A \cap B$, we have $\bag'(y) \cap \bag'(t) = A \cap B$, which has size $<k$.

Then, for the sake of contradiction, assume that there is a leaf node $x \in V(T) \setminus \{t,y\}$ so that $|\bag'(x) \cap \bag'(t)| > |\bag(x) \cap \bag(t)|$.
Denote also $(A_x,B_x) = (V(G) \setminus (\bag(x) \setminus \bag(t)), \bag(x))$, and note that $A_x \cap B_x = \bag(t) \cap \bag(x)$.
Because both $(A,B)$ and $(A_x,B_x)$ are vertex cuts, also $(A_x \cap A, B_x \cup B)$ is a vertex cut.
\begin{claim}
The order of $(A_x \cap A, B_x \cup B)$ is less than the order of $(A,B)$.
\end{claim}
\begin{claimproof}
For the sake of contradiction, assume that the order of $(A_x \cap A, B_x \cup B)$ is at least the order of $(A,B)$.
By submodularity, then the order of the vertex cut $(A_x \cup A, B_x \cap B)$ is at most the order of $(A_x, B_x)$.
In other words, then
\begin{align*}
|(A_x \cup A) \cap B_x \cap B| \le |A_x \cap B_x|\\
|(\bag(t) \cup A) \cap \bag(x) \cap B| \le |A_x \cap B_x|\\
|((\bag(t) \setminus A) \cup (A \cap B)) \cap \bag(x) \cap B| \le |A_x \cap B_x|\\
|\bag'(t) \cap \bag'(x)| \le |\bag(t) \cap \bag(x)|
\end{align*}
However, this contradicts the assumption that $|\bag'(x) \cap \bag'(t)| > |\bag(x) \cap \bag(t)|$.
\end{claimproof}

However, $A \cap \bag(t) \subseteq A_x \cap A$, and $A \cap B \subseteq B_x \cup B$, so the vertex cut $(A_x \cap A, B_x \cup B)$ contradicts that $(A,B,t,t)$ is $A$-linked.
\end{proof}

Now we are ready to prove \Cref{lem:leanalglarge} by giving an algorithm that repeatedly applies the refinement operation.

\leanalglarge*
\begin{proof}
We implement an iterative improvement process that maintains a star decomposition $\Tc = (T,\bag)$ of $G$ with root $t = \troot(T)$.
During the process, $\Tc$ is not represented by the usual representation of tree decompositions.
Instead of storing $\bag(t)$ as a linked list, it is stored by an array indexed by $V(G)$ that tells for each vertex whether it is in $\bag(t)$.
This allows to insert and delete vertices from $\bag(t)$ in $\OO(1)$ time, and query whether a vertex is in $\bag(t)$ in $\OO(1)$ time.
Otherwise, the representation of $\Tc$ is as usual, i.e., $T$ is explicitly represented and the non-root bags are stored as linked lists with the nodes of $T$ storing pointers pointing to their bags.
In addition to this, for each vertex $v \in V(G)$ we store a linked list listing the non-root nodes of $T$ in whose bags $v$ is in.

The invariants that our process maintains will be that $\adhsize(\Tc) < k$ and all non-root bags of $\Tc$ have size $<s$.
By \Cref{lem:bigbagexists}, this implies that $|\bag(t)| \ge 2s$.

In our process, we maintain a queue $Q$ storing vertices of $G$ with degree $<s$.
This queue will satisfy the invariant that if there exists a well-structured non-$k$-lean-witness of form $(A,B,t,t)$, then there exists $v \in Q$ with $v \in A$.
Initially, we set $Q$ to contain all vertices of degree $<s$, which satisfies the invariant by \Cref{lem:wellstructwitness}.

Now, our algorithm works as follows.
As long as $Q$ is non-empty, we let $v \in Q$ be the front vertex in $Q$, and use the algorithm of \Cref{lem:compwelstrucleanwit} to in time $s^{\OO(k)}$ either (1) return a set $C \subseteq V(G)$ so that $(N[C], V(G) \setminus C, t, t)$ is a well-structured non-$k$-lean-witness, or (2) conclude that no well-structured non-$k$-lean-witness of form $(A,B,t,t)$ with $v \in A$ exists.
In the latter case, we simply pop $v$ from $Q$.

In the former case, we use the algorithm of \Cref{lem:alinkedwitness} to in time $s^{\OO(1)}$ obtain sets $S \subseteq A \subseteq N[C]$ so that $(A, V(G) \setminus (A \setminus S), t, t)$ is an $A$-linked non-$k$-lean-witness.
Denote the refinement of $\Tc$ with $(A, V(G) \setminus (A \setminus S), t, t)$ by $\Tc' = (T',\bag')$.
By \Cref{lem:refinresultingdecomp2,lem:refinresultingdecomp}, $\Tc'$ satisfies the invariants, and furthermore, $|\bag'(t)| < |\bag(t)|$.

We transform $\Tc$ into $\Tc'$ by (1) removing vertices in $A \setminus S$ from all bags, (2) inserting vertices in $S$ to $\bag(t)$, and (3) inserting a new non-root node $y$ with $\bag'(y) = A$.
Clearly, (2) and (3) can be implemented in time $\OO(s)$.
The operation (1) runs in time $\OO(s)$ times the number of bags containing vertices from $A \setminus S$ (as all non-root bags have size less than $s$).
Across all refinements, this takes in total $\OO(|V(G)| \cdot s^2)$ time, as each iteration increases the sum of the bag sizes by at most $\OO(s)$, but by \Cref{lem:refinresultingdecomp}, $|\bag(t)|$ decreases in each refinement, so there are at most $|V(G)|$ iterations.

Finally, we need to update $Q$ to satisfy its invariant.
We add all vertices in $A$ with degree $<s$ to $Q$.
To prove that $Q$ now satisfies the invariant, suppose otherwise, and let $(A',B',t,t)$ be a well-structured non-$k$-lean-witness of $\Tc'$ so that $A'$ is disjoint with $Q$.
Now, $A' \setminus B'$ must be disjoint with $A$, and $A' \cap B'$ intersects $A$ only in vertices of degree $\ge s$.
Because vertices of degree $\ge s$ are in $\bag'(t)$ and $\bag(t)$ (by \Cref{lem:bigbagexists,lem:degnodecenter}), it follows that $A' \cap \bag'(t) = A' \cap \bag(t)$.
From this, it follows that $(A',B',t,t)$ is in fact also a well-structured non-$k$-lean-witness of $\Tc$, which would contradict the invariant of $Q$ at the previous timestep.

As in each refinement we add $\OO(s)$ vertices to $Q$, and processing one vertex from $Q$ is done in time $s^{\OO(k)}$, and there are at most $|V(G)|$ refinements, the algorithm runs in total $\OO(\|G\|) + |V(G)| \cdot s^{\OO(k)}$ time.
\end{proof}

\subsection{Lean tree decompositions of small graphs}
\label{subsec:leandecompsmallgraph}
This subsection is dedicated to the proof of \Cref{lem:leanalgsmall}, which we re-state now.

\leanalgsmall*

The proof of \Cref{lem:leanalgsmall} will be an algorithmic version of the proof of Bellenbaum and Diestel~\cite{bellenbaum2002two}.
A similar algorithmic procedure, but only for obtaining an unbreakable tree decomposition, was given before by~\cite{DBLP:journals/talg/CyganKLPPSW21}.

We start by defining a special type of a non-$k$-lean-witness, following an analogous definition in~\cite{bellenbaum2002two}.

Let $G$ be a graph and $\Tc = (T,\bag)$ a tree decomposition of $G$.
For two nodes $t_1, t_2 \in V(T)$ and a vertex $v \in V(G)$, let us denote by $d_{t_1,t_2}(v)$ the distance in $T$ from the the unique $t_1$-$t_2$-path to the closest node of $T$ whose bag contains $v$.
Note that if $v$ is in a bag whose node is in the $t_1$-$t_2$-path, then $d_{t_1,t_2}(v) = 0$, and otherwise a node $t_v$ with $v \in \bag(t_v)$ minimizing the distance to the $t_1$-$t_2$-path is unique.

We say that a non-$k$-lean-witness $(A,B,t_1,t_2)$ for $\Tc$ is a \emph{minimal non-$k$-lean-witness} for $\Tc$ if among all non-$k$-lean-witnesses it
\begin{enumerate}
\item minimizes the distance between $t_1$ and $t_2$ in $T$, and\label{enu:minkleancond1}
\item subject to \Cref{enu:minkleancond1}, among all $(\bag(t_1) \cap A, \bag(t_2) \cap B)$-separators, $A \cap B$ minimizes $|A \cap B|$, and\label{enu:minkleancond2}
\item subject to \Cref{enu:minkleancond1,enu:minkleancond2}, among all $(\bag(t_1) \cap A, \bag(t_2) \cap B)$-separators, $A \cap B$ minimizes $\sum_{v \in A \cap B} d_{t_1,t_2}(v)$.\label{enu:minkleancond3}
\end{enumerate}

We then give an algorithm for finding a minimal non-$k$-lean-witness whenever $\Tc$ is not $k$-lean.

\begin{lemma}
\label{lem:smallleanfindminiwitness}
There is an algorithm that, given a graph $G$, integers $s \ge k \ge 1$ so that $G$ is $(s,k)$-unbreakable, and a tree decomposition $\Tc$ of $G$ with $\adhsize(\Tc) < k$, in time $s^{\OO(k)} \cdot \|G\|^{\OO(1)} \cdot \|\Tc\|^{\OO(1)}$ either returns a minimal non-$k$-lean-witness for $\Tc$ or concludes that $\Tc$ is $k$-lean.
\end{lemma}
\begin{proof}
The algorithm boils down to two parts, first finding a non-$k$-lean-witness $(A,B,t_1,t_2)$ minimizing the distance between $t_1$ and $t_2$ in $T$, and then transforming it so that it satisfies the conditions of \Cref{enu:minkleancond2,enu:minkleancond3} in the definition of a minimal non-$k$-lean-witness.

\paragraph{Finding a non-$k$-lean-witness.}
We first focus on the problem of finding a non-$k$-lean-witness $(A,B,t_1,t_2)$ that minimizes the distance between $t_1$ and $t_2$.
When allowing a factor of $\|\Tc\|^{\OO(1)}$ in the running time, this boils down to the problem of finding a non-$k$-lean-witness $(A,B,t_1,t_2)$ when $t_1$ and $t_2$ are given, or concluding that none exist.
By \Cref{lem:atmostonelargebag}, $\Tc$ has at most one node $t$ with $|\bag(t)| \ge 2s$, and by \Cref{lem:leanwitstruct} we can assume that either $t_1 = t_2 = t$ or $t_1 \neq t$ and $t_2 \neq t$.

When $t_1 = t_2 = t$ and $|\bag(t)| \ge 2s$, we can use \Cref{lem:compwelstrucleanwit} to in time $s^{\OO(k)} \cdot \|G\|^{\OO(1)}$ find a non-$k$-lean-witness of form $(A,B,t,t)$.

Assume then that $|\bag(t_1)| < 2s$, $|\bag(t_2)| < 2s$, and the minimum size of an $(t_1,t_2)$-adhesion is $\alpha$. (Let $\alpha = \infty$ if $t_1 = t_2$.)
We iterate through all subsets $X_1 \subseteq \bag(t_1)$ and $X_2 \subseteq \bag(t_2)$ of size $|X_1| = |X_2| \le k$, and compute $\flow(X_1,X_2)$ with the Ford-Fulkerson algorithm.
If $\flow(X_1,X_2) \ge \min(|X_1|, \alpha)$ for all such pairs, then there is no non-$k$-lean-witness of form $(A,B,t_1,t_2)$.
Otherwise, we find a non-$k$-lean-witness $(A,B,t_1,t_2)$ with $X_1 \subseteq A$, $X_2 \subseteq B$, and $|A \cap B| = \flow(X_1,X_2)$.
This procedure runs in time $(2s)^{2k} \cdot \|G\|^{\OO(1)} = s^{\OO(k)} \cdot \|G\|^{\OO(1)}$.

\paragraph{Minimalizing the witness.}
We then transform the non-$k$-lean-witness $(A,B,t_1,t_2)$ into a minimal non-$k$-lean-witness.
The first part already guarantees that it satisfies \Cref{enu:minkleancond1} of the definition of a minimal non-$k$-lean-witness.

We assign for each vertex $v \in V(G)$ a weight $w(v) = |V(T)| \cdot |V(G)| + d_{t_1,t_2}(v)$.
Then, we compute a $(\bag(t_1) \cap A, \bag(t_2) \cap B)$-separator $S$ minimizing $\sum_{v \in S} w(v)$.
Because the weights are bounded by $\|G\|^{\OO(1)} \cdot \|\Tc\|^{\OO(1)}$, this can be done in $\|G\|^{\OO(1)} \cdot \|\Tc\|^{\OO(1)}$ time with the Ford-Fulkerson algorithm.
By the choice of the weights, such a separator $S$ corresponds to an $(\bag(t_1) \cap A, \bag(t_2) \cap B)$-separator that minimizes $|S|$, and subject to that minimizes $\sum_{v \in S} d_{t_1,t_2}(v)$.

Now we let $(A',B')$ be the separation with $\bag(t_1) \cap A \subseteq A'$, $\bag(t_2) \cap B \subseteq B'$, and $A' \cap B' = S$.
It follows that $(A',B',t_1,t_2)$ is a minimal non-$k$-lean-witness.
\end{proof}

We then define the \emph{refinement} of a tree decomposition $\Tc = (T,\bag)$ with a minimal non-$k$-lean-witness $(A,B,t_1,t_2)$.
This will be similar to the refinement operation of the previous subsection, in particular the refinement of the previous subsection is a simplified version of the refinement here.
From now, instead of denoting a non-$k$-lean-witness by $(A,B,t_1,t_2)$, we will denote a non-$k$-lean-witness by $(C_1,C_2,t_1,t_2)$, as this will make stating the definitions easier.
We also point out that in what follows, we will always work with minimal non-$k$-lean-witnesses instead of arbitrary non-$k$-lean-witnesses.

When $G$ is a graph, $\Tc = (T,\bag)$ a rooted tree decomposition of $G$, and $v \in V(G)$, we define the \emph{forget-node} of $v$ to be the closest node $t \in V(T)$ to the root with $v \in \bag(t)$.
We denote the forget-node of $v$ by $\forget_{\Tc}(v)$.

For the refinement, we construct a rooted tree decomposition $\Tc^i = (T^i,\bag^i)$ of $G[C_i]$ for each $i \in [2]$ as follows.
Denote $j = 3-i$ and $S = C_1 \cap C_2$.
Consider the tree decomposition $\Tc = (T,\bag)$ as rooted at the node $t_j$ (note $t_j$, not $t_i$).
Then, define for all $t \in V(T)$ that
\[\pull^i(t) = \{v \in S : \forget_{\Tc}(v) \text{ is a strict descendant of } t \text{ in } T\}.\]
Note that in this definition it is important that we consider $\Tc$ to be rooted at $t_j$.
Then, we let $T^i = T$, and set
\[\bag^i(t) = (\bag(t) \cap C_i) \cup \pull^i(t)\]
for all $t \in V(T)$.
We then prove the basic properties of $\Tc^i$.

\begin{lemma}
$\Tc^i = (T^i,\bag^i)$ is a tree decomposition of $G[C_i]$, and $S \subseteq \bag^i(t_j)$.
\end{lemma}
\begin{proof}
First, to prove that $S \subseteq \bag^i(t_j)$, we observe that because $t_j$ is the root, for all $v \in S$ either $v \in \bag(t_j)$, implying that $v \in \bag(t_j) \cap C_i$, or the forget-node of $v$ is a strict descendant of $t_j$, implying that $v \in \pull^i(t_j)$.

The edge condition of tree decompositions follows from the fact that we do not remove any occurrences of vertices in $C_i$ from the bags of $\Tc^i$ compared to $\Tc$.
Also, the occurrences of vertices in $C_i \setminus S$ are not altered at all compared to $\Tc$, which implies the vertex condition for them.
It remains to prove the vertex condition for vertices in $S$.

For this, we observe that if $v \in \bag^i(t) \cap S$, and $p$ is the parent of $t$, then either $v \in \bag(p)$ or $v \in \pull^i(p)$, implying that $v \in \bag^i(p)$.
Therefore, the occurrences of each $v \in S$ in $\Tc^i$ form a prefix of $T^i$, implying that the vertex condition is satisfied for vertices in $S$.
\end{proof}

Then we bound the adhesion size of $\Tc^i$.

\begin{lemma}
It holds that $\adhsize(\Tc^i) \le \adhsize(\Tc)$.
\end{lemma}
\begin{proof}
We prove this for $i=1$, but by symmetry this implies it also for $i=2$.
Let $xy \in E(T)$, and denote by $A = \adh_{\Tc}(xy)$ the adhesion at $xy$ in $\Tc$ and by $A^1 = \adh_{\Tc^1}(x y)$ the adhesion at $xy$ in $\Tc^1$.
We will prove that $|A^1| \le |A|$.

Consider $\Tc$ as rooted at $t_2$, and assume that $y$ is the parent of $x$.
Observe that $\pull^1(y) \subseteq \bag^1(x)$ and $\pull^1(x) \subseteq \pull^1(y)$.
Using this, we observe that
\[A^1 = (A \cap C_1) \cup \pull^1(y).\]
Now, to prove that $|A^1| \le |A|$, it suffices to prove that $|\pull^1(y)| \le |A \setminus C_1|$.

Denote $X_1 = C_1 \cap \bag(t_1)$ and $X_2 = C_2 \cap \bag(t_2)$.
Because $S$ is a minimum-size $(X_1,X_2)$-separator, by Menger's theorem there exists a collection $\paths = P_1,\ldots,P_{|S|}$ of $|S|$ vertex-disjoint $(S,X_2)$-paths.
Because $(C_1,C_2)$ is a separation, $S = C_1 \cap C_2$, and $X_2 \subseteq C_2$, these paths are disjoint from $C_1 \setminus C_2$, and in particular, all of their vertices except the first ones are in $C_2 \setminus C_1$.
Because $\pull^1(y) \subseteq S$, we can select a subcollection $\paths' \subseteq \paths$ of $|\pull^1(y)|$ vertex-disjoint $(\pull^1(y),X_2)$-paths.

Each path $P \in \paths'$ intersects a bag of a descendant of $x$ because the forget-nodes of vertices in $\pull^1(y)$ are descendants of $x$.
Also, each path $P \in \paths'$ intersects the root-bag $\bag(t_2)$ because $X_2 \subseteq \bag(t_2)$.
Therefore, as $A$ separates $\pull^1(y)$ from $\bag(t_2)$, we have that each $P \in \paths'$ intersects $A$.
Furthermore, as $A$ and $\pull^1(y)$ are disjoint, and each $P \in \paths'$ intersects $C_1$ only in $\pull^1(y)$, we get that each $P \in \paths'$ intersects $A \setminus C_1$.
It follows that $|A \setminus C_1| \ge |\paths'| = |\pull^1(y)|$.
\end{proof}

Then, we define that the \emph{pre-refinement} $\Tc' = (T',\bag')$ of $\Tc$ with $(C_1,C_2,t_1,t_2)$ is the tree decomposition obtained by taking the disjoint union of $\Tc^1$ and $\Tc^2$, adding a new node $r$ with a bag $\bag'(r) = S$, and inserting an edge between the root $t_2$ of $\Tc^1$ and the root $t_1$ of $\Tc^2$.

\begin{lemma}
The pre-refinement $\Tc'$ is a tree decomposition of $G$ with $\adhsize(\Tc') \le \max(\adhsize(\Tc), k-1)$.
\end{lemma}
\begin{proof}
The fact that $\Tc'$ is a tree decomposition of $G$ follows from the facts that $\Tc^1$ is a tree decomposition of $G[C_1]$, $\Tc^2$ is a tree decomposition of $G[C_2]$, the root bags of $\Tc^1$ and $\Tc^2$ contain $S = C_1 \cap C_2$, and $(C_1,C_2)$ is a separation of $G$.

The adhesion size of $\Tc'$ is at most 
\[\max(\adhsize(\Tc^1), \adhsize(\Tc^2), |S|) \le \max(\adhsize(\Tc), k-1),\]
because the only adhesions of $\Tc'$ that are not adhesions of either $\Tc^1$ or $\Tc^2$ are subsets of $S$, and we have that $|S| < k$ because $(C_1,C_2,t_1,t_2)$ is a non-$k$-lean-witness.
\end{proof}

For finally defining the refinement of $\Tc$ with $(C_1,C_2,t_1,t_2)$, we need one more lemma.

\begin{lemma}
\label{lem:leansimplifydecomp}
There is an algorithm that, given a tree decomposition $\Tc = (T,\bag)$ and a node $r \in V(T)$, in time $\|\Tc\|^{\OO(1)}$ returns a tree decomposition $\Tc' = (T',\bag')$ with
\begin{itemize}
\item $\adhsize(\Tc') \le \adhsize(\Tc)$
\end{itemize}
and so that
\begin{itemize}
\item there exists an injection $\phi \colon V(T') \to V(T)$ so that $\bag'(t) = \bag(\phi(t))$ for all $t \in V(T')$, and
\item there is exactly one node $p \in V(T')$ so that $\bag'(p) \subseteq \bag(r)$. For this node $p$ it holds that $\phi(p) = r$.
\end{itemize}
\end{lemma}
\begin{proof}
Consider $\Tc = (T,\bag)$ as rooted at $r$.
We will do a process where as long as there exists a node $t \in V(T)$ with a parent $p \in V(T)$ so that $\bag(t) \subseteq \bag(p)$, we contract the edge $tp$, i.e., remove $t$ and connect all children of $t$ as children of $p$.
Clearly, this process can be implemented in time $\|\Tc\|^{\OO(1)}$.
It also does not increase the adhesion size, and the required injection $\phi$ can be maintained throughout the process.

To prove there are no other nodes $t$ in the resulting tree decomposition $\Tc' = (T',\bag')$ with $\bag'(t) \subseteq \bag(r)$ than the root-node $r$, we observe that for the resulting tree decomposition it holds that every non-root node is the forget-node of at least one vertex, which therefore cannot be in the root-bag.
\end{proof}

Now we define the \emph{refinement} $\Tc''$ of $\Tc$ with $(C_1,C_2,t_1,t_2)$ to be the result of applying the algorithm of \Cref{lem:leansimplifydecomp} to the pre-refinement $\Tc'$ with the node $r$.
Note that $\adhsize(\Tc'') \le \adhsize(\Tc') \le \max(\adhsize(\Tc), k-1)$.

We will then argue that the process of iteratively replacing $\Tc$ by a refinement of $\Tc$ terminates in $2^{\OO(k)} \cdot \|\Tc\|$ steps.
For this, the following is the key lemma.

\begin{lemma}
\label{lem:bdanalysis}
For every $t \in V(T)$ and $i \in [2]$, we have that $|\bag^i(t)| \le |\bag(t)|$, and moreover, the equality $|\bag^i(t)| = |\bag(t)|$ holds only if $\bag^i(t) = \bag(t) \subseteq C_i$.
\end{lemma}
\begin{proof}
We prove this for $i=1$, but by symmetry this implies it also for $i=2$.
Recall that
\[\bag^1(t) = (\bag(t) \cap C_1) \cup \pull^1(t),\]
so it suffices to prove that either $\pull^1(t)$ is empty or $|\pull^1(t)| < |\bag(t) \setminus C_1|$.
For the sake of contradiction, suppose that $\pull^1(t)$ is non-empty and $|\pull^1(t)| \ge |\bag(t) \setminus C_1|$.

Denote $X_1 = C_1 \cap \bag(t_1)$ and $X_2 = C_2 \cap \bag(t_2)$.
We view $\Tc$ as rooted at $t_2$.

\begin{claim}
\label{lem:bdanalysis:claim1}
The set $S' = (S \setminus \pull^1(t)) \cup (\bag(t) \setminus C_1)$ is an $(X_1,X_2)$-separator and a $(\bag(t), X_2)$-separator.
\end{claim}
\begin{claimproof}
We prove that $S'$ is a $(C_1, X_2)$-separator, which implies the claim because $X_1 \subseteq C_1$ and $\bag(t) \subseteq C_1 \cup S'$.

Assume not, and let $P$ be a shortest $(C_1 \setminus S' ,X_2 \setminus S')$-path in $G \setminus S'$.
Because $S$ is an $(C_1,X_2)$-separator in $G$, we have that $\pull^1(t)$ is an $(C_1 \setminus S', X_2 \setminus S')$-separator in $G \setminus S'$.
Therefore, we can assume that the first vertex of $P$ is in $\pull^1(t)$, and all other vertices in $C_2 \setminus C_1$.
Because $P$ is a path that intersects a descendant of $t$ and the root $t_2$, it must intersect $\bag(t)$.
However, $\bag(t)$ is disjoint with $\pull^1(t)$ and $\bag(t) \setminus C_1 \subseteq S'$, so no such path $P$ exists in $G \setminus S'$.
\end{claimproof}

Now we have that if $|\pull^1(t)| > |\bag(t) \setminus C_1|$, then $|S'| < |S|$, contradicting that $S$ is a minimum-size $(X_1,X_2)$-separator.
Otherwise, $|S'| = |S|$ and $\pull^1(t)$ is non-empty.
We then consider two cases.

\paragraph{Case 1.}
First suppose that $t \neq t_1$ and $t$ is on the $(t_1,t_2)$-path in $T$.
Let $(C_1',C_2')$ be a separation with $C_1' \cap C_2' = S'$, $\bag(t) \subseteq C_1'$, and $X_2 \subseteq C_2'$, which exists by \Cref{lem:bdanalysis:claim1}.
We claim that $(C_1',C_2',t,t_2)$ is a non-$k$-lean-witness that contradicts the minimality of $(C_1,C_2,t_1,t_2)$.
As the distance between $t$ and $t_2$ is smaller than the distance between $t_1$ and $t_2$, it suffices to prove that $(C_1',C_2',t,t_2)$ is a non-$k$-lean-witness.

Because $X_2 \subseteq C_2'$ we have that $|C_2' \cap \bag(t_2)| > |S'|$, and because $t$ is on the $(t_1,t_2)$-path, we have that $|\bag(t)| > |S| = |S'|$, implying that $|C_1' \cap \bag(t)| > |S'|$.
Furthermore, as the $(t,t_2)$-path is a subpath of the $(t_1,t_2)$-path, all $(t,t_2)$-adhesions are larger than $|S'| = |S|$.

\paragraph{Case 2.}
The second case is that $t$ is not on the $(t_1,t_2)$-path or $t = t_1$.
We claim that $\sum_{v \in S'} d_{t_1,t_2}(v) < \sum_{v \in S} d_{t_1,t_2}(v)$, which would contradict the fact that $(C_1,C_2,t_1,t_2)$ is a minimal non-$k$-lean-witness.
We observe that the distance between a strict descendant of $t$ and the $(t_1,t_2)$-path is strictly larger than the distance between $t$ and the $(t_1,t_2)$-path.
Therefore, for all $v \in \pull^1(t)$ and $u \in \bag(t)$, we have that $d_{t_1,t_2}(v) > d_{t_1,t_2}(u)$, which by the fact that $\pull^1(t)$ is non-empty implies that $\sum_{v \in S'} d_{t_1,t_2}(v) < \sum_{v \in S} d_{t_1,t_2}(v)$.
\end{proof}

Now we consider the following potential function on tree decompositions.
Let $\Tc = (T,\bag)$ be a tree decomposition.
We define that the potential $\Phi_{\Tc}(t)$ of a node $t \in V(T)$ is
\[
\Phi_{\Tc}(t) =
\begin{cases}
4^{4k} \cdot (|\bag(t)| - 3k) & \text{ if } |\bag(t)| > 3k\\
4^{|\bag(t)|} & \text{ otherwise,}
\end{cases}
\]
and the potential $\Phi(\Tc)$ of $\Tc$ is
\[\Phi(\Tc) = \sum_{t \in V(T)} \Phi_{\Tc}(t).\]

Our goal is to show that the potential of the refinement of $\Tc$ is smaller than the potential of $\Tc$.
For this, we start with the following.

\begin{lemma}
\label{lem:smallleanfirstpotanalysis}
If both $C_1 \setminus C_2$ and $C_2 \setminus C_1$ intersect $\bag(t)$, then $\Phi_{\Tc^1}(t) + \Phi_{\Tc^2}(t) \le \Phi_{\Tc}(t) - 4^{\min(|\bag(t)|, k)}/2$.
\end{lemma}
\begin{proof}
By \Cref{lem:bdanalysis}, we have that in this case $|\bag^i(t)| < |\bag(t)|$ for both $i \in [2]$.
Suppose first that $|\bag(t)| \le 3k$.
In this case,
\[\Phi_{\Tc^1}(t) + \Phi_{\Tc^2}(t) \le 4^{|\bag^1(t)|} + 4^{|\bag^2(t)|} \le 2 \cdot 4^{|\bag(t)|-1} \le \Phi_{\Tc}(t)-4^{|\bag(t)|}/2.\]
Then it remains to consider the case of $|\bag(t)| > 3k$.
First, if $|\bag^i(t)| \le 3k$ for both $i \in [2]$, we have that
\[\Phi_{\Tc^1}(t) + \Phi_{\Tc^2}(t) \le 2 \cdot 4^{3k} \le 4^{4k}-4^k/2 \le \Phi_{\Tc}(t)-4^k/2.\]
Then, if $|\bag^i(t)| \le 3k$ for one $i \in [2]$, we have that
\[\Phi_{\Tc^1}(t) + \Phi_{\Tc^2}(t) \le 4^{3k} + \Phi_{\Tc}(t)-4^{4k} \le \Phi_{\Tc}(t)-4^k/2.\]
Finally, it remains to consider the case when $|\bag^i(t)| > 3k$ for both $i \in [2]$.
In this case, a crucial observation is that $|\bag^1(t)|+|\bag^2(t)| \le |\bag(t)|+2k$, because $\bag^1(t) \cup \bag^2(t) \subseteq \bag(t) \cup S$, $\bag^1(t) \cap \bag^2(t) \subseteq S$, and $|S| < k$.
Using this, we can prove that
\begin{align*}
\Phi_{\Tc^1}(t) + \Phi_{\Tc^2}(t) &\le 4^{4k} \cdot (|\bag^1(t)|-3k) + 4^{4k} \cdot (|\bag^2(t)|-3k)\\
&\le 4^{4k} \cdot (|\bag^1(t)|+|\bag^2(t)| - 6k)\\
&\le 4^{4k} \cdot (|\bag(t)|+2k-6k) \le 4^{4k} \cdot (|\bag(t)|-4k) \le \Phi_{\Tc}(t)-4^k/2.
\end{align*}
\end{proof}

Now, the final prerequisite for showing that the potential of a refinement of $\Tc$ is smaller than the potential of $\Tc$ is to show that there exists a node $t \in V(T)$ so that $|\bag(t)| > |S|$ and $\bag(t)$ intersects both $C_1 \setminus C_2$ and $C_2 \setminus C_1$.

\begin{lemma}
\label{lem:smallleansplitnodeexists}
There exists $t \in V(T)$ so that $|\bag(t)| > |S|$ and $\bag(t)$ intersects both $C_1 \setminus C_2$ and $C_2 \setminus C_1$.
\end{lemma}
\begin{proof}
Consider the $(t_1,t_2)$-path in $T$.
Because $(C_1,C_2,t_1,t_2)$ is a non-$k$-lean-witness and $S = C_1 \cap C_2$, all bags of nodes on this path are larger than $|S|$, as otherwise there would be an $(t_1,t_2)$-adhesion of size $\le |S|$.
We claim that there is a node $t$ on this path so that $\bag(t)$ intersects both $C_1 \setminus C_2$ and $C_2 \setminus C_1$.

Suppose not, implying that each bag on this path is either a subset of $C_1$ or a subset of $C_2$.
Note that no bag is a subset of both, because then the bag would be a subset of $S$.
Because $(C_1,C_2,t_1,t_2)$ is a non-$k$-lean-witness, $\bag(t_1)$ must intersect $C_1 \setminus C_2$ and $\bag(t_2)$ must intersect $C_2 \setminus C_1$.
Therefore, there are two consecutive nodes $x,y$ on the $(t_1,t_2)$-path so that $\bag(x) \subseteq C_1$ and $\bag(y) \subseteq C_2$.
However, then $\adh(xy) = \bag(x) \cap \bag(y) \subseteq C_1 \cap C_2$, contradicting that $(C_1,C_2,t_1,t_2)$ is a non-$k$-lean-witness.
\end{proof}

Now we can prove that the potential decreases in the refinement.

\begin{lemma}
\label{lem:smallleanpotanal}
Let $\Tc$ be a tree decomposition, $(C_1,C_2,t_1,t_2)$ a minimal non-$k$-lean-witness of $\Tc$, and $\Tc''$ the refinement of $\Tc$ with $(C_1,C_2,t_1,t_2)$.
Then, $\Phi(\Tc'') \le \Phi(\Tc) - 1$.
\end{lemma}
\begin{proof}
In the pre-refinement $\Tc' = (T',\bag')$, there are for each node $t \in V(T)$ two nodes $t^1, t^2 \in V(T')$ with bags $\bag'(t^1) = \bag^1(t)$ and $\bag'(t^2) = \bag^2(t)$.
Additionally, there is a node $r$ with the bag $\bag'(r) = S$.

By \Cref{lem:leansimplifydecomp}, the collection of the bags of the refinement $\Tc''$ is a subcollection of the bags of the pre-refinement $\Tc'$, obtained by removing all bags that are subsets of $S$, except the bag $\bag'(r) = S$.
Therefore, we have that

\begin{equation}
\label{lem:smallleanpotanal:eq1}
\Phi(\Tc'') = 4^{|S|} + \sum_{\{t \in V(T) : \bag^1(t) \not\subseteq S\}} \Phi_{\Tc^1}(t) + \sum_{\{t \in V(T) : \bag^2(t) \not\subseteq S\}} \Phi_{\Tc^2}(t).
\end{equation}

We say a node $t \in V(T)$ is \emph{split} by $(C_1,C_2)$ if $\bag(t)$ intersects both $C_1 \setminus C_2$ and $C_2 \setminus C_1$.
We observe that if a node is not split by $(C_1,C_2)$, then there is $i \in [2]$ so that
\begin{itemize}
\item $\bag(t) \subseteq C_i$, 
\item $\bag^i(t) = \bag(t)$ (by \Cref{lem:bdanalysis}), and
\item $\bag^j(t) \subseteq S$ for $j = 3-i$.
\end{itemize}
Let $R \subseteq V(T)$ denote the set of all nodes of $T$ split by $(C_1,C_2)$.
By using \Cref{lem:smallleanpotanal:eq1}, we obtain that
\begin{align*}
\Phi(\Tc'') \le& 4^{|S|} + \sum_{t \in R} \left(\Phi_{\Tc^1}(t) + \Phi_{\Tc^2}(t)\right) + \sum_{t \in V(T) \setminus R} \Phi_{\Tc}(t)\\
\le& \Phi(\Tc) + 4^{|S|} + \sum_{t \in R} \left(\Phi_{\Tc^1}(t) + \Phi_{\Tc^2}(t) - \Phi_{\Tc}(t)\right).
\intertext{By applying \Cref{lem:smallleanfirstpotanalysis}, we get then that}
\Phi(\Tc'') \le& \Phi(\Tc) + 4^{|S|} - \sum_{t \in R} 4^{\min(|\bag(t)|,k)}/2,
\intertext{which by \Cref{lem:smallleansplitnodeexists} implies that}
\Phi(\Tc'') \le& \Phi(\Tc) + 4^{|S|} - 4^{\min(|S|+1, k)}/2\\
\le& \Phi(\Tc) - 1.
\end{align*}

\end{proof}

Now we can put these results together to complete the proof of \Cref{lem:leanalgsmall}, which we re-state.

\leanalgsmall*
\begin{proof}
The idea is to iteratively apply the refinement operation.
We maintain a tree decomposition $\Tc$ with $\adhsize(\Tc) < k$, and as long as $\Tc$ is not $k$-lean, apply the refinement operation to decrease its potential.

We start with $\Tc$ being a tree decomposition with one bag containing all vertices.
We have that $\Phi(\Tc) = 2^{\OO(k)} \cdot |V(G)|$.
Then, we repeat the following process.

We first use the algorithm of \Cref{lem:smallleanfindminiwitness} to in time $s^{\OO(k)} \cdot \|G\|^{\OO(1)} \cdot \|\Tc\|^{\OO(1)}$ either conclude that $\Tc$ is $k$-lean or return a minimal non-$k$-lean-witness $(C_1,C_2,t_1,t_2)$ for $\Tc$.
If it concludes that $\Tc$ is $k$-lean, we are done.
Otherwise, we use $(C_1,C_2,t_1,t_2)$ to transform $\Tc$ into the refinement of $\Tc$ with $(C_1,C_2,t_1,t_2)$.
It is not hard to see that this transformation can be implemented in time $\|\Tc\|^{\OO(1)} \cdot \|G\|^{\OO(1)}$.
Therefore, one iteration of the process takes time $s^{\OO(k)} \cdot \|\Tc\|^{\OO(1)} \cdot \|G\|^{\OO(1)}$.

As the refinement decreases the potential of $\Tc$ by at least one (\Cref{lem:smallleanpotanal}), there are at most $2^{\OO(k)} \cdot |V(G)|$ iterations.
Furthermore, we can bound $\|\Tc\| = \OO(\Phi(\Tc))$, and therefore as $\Phi(\Tc)$ is always bounded by $2^{\OO(k)} \cdot |V(G)|$, each iteration runs in time $s^{\OO(k)} \cdot \|G\|^{\OO(1)}$, yielding the overall running time of $s^{\OO(k)} \cdot \|G\|^{\OO(1)}$.
\end{proof}

\bibliographystyle{alpha}
\bibliography{book_kernels_fvf}

\end{document}